\documentclass[pra,aps,floatfix,amsmath,superscriptaddress,twocolumn,longbibliography,nofootinbib]{revtex4-2}%\documentclass[prx,aps,floatfix,amsmath,superscriptaddress,twocolumn]{revtex4}
\usepackage{amssymb,enumerate}

\usepackage{amssymb,enumerate}
\usepackage{graphicx}
\usepackage{graphics}
\usepackage{amsmath}
\usepackage{amsthm,bbm}
\usepackage{color}
\usepackage{dsfont}
\usepackage{hyperref}

\usepackage[capitalise,nameinlink]{cleveref}

\usepackage{qcircuit}

\usepackage{graphicx}

\usepackage{xcolor}
\hypersetup{
    colorlinks,
    linkcolor={red!50!black},
    citecolor={blue!50!black},
    urlcolor={blue!80!black}
}

\usepackage{xfrac}

\usepackage{amsmath}
\usepackage{tikz-cd}
\usepackage{empheq}

\usepackage{enumitem}

\usepackage{graphicx}% http://ctan.org/pkg/graphicx

\usepackage{hyperref}
\usepackage[capitalise]{cleveref}
\usepackage{youngtab}

\usepackage{tikz}
\usepackage{mathtools}
\usepackage{amsmath}
\usepackage{amsthm,bbm}

\newcommand{\Tr}{\operatorname{Tr}}

\newcommand{\<}{\langle}
\renewcommand{\>}{\rangle}

\newtheorem{thm}{Theorem}

\usepackage{lipsum}
\usepackage{lmodern}
\usepackage{tcolorbox}

\usepackage{graphicx,graphics,epsfig,times,bm,bbm,amssymb,amsmath,amsfonts,mathrsfs}
\usepackage[normalem]{ulem}
\usepackage{setspace}
\usepackage{subfigure}
\usepackage{braket}
\usepackage{upgreek }
\usepackage{tikz}
\usepackage{natbib}
\usepackage{chngcntr}

\newtheorem{theorem}{Theorem}

\newtheorem{corollary}[theorem]{Corollary}

\newtheorem{lemma}[theorem]{Lemma}

\newtheorem{proposition}[theorem]{Proposition}

\makeatletter 
\renewcommand\onecolumngrid{% 
  \do@columngrid{one}{\@ne}%
  \def\set@footnotewidth{\onecolumngrid}%
  \def\footnoterule{\kern-6pt\hrule width 1.5in\kern6pt}%
}
\makeatother

\newcommand{\bes} {\begin{subequations}}
\newcommand{\ees} {\end{subequations}}
\newcommand{\bea} {\begin{eqnarray}}
\newcommand{\eea} {\end{eqnarray}}
\newcommand{\be} {\begin{equation}}
\newcommand{\ee} {\end{equation}}

\def\>{\rangle}
\def\<{\langle}
\def\Tr{\textrm{Tr}}

\newcommand{\ketbra}[2]{|{#1}\>\!\<#2|}

\newcommand{\ignore}[1]{}

\crefformat{step}{#2Step #1#3}

\begin{document}

\title{Optimal 
 Distillation of Coherent States with Phase-Insensitive Operations}

\author{Shiv Akshar Yadavalli}\affiliation{Department of Physics, Duke University, Durham, NC 27708, USA}
\author{Iman Marvian}
\affiliation{Department of Physics, Duke University, Durham, NC 27708, USA}
\affiliation{Department of Electrical and Computer Engineering, Duke University, Durham, NC 27708, USA}
\affiliation{Duke Quantum Center, Durham, NC 27708, USA}

\begin{abstract}
By combining multiple copies of noisy coherent states of light (or other bosonic systems), it is possible to obtain a single mode in a state with  lesser noise, a process known as distillation or purification of coherent states. We investigate the distillation of coherent states from coherent thermal states under general phase-insensitive operations, and find a distillation protocol that is optimal in the asymptotic regime, i.e., when the number of input copies is much greater than 1. Remarkably, we find that in this regime, the error -- as quantified by infidelity (one minus the fidelity) of the output state with the desired coherent state -- is proportional to the inverse of the purity of coherence of the input state, a quantity obtained from the Right-Logarithmic-Derivative (RLD) Fisher information metric, hence revealing an operational interpretation of this quantity. The heart of this protocol is a phase-insensitive channel that optimally converts an input coherent thermal state with high amplitude, into an output with significantly lower amplitude and temperature. Under this channel, the purity of coherence remains asymptotically conserved. While both the input and desired output are Gaussian states, we find that the optimal protocol cannot be a Gaussian channel. Among Gaussian phase-insensitive channels, the optimal distillation protocol is a simple linear optical scheme that can be implemented with beam splitters.
\end{abstract}
\maketitle

\section{Introduction}

Coherent states are central objects in the theory of quantum optics \cite{ECG_coherent1, GlauberOpticalCoherence, Glauber63, ECG_coherent2}.  \textit{Pure} coherent states minimize the Heisenberg uncertainty in continuous variable systems \cite{schrodinger1926stetige}. Furthermore, they model the `classical-like' behavior of electromagnetic fields produced by coherent sources such as lasers \cite{Mandel_Wolf_1995}. Such properties make pure coherent states useful in varied applications ranging from quantum communication and cryptography, to metrology and sensing \cite{khan2017quantum, Caves94,app3,app4,app5,app6,app7,app8,app9, Lutkenhaus2014, pirandola2018advances, TestingQuantumnessLami}. Thus, generating pure coherent states of 
 light or, more generally, bosonic modes, is a central task in quantum physics.

In practice, however, because of various fundamental or practical limitations, a quantum system will never be in a pure state. Furthermore, even if one can prepare a system in a state that is very close to a desired pure state, {it gradually becomes more mixed because of different instrumental imperfections and noise processes.} 

For example, the transmission of a pure coherent state through a communication channel adds noise to the state, which {can often be} modeled by a thermal attenuator channel \cite{Caves94, Holevo+2013, GaussianQI}. Such channels describe a broad range of physical processes involving energy loss to an environment at non-zero temperatures, such as the transmission of low-frequency electromagnetic signals through optical fiber or free space at non-zero temperatures \cite{rosati2018narrow, ThermalMicrowaveNetworks}. Then, transmitting a pure coherent state through a thermal attenuator channel will result in a
  \textit{coherent thermal state} \cite{bishop1987coherent, CTS1, GlauberMatrixElem, MarianMatrixElem}, also known as a displaced thermal state. 
    As the latter name suggests, coherent thermal states also arise when one displaces a bosonic mode initially in a thermal state. This is a standard method for preparing coherent states, and results in a perfect coherent state when the initial state is at zero temperature (i.e., the vacuum state) \cite{CTS1,CTS2,CTS3}.

Given such limitations, in this paper, we ask: {how well can we distill pure coherent states from multiple copies of coherent thermal states?} More precisely, we study \textit{coherence distillation} processes \cite{marvian2020coherence} that combine multiple modes of coherent thermal states at a higher temperature to generate a \textit{single mode} {that is close to} a coherent thermal state at a lower temperature. Such lowering of temperature essentially {purifies} the coherent thermal states, where the zero temperature limit would correspond to the pure coherent state (See Fig.\ref{fig:distill}). {As we further explain in Sec.~\ref{section: setup} and Sec.~\ref{Sec:discussion}, } this study can  be interpreted as single-copy resource distillation in the context of the resource theory of asymmetry \cite{gour2008resource, Marvian_thesis, marvian2013theory, marvian2014extending, marvianAsymPureStates, PianiAsym} and (unspeakable) coherence \cite{marvian2016quantify}, 
  where the set of free operations are phase-insensitive operations.

\begin{figure}
    \centering
\includegraphics[width=0.47\textwidth]{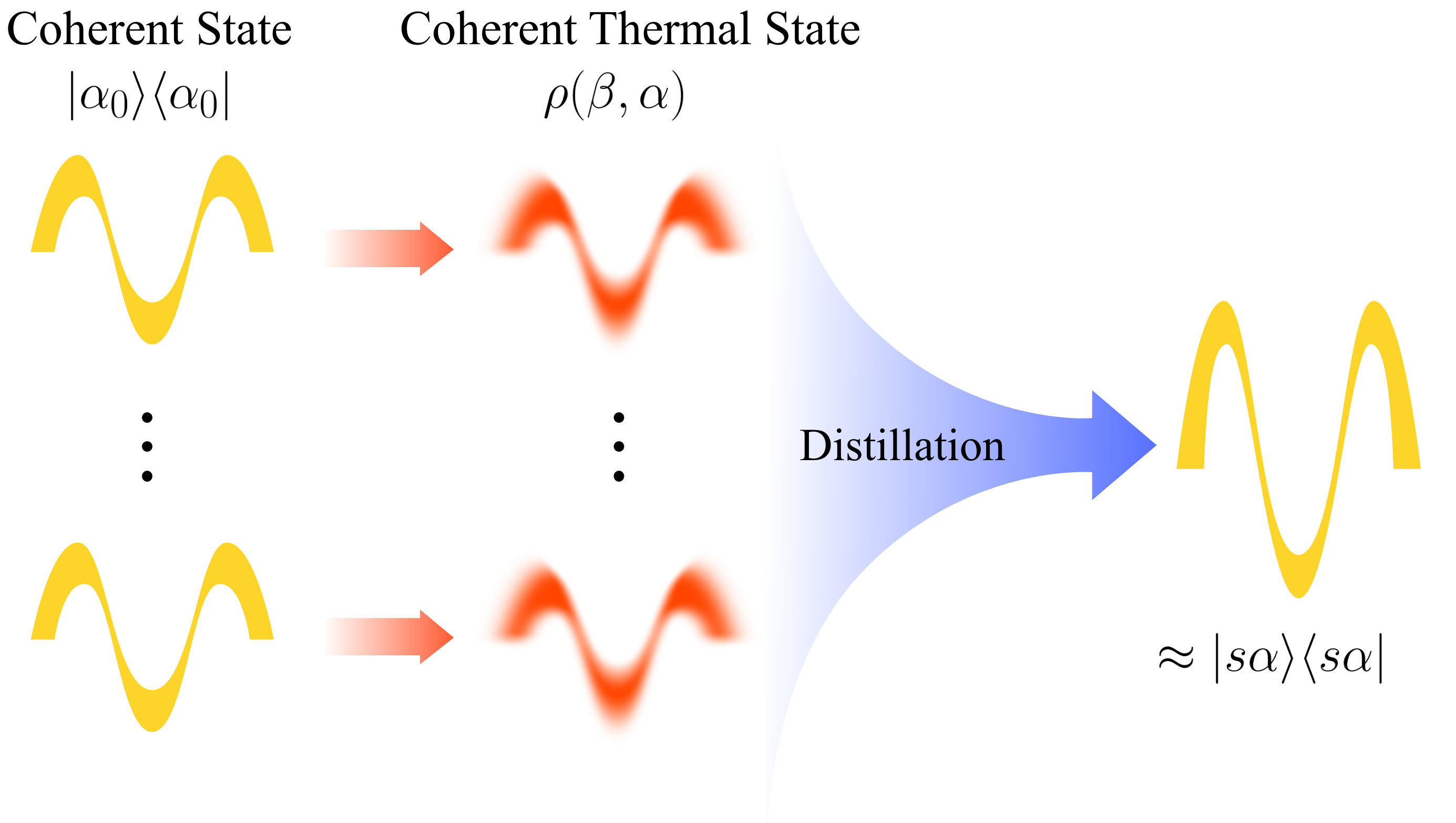}
\caption{
\textbf{Phase-insensitive distillation of coherent states -- }
{The goal is to convert $n\gg1$ copies of a \textit{coherent thermal state} $\rho(\beta,\alpha)=D(\alpha)\rho_{\text{th}}(\beta)D(\alpha)^\dag$,  where $D(\alpha)$ is the Weyl displacement operator, to a pure coherent state $|\alpha\rangle$, or more generally $|s\alpha\rangle$ for some constant $s\ll n$, using phase-insensitive operations.} The state $\rho(\beta,\alpha)$ describes, e.g., the output of a thermal attenuator channel given a pure coherent state $|\alpha_0\rangle$ as the input. The phase-insensitivity of the distillation operation ensures that the final output state is in-phase, i.e.,  synchronized with the input states  $\rho(\beta,\alpha)$. We then ask, what is the \textit{optimal distillation protocol} that achieves the highest fidelity with the desired state $|s \alpha\rangle$, and what is the value of this optimal fidelity?} 
    \label{fig:distill}
\end{figure}

\begin{figure}
  \centering
     \includegraphics[width=0.45\textwidth]{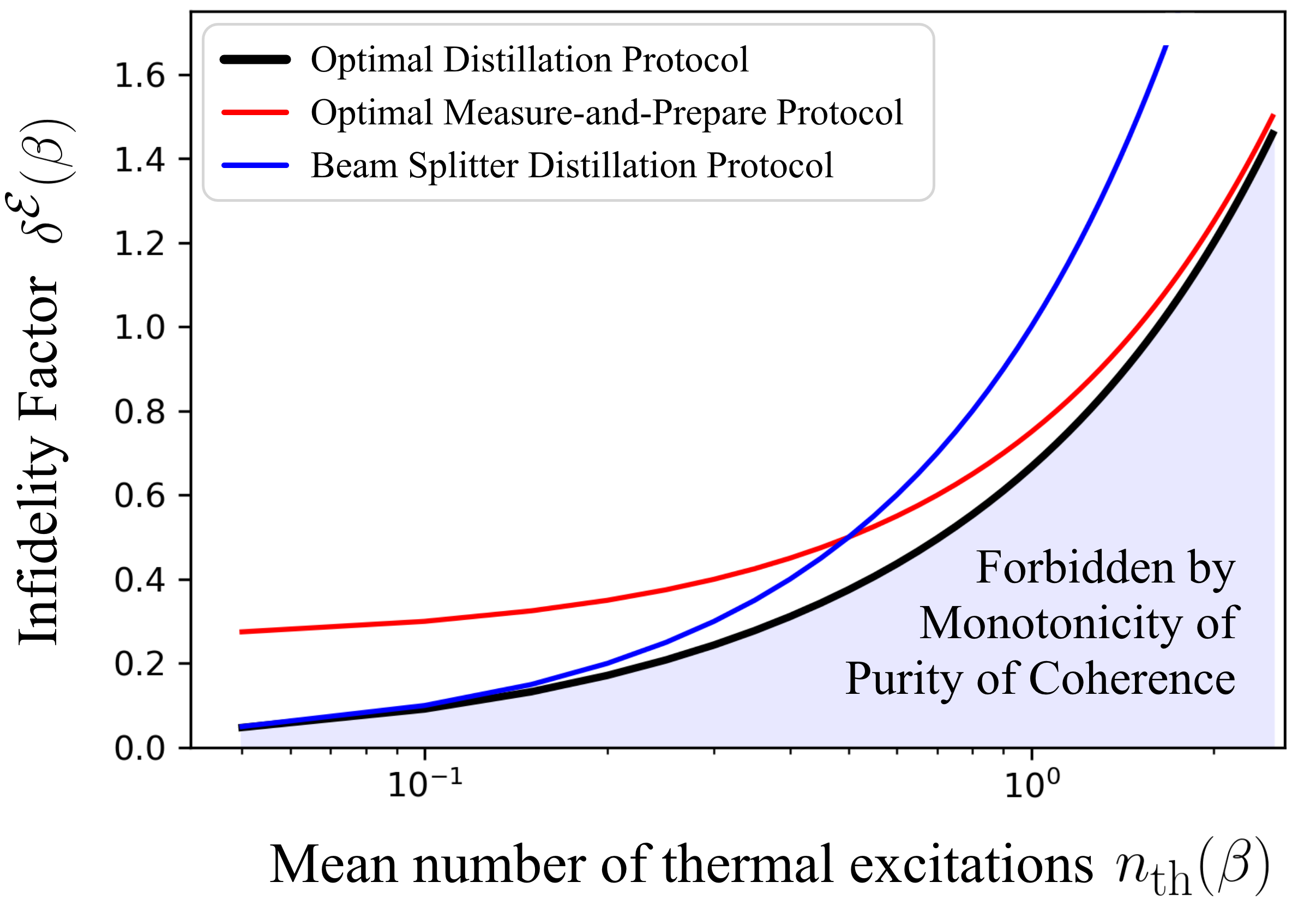}
\caption{\textbf{Performance of different distillation protocols -- }
This plot shows the \textit{infidelity factor} $\delta^{\mathcal{E}}(\beta, \alpha)$,  {defined in Eq.(\ref{eqn: intro delta}),  for different distillation protocols,} as a function of the mean number of thermal excitations $n_{\text{th}}(\beta)$ in the input coherent thermal states. $\delta^{\mathcal{E}}(\beta, \alpha)$ characterizes the asymptotic performance of a distillation protocol in the many-copy regime. 
Using the properties of the purity of coherence \cite{marvian2020coherence},  one can  show that $\delta^{\mathcal{E}}(\beta, \alpha)$ is lower bounded by  ${n_\text{th}(\beta)}/{2}+{n_\text{th}(\beta)}/{[2+4n_\text{th}(\beta)]}$. We prove that this can be achieved with a novel distillation protocol introduced in this paper. The red and blue curves correspond respectively to a simple protocol discussed in Fig.\ref{first real}, that is realizable with linear optical elements, and to a measure-and-prepare protocol. See below for further details. }
 \label{fig89}
% \vspace{-2mm}
\end{figure}

\subsection*{Summary of the Main Results}
In this paper, we investigate coherence distillation using  \textit{phase-insensitive} operations, i.e., operations that are covariant under phase shifts 
 (see Sec.~\ref{section: setup} for the formal definitions). Such operations preserve information about the phase {of the input}, and therefore appear naturally, for instance, in the context of electronic and optical amplifiers \cite{CavesLinearAmp, CavesAmplifier12, PhaseEstimation12}. {Equivalently, these are operations that respect the  time-translation symmetry generated by the intrinsic Hamiltonians of periodic systems with equal periods \cite{marvian2020coherence}. }

We study various protocols and find the ultimate limits of such phase-insensitive coherence distillation. {In the following, $\rho(\beta, \alpha)=D(\alpha)\rho_\text{th}(\beta)D^\dag(\alpha)$ denotes coherent thermal state at temperature $T=\beta^{-1}/k_\text{B} \ge 0$, where $k_\text{B}$ is the Boltzmann constant and $\alpha\in\mathbb{C}$ is the displacement.} Consider a sequence of phase-insensitive  channels $\mathcal{E}_n$ indexed by the number of coherent thermal states $\rho(\beta, \alpha)^{\otimes n}$ it processes. Then, {unless the protocol exhibits sub-optimal error scaling,} the fidelity of the output of these distillation channels with the pure coherent state $\ket{\alpha}$ takes the form 
\begin{align}
\braket{\alpha|\mathcal{E}_n(\rho(\beta, \alpha)^{\otimes n})|\alpha} = 1 - \frac{\delta^{\mathcal{E}}(\beta,\alpha)}{n} + {{o}}\Big(\frac{1}{n}\Big) \ ,
\end{align}
where  $\delta^{\mathcal{E}}\ge 0$, called the {\emph{infidelity factor}},  determines the {leading-order coefficient in $1/n$} of distillation error for the sequence of phase-insensitive distillation channels under consideration (or `protocol $\mathcal{E}$' for short)\footnote{Recall that $A_n={{o}}(\frac{1}{n})$ means that $\lim_{n\rightarrow\infty} n A_n=0$.}. Therefore, for any single-copy distillation protocol, we are interested in the quantity
\begin{align}\label{eqn: intro delta}
    \delta^\mathcal{E}(\beta, \alpha)= {\lim_{n\rightarrow \infty}} n \times \epsilon_n \ , 
\end{align}
where     $\epsilon_n= 1 - \braket{\alpha|\mathcal{E}_n(\rho(\beta, \alpha)^{\otimes n})|\alpha}$ is the infidelity. {Here, we have assumed that for the distillation protocol under consideration, the above limit exists (this is the case for all the protocols considered in this work). If not, we  characterize the asymptotic behavior of the protocol by considering both the $\liminf_{n\rightarrow\infty} n\times \epsilon_n$ and $\limsup_{n\rightarrow\infty} n\times \epsilon_n$.} 
Thus, the optimal protocol's performance is determined by 
\begin{align}\label{eqn: intro optimal delta}
    \delta^{\text{opt}}(\beta, \alpha) = \inf_{\mathcal{E}} \delta^{\mathcal{E}}(\beta, \alpha)=\inf_{\mathcal{E}} {\lim_{n\rightarrow \infty}} n \times \epsilon_n \ ,
\end{align}
where the minimization is over all phase-insensitive protocols. 

{It is worth noting that, in general, the   distillation protocol $\mathcal{E}_n$ 
 depends on the amplitude $|\alpha|$ of the desired coherent state $|\alpha\rangle$. That is, we assume the amplitude of the desired output is known, whereas its phase may or may not be known (it can be easily seen that under the restriction to phase-insensitive distillation protocols, the knowledge of this phase is not useful). Similarly, in principle,  the infidelity factor may also depend on the magnitude $|\alpha|$. Since the phase of $\alpha$ is irrelevant, we often assume $\alpha$ is real and positive.}\\

\noindent\textbf{Main result:} We construct the \textit{optimal} phase-insensitive coherence distillation protocol and show that the lowest achievable infidelity factor is determined by the ratio of the minimal and maximal Quantum Fisher Information metrics \cite{PETZ_199681, Bengtsson_Zyczkowski_2006, QFI_Metrics_2016} for the desired pure output and the actual input, namely  
\begin{align}\label{main}
 \delta^{\text{opt}}(\beta, \alpha) =\frac{F_H(|\alpha\rangle\langle\alpha|)}{4 \times P_H(\rho(\beta,\alpha))}=\frac{|\alpha|^2\omega^2}{ P_H(\rho(\beta,\alpha))}\ .
\end{align}
{Functions  $F_H$ and $P_H$ 
 are  called 
   the Quantum Fisher Information (QFI) and the purity of coherence, which are obtained from} the Symmetric-Logarithmic-Derivative (SLD) and the Right-Logarithmic-Derivative (RLD) Fisher information metrics, respectively \cite{PETZ_199681, Bengtsson_Zyczkowski_2006, QFI_Metrics_2016} (see Sec.~\ref{section: universal bounds proof} for definitions and further discussion on such metrics). In particular, $F_H(|\alpha\rangle\langle\alpha|)$ is the QFI of the pure coherent state $|\alpha\rangle$. {For pure states, QFI is four times the energy variance of the state,} i.e., 
\begin{align}
F_H(|\alpha\rangle\langle\alpha|)&=4[\langle\alpha|H^2|\alpha\rangle-\langle\alpha|H|\alpha\rangle^2] =4 \omega^2 |\alpha|^2\  ,
\end{align}
where $H= \omega a^\dag a$ is the Hamiltonian of a Harmonic oscillator with angular frequency $\omega$, and {the purity of coherence is }
\begin{align}
P_H(\rho)&=-\Tr([H,\rho]^2\rho^{-1})\\  &=
\omega^2 |\alpha|^2 \times \frac{2n_\text{th}(\beta)+1}{n_\text{th}(\beta) (n_\text{th}(\beta)+1)} \ ,
\end{align}
where $n_{\text{th}}(\beta)= ({e^{\beta\omega}}-1)^{-1}$ is the mean  number of thermal excitations in the input $\rho(\beta, \alpha)$.  Therefore, the optimal infidelity factor 
\begin{align}\label{opt bound}
    \delta^{\text{opt}}(\beta)=
\frac{n_{\text{th}}(\beta)}{2}+\frac{n_{\text{th}}(\beta)}{4n_{\text{th}}(\beta)+2}\ ,
\end{align}
is independent of  amplitude $\alpha$ (hence, we drop the dependence on $\alpha$). Fig.~\ref{fig89} compares 
 this function with the infidelity factor of two other distillation protocols discussed below.

{As we show in Sec.~\ref{section: universal bounds proof}, the fact that $\delta^{\text{opt}}(\beta)$ } is lower bounded by the right-hand side of Eq.(\ref{main}) follows from the general results of  \cite{marvian2020coherence}, which introduced the purity of coherence as a useful quantification of coherence in noisy systems, especially in the context of coherence distillation.\footnote{{The purity of coherence and the QFI are used in \cite{marvian2020coherence} to show that it is impossible to covariantly distill a pure state containing coherence from copies of a generic mixed state that contains coherence, at a non-zero rate with vanishing error. {Similarly, it is impossible to perfectly distill a pure state from a finite number of copies of a generic mixed state. But, as discussed in \cite{marvian2020coherence}, this no-go theorem does not forbid the covariant distillation of a sub-linear number of output copies or, for that matter, a \textit{single} copy of a pure state, in the limit of large number of inputs}. This is our focus in this paper.}} Hence our result that this bound is indeed achievable reveals an operational interpretation of the purity of coherence and RLD Fisher information in this context.  
This can be compared with the operational interpretation of SLD QFI in the context of coherence formation \cite{marvian2022operational}, where one converts pure coherent states to mixed states in the asymptotic independent and identically distributed (i.i.d.) regime. \\

\begin{figure}
    \centering
    \includegraphics[width=0.45\textwidth]{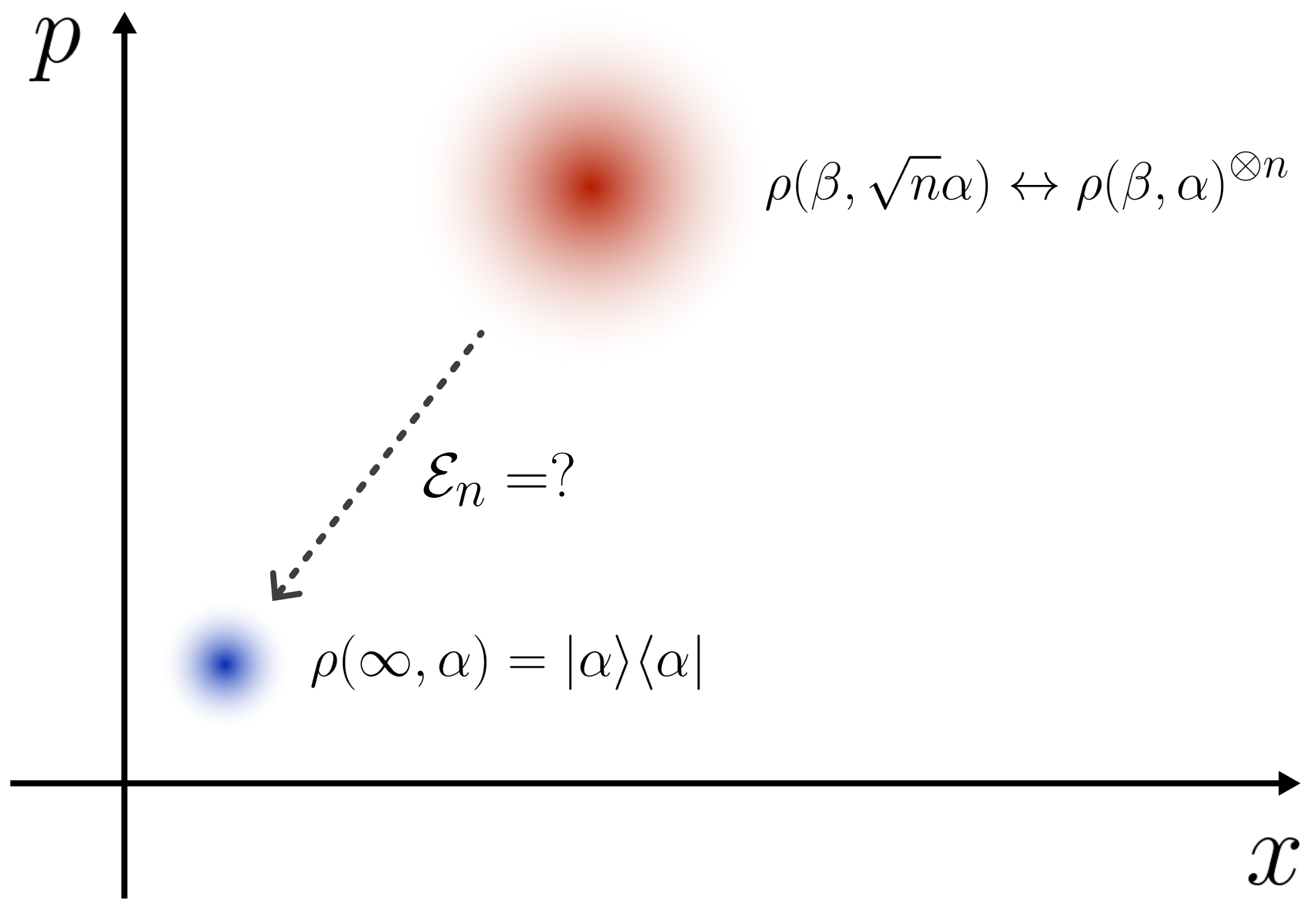}
\caption{\textbf{Coherence  distillation in phase space --}  $n$ copies of the coherent thermal state $\rho(\beta,\alpha)$ can be reversibly transformed to state $\rho(\beta,\sqrt{n}\alpha)$ via phase-insensitive channels. Hence, any coherence distillation protocol can be understood as a sequence of phase-insensitive channels $\mathcal{E}_n$ that transform the state $\rho(\beta, \sqrt{n}\alpha)$ to a state close to  $\rho(\infty,\alpha)= \ketbra{\alpha}{\alpha}$. 
Both the input and the desired output states are described by  Gaussian Wigner distributions,  
with  $x$ and $p$ variances $ \braket{\Delta x^2}_\text{in} = \braket{\Delta p^2}_\text{in} =n_\text{th}(\beta)+1/2$, and  $\braket{\Delta x^2}_\text{out} = \braket{\Delta p^2}_\text{out} =1/2$, respectively.   
The centers of the two Gaussian distributions have radii $r_\text{in}=\sqrt{n}|\alpha|$ and $r_\text{out}=|\alpha|$, respectively.  From a  classical perspective, one may expect that under the rescaling $x\rightarrow x/\sqrt{n}$ and $p\rightarrow p/\sqrt{n}$, which is a phase-insensitive map,  the input distribution can be transformed to the output, provided that $n$ is equal to or larger than the ratio of the variances, namely $2 n_\text{th}(\beta)+1$. That is,  using $n\ge 2 n_\text{th}(\beta)+1$ copies of the input $\rho(\beta,\alpha)$, one should be able to obtain an exact copy of the coherent state $|\alpha\rangle$. Indeed, this is exactly the same bound one obtains  by considering the ratio of QFI $F_H$ for the output and input states (see Sec.~\ref{section: universal bounds proof}). However, the above rescaling cannot be implemented as a physical process. Indeed, unless $\alpha=0$ or $\beta=\infty$, one needs an infinite number of copies of $\rho(\beta,\alpha)$ to obtain an exact copy of the pure coherent state $|\alpha\rangle$.
As shown in \cite{marvian2020coherence}, this can be established  using the properties of the purity of coherence $P_H$. Furthermore, as we show in this letter, in the regime $n \gg 1$, 
this quantity determines the minimum achievable error in the output state (see Eq.(\ref{opt bound})). It is also worth noting that the above rescaling can be realized with a beam splitter, with an order $\mathcal{O}(n^{-1})$ correction that comes from the vacuum noise in the other input mode (See Fig.~\ref{first real}).}
\label{fig:phase space}
\end{figure}

\noindent\textbf{Other {distillation protocols}: }Interestingly, 
even though the input and the desired output states are both Gaussian, we find that the optimal protocol that achieves the above performance is not a Gaussian process.  Indeed, we find that among  \textit{Gaussian} phase-insensitive protocols, the optimal protocol is a simple scheme that has been considered before in the context of purification \cite{BeamSplitterProtocol1, BeamSplitterProtocol2,  Chiribella}, which achieves the infidelity factor
\begin{align}
       {\delta^{\text{opt-Gauss}}}(\beta, \alpha) = n_{\text{th}}(\beta)\ .
    \end{align}
This protocol can be realized  {using only linear optical elements and a single ancilla mode at zero temperature (See Fig.~\ref{first real}).}

{Another important class of  distillation protocols studied in this paper are based on \textit{measure-and-prepare}  phase-insensitive operations. It turns out that such protocols also can not achieve} the optimal infidelity factor. {As we further discuss in Sec.~\ref{section: universal bounds proof}, the general results of \cite{marvian2020coherence} imply that the lowest achievable infidelity in phase-insensitive} measure-and-prepare channels is determined by the ratio of the {QFI of the desired output and input states, i.e.,} 
\begin{align}\label{opt mp intro}
 \delta^{\text{opt-MP}}(\beta, \alpha) &=\frac{F_H(|\alpha\rangle\langle\alpha|)}{4 \times  F_H(\rho(\beta,\alpha))}=
 \frac{n_{\text{th}}(\beta)}{2}+\frac{1}{4}\ .
\end{align} 
{The distillation protocol of \cite{marvian2020coherence} that achieves this optimal performance among measure-and-prepare phase-insensitive protocols is based on {a rather complicated estimation scheme that achieves the Cramer-Rao bound, via the maximum-likelihood estimator \cite{Caves94, Barndorff-Nielsen_2000}.}  However, as we explain in Appendix \ref{appendix: measure prepare}, when the input states are coherent thermal states, this measurement can be chosen to be one that uses  the single-mode 
\textit{canonical phase measurement} \cite{holevo2011probabilistic}, with POVM $\{M(\phi): \phi\in[0,2\pi)\}$ with matrix elements in the Fock basis given by  $M_{mn}(\phi)={e^{\mathrm{i}\phi (m-n)}d\phi}/{2\pi}$. }

{Note that in contrast to the optimal distillation protocol, the infidelity factor for measure-and-prepare distillation protocols does not vanish}, even in the limit of zero temperature input $n_\text{th}(\beta)=0$.  On the other hand, in the infinite temperature limit $n_\text{th}(\beta) \gg 1$, the performance of this protocol approaches the optimal protocol, i.e., it saturates the upper bound imposed by the purity of coherence. As we further explain in Sec.~\ref{section: universal bounds proof}, using the results of \cite{marvian2020coherence}, this is a consequence of the fact that in the high-temperature limit, the purity of coherence is approximately equal to QFI, i.e., $P_H(\rho)\approx F_H(\rho)$. {In Appendix \ref{appendix: measure prepare},} {we also study {a measure-and-prepare protocol} that employs Heterodyne measurements and show that it achieves the infidelity factor
\begin{align}
    \delta^{\text{Heterodyne}}(\beta, \alpha) = n_\text{th}(\beta)+1\ , 
\end{align}
{which is worse than the infidelity obtained via the canonical phase measurement in Eq.(\ref{opt mp intro}).}

{It is worth noting that, the infidelity factor of the optimal protocol, as well as all the other protocols discussed above, {are} independent of the amplitude $\alpha$, and {depends only on the} temperature $\beta$.}\\

\noindent{\textbf{Distillation combined with 
amplification/attenuation:}  The infidelity factor in Eq.(\ref{eqn: intro optimal delta}) characterizes the performance of a distillation protocol that converts $n$ copies of state $\rho(\beta,\alpha)$ to $|\alpha\rangle$, in the large $n$ regime. One can consider more general distillation protocols that convert the given copies of state $\rho(\beta,\alpha)$ to state $|s\alpha\rangle$ for a general complex number $s$. This corresponds to amplification when $|s|>1$, or attenuation when $|s|<1$, of the coherent state $|\alpha\rangle$.

As we further explain in Sec. \ref{optimal amplify/attenuate subsection}, any such protocol $\mathcal{A}_n$ can be thought of as a process that first combines the given $n$ copies of state $\rho(\beta,\alpha)$ to obtain $n'\approx n/|s|^2 $ copies of $\rho(\beta,\alpha')$ where $s=\alpha'/\alpha$, and then uses them to distill a copy of the coherent state $|\alpha'\rangle$, via a distillation protocol $\mathcal{E}_{n'}$, such that
$$\mathcal{A}_{n}((\rho(\beta,\alpha))^{\otimes n})=\mathcal{E}_{n'}((\rho(\beta,\alpha'))^{\otimes n'})\approx |\alpha'\rangle\langle \alpha'|\ .$$ 
{Then, in this case,  for the optimal protocol, the infidelity $\epsilon_n=1-\langle \alpha'|\mathcal{A}_n(\rho(\beta,\alpha)^{\otimes n })|\alpha'\rangle$ 
times $n$ converges to 
\begin{align}
\lim_{n\rightarrow \infty} n \times \epsilon_n = \frac{|\alpha'|^2}{|\alpha|^2}  \times \delta^{\mathcal{E}}(\beta,\alpha)=\frac{1}{4}\frac{F_H(|\alpha'\rangle\langle \alpha'|)}{P_H(\rho(\beta,\alpha))}  \ ,
\end{align}
where $\delta^{\mathcal{E}}$ is the infidelity factor of the distillation process $\mathcal{E}$. This is again equal to the ratio of QFI of the desired output to the purity of coherence 
of each input copy (up to a factor of 4).}

\section{The setup}\label{section: setup}

\subsection{  Coherent Thermal  States}

Consider $n$ bosonic modes, e.g., optical modes with identical frequency. Or, equivalently, $n$ Harmonic Oscillators with the total Hamiltonian 
\be\label{eqn: total hamiltonian}
{H}_n=\omega \sum_{i=1}^n {a}^\dag_i{a}_i=\omega \sum_{i=1}^n \frac{x_i^2+p_i^2-1}{2} \ ,
\ee
where ${a}_i: i=1,\cdots, n$ index the annihilation operators of the $n$ modes satisfying the standard commutation relations $[{a}_i, {a}^\dag_j]=\delta_{i,j}$, and $[{a}_i, {a}_j]=0$, where $x_i={(a_i+a_i^\dagger)}/{\sqrt{2}}$, $p_i={(a_i-a_i^\dagger)}/{\sqrt{2} \mathrm{i}}$ and $\omega$ is the frequency (we assume $\hbar=1$).  

Suppose initially these modes are  prepared in the joint uncorrelated state ${\rho}(\beta,\alpha)^{\otimes n}$, 
where the state of each individual input mode is a coherent thermal state, also known as a displaced thermal state,  
\be\label{form}
{\rho}(\beta,\alpha)={D}(\alpha) {\rho}_\text{th}(\beta) {D}^\dag(\alpha)\ ,
\ee
where $\alpha\in\mathbb{C}$ and $\beta>0$. Here,
\be
{\rho}_\text{th}(\beta)=\frac{\exp({-\beta \omega {a}^\dag {a}})}{\Tr(\exp({-\beta \omega {a}^\dag {a}}))}\ ,
\ee
is the thermal (Gibbs) state at temperature $T=(k_\text{B}\beta)^{-1}$, where $k_\text{B}$ is the Boltzman's constant, 
and 
${D}(\alpha)=\exp(\alpha {a}^\dag-\alpha^\ast {a})$ is the Weyl displacement operator.  
The only pure states in this family are obtained in the zero temperature limit, i.e., $\beta\rightarrow \infty$, and are the coherent states 
\be
|\alpha\rangle=D(\alpha)|0\rangle=e^{-\frac{|\alpha|^2}{2}} \sum_{n=0}^\infty \frac{\alpha^n}{\sqrt{n!}} |n\rangle  \ \ \ \ : \alpha\in\mathbb{C} \ ,
\ee
where $\{|n\rangle\}$ denotes the Fock basis, i.e., the normalized eigenvectors of $a^\dag a$ with the corresponding integer eigenvalues $n\ge 0$, which satisfy $a |n\rangle=\sqrt{n}|n-1\rangle$ (equivalently, $|\alpha\rangle$ is defined by $a|\alpha\rangle=\alpha|\alpha\rangle$). The ground state of the Harmonic Oscillator, namely state $|0\rangle$, is often called the vacuum state. For general values of $\beta > 0$,  coherent thermal states describe noisy coherent states, and {as discussed in the introduction, they} naturally appear in different contexts in quantum optics and communication.

The expected energy of each input mode of $\rho(\beta,\alpha)$ is
\begin{align}\label{energy}
 \Tr(\omega{a}^\dag {a}\ {\rho}(\beta,\alpha))=\omega n_\text{th}(\beta)+\omega |\alpha|^2 \ ,
\end{align}
where 
\be
n_\text{th}(\beta)=\Tr(a^\dag a \rho_\beta)=\frac{1}{e^{\beta\omega}-1}\ ,
\ee
is the mean number of thermal excitations {(sometimes denoted as $n_\beta$).} Therefore, the term $\omega n_\text{th}(\beta)$ is the contribution of thermal excitations in the total energy, whereas  $\omega |\alpha|^2$ is the energy associated to the  coherence.

The distance between the coherent thermal state ${\rho}(\beta,\alpha)$ and the corresponding pure coherent state $|\alpha\rangle\langle\alpha|=\rho(\infty, \alpha)$ can be quantified by their infidelity, i.e., one minus their Uhlmann fidelity 
\be\label{eqn: infid formula}
\epsilon=1-\langle\alpha|{\rho}(\beta,\alpha)|\alpha\rangle
=e^{-\beta \omega}=\frac{n_\text{th}(\beta)}{n_\text{th}(\beta)+1}\ .
\ee

Therefore, in the low temperature regime, where $\beta \omega \gg 1$, we have $\epsilon \approx n_\text{th}(\beta)$, which means that infidelity is approximately equal to the mean number of thermal excitations.

\begin{figure*}
\centering\includegraphics[width=0.75\textwidth]{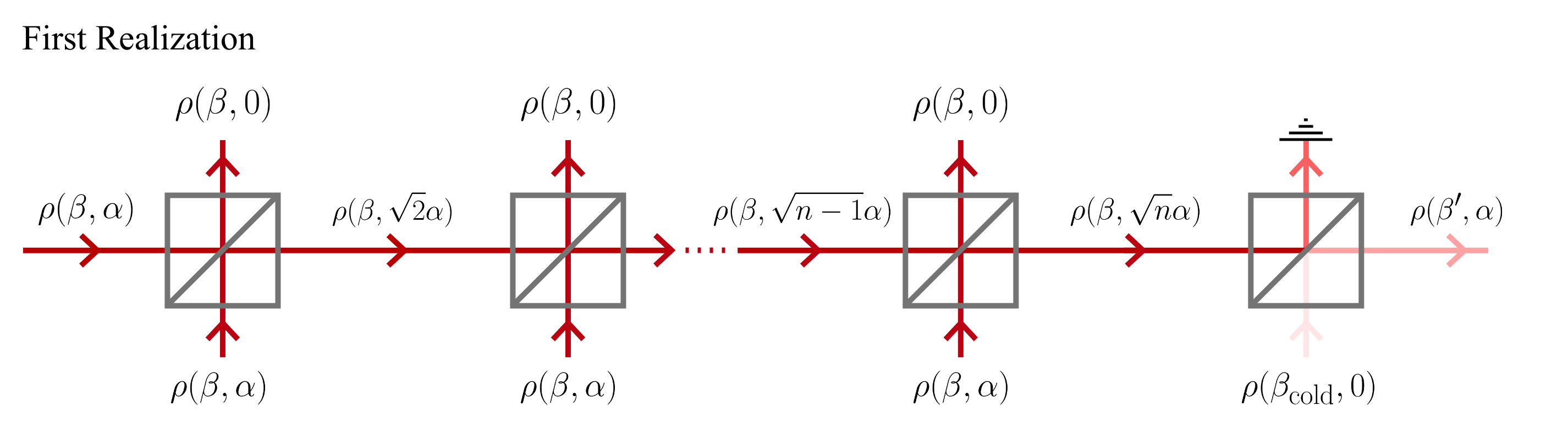}
\includegraphics[width=0.75\textwidth]{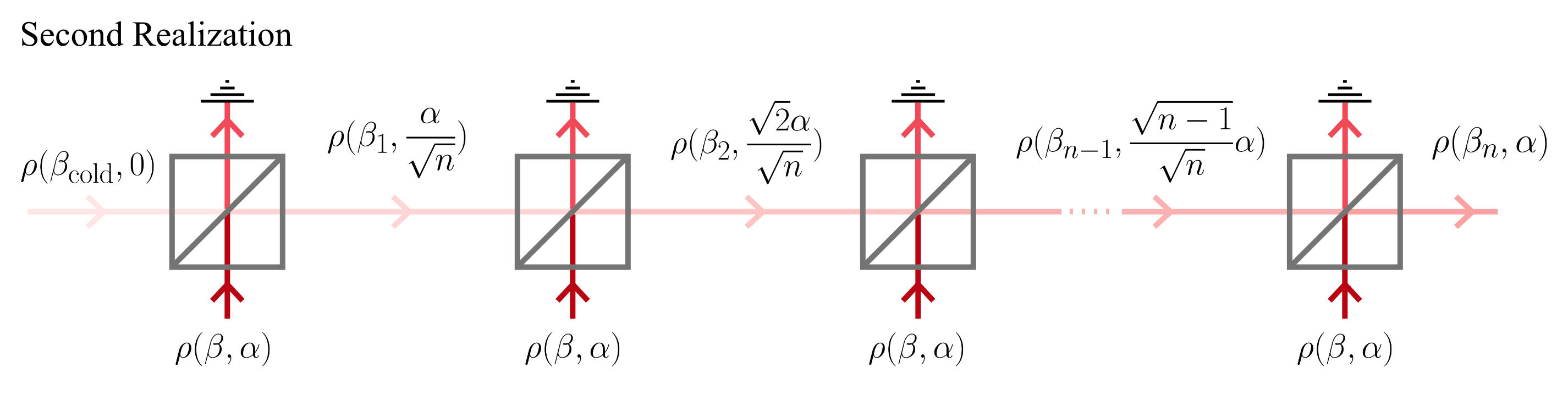}
\caption{\textbf{A suboptimal coherence distillation protocol using beam splitters --} {Two different realizations of a distillation protocol that} consumes $n$ copies of coherent thermal state $\rho(\beta, \alpha)$ to produce single mode $\rho(\beta_\text{out}, \alpha)$ in a lower temperature are presented (the change in temperature is indicated using \textit{darker} and  \textit{lighter} shades of red for hot and cold states, respectively). This protocol achieves the optimal performance among Gaussian distillation protocols (see Proposition \ref{prop: gaussian}). The transmission/reflection coefficients of each of the beam splitters are adjusted to attain the state transformations indicated in the figure as per Eq.(\ref{eq: bs alpha}). \textit{\textbf{First Realization}} -- In the first stage of this realization the coherence in $n$ input modes is \textit{concentrated} into a single mode $\rho(\beta,\alpha)^{\otimes n} \longrightarrow \rho(\beta,\sqrt{n}\alpha)$ 
    (we refer to this passive transformation as the \textit{concentration map}; see Sec.~\ref{subsec: concentration dilution reversibility} for further discussion). Then, to reduce the temperature we combine this state with an ancilla mode in a lower temperature $\beta_\text{cold}$ to realize $\rho(\beta,\sqrt{n}\alpha)\otimes \rho(\beta_\text{cold},0)\ \longrightarrow \ \rho(\beta',\alpha)\ ,$ where $n_\text{th}(\beta')$ is given by  Eq.(\ref{eq: bs n}). If the ancilla mode is {in the vacuum state}, (i.e., $n_\text{th}(\beta_{\text{cold}})=0$), then $n_\text{th}(\beta')=n_\text{th}(\beta)/n$. From Eq.(\ref{eqn: infid formula}), for $n\gg 1$, the infidelity with the desired state $|\alpha\rangle$ is $\epsilon_n={n_\text{th}(\beta)}/{n}+\mathcal{O}(n^{-2})$. \textit{\textbf{Second Realization}} -- In this realization, a cold thermal ancilla mode in the initial state $\rho_\text{th}(\beta_{\text{cold}})$ sequentially and weakly interacts with each copy of the state $\rho(\beta,\alpha)$ where $\beta_\text{cold}> \beta$ ({See Appendix \ref{appendix: gaussian prop} for further discussion}). {Each interaction  displaces the state of ancilla by $\approx \alpha/\sqrt{n}$.}  }
    \label{first real}
\end{figure*}

\subsection{Phase-insensitive operations}

{Starting from multiple copies of a coherent thermal state,} we are interested in distilling a single mode (or multi-mode) state that has higher fidelity with the desired pure coherent state $|\alpha\rangle$ using {phase-insensitive operations}. That is, we consider the most general time evolution described by a completely positive trace-preserving (CPTP) map $\mathcal{E}_n$ from density operators on $n$ bosonic input modes to density operators on a single mode (or more generally $m\ge 1$ output modes), satisfying the covariance condition  
\be\label{phase-insensitive condition}
\mathcal{E}_n(R_\phi^{\otimes n}(\cdot)R_{-\phi}^{\otimes n})=R^{\otimes m}_\phi\mathcal{E}_n(\cdot)R_{-\phi}^{\otimes m} \ \ \ \ \ : \phi\in [0, 2\pi)\ ,
\ee
where $R_\phi=\exp(i \phi a^\dag a)$ is the phase shift unitary acting on a single mode. Note that this condition can be equivalently stated as invariance under time translation symmetry, i.e.,
\be\label{cov3}
\mathcal{E}_n(e^{-\mathrm{i} t H_n}(\cdot)e^{\mathrm{i} t H_n})=e^{-\mathrm{i} t H_m}\mathcal{E}_n(\cdot)e^{\mathrm{i} t H_m}  \ \ \ \ \ : t\in [0, \frac{2\pi}{\omega})\ 
\ee
where $H_n$ is defined as in Eq.(\ref{eqn: total hamiltonian}). The fact that the operation $\mathcal{E}_n$ is phase-insensitive, in particular, implies that the transformation can be realized without knowing the phase of the displacement $\alpha$, which is crucial for applications in the context of metrology, reference frames, and communication \cite{PeresClocks, giovannetti2001quantum, Buzek99, GiovannettiClockSync, KwonClockTradeOff, BartlettRefFrames, GiovannettiMetrology, giovannetti2011advances, chiribella2013quantum, helstrom1969quantum}.

\subsection*{Resource theory of U(1) asymmetry and unspeakable coherence}
It is worth noting that phase-insensitive operations are indeed the set of free operations in the resource theory of asymmetry for U(1) symmetry \cite{gour2008resource, Marvian_thesis, marvian2013theory, marvian2014extending}. Therefore,  the problem studied in this paper can be interpreted as a question in the context of this resource theory. As argued in \cite{marvian2016quantify, marvian2014modes}, this resource theory provides an approach to quantifying and characterizing unspeakable coherence (see also \cite{streltsov2017colloquium}). 
 {From this point of view, the functions $F_H$ and $P_H$ are indeed examples of measures of asymmetry.}  We discuss more about this interpretation in Section \ref{Sec:discussion}.

\section{Example: A simple distillation protocol using beam splitters}\label{section: beam splitter}

As an illuminating example, {we start with a simple, though suboptimal,  protocol that has been previously studied in the context of state purification} \cite{BeamSplitterProtocol1, BeamSplitterProtocol2}. In Fig.~\ref{first real}, {we present two slightly different realizations of this protocol that use beam splitters and {a low-temperature single ancilla mode that is ideally in the vacuum state}. This protocol takes advantage of an essential property of coherence, namely constructive interference: as we review in Appendix~\ref{appendix: general transformations passive}, after combining two coherent thermal states $\rho(\beta_1,\alpha_1)$  and $\rho(\beta_2,\alpha_2)$  via a  beam splitter, the reduced state of each mode   is also a coherent thermal state as 
\be\label{eq: bs alpha}
\rho(\beta_1,\alpha_1)\otimes \rho(\beta_2,\alpha_2) \longrightarrow \rho(\beta',t\alpha_1+r\alpha_2)\ ,
\ee
{where $t$ and $r$ are respectively the transmissivity and reflectivity of the beam splitter satisfying $|t|^2+|r|^2=1$,} and $\beta'$ is determined by the following identity on the expected number of thermal exictations 
\be\label{eq: bs n}
n_\text{th}(\beta')=|t|^2 \times n_\text{th}(\beta_1) +|r|^2 \times  n_\text{th}(\beta_2)\ . 
\ee
{In other words, while thermal excitations behave stochastically under beam splitters, coherence in the inputs exhibits interference effects.}
Applying this rule, it can be easily shown  that  {by properly choosing the beam splitters in these schemes}, one can achieve the infidelity
$\epsilon_n={n_\text{th}(\beta)}/{n}\ +\ \mathcal{O}(n^{-2})$, which means the infidelity factor is
\be\label{BS delta}
\lim_{n\rightarrow \infty } n\times \epsilon_n= n_\text{th}(\beta)\ .
\ee
See Appendix \ref{appendix: gaussian prop} for further details.

It is worth noting that Ref. \cite{Chiribella} has shown that this protocol is
optimal among all Gaussian \textit{and} non-Gaussian protocols, {as quantified by the fidelity with the desired coherent state}, under the assumption that the amplitude $\alpha$ is not known, but $\alpha$ is sampled from a Gaussian distribution centered at $\alpha=0$. However, our results imply that this protocol is \emph{not} optimal if one knows the magnitude $|\alpha|$ (note that because all the protocols under consideration in our paper are phase-insensitive, the phase of $\alpha$ does not affect the fidelity).  Nevertheless, we find 
that even when the magnitude $|\alpha|$ is known, this protocol is  still optimal among \textit{Gaussian} phase-insensitive operations.  As we show in the Appendix~\ref{appendix: gaussian prop},

\begin{proposition}\label{prop: gaussian}
For the optimal  Gaussian phase-insensitive distillation protocol, the infidelity factor is $ {\delta^{\text{opt-Gauss}}}(\beta, \alpha)= n_\text{th}(\beta)$.   Furthermore, the optimal Gaussian protocol can be realized using the scheme in Fig.~\ref{first real}.
\end{proposition}

{As we further discuss in Appendix~\ref{appendix: phase insenstive gaussian channels}, it can be easily seen that} for the input state $\rho(\beta,\alpha)^{\otimes n}$, the output of any  phase-insensitive Gaussian channel is also a coherent thermal state in the form $\rho(\beta' , A \sqrt{n} \alpha )$,  
where $A\in \mathbb{C}$,  and the inverse temperature $\beta'$ satisfies
\be\label{noi}
n_\text{th}(\beta')\ge |A|^2\times  n_\text{th}(\beta)+\max\{0, |A|^2-1\} \ . 
\ee
Furthermore, for any $\beta'\ge 0$ satisfying this constraint, the transformation $\rho(\beta,\alpha)^{\otimes n}\rightarrow \rho(\beta',A\sqrt{n}\alpha)$ is realizable by a phase-insensitive Gaussian channel. The second term in the right-hand side of Eq.(\ref{noi}), which is positive for $A>1$, corresponds to the so-called quantum noise 
in amplifiers \cite{CavesLinearAmp}.  For $A=1/\sqrt{n}$, one obtains $n_\text{th}(\beta')\ge n_\text{th}(\beta)/n$, {which is saturated by the protocol described in Fig.~\ref{first real}.}

\section{Optimal Distillation Protocol }\label{section: optimal protocol}

{Next, we introduce a distillation protocol that achieves the optimal bound mentioned in Eq.(\ref{opt bound}). In particular, under this protocol, 
the purity of coherence is conserved in the limit $n\rightarrow \infty$. To construct this protocol we first focus on the so-called \textit{strong-input weak-output} regime, i.e., where the input coherent thermal state $\rho(\beta, \alpha')$ satisfies $|\alpha'|\gg 1$, and the desired output coherent state $|\alpha\rangle$ satisfies  $|\alpha|\ll 1$. This protocol is discussed in Sec.~\ref{subsection: weak output regime}. Then, we extend this to input states with arbitrary $\alpha$ via a technique presented in Fig.~\ref{fig: divide distill} that we call ``divide and distill".

\subsection{Divide and Distill Strategy}\label{subsection: divide distill}
In the first realization of the sub-optimal distillation protocol in Fig.~\ref{first real}, we see that  $r$ copies of a coherent thermal state can be combined to obtain a coherent thermal state as
\begin{align}
  \forall \alpha\in\mathbb{C},\  \forall \beta\ge 0:\  \  \ \ \  \mathcal{C}_{r}(\rho(\beta, \alpha)^{\otimes r}) = \rho(\beta, \sqrt{r}\times \alpha)\ .
\end{align}
We sometimes refer to $\mathcal{C}_{r}$ as the concentration map. As we discuss in Sec.~\ref{subsec: concentration dilution reversibility}, when all the input modes of the concentration map $\mathcal{C}_{r}$ have the same temperature, and only in this case, this process is reversible.  {That is,  
there exists a phase-insensitive map}, which can be called the `dilution' map such that
\begin{align}
     \forall \alpha\in\mathbb{C}:\  \  \ \ \    \mathcal{D}^{(\beta)}_{r}(\rho(\beta, \alpha))= \rho(\beta, \frac{\alpha}{\sqrt{r}})^{\otimes r}\ , 
\end{align}
which means that the composed map $\mathcal{D}^{(\beta)}_{r}\circ \mathcal{C}_{r}$ leaves 
state $\rho(\beta,\alpha)^{\otimes r}$ unchanged. We can realize a dilution map by reversing the arrow of time in the first realization in Fig.~\ref{first real} ({note that in this scheme each discarded output mode is in the thermal state $\rho(\beta,0)$ and is uncorrelated with the rest of the output modes}). It is crucial to note that in contrast to the concentration map, the dilution map $\mathcal{D}_r^{(\beta)}$ depends on the temperature. {We note that the concentration map $\mathcal{C}_{r}$  has been used before in \cite{BeamSplitterProtocol1, BeamSplitterProtocol2, Chiribella}.}

\begin{figure}
    \centering
    \includegraphics[width=0.45\textwidth]{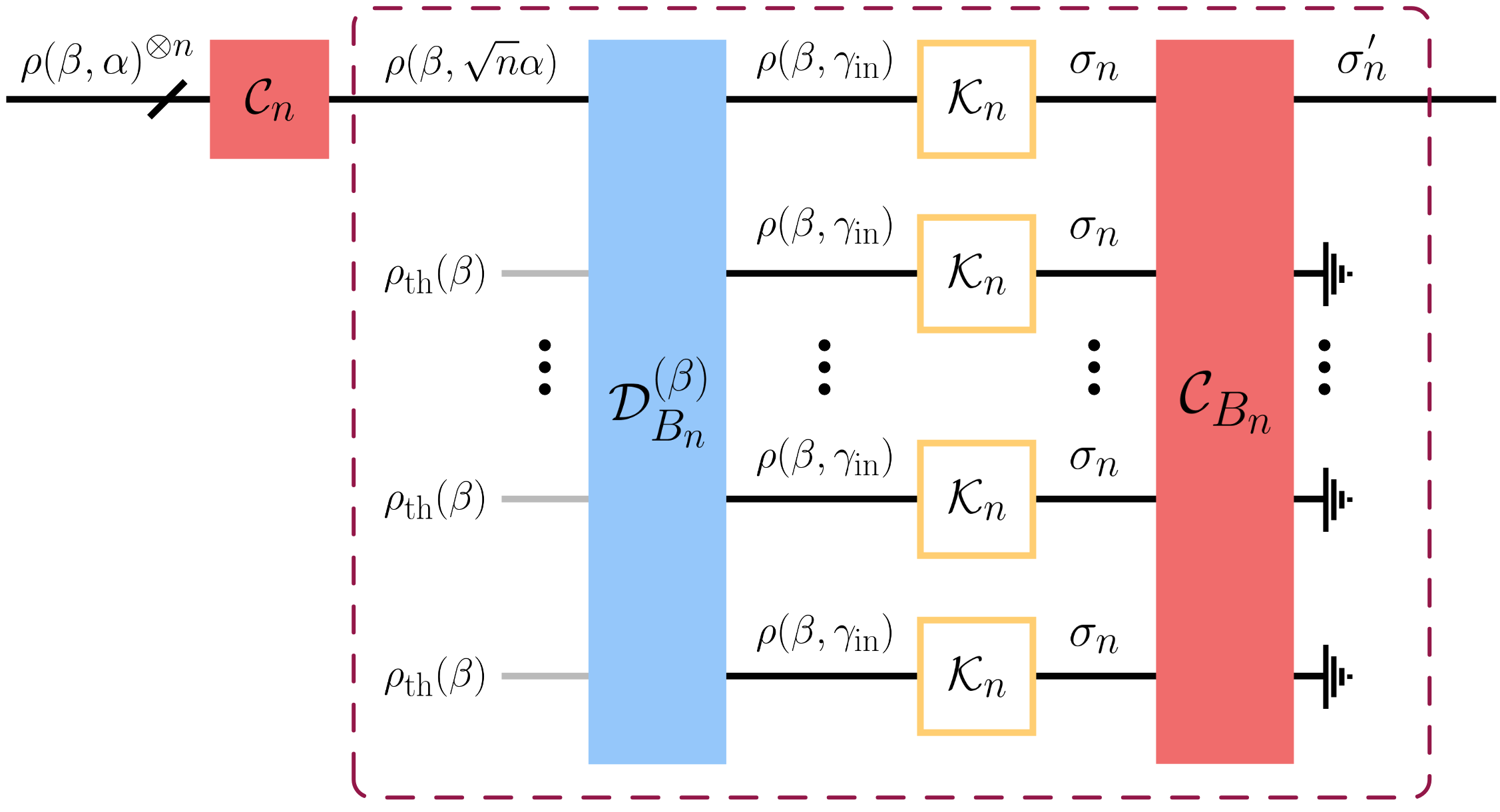}
   \caption{\textbf{Optimal Distillation using the Divide and Distill Strategy -- } {The circuit above illustrates the optimal protocol, along with the `divide and distill' strategy highlighted within the dashed box:} Suppose the phase-insensitive channel $\mathcal{K}_n$ distills an approximate pure state $\sigma_n \approx \ketbra{\gamma_\text{out}}{\gamma_\text{out}}$ from $\rho(\beta, \gamma_\text{in})$ in the regime where $|\gamma_\text{out}|\ll1$ and   $|\gamma_\text{in}|\gg1$ (discussed in Sec.~\ref{subsection: weak output regime}). We can extend this optimal behavior to general inputs, i.e., $\rho(\beta, \sqrt{n} \alpha)$ for any $\alpha$, by first diluting the input into an appropriate number of copies ($B_n$) such that each copy is in the strong-input weak-output regime {for $n \gg 1$}, then applying $\mathcal{K}_n$ in parallel to each of the $B_n$ copies, and finally concentrating all these outputs $\sigma_n$ into a single mode $\sigma'_n$. In the limit of large $n$, this procedure achieves the highest possible fidelity $\braket{\alpha|\sigma'_n|\alpha}$ (discussed in Sec.~\ref{subsection: optimal general}). It is worth noting that the first step of concentrating $\rho(\beta, \alpha)^{\otimes n} \leftrightarrow \rho(\beta, \sqrt{n}\alpha)$ is a reversible transformation (discussed in Sec.~\ref{subsec: concentration dilution reversibility}).}
    \label{fig: divide distill}
\end{figure}

Using these channels, we can achieve the following transformation,

\begin{lemma}\label{lemma: divide and distill}
Let  $\mathcal{E}$  
be a single mode (possibly non-Gaussian) phase-insensitive channel
that maps the coherent thermal state $\rho(\beta,\alpha)$ to the output state $\sigma:=\mathcal{E}(\rho(\beta,\alpha))$. Then, for any integer $m$, the phase-insensitive channel 
 \begin{align}\label{divide distill scheme}
    \mathcal{E}_m'\equiv \mathcal{C}_m \circ \mathcal{E}^{\otimes m} \circ \mathcal{D}^{(\beta)}_m\ 
\end{align}
maps  the input state $\rho(\beta,\alpha')$ to the output 
$\sigma':=\mathcal{E}'_m(\rho(\beta,\alpha'))$  {with $\alpha'=\sqrt{m}\alpha$},  whose infidelity  with the coherent state $|\alpha'\rangle$   is upper bounded by
\begin{align}\label{eqn: infid divide distill}
     1-\braket{\alpha'|\sigma'|\alpha'}\nonumber &\leq 
     \Tr(\sigma' a^\dagger a) - |\Tr(\sigma' a)|^2 + |\Tr(\sigma' a)-\alpha'|^2\nonumber\\
     &=\Tr(\sigma a^\dagger a) - |\Tr(\sigma a)|^2 + m|\Tr(\sigma a)-\alpha|^2\ .
\end{align}
 \end{lemma}
The circuit within the dashed box in Fig.~\ref{fig: divide distill} illustrates the transformation in Eq.(\ref{divide distill scheme}) {in the context of the distillation protocol discussed in the next section (in this example   $\mathcal{E}=\mathcal{K}_n$ and $m=B_n$).} Note that 
 in the first line 
 of Eq.(\ref{eqn: infid divide distill}), 
 the right-hand side 
 can be interpreted as the sum of the variance $\Tr(\sigma' a^\dagger a) - |\Tr(\sigma' a)|^2$,  and the squared `bias'  $|\Tr(\sigma' a)-\alpha'|^2$ in the phase space.  

{In summary, even though the states $\sigma$ and  $\sigma'$ are not necessarily Gaussian, this simple lemma, {which is a  consequence of Markov's inequality,} allows us to bound the infidelity of the output state $\sigma'$ with the desired coherent state $|\sqrt{m}\alpha\rangle$, based on the first and second moments of state $\sigma$. In particular, if the first moments of state $\sigma$ and the coherent state $|\alpha\rangle$ are equal, then by combining $m$ copies of $\sigma$ via the concentration map we obtain a single-mode state $\sigma'$ whose first moment is identical to the coherent state 
$|\sqrt{m}\alpha\rangle$. Then, the infidelity of these states will be bounded by the single-copy variance $\Tr(\sigma a^\dagger a) - |\Tr(\sigma a)|^2$.}

\subsection{Optimal protocol 
in the strong-input weak-{output} regime}\label{subsection: weak output regime}

Next, we  introduce a novel  phase-insensitive distillation channel $\mathcal{K}$ and show that in the proper input-output regime it is optimal, and, in particular, it preserves the purity of coherence. This protocol serves as the essential `building block' for the general optimal protocol that will be discussed in Sec.~\ref{subsection: optimal general}. 

We consider the regime in which the desired output is a weak  coherent state $|\gamma_\text{out}\rangle$ with  $|\gamma_\text{out}|\ll 1$, and the input is   $\rho(\beta,\gamma_\text{in})$
with $|\gamma_\text{in}|\gg 1$. More precisely, {we consider the regime} 
\begin{align}\label{asym conditions}
|\gamma_\text{in}|
 \gg 1 \ \  ,\ \text{and}\ \   |\gamma_\text{out}| \times |\gamma_\text{in}|\ll 1\ .
\end{align}
For instance, one can choose 
$|\gamma_\text{in}|= A |\gamma_\text{out}|^{-\delta}$, where $0<\delta<1$ and $A>0$, and take the limit $|\gamma_\text{out}|\rightarrow 0$. For convenience, without loss of generality, in the following, we assume $\gamma_\text{in}$ and $\gamma_\text{out}$ are real and positive (otherwise, we can always bring them to this form by applying a proper phase shift, which is a phase-insensitive unitary). 

In the weak-output regime, the desired output coherent state in the Fock basis can be approximated as 
\be\nonumber
|\gamma_{\text{out}}\rangle= |0\rangle+\gamma_{\text{out}}|1\rangle+\mathcal{O}(\gamma_{\text{out}}^2)\ ,
\ee 
which means that with high probability the state is restricted to the 2D subspace spanned by $|0\rangle$ and $|1\rangle$. Hence, we consider a  channel $\mathcal{K}$ such that its output is  restricted to this subspace, with Kraus decomposition 
\be\label{weak optimal channel}
\mathcal{K}(\cdot)=\sum_{l=0}^\infty K_l(\cdot) K_l^\dag\ ,
 \ee
where the Kraus operator for $l\geq0$ are 
\begin{align}\label{kraus operators 1}
    K_l = \frac{1}{\sqrt{1+|c_l|^2 \gamma_\text{out}^2 }} \ketbra{0}{l} + \frac{ c_{l+1} \gamma_\text{out}}{\sqrt{1+|c_{l+1}|^2 \gamma_\text{out}^2 }} \ketbra{1}{l+1}\ ,
\end{align}
with $c_0=0$ and $c_l \in\mathbb{C}: l>0$ is arbitrary. Then, for any choice of $c_l : l>0$, the Kraus operators satisfy the completeness relation $\sum_{l=0}^\infty K^\dag_lK_l=\mathbb{I}$, and the symmetry transformation $e^{-\mathrm{i}\phi a^\dag a} K_l e^{\mathrm{i}\phi a^\dag a}=K_l e^{\mathrm{i} l \phi}$ for all $\phi\in [0,2\pi)$, which implies that $\mathcal{K}$ is phase-insensitive. The output of this channel for arbitrary input $\rho$, has the following density matrix in the Fock basis $\{|0\rangle, |1\rangle\}$:
\begin{align}\label{eqn: tau out defn}
\tau=\mathcal{K}(\rho)  =\begin{pmatrix}
   1-\Tr(a^\dagger a\ \tau) & \Tr(a\ \tau) \nonumber\\
 \Tr(a\ \tau) & \Tr(a^\dagger a\ \tau)\ 
\end{pmatrix}\ .
\end{align}
For the optimal protocol $\mathcal{K}^{\text{opt}}$, we choose  
\be\label{weak optimal channel end}
c^{\text{opt}}_l{(\gamma_\text{in})}=\frac{\braket{l-1|\rho(\beta, \gamma_\text{in})|l-1} }{\braket{l-1|\rho(\beta, \gamma_\text{in})|l}}\ \ \ \ \  \ \ \ : \ l\geq 1 \ ,
\ee
which is the ratio of a diagonal matrix element of $\rho(\beta, \gamma_\text{in})$ in the Fock basis to the first off-diagonal element below it (it can be easily seen that unless $\gamma_\text{in}=0$ the matrix element {$\braket{l-1|\rho(\beta, \gamma_\text{in})|l}$} is  non-zero, and therefore the denominator does not vanish).

Then, as we show in  Appendix~\ref{appendix: performance of strong-input weak output}, the first and second moments and the infidelity with the desired coherent state $|\gamma_\text{out}\rangle$  respectively satisfy
\bes\label{eqn: moment asymptotics}
\begin{align}
     \gamma_\text{out}-\Tr(\tau a)&=   \gamma_\text{out}^3 \times  [1+ \mathcal{O}(\gamma_\text{in}^{-2})] + \mathcal{O}(\gamma_\text{out}^5)\ , \label{eqn: m1 main}\\   
        {\Tr(\tau  a^\dagger a)}-\gamma_\text{out}^2&=  \gamma_\text{out}^2 \times E(\gamma_\text{in}) +  \mathcal{O}(\gamma_\text{out}^4)\ ,\label{eqn: m2 main}
    \\
   1-\langle\gamma_\text{out}|\tau|\gamma_\text{out}\rangle
 &= \gamma_\text{out}^2 \times E(\gamma_\text{in})
+\mathcal{O}(\gamma_\text{out}^4)\ \label{summ 4}\ ,
\end{align}
\ees
where here the $\mathcal{O}$ notation suppresses terms that are independent of $\gamma_\text{in}$ and  $\gamma_\text{out}$ upto leading order, and function $E(\gamma_\text{in})$ is defined in the following.

\begin{lemma}\label{lem4}
Consider the coherent thermal state $\rho(\beta,\gamma)=D(\gamma) e^{-\beta a^\dag a} D(\gamma)^\dag/\Tr(e^{-\beta a^\dag a} )$ {with $\gamma>0$}. Define the function
\be\label{Def:E}
E(\gamma)=\sum_{l=0}^\infty \rho_{l,l} \big[|c^{\text{opt}}_l(\gamma)|^2-1\big]\ ,
\ee
where $\rho_{l,m}=\braket{l|\rho(\beta,\gamma)|m}$ are the Fock basis matrix elements,  $c^\text{opt}_0=0$ and $c^\text{opt}_l(\gamma)={\rho_{l-1,l-1}}/{\rho_{l-1,l}} : l>0$. Then, in the {$\gamma \gg \max\{1, n_\text{th}(\beta)^2\}$} regime,
\begin{align}\label{E with error}
E(\gamma) =  \frac{\omega^2 }{P_H(\rho(\beta, \gamma))} + \mathcal{O}\bigg(\bigg[\frac{\sqrt{1+2 n_\text{th}(\beta)}}{\gamma}\bigg]^3\bigg) \ ,
\end{align}
Therefore, 
\begin{align}\label{limit E}
{\lim_{\gamma\rightarrow \infty}} \gamma^2 E(\gamma)&=\delta^{\text{opt}}(\beta)=\frac{1}{4}\frac{F_H(|\alpha\rangle\langle\alpha|)}{P_H(\rho(\beta,\alpha))} \ ,
\end{align}
where  $H=\omega a^\dag a$ is the Hamiltonian, and the second equation holds for any $\alpha\in\mathbb{C}$ \footnote{{We note that $\mathcal{O}((1+2n_\beta)^{3/2})$ is short for $\mathcal{O}((1+2n_\beta)^{3})/(1+2n_\beta)^{3/2}$. And here, $\mathcal{O}((1+2n_\beta)^{3})$ indicates a $\gamma_\text{in}$-independent rational function that is a ratio of polynomials in $1+2n_\beta$ where the difference between the degree of the numerator and denominator polynomials is $3$.}}.
 \end{lemma}

\begin{figure}
    \centering
    \includegraphics[width=0.45\textwidth]{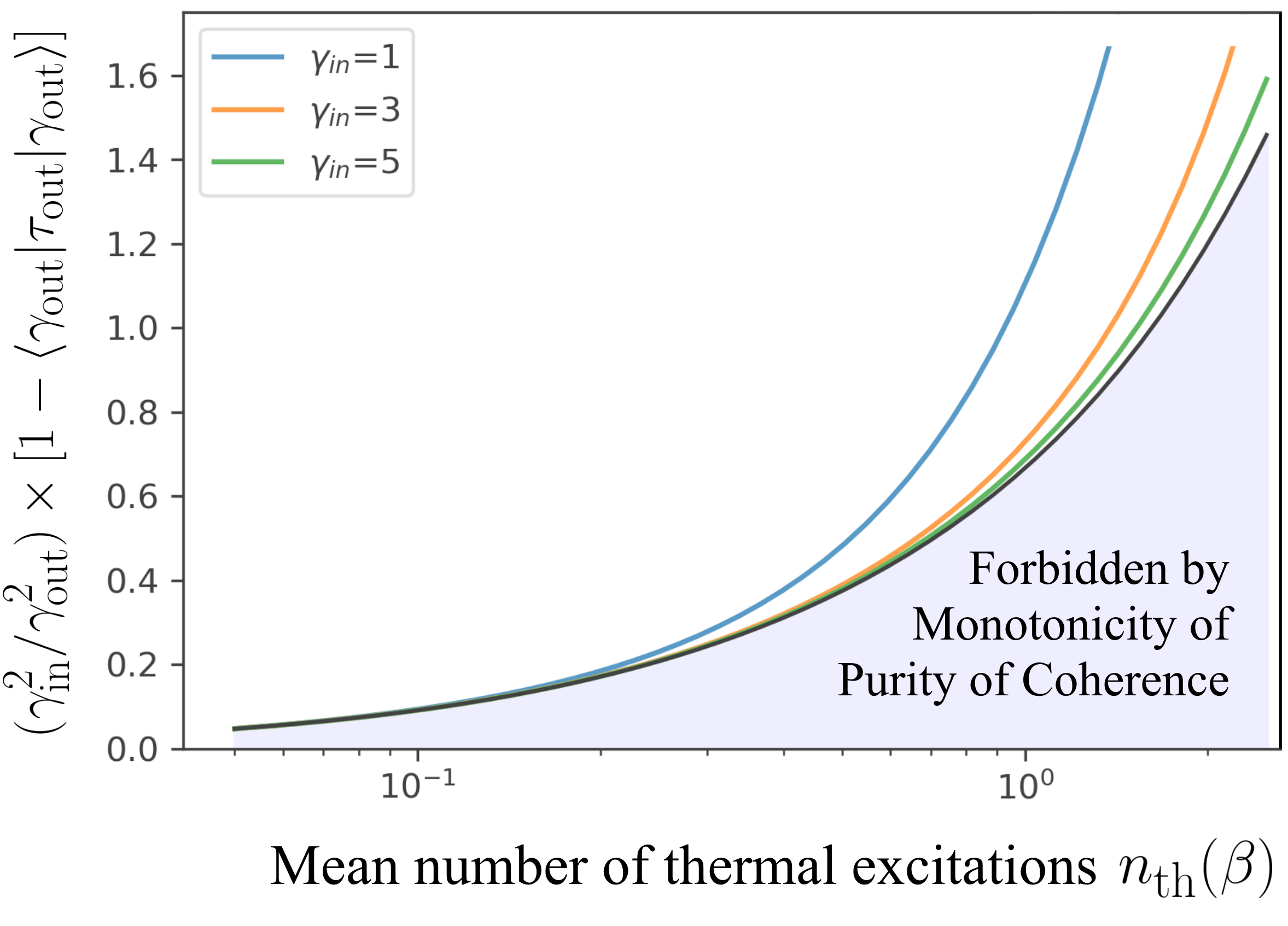}
    \caption{{\textbf{The infidelity in the strong-input weak-output regime -- }} 
    Here, we numerically study the performance of the channel $\mathcal{K}^{\text{opt}} $ introduced in Eq.(\ref{weak optimal channel}-\ref{weak optimal channel end}), for converting a single-copy of the input $\rho(\beta,\gamma_\text{in})$ to a pure coherent state  $|\gamma_\text{out}\rangle$, with $\gamma_\text{out}\ll 1$.  Specifically, we choose $\gamma_\text{out}=10^{-4}$ with values of $\gamma_\text{in}$ specified in the plot. The vertical axis is   $({\gamma_\text{in}^2}/{\gamma_\text{out}^2})\times [1-\braket{\gamma_\text{out}|\tau_\text{out}|\gamma_\text{out}}]$, i.e., the squared of the ratio of the input to output amplitudes times  the infidelity.  The shaded area is forbidden by the monotonicity of the purity of coherence. As predicted by the limit in Eq.(\ref{kjj}), we observe that as $\gamma_\text{in}$ grows, for a wide range of temperatures,  this quantity converges to the optimal infidelity factor $\delta^{\text{opt}}(\beta)$ in Eq.(\ref{opt bound}), which is dictated by the conservation of the purity of coherence. To numerically calculate fidelity $\braket{\gamma_\text{out}|\tau_\text{out}|\gamma_\text{out}}$, we include $400$ Kraus operators in Eq.(\ref{weak optimal channel}). As detailed in Appendix \ref{final infid}, this truncation of the summation introduces an error of at most {$\mathcal{O}(10^{-25})$} in the numerical estimate of $({\gamma_\text{in}^2}/{\gamma_\text{out}^2})\times [1-\braket{\gamma_\text{out}|\tau_\text{out}|\gamma_\text{out}}]$, which is clearly negligible for purposes of our plot.}
    \label{fig: num verify E}
\end{figure}
To establish this result, we find a new approximation to energy distribution for coherent thermal 
 states, which goes beyond the standard Gaussian approximation and is  of independent 
interest. See Sec.~\ref{subsection: gaussian approx} for a short overview of this approximation and Appendix \ref{appendix: matrix elements large alpha} for the detailed proof of this lemma.
In summary, $E(\gamma_\text{in})$, which determines the infidelity of the output state $\tau_\text{out}$ with the desired coherent state $|\gamma_\text{out}\rangle$, is itself determined by $P_H(\rho(\beta, \gamma_\text{in}))$,  the purity of coherence of the input.   Furthermore, in the limit of large $\gamma_\text{in}$, Eq.(\ref{summ 4}) implies that
\be\label{kjj}
{\lim_{\gamma_\text{in}\rightarrow \infty} }\frac{\gamma_\text{in}^2}{\gamma_\text{out}^2}\times [1-\braket{\gamma_\text{out}|\tau_\text{out}|\gamma_\text{out}}]=\delta^{\text{opt}}(\beta) ,
\ee
where again 
we have assumed that $\gamma_\text{in}= A \gamma_\text{out}^{-\delta}$, where $\delta $ and $A$ are fixed and satisfy $0<\delta<1$ and $A>0$.
{In Fig.~\ref{fig: num verify E} we consider the left-hand side of Eq.(\ref{kjj})
for finite values of $\gamma_\text{in}$ and $\gamma_\text{out}$ and compare it with the right-hand side. As we see in the figure, by fixing a small $\gamma_\text{out}$ and increasing $\gamma_\text{in}$, the quantity  $({\gamma_\text{in}^2}/{\gamma_\text{out}^2})\times [1-\braket{\gamma_\text{out}|\tau_\text{out}|\gamma_\text{out}}]$ approaches $\delta^\text{opt}(\beta)$.}

We conclude that this protocol saturates the bound on infidelity {set by the monotonicity of the purity of coherence}. In particular, the purity of coherence of the output is 
\begin{align}\label{qubit}
P_H(\tau_\text{out})=\frac{F_H(\tau_\text{out})}{2(1-\Tr(\tau_\text{out}^2))} =\frac{\omega^2}{E(\gamma_\text{in})} + \mathcal{O}(\gamma_\text{in}^2 \times \gamma_\text{out}^2) \ ,
\end{align}
where the first equality holds for general 2-level systems, as established in  \cite{marvian2020coherence}\footnote{We mention the useful formula $$P_H(\tau_\text{out})= \frac{(1-2p)^2}{p(1-p)}\times V_H(\ketbra{\psi}{\psi})$$ where
$\tau_{\text{out}}=p\ketbra{\psi}{\psi}+(1-p)\ketbra{\psi^\perp}{\psi^\perp}$ is its spectral decomposition.}.  Therefore, comparing with Eq.(\ref{E with error}) we find that 
{$$  {P_H(\rho(\beta,\gamma_\text{in}))}-{P_H(\tau_\text{out})}=\mathcal{O}(\gamma_\text{in}^2 \times \gamma_\text{out}^2)\ ,$$ 
which vanishes in  the  strong-input weak-output  regime, i.e.,  
$\gamma_\text{in}\rightarrow \infty$ while $\gamma_\text{in} \times \gamma_\text{out} \rightarrow 0$.  Therefore,  in the limit under consideration, the 
state transition $\rho(\beta, \gamma_\text{in})\rightarrow \tau_\text{out}$ fully conserves the purity of coherence. This  is crucial, especially because in the divide and distill strategy, discussed next, we split the input state and apply the above protocol to many copies of a coherent thermal state in parallel. If for each copy we waste non-negligible amount of purity of coherence, then the overall protocol may not be optimal. 
}

\subsection{The optimal protocol for general input }\label{subsection: optimal general}

Next, we use this protocol as a building block within the `divide and distill' scheme of Sec.~\ref{subsection: weak output regime} to construct a sequence of channels $\mathcal{G}_n$ 
that achieves the optimal infidelity factor for general values of $\alpha$ and $\beta$.

Suppose we are given $n$ copies of state $\rho(\beta,\alpha)$. {To go to the appropriate strong-input regime, we first convert them to $B_n$ modes such that
\be\nonumber
\rho(\beta,\alpha)^{\otimes n}\longrightarrow   \rho\big(\beta, \frac{\sqrt{n}}{\sqrt{B_n}} \alpha\big) ^ {\otimes B_n}\ ,
\ee
where the value of $B_n\ll n$ is fixed in the following (For concreteness, one can choose, e.g., $B_n=\lfloor n^{3/4}\rfloor$).  Then, we apply the optimal channel $\mathcal{K}^{\text{opt}}$ defined in Eq.(\ref{weak optimal channel}-\ref{weak optimal channel end}), for the value of parameters
\begin{align}\label{gamma in out}
    \gamma_\text{in}=\frac{\sqrt{n}\alpha}{\sqrt{B_n}} \  \ \ \ \text{and}\ \  \  \    \gamma_\text{out}=\frac{\alpha}{\sqrt{B_n}} \ .
\end{align}
We label this channel as $\mathcal{K}_n$, and the corresponding output state as       $\sigma_n := \mathcal{K}_n\big(\rho\big(\beta,\gamma_\text{in}\big)\big)$. 
In the relevant regime of parameters, this state will be close to the coherent state $|\gamma_\text{out}\rangle$. 
Finally, combining the obtained $B_n$ copies of  $\sigma_n$  via  the concentration channel $\mathcal{C}_n$, we generate the output state  
\begin{align}
    \sigma'_n:= \mathcal{C}_{B_n}(\sigma_n^{\otimes B_n})\ .
\end{align}

Applying Lemma \ref{lemma: divide and distill}, {$n$ times the infidelity} of this state with the desired output coherent state $|\alpha\rangle$ is upper bounded by
\begin{align}\label{eq: infid mid}
    n[1-\braket{\alpha|\sigma'_n|\alpha}]
     &\le n\big[\Tr(\sigma_n a^\dagger a) - \gamma_\text{out}^2\big]\\
     &+ n \big[\gamma_\text{out}^2-\Tr(\sigma_n a)^2\big]\nonumber\\ 
    &+ n {B_n}[
    \gamma_\text{out}-\Tr(\sigma_n a)]^2\nonumber.
\end{align}
{Using Eq.(\ref{eqn: moment asymptotics}), in Appendix \ref{Sec:overview} we show that {by} assuming $\gamma_{\text{in}}\times \gamma_{\text{out}}=\alpha^2\sqrt{n}/B_n$ vanishes in the limit of large $n$, 
the second and third lines in the right-hand side of Eq.(\ref{eq: infid mid}) vanish (this is achieved, e.g., for $B_n=\lfloor n^{3/4}\rfloor$). Then, in this limit Eq.(\ref{eqn: m2 main}) implies}
\bes
\begin{align}
   n[1-\braket{\alpha|\sigma'_n|\alpha}]
     &\le   \gamma_{\text{in}}^2\times E(\gamma_\text{in})+\mathcal{O}(\gamma_\text{in}^2\times \gamma_\text{out}^2)\ \\ &= \delta_{\text{opt}}(\beta)+\mathcal{O}\Big(\frac{n}{B_n^2}\Big)+ \mathcal{O}\Big(\frac{\sqrt{B_n}}{\sqrt{n}}\Big)
     \ ,
\end{align}
\ees
{where,} to get the second line we used {Lemma \ref{lem4} and assumed}  $\gamma_\text{in}=\alpha {\sqrt{n}}/{\sqrt{B_n}} \gg 1$, which again holds for the above choice $B_n=\lfloor n^{3/4}\rfloor$. We conclude that 
  
\begin{theorem}\label{optimal theorem}
The sequence of phase-insensitive channels
\begin{align}\label{eqn: optimal channel}
    \mathcal{G}_n \equiv \mathcal{C}_{B_n}\circ \mathcal{K}_n^{\otimes B_n}\circ\mathcal{D}_{B_n}\circ \mathcal{C}_n
\end{align}
with {$B_n=\lfloor n^{3/4} \rfloor$},  achieves the infidelity factor
\begin{align}
   \lim_{n\rightarrow \infty} n \times [1-\langle\alpha|\mathcal{G}_n(\rho(\beta,\alpha)^{\otimes n})|\alpha\rangle]=\delta^{\text{opt}}(\beta)\ .
\end{align}
Therefore, it achieves the lower bound set by the purity of coherence.
\end{theorem}
{For the above choice of $B_n$, the finite $n$ regime correction to $n$ times infidelity is $\mathcal{O}(1/\sqrt{n})$.} See Appendix \ref{Sec:overview} for a detailed analysis and the proof of this theorem.

\subsection{Distillation combined with amplification/attenuation}\label{optimal amplify/attenuate subsection}
Next, we show how this protocol can be extended to optimally distill coherent state $|\alpha'\rangle$ {for arbitrary $\alpha'$,  from copies of state $\rho(\beta,\alpha)$}, by using the freedom to reversibly amplify and attenuate coherent thermal states (reversibility of these operations is further discussed in Sec. \ref{subsec: concentration dilution reversibility}).
\begin{corollary}\label{cor5}
There exists a sequence of phase-insensitive channels $\mathcal{A}^{\text{opt}}_n$,  transforming $n$ copies of a coherent thermal state $\rho(\beta,\alpha)$ to state $\sigma_n=\mathcal{A}^{\text{opt}}_n(\rho(\beta,\alpha)^{\otimes n})$, 
such that 
\begin{align}\label{corr eqn}
 \lim_{n \rightarrow \infty} n \times  (1-\langle \alpha'|\sigma_n|\alpha'\rangle)= \frac{|\alpha'|^2}{|\alpha|^2}  \delta^\text{opt}(\beta)= \frac{1}{4}\frac{F_H(\ketbra{\alpha'}{\alpha'})}{P_H(\rho(\beta,\alpha))}
  ,\
\end{align}
where $1-\langle \alpha'|\sigma_n|\alpha'\rangle$ is the infidelity of the output state with the coherent state $|\alpha'\rangle$. Furthermore, this is the lowest achievable infidelity. That is, for any phase-insensitive sequence of channels $\mathcal{A}_n$, 
\begin{align}\label{corr eqn lowerbound}
    \liminf_{n \rightarrow \infty} n \times  (1-\langle \alpha'|\sigma_n|\alpha'\rangle) & \geq \frac{|\alpha'|^2}{|\alpha|^2} \times \delta^\text{opt}(\beta) \ .
\end{align}
\end{corollary}

\begin{proof}
The bound in Eq.(\ref{corr eqn lowerbound}) follows from 
Eq.(\ref{thm first part}) in Theorem \ref{theorem: universal bounds} for the special case of $\rho=\rho(\beta, \alpha)$ and $\ket{\phi}=\ket{\alpha'}$, where we have used the fact that $F_H(\ketbra{\alpha'}{\alpha'})= F_H(\ketbra{\alpha}{\alpha}) \times  |\alpha'|^2/|\alpha|^2$, along with the definition of $\delta^\text{opt}(\beta)$ in Eq.(\ref{limit E}). 
To see that this bound is attainable, first, we note that the state conversion  
\be\nonumber
\rho(\beta,\alpha)^{\otimes n} \longrightarrow \rho(\beta,\alpha')^{\otimes 
m(n)}\
\ee
with $m(n)= \lfloor n\times |\alpha|^2/|\alpha'|^2 \rfloor$,  can be realized with passive transformations, e.g., by composing the concentration and dilution maps $\mathcal{C}$ and $\mathcal{D}$ {discussed in Sec. \ref{subsection: divide distill}}. Then, applying the map {$\mathcal{G}_{m(n)}$} in Eq.(\ref{eqn: optimal channel}), we can convert  $\rho(\beta,\alpha')^{\otimes 
m}$ to state $\sigma_{m(n)}$ whose infidelity with  $|\alpha'\rangle$ 
 satisfies
\begin{align}
\lim_{n\rightarrow \infty} m(n) \times (1-\langle \alpha'|\sigma_{m(n)}|\alpha'\rangle) = \delta^{\text{opt}}(\beta)\ .
\end{align}
This, in turn, implies 
\begin{align}
\lim_{n\rightarrow \infty} n \times (1-\langle \alpha'|\sigma_{m(n)}|\alpha'\rangle) = \frac{|\alpha'|^2}{|\alpha|^2} \times  \delta^{\text{opt}}(\beta)\ ,
\end{align}
which completes the proof.

\end{proof}

\newpage

\subsection{Coherent thermal states in the energy basis: \\ Deviations from Gaussianity}\label{subsection: gaussian approx}

\begin{figure}
    \centering
    \includegraphics[width=0.45\textwidth]{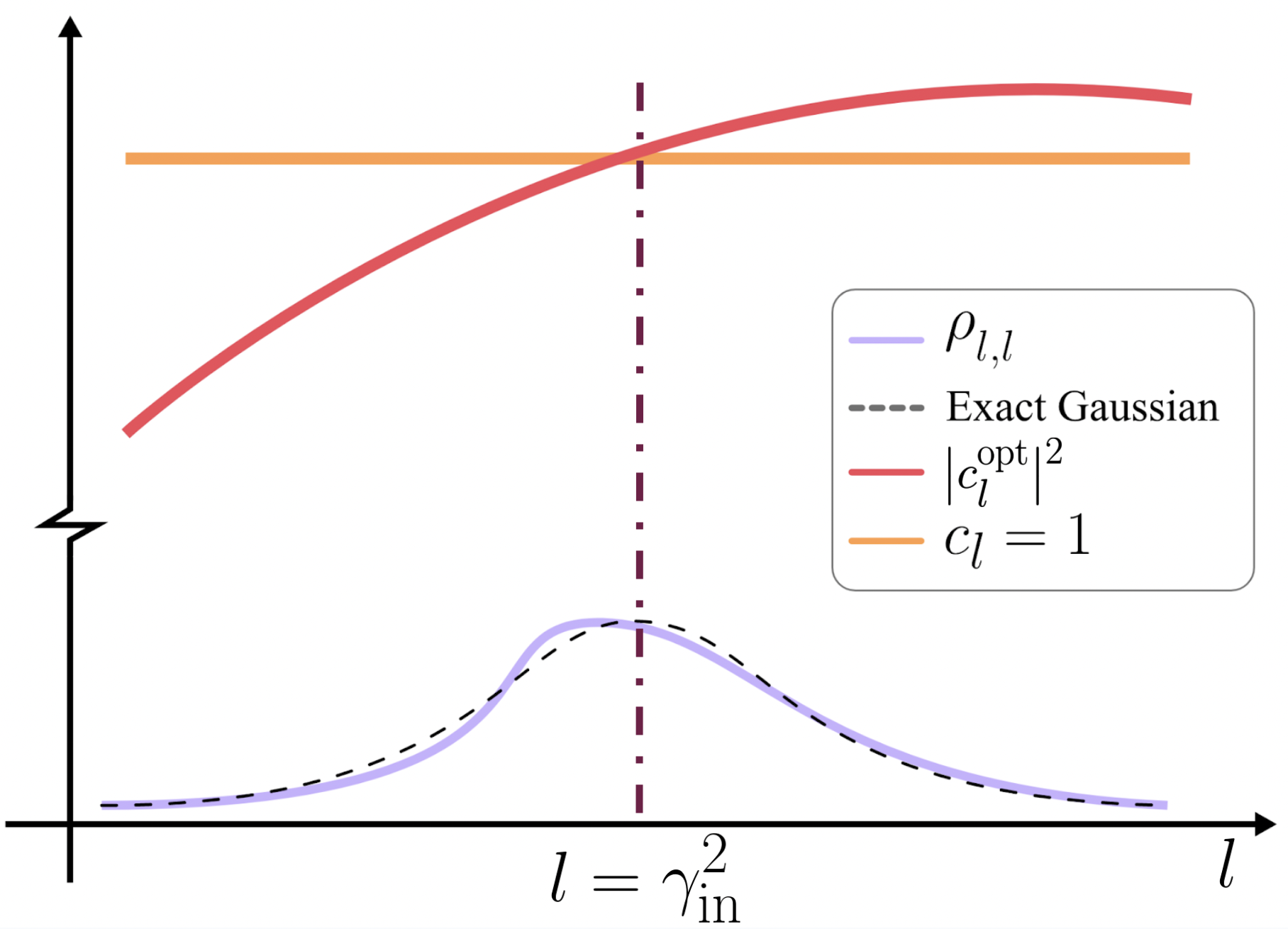}
    \caption{\textbf{Deviations from Gaussianity -- } 
    Consider the energy distribution associated to the coherent thermal state $\rho(\beta, \gamma_\text{in})$, or, equivalently, the distribution associated to the diagonal elements $\rho_{l,l}\equiv \braket{l|\rho(\beta, \gamma_\text{in})|l}$.  In the large $\gamma_\text{in}$ regime, this distribution can be approximated by a Gaussian distribution (the dashed curve). However, to understand the performance of the proposed optimal protocol we need to go beyond this approximation  (see Lemma \ref{eqn: matrix elements final}). In particular, the deviations from Gaussianity are relevant {to study the behaviour of Kraus operators} of the optimal distillation protocol, or, more precisely coefficients $c^{\text{opt}}_l$ defined in Eq.(\ref{weak optimal channel end}) as the ratio of two matrix elements{, {around $l \approx \gamma_\text{in}^2$}}. As we show in Appendix \ref{c_l subsection}, for $l\approx \gamma_\text{in}^2$, and in the large $\gamma_\text{in}$ regime,   $|c^{\text{opt}}_l|^2$ can be approximated as a quadratic polynomial in $(l-\gamma_\text{in}^2)$ (see Lemma \ref{c_l lemma2}).}
    \label{fig:deviations}
\end{figure}

{To analyze the error in 
the strong-input weak-output regime and 
establish the limit in Lemma \ref{lem4}, we need to analyze the behavior of function   $E(\gamma)=\sum_{l=0}^\infty \rho_{l,l} \big[|c^{\text{opt}}_l(\gamma)|^2-1\big]$ in the large $\gamma$ regime. This requires us to find an expansion of  $\rho_{l,l} $ the matrix elements of coherent thermal state as well as the ratio  $c^\text{opt}_l(\gamma)={\rho_{l-1,l-1}}/{\rho_{l-1,l}} : l>0$, which can be of independent interest (see Appendix~\ref{appendix: matrix elements large alpha}).} Recall that {a Harmonic Oscillator with Hamiltonian $\omega a^\dag a=\omega \sum_{l=0}^\infty l|l\rangle\langle l|$ in state  $\rho(\beta,\gamma_\text{in})$ in the regime  $\gamma_\text{in}\gg 1$, 
with high probability } is found in energy $l\omega$ with $l\approx \gamma_\text{in}^2+n_\text{th}(\beta)$. More precisely, in this regime, the diagonal matrix elements $\rho_{l,l}$ of this density operator can be approximated as a Gaussian distribution  
$$\rho_{l,l}\approx  \frac{1}{\sqrt{2 \pi}\sigma} \exp\Big[-\frac{1}{2} \Big(\frac{l-\gamma_\text{in}^2}{\sigma}\Big)^2\Big] \ ,\ $$
{with the variance}
\begin{align}
    \sigma^2 &:= (1+2 n_\text{th}(\beta)) \gamma_\text{in}^2 \ .
\end{align}
However, to determine the limit in Eq.(\ref{limit E}), one needs to go beyond this Gaussian approximation.\footnote{We note that there are standard {generic methods for approximating the deviations of a distribution from a Gaussian distribution based on its cumulants, 
such as the Edgeworth Series and the Gram-Charlier A Series \cite{nearlyGaussian}. In Appendix~\ref{appendix: matrix elements large alpha} we compare our formulae in Lemma \ref{eqn: matrix elements final}, with these generic approximations for $\rho_{ll} $and 
show that it yields a significantly better approximation.} }

{We study this in Appendix~\ref{appendix: matrix elements large alpha} and establish the following approximation for $l$ that is sufficiently close to $\gamma_\text{in}^2$, such that  $|r|\ll \gamma_\text{in}^{1/3}$, where  $r=\frac{l-\gamma_\text{in}^2}{\sigma}$ quantifies the distance of $l$ from $\gamma_\text{in}^2$ in terms of the number of {standard deviations $\sigma$} (note that this range is growing with $\gamma_\text{in}$, and in the limit $\gamma_\text{in}\rightarrow\infty$ it contains all the probability).  In particular, we find that in this range, the diagonal matrix elements  $\rho_{l,l}$ can be approximated as a Gaussian multiplied by a degree-6 polynomial of $r$, as   }
\begin{align}
    \rho_{l,l} = \frac{e^{-\frac{r^2}{2}}}{\sqrt{2 \pi}\sigma} \Big[1 + \frac{f_1(n_\beta,r)}{\sigma} &+ \frac{1}{2}\frac{f_2(n_\beta,r)}{\sigma^2}\nonumber\\ 
    &\ + \mathcal{O}\Big(\frac{(1+2n_\beta)^3}{\sigma^3}r^9\Big)\Big]\ , 
\end{align}
where $f_1$ is a degree 3, and $f_2$ is a degree 6 polynomial in $r$ with rational coefficients in terms of $n_\beta$, and $\mathcal{O}(1+2n_\beta)$ indicates a $\gamma_\text{in}$-independent rational function that is a ratio of polynomials in $1+2n_\beta$ (see Lemma \ref{eqn: matrix elements final}). It is worth noting that in the regime under consideration, namely $|r|\ll \gamma_\text{in}^{1/3}$, the contribution of ${f_1(n_\beta,r)}/{\sigma}$ and ${f_2(n_\beta,r)}/{\sigma^2}$ remains small. Nevertheless, to find the limit in Eq.(\ref{limit E}), which yields the purity of coherence, their contributions are non-negligible.   
We also obtain a similar expansion for the relevant off-diagonal terms $\rho_{l,l+1}$. Using these approximations, we prove the following quadratic approximation for $|c_l^\text{opt}|^2=|\rho_{l-1,l-1}/\rho_{l-1,l}|^2$, which in turn implies Lemma \ref{lem4}.
\begin{lemma}\label{c_l lemma2}
For $\gamma_\text{in} \gg \max\{1, n_\beta^2\}$ and  $|r|\ll \gamma_\text{in}^{1/3}$ where $r=(l-\gamma_\text{in}^2)/\sigma$, {the matrix elements of the coherent thermal state $\rho(\beta,\gamma_\text{in})$ satisfy}
\begin{align}\label{c_l eqn_copy}
    \Big|\frac{\rho_{l-1,l-1}}{\rho_{l-1,l}}\Big|^2 = 1 + \frac{r}{\sigma} + \frac{{2 n_\beta (n_\beta+1)^2}-{n_\beta (3 n_\beta+2)}r^2}{(1+2n_\beta)\ \sigma^2} \ ,
\end{align}
where $\sigma=\sqrt{1+2n_\beta} \gamma_\text{in}$ and we have ignored terms of  $\mathcal{O}({(1+2n_\beta)^{3}}/{\sigma^3})$.
\end{lemma}
Note that Eq.(\ref{c_l eqn_copy}) captures the deviation from
\be
\Big|\frac{\rho_{l-1,l-1}}{\rho_{l-1,l}}\Big|^2=\frac{l}{\gamma_\text{in}^2}=1+\frac{r}{\gamma_\text{in}}\ ,
\ee
which holds exactly in the special case of pure coherent state $\rho=|\gamma_\text{in}\rangle\langle \gamma_\text{in}|$. In Appendix \ref{app num verify all 2}, we numerically test the quadratic formula in Eq.(\ref{c_l eqn_copy}) and show that it holds for a wide range of temperatures. In Appendix \ref{limit comp app} we show that this formula implies the limit in Eq.(\ref{limit E}), or, equivalently, Eq.(\ref{kjj}).

The quadratic behavior of $|c_l^\text{opt}|^2$ and its relation to the Gaussian approximation for $\rho_{l,l}$ is shown in the schematic in Figure \ref{fig:deviations}. See Appendix \ref{appendix: matrix elements large alpha} for further details.

\section{Universal lower and upper bounds on distillation error}\label{section: universal bounds proof}
Next, we review the universal lower and upper bounds on {the error for optimal coherence distillation} previously found in \cite{marvian2020coherence}, which, in particular, {implies the lower bound  $\delta^{\text{opt}}(\beta, \alpha)\ge F_H(|\alpha\rangle\langle\alpha|)/(4 P_H(\rho(\beta,\alpha)))$ on the infidelity factor.} Note that in contrast to the rest of this paper, the results discussed in this section are applicable to general quantum systems, including finite-dimensional Hilbert spaces and infinite-dimensional Hilbert spaces with non-Gaussian states.

The  bounds found in \cite{marvian2020coherence} are in terms of two central quantities,  namely the Quantum Fisher Information and the purity of coherence: for a system with Hamiltonian $H$ and density operator $\rho$, these quantities can be defined as the second derivative of fidelity and Petz-Rényi entropy for $\alpha=2$, i.e.,
\bes
\begin{align}
F_H(\rho)&=-\frac{d^2}{dt^2} \text{Fid}(\rho,e^{-\mathrm{i} H t} \rho e^{\mathrm{i} H t})|_{t=0}\\
P_H(\rho)&=-\frac{d^2}{dt^2} D_{\alpha=2}(\rho,e^{-\mathrm{i} H t} \rho e^{\mathrm{i} H t})|_{t=0}\\ &= -\Tr(\rho^{-1}[H,\rho]^2)\ ,
\end{align}
\ees
{where $\text{Fid}(\rho_1, \rho_2)=\|\sqrt{\rho_1}\sqrt{\rho_2}\|^2_1$, and  $D_\alpha(\rho_1\|\rho_2)=(\alpha-1)^{-1} \log(\rho_1^{\alpha}\rho_2^{1-\alpha})$ when $\text{supp}(\rho_1)\subseteq \text{supp}(\rho_2)$, and $D_\alpha(\rho_1\|\rho_2)=\infty$ otherwise.}

As we further explain in Sec.~\ref{section: qfi},   these functions are indeed special cases of Fisher information metrics \cite{PETZ_199681, Bengtsson_Zyczkowski_2006, QFI_Metrics_2016, Katariya_2021, katariya2021geometric}.  {Based on this connection we find general formulae for $F_H(\rho)$ and $P_H(\rho)$ for the general class of displaced incoherent states, which include coherent thermal states as a special case.} {We also note that if  $\rho=\sum_i p_i \ketbra{\psi_i}{\psi_i}$
is the spectral decomposition
of $\rho$, then these functions can be rewritten as
\bes\label{eq: gen F_H P_H}
\begin{align}
F_H(\rho)&=\sum_{k,l}\frac{2(p_k-p_l)^2}{p_k+p_l}  \times |\langle\psi_k|H|\psi_l\rangle|^2\ ,\\  P_H(\rho)&=\sum_{k,l} \frac{p_k^2-p_l^2}{p_l} \times |\langle\psi_k|H|\psi_l\rangle|^2\ .
\end{align}
\ees
  
Functions $F_H$ and $P_H$ are examples of measures of asymmetry and (unspeakable) coherence in the resource theory of asymmetry. They both are additive {for uncorrelated systems} and monotone under covariant CPTP maps \cite{marvian2020coherence}. This, in particular, means that   a single mode state $\sigma$ can be generated from an $n$-mode state $\rho^{\otimes n}$ only if  
\begin{align}\label{bound4}
n \ge \max\big\{\frac{F_H(\sigma)}{F_H(\rho)}\ , \frac{P_H(\sigma)}{P_H(\rho)}  \big\}\ .
\end{align}

{Using the formulae in Eq.(\ref{formulaP}) for $P_H$ and $F_H$, we find that } for the desired state transformation $\rho(\beta,\alpha)^{\otimes n}\rightarrow |\alpha\rangle\langle\alpha|$ the bound in terms of QFI implies that $n\ge 2n_\beta+1$, 
whereas the bound in terms of the purity of coherence implies that $n=\infty$ (unless $\beta=\infty$ or $\alpha=0$). In other words, the purity of coherence captures the unreachability of pure coherent states via phase-insensitive distillation \cite{marvian2020coherence}.  
 It is also worth noting that for measure-and-prepare phase-insensitive channels, one obtains the bound $n\ge{F_H(\sigma)}/{P_H(\rho)}$ \cite{marvian2020coherence} which is stronger than the bounds in Eq.(\ref{bound4}), because $P_H(\rho)\ge F_H(\rho)$ .

The general results of \cite{marvian2020coherence}
are applicable to arbitrary systems provided that the systems under consideration satisfy certain regularity conditions: let $H$ be the system Hamiltonian and $\sigma$ be its density operator. Then, we assume
\begin{itemize}
\item {Finite first and second moments of energy,  i.e., $\Tr(H^k\sigma)<\infty : k=1,2$ , }
\item The dynamics of system is periodic with finite period  $\tau_H(\sigma)<\infty$, where  
\be
\tau_H(\sigma)=\inf\{ t>0: e^{-\mathrm{i} H t} \sigma e^{\mathrm{i} H t}=\sigma\}\ .
\ee
\end{itemize}
Note that the covariance condition in Eq.(\ref{cov3}) implies that the period of the input system is an integer multiple of the output system (assuming the output period is non-zero). Using the above conditions we can now state the main result of \cite{marvian2020coherence}.
In the following $V_H(\sigma)=\Tr(H^2\sigma)-\Tr(H\sigma)^2$ is the energy variance of state $\sigma$. 

\begin{theorem}\label{theorem: universal bounds} (Based on \cite{marvian2020coherence})
Let $\mathcal{E}_n$  be a sequence of CPTP maps satisfying the covariance condition in Eq.(\ref{cov3}), which converts $n$ copies of 
a system with state $\rho$ and Hamiltonian $H$ to a single copy of a system with Hamiltonian $H$ and {state $\mathcal{E}_n(\rho^{\otimes n})$ whose infidelity with the pure state $|\phi\rangle$} is $\epsilon_n=1-\langle\phi|\mathcal{E}_n(\rho^{\otimes n})|\phi\rangle$. Suppose  {$|\phi\rangle$ and $\rho$ both have} finite {first and} second moments of energy. Then, 
\be\label{thm first part}
\liminf_{n\rightarrow \infty } n\times \epsilon_n\ge \frac{V_H(\phi)}{P_H(\rho)} \ . 
\ee
Furthermore, if $\tau_H(\rho)/\tau_H(\phi)$ is an integer (i.e., the period of state $\rho$ under Hamiltonian $H$ is an integer multiple of the period of state $\phi$ under Hamiltonian $H$)  then there exists a covariant measure-and-prepare protocol such that
\be\label{thm second part}
{\lim_{n\rightarrow \infty }  n\times \epsilon_n}= \frac{V_H(\phi)}{F_H(\rho)} \ . 
\ee
\end{theorem}

In summary, we conclude that under the assumptions of this theorem, {the infidelity factor of the optimal {phase-insensitive} distillation process {that converts $n$ copies of $\rho$ to a single copy of pure state $\phi$,}
denoted by $\delta^{\text{opt}} $}, satisfies 
\be
\frac{V_H(\phi)}{P_H(\rho)}\le  \delta^{\text{opt}}
  \le \frac{V_H(\phi)}{F_H(\rho)}\ ,
\ee
where ${\delta^{\text{opt}}} := \inf_{\mathcal{E}} \lim_{n\rightarrow \infty}  n\times \epsilon_n$,   and the minimization is over all {phase-insensitive} distillation protocols $\mathcal{E}$. {Eq.(\ref{formulaP}) in the next section {gives} $P_H(\rho)$ and $F_H(\rho)$} for {the} coherent thermal state $\rho=\rho(\beta, \alpha)$ and the harmonic oscillator Hamiltonian $H=\omega a^\dag a$. {Putting these values in the above bound,} we find that in the state conversion $(\rho(\beta, \alpha))^{\otimes n} \rightarrow |\alpha\rangle\langle \alpha|$, the infidelity factor is 
\be
 \frac{n_\text{th}(\beta)}{2} + \frac{n_\text{th}(\beta)}{4n_\text{th}(\beta) + 2} \le \delta^{\text{opt}} \leq   \frac{n_\text{th}(\beta)}{2} + \frac{1}{4} \ ,
\ee
and the upper bound is achievable by a measure-prepare channel. As mentioned before, in Appendix \ref{appendix: optimal MP canonical}, we show that this optimal measure-and-prepare channel can be chosen to be one that uses the well-known canonical phase measurement \cite{holevo2011probabilistic}.\\

We provide a detailed derivation of the first part of the theorem in Appendix~\ref{appendix: universal bounds}. The second part of the theorem follows straightforwardly from the following result regarding general $m$-copy distillation.
\begin{lemma}\label{period lemma}
\cite{marvian2020coherence} Suppose the period of  system 1 with state $\rho$ and Hamiltonian $H$ is an integer multiple of the period of system 2 with state $\rho'$ and Hamiltonian $H'$. {Assume $\rho$ and $\rho'$  both have finite first and second moments of energy}. Then, there exists a time-translationally-invariant measure-and-prepare channel that transforms {$n \gg m$} copies of system 1 to  $m \ge 1$ copies of system 2, in state $\sigma_m$ such that the infidelity of $\sigma_m$ and $\rho'^{\otimes m}$ is  upper bounded by
\be\label{jdhd}
\epsilon_n:=1-\text{Fid}(\sigma_m, {\rho'}^{\otimes m}) \le \frac{1}{4}\times\frac{m}{n}\times  \frac{F_{H'}(\rho')}{ F_{H}(\rho)}+o\Big(\frac{m}{n}\Big)\ .
\ee
\end{lemma}
Note that Eq.(\ref{thm second part}) corresponds to the special case of $m=1$, where we have used that for the pure target state $\ket{\phi}$, $F_H(\ketbra{\phi}{\phi})= 4 \times V_H(\ketbra{\phi}{\phi})$. 

{The measure-and-prepare protocol that achieves the infidelity in Eq.(\ref{jdhd}) is based on the optimal maximum likelihood estimator that asymptotically saturates the Cramer-Rao bound \cite{Caves94, Barndorff-Nielsen_2000}. The protocol first 
uses $n$ copies of state $\rho(t)=\exp(-\mathrm{i} H t) \rho \exp(\mathrm{i} H t): t\in[0,\tau_H(\rho))$, to find an estimate of parameter $t$, denoted by $t_{\text{est}}$  with mean squared error $\langle \delta t^2\rangle=[n F_H(\rho)]^{-1}+o(1/n)$.   
Then, it prepares $m$ copies of state $\exp(-\mathrm{i} H' t_{\text{est}})\rho' \exp(\mathrm{i} H' t_{\text{est}})$. As shown in \cite{marvian2020coherence}, this protocol respects the time translation symmetry in Eq.(\ref{cov3}), and the infidelity of its output with the desired state $\rho'^{\otimes m}$ satisfies  Eq.(\ref{jdhd}).} {Note that because $t_{\text{est}}$ is probabilistic, in the resulting state $\sigma_m$ the $m$ output systems will be correlated. This explains why the above upper bound on infidelity of the total output state $\sigma_m$ with ${\rho'}^{\otimes m}$ grows linearly with $m$. On the other hand, if instead of the overall output state $\sigma_m$, we only focus on the reduced state of a single output system, then the infidelity 
of this copy with the desired output state $\rho'$ is independent of $m$, and is bounded by the right-hand side of Eq.(\ref{jdhd}) for $m=1$.}

\subsection{QFI Metrics for displaced incoherent states}\label{section: qfi}

The QFI {$F_H$ and the purity of coherence $P_H$ belong} to a general family of Quantum Fisher Information metrics \cite{BraunsteinGeometry, morozova1991markov, Barndorff-Nielsen_2000, matsumoto2005reverseestimationtheorycomplementality, paris2009quantum}. By definition, any Quantum Fisher Information metric $\textbf{g}$ is a contractive metric on the space of density operators. Such contractive metrics are fully classified by the work of Morozova, Cencov and Petz (here, we follow the presentation of \cite{PETZ_199681, Bengtsson_Zyczkowski_2006, QFI_Metrics_2016}).  Namely, any such metric is characterized by a function $f: \mathbb{R}_+\rightarrow \mathbb{R}_+$, such that for a density operator $\rho$ with spectral decomposition $\rho=\sum_i p_i |i\rangle\langle i|$, it holds that
\be\label{eqn: qfi metric}
\textbf{g}^{f}_\rho(\dot{\rho}, \dot{\rho})=\sum_{i} \frac{\dot{\rho}^2_{ii}}{p_i} +2 \sum_{i<j} c_f(p_i,p_j) |\dot{\rho}_{ij}|^2\ ,
\ee
where $\rho_{ij}=\langle i|\rho|j\rangle$, and  $c_f(x,y)=\frac{1}{yf(x/y)}$
where $f: \mathbb{R}_+\rightarrow \mathbb{R}_+$ is (i) an operator monotone (i.e., $f(A)\le f(B)$ for $A\le B$), (ii) selfinverse $f(t)=t f(1/t)$, and (iii) normalized $f(1)=1$.\footnote{Note that our definition of $\textbf{g}^f_\rho$ differs from the same in \cite{PETZ_199681, Bengtsson_Zyczkowski_2006, QFI_Metrics_2016} by a factor of $4$ to ensure that $F_H(\psi)=4 \times V_H(\psi)$ for all pure states $\psi$.} Interestingly,  any function $f$ satisfying the above condition is bounded from below and above by the following functions, which themselves satisfy these conditions:  
\begin{align}\label{SLD}
f^{\text{SLD}}(x)=\frac{1+x}{2}\ \ ,\  \ 
f^{\text{RLD}}(x)=\frac{2x}{1+x}\ .
\end{align}
The corresponding 
QFI metrics are called the Symmetric-Logarithmic-Derivative (SLD) and the Right-Logarithmic-Derivative (RLD) QFI, which are, respectively, the minimum and maximum QFI metrics. In particular, in the case of family of states $\exp(-\mathrm{i} t H)\rho \exp(\mathrm{i} t H): t\in \mathbb{R}$, in Appendix \ref{F_H and P_H deriv appendix} we verify that the corresponding QFI metrics are indeed the functions $F_H
$ and $P_H$ mentioned in Eq.(\ref{eq: gen F_H P_H}).

In Appendix~\ref{appendix: incoherent qfi computation} we calculate the QFI metrics for {the family of \textit{displaced incoherent states} defined as $D(\alpha) \rho_\text{incoh} D^\dagger(\alpha)$, where $\rho_\text{incoh}$ is an arbitrary incoherent state satisfying $[\rho_\text{incoh},H]=0$ with $H=\omega a^\dagger a$}, and show that
\begin{lemma}\label{lemma: qfi computation}
    Let $\rho=D(\alpha)\rho_\text{incoh}D(\alpha)^\dag$, where  $\rho_\text{incoh}=\sum_i p_i |i\rangle\langle i|$ is an incoherent state of a Harmonic oscillator, and $\{|i\rangle\}$ is the  eigenbasis of  Hamiltonian $H=\omega a^\dag a$.  Then, for the family of states {$\rho(t)=e^{-\mathrm{i} H t} \rho e^{\mathrm{i} H t}$} the Fisher information metric in Eq.(\ref{eqn: qfi metric})
     evaluates to
\begin{align}\label{eqn: incoherent qfi}
    \textbf{g}^{f}_\rho(\dot{\rho}, \dot{\rho})=2{\omega^2 |\alpha|^2}\sum_i (i+1) p_{i+1} \times  \frac{({p_i}/{p_{i+1}}-1)^2}{f(p_i/p_{i+1})}\ .
\end{align}
\end{lemma}

For the special case of thermal states $\rho_\text{incoh}=e^{-\beta H}/\Tr(e^{-\beta H})$, the ratio ${p_i}/{p_{i+1}}$ is independent of the energy level and is equal to the Boltzmann factor
\be
\frac{p_i}{p_{i+1}}=e^{-\beta \omega}=\frac{n_{\text{th}}(\beta)}{1+n_{\text{th}}(\beta)}\ .
\ee
This implies
\begin{align}\label{qfi eqn cts}
    \textbf{g}^{f}_\rho(\dot{\rho}, \dot{\rho})=  \frac{{2}\omega^2|\alpha|^2}{n_\text{th}(\beta)\times f(1+n_\text{th}(\beta)^{-1})}\ ,
\end{align}
where we noted that $\sum_i p_{i+1} (i+1) = n_{\text{th}}(\beta)$ by definition for thermal states.  Then, using Eq.(\ref{SLD}) we obtain
\bes\label{formulaP}
\begin{align}
F_H(\rho(\beta, \alpha))=\textbf{g}^{\text{SLD}}_\rho(\dot{\rho}, \dot{\rho})&=\omega^2|\alpha|^2\times \frac{4}{2n_\text{th}(\beta)+1  }   \label{eq: sld fisher info} \\ 
P_H(\rho(\beta, \alpha))=\textbf{g}^{\text{RLD}}_\rho(\dot{\rho}, \dot{\rho})&=\omega^2|\alpha|^2 \times \frac{2n_\text{th}(\beta)+1}{n_\text{th}(\beta) (n_\text{th}(\beta)+1)} 
\end{align}
\ees
which are the QFI and the purity of coherence, respectively. It is worth noting  that in the infinite temperature limit $n_\text{th}(\beta)\gg 1$,  we have
\be
P_H(\rho)\approx F_H(\rho)\approx  \frac{2\omega^2|\alpha|^2}{n_{\text{th}}(\beta)}\ .
\ee
Furthermore, in the opposite limit where temperature goes to zero, $n_\text{th}(\beta)\ll 1$,  we have
\be
P_H(\rho)\approx F_H(\rho)
\times \frac{1}{4 n_\text{th}(\beta)}\ .
\ee

\subsection{{Reversibility and irreversibility of the concentration map }}\label{subsec: concentration dilution reversibility}

Finally, we explain another interesting application of QFI metrics in the context of coherence distillation. Recall the concentration map 
discussed in Sec.~\ref{subsection: divide distill}, which can be realized by the top circuit in Fig.~\ref{fig: divide distill}. Suppose the input modes in the circuit are in  state $\otimes_{j=1}^n \rho(\beta_j,\alpha)$. The beam splitters can be chosen such that full constructive interference occurs in only one of the output modes, which means that the state of this output mode is $\rho(\beta',\sqrt{n}\alpha)$. On the other hand, in the other $n-1$ output modes, destructive interference occurs, and as a result, their reduced state is a thermal state, i.e., they do not contain any coherence. Hence, the concentration map $\mathcal{C}_n$ discards these $n-1$ modes. In Sec.~\ref{subsection: divide distill} we mentioned that when all the input modes have the same temperature, i.e., when 
\be\label{equal temp condition}
\beta_1=\cdots=\beta_n:=\beta\ , 
\ee
the transformation from $n$ input modes to the single output mode is reversible (in this case, the output mode with constructive interference is in states $\rho(\beta,\sqrt{n}\alpha)$, i.e., have the same temperature). 
That is, by applying the phase-insensitive channel $\mathcal{D}_n^{(\beta)}$, we can recover the initial state of $n$ modes from this single output mode (see Appendix \ref{appendix: gaussian prop} for further discussion). {Recall that we used this reversibility to argue that using phase-insensitive operations state $\rho(\beta,\alpha)^{\otimes n}$ can be converted to state $\rho(\beta',\alpha')^{\otimes m}$ with infidelity $\epsilon$ if, and only if,  state $\rho(\beta,\sqrt{n}\alpha)$ can be converted to state $\rho(\beta',\sqrt{m}\alpha')$ with the same infidelity.}

{Now suppose that the input modes of the concentration map have different temperatures.  
Then, again it is not hard to show (see Appendix \ref{appendix: gaussian prop}) that the same constructive interference happens and creates a coherent thermal state in the form $\rho(\beta',\sqrt{n}\alpha)$, where the output temperature is determined by
\be
n_{\text{th}}(\beta')=\frac{1}{n}\sum_{j=1}^n n_{\text{th}}(\beta_j)\ .
\ee
Furthermore, similar to the equal temperature case in Eq.(\ref{equal temp condition}), all the remaining $n-1$ modes are in thermal states, and do not contain any coherence.  Therefore, one may expect that the concentration map is still reversible, i.e., there exists a phase-insensitive process recovering the initial state $\otimes_{j=1}^n \rho(\beta_j,\alpha)$ from this single output mode in state $\rho(\beta',\sqrt{n}\alpha)$. }

However, despite the lack of coherence in the discarded $n-1$ modes, it turns out that, unless  all the input modes have the same temperature, i.e.,  Eq.(\ref{equal temp condition}) holds, the process {$\otimes_{j=1}^n \rho(\beta_j,\alpha) \rightarrow \rho(\beta',\sqrt{n}\alpha)$ is not reversible via phase-insensitive operations.} Here, we show how this claim, and its converse, i.e., the reversibility of the concentration map when the equal-temperature condition holds, can both 
 be understood using only the properties of QFI metrics.

Consider QFI metrics in Eq.(\ref{qfi eqn cts}), for the family of input states 
$$\bigotimes_{j=1}^n \exp(-\mathrm{i} H t)\rho(\beta_j,\alpha)  \exp(\mathrm{i} H t)\ \ \ \ \  : t\in \mathbb{R}\ ,$$
and for the corresponding output states 
$$\exp(-\mathrm{i} H t)\rho(\beta',\sqrt{n}\alpha)  \exp(\mathrm{i} H t)\ \ \ \ : t\in \mathbb{R}\ .$$
Then, the monotonicity of QFI metrics requires that the input QFI should be larger than or equal to the output QFI, which implies
\begin{align}\label{input qfi}
     \omega^2|\alpha|^2 \times \sum_{j=1}^n g^f(n_\text{th}(\beta_j)) \ge \omega^2|\sqrt{n}\alpha|^2 \times g^f\big[\frac{1}{n}\sum_{j=1}^n n_\text{th}(\beta_j)\big]\ ,
\end{align}
where
\be\label{hg}
g^f(s)=  \frac{{2}}{s\times f(1+s^{-1})}\ ,
\ee
{and $f$ is the monotone function defining the  QFI metric under consideration through
the Morozova-Cencov-Petz characterization in 
Eq.(\ref{eqn: qfi metric}).} Here, to calculate the QFI metric for the input state in the left-hand side of Eq.(\ref{input qfi}), we have used the additivity of Fisher information metrics for tensor product states. If Eq.(\ref{input qfi}) is a strict inequality for some QFI metrics, then the process is irreversible.  
As an example, consider SLD Fisher information (i.e., $F_H$)  which corresponds to the function 
$f(x)=f^{\text{SLD}}(x)={(1+x)}/{2}$ in Eq.(\ref{hg}). Then, since the function {$(2s+1)^{-1}$ is strictly convex for all $s>-0.5$}, it follows that the equality holds only if the equal-temperature condition in Eq.(\ref{equal temp condition}) is satisfied. Thus, if this condition is not satisfied, the value of this metric decreases during the concentration process. 

Therefore, interestingly, we conclude that when the equal-temperature condition in Eq.(\ref{input qfi}) is not satisfied, even though the discarded $n-1$ modes do not contain any coherence,   {discarding them makes the process irreversible}. This happens because the discarded modes contain correlations with the output mode, and therefore the actual coherence of the system is reduced by discarding them. 

It is also worth noting that the converse statement, i.e., the fact that the process is reversible when the equal-temperature condition holds, can also be seen using 
a recent result of  \cite{gao2023sufficient}. According to this result,  
if any of a subclass of Fisher information metrics that satisfy certain regularity conditions remain conserved, then the process is reversible (this subclass includes, for instance,  skew information). But, comparing the two sides of Eq.(\ref{input qfi}), we clearly see that when the equal-temperature condition holds, all QFI metrics remain conserved in the process, which by the result of \cite{gao2023sufficient}, implies the process is reversible.

\section{Conclusion}\label{Sec:discussion}
Inspired by the notion of entanglement distillation \cite{Bennett:96}, in recent years many researchers have explored various notions of resource distillation, and more specifically coherence distillation with respect to various sets of free operations in these resource theories (See, e.g., \cite{winter2015operational, streltsov2017colloquium, regula2018one, zhao2018one, lami2018generic, ProbDistillation, CoherenceManipulation, GenericBoundCoherence, ChitambarResourceTheories, LamiGrandTour}). While from a more abstract perspective, the notion of resource distillation is a well-defined and interesting problem to study, in general resource theories (including various versions of the resource theory of coherence) it may not always have a clear physical interpretation. This is because, in contrast to entanglement theory, in general resource theories the choice of free operations is not always justifiable based on physical or practical considerations (see \cite{marvian2016quantify} for further discussion).

{On the other hand}, the notion of distillation studied in \cite{marvian2020coherence} and {in the current} paper has clear physical interpretations. First, as schematically presented in Fig.~\ref{fig:distill},  coherence distillation in this context is a process for recovering a pure signal from its noisy version while preserving its phase information. Second, as further discussed in \cite{marvian2020coherence}, 
this notion of coherence distillation arises in quantum thermodynamics, assuming one can consume an arbitrary amount of work (or, equivalently, low-entropy ancilla systems), but one does not have access to a source of coherence. 

Besides these practical motivations (and probably related to them), there is an  information-theoretic motivation to further investigate coherence distillation and, more generally, manipulations of coherence and asymmetry, from the perspective of the resource theory of asymmetry. Namely, this perspective reveals novel operational interpretations of the SLD and RLD Fisher information metrics, which as shown by Petz \cite{PETZ_199681}, are respectively the minimal and maximal monotone metrics. More precisely, the SLD Fisher information determines the coherence cost, i.e., the rate of transformation from pure states containing coherence to general states in asymptotic (i.i.d.) regimes \cite{marvian2022operational}. On the other hand, as we demonstrated in this paper, the RLD Fisher information determines the minimum achievable error in the distillation of a single copy pure state, at least, in the case of coherent thermal states. 
This remarkable fact about the distinguished roles of the SLD and RLD Fisher information in the resource  theory of asymmetry, motivates further investigation of this resource theory, and its approach to (unspeakable) coherence (See, e.g., \cite{tajima2022universal, yamaguchi2023beyond, tajima2024gibbs, takagi2022correlation} for recent works in this direction).\\

\section*{Acknowledgements}
This work is supported by a collaboration between the US DOE and other Agencies. This material is based upon work supported by the U.S. Department of Energy, Office of Science, National Quantum Information Science Research Centers, Quantum Systems Accelerator. Additional support is
acknowledged from  
NSF Phy-2046195, NSF FET-2106448, and 
NSF QLCI grant OMA-2120757. We also acknowledge helpful discussions with Y. Chitgopekar, N. Koukoulekidis, S. Kazi, D. Jakab, P. Deliyannis, S. Brahmachari, A. Hulse and G. L. Sidhardh. SAY sincerely thanks Saathwik Yadavalli for all the time and effort he volunteered for making figures for this paper despite his busy schedule.

\bibliography{main}

\begin{thebibliography}{89}%
\makeatletter
\providecommand \@ifxundefined [1]{%
 \@ifx{#1\undefined}
}%
\providecommand \@ifnum [1]{%
 \ifnum #1\expandafter \@firstoftwo
 \else \expandafter \@secondoftwo
 \fi
}%
\providecommand \@ifx [1]{%
 \ifx #1\expandafter \@firstoftwo
 \else \expandafter \@secondoftwo
 \fi
}%
\providecommand \natexlab [1]{#1}%
\providecommand \enquote  [1]{``#1''}%
\providecommand \bibnamefont  [1]{#1}%
\providecommand \bibfnamefont [1]{#1}%
\providecommand \citenamefont [1]{#1}%
\providecommand \href@noop [0]{\@secondoftwo}%
\providecommand \href [0]{\begingroup \@sanitize@url \@href}%
\providecommand \@href[1]{\@@startlink{#1}\@@href}%
\providecommand \@@href[1]{\endgroup#1\@@endlink}%
\providecommand \@sanitize@url [0]{\catcode `\\12\catcode `\$12\catcode `\&12\catcode `\#12\catcode `\^12\catcode `\_12\catcode `\%12\relax}%
\providecommand \@@startlink[1]{}%
\providecommand \@@endlink[0]{}%
\providecommand \url  [0]{\begingroup\@sanitize@url \@url }%
\providecommand \@url [1]{\endgroup\@href {#1}{\urlprefix }}%
\providecommand \urlprefix  [0]{URL }%
\providecommand \Eprint [0]{\href }%
\providecommand \doibase [0]{https://doi.org/}%
\providecommand \selectlanguage [0]{\@gobble}%
\providecommand \bibinfo  [0]{\@secondoftwo}%
\providecommand \bibfield  [0]{\@secondoftwo}%
\providecommand \translation [1]{[#1]}%
\providecommand \BibitemOpen [0]{}%
\providecommand \bibitemStop [0]{}%
\providecommand \bibitemNoStop [0]{.\EOS\space}%
\providecommand \EOS [0]{\spacefactor3000\relax}%
\providecommand \BibitemShut  [1]{\csname bibitem#1\endcsname}%
\let\auto@bib@innerbib\@empty
%</preamble>
\bibitem [{\citenamefont {Sudarshan}(1963)}]{ECG_coherent1}%
  \BibitemOpen
  \bibfield  {author} {\bibinfo {author} {\bibfnamefont {E.~C.~G.}\ \bibnamefont {Sudarshan}},\ }\bibfield  {title} {\bibinfo {title} {Equivalence of semiclassical and quantum mechanical descriptions of statistical light beams},\ }\href {https://doi.org/10.1103/PhysRevLett.10.277} {\bibfield  {journal} {\bibinfo  {journal} {Phys. Rev. Lett.}\ }\textbf {\bibinfo {volume} {10}},\ \bibinfo {pages} {277} (\bibinfo {year} {1963})}\BibitemShut {NoStop}%
\bibitem [{\citenamefont {Glauber}(1963{\natexlab{a}})}]{GlauberOpticalCoherence}%
  \BibitemOpen
  \bibfield  {author} {\bibinfo {author} {\bibfnamefont {R.~J.}\ \bibnamefont {Glauber}},\ }\bibfield  {title} {\bibinfo {title} {The quantum theory of optical coherence},\ }\href {https://doi.org/10.1103/PhysRev.130.2529} {\bibfield  {journal} {\bibinfo  {journal} {Phys. Rev.}\ }\textbf {\bibinfo {volume} {130}},\ \bibinfo {pages} {2529} (\bibinfo {year} {1963}{\natexlab{a}})}\BibitemShut {NoStop}%
\bibitem [{\citenamefont {Glauber}(1963{\natexlab{b}})}]{Glauber63}%
  \BibitemOpen
  \bibfield  {author} {\bibinfo {author} {\bibfnamefont {R.~J.}\ \bibnamefont {Glauber}},\ }\bibfield  {title} {\bibinfo {title} {Coherent and incoherent states of the radiation field},\ }\href {https://doi.org/10.1103/PhysRev.131.2766} {\bibfield  {journal} {\bibinfo  {journal} {Phys. Rev.}\ }\textbf {\bibinfo {volume} {131}},\ \bibinfo {pages} {2766} (\bibinfo {year} {1963}{\natexlab{b}})}\BibitemShut {NoStop}%
\bibitem [{\citenamefont {Mehta}\ \emph {et~al.}(1967)\citenamefont {Mehta}, \citenamefont {Chand}, \citenamefont {Sudarshan},\ and\ \citenamefont {Vedam}}]{ECG_coherent2}%
  \BibitemOpen
  \bibfield  {author} {\bibinfo {author} {\bibfnamefont {C.~L.}\ \bibnamefont {Mehta}}, \bibinfo {author} {\bibfnamefont {P.}~\bibnamefont {Chand}}, \bibinfo {author} {\bibfnamefont {E.~C.~G.}\ \bibnamefont {Sudarshan}},\ and\ \bibinfo {author} {\bibfnamefont {R.}~\bibnamefont {Vedam}},\ }\bibfield  {title} {\bibinfo {title} {Dynamics of coherent states},\ }\href {https://doi.org/10.1103/PhysRev.157.1198} {\bibfield  {journal} {\bibinfo  {journal} {Phys. Rev.}\ }\textbf {\bibinfo {volume} {157}},\ \bibinfo {pages} {1198} (\bibinfo {year} {1967})}\BibitemShut {NoStop}%
\bibitem [{\citenamefont {Schr{\"o}dinger}(1926)}]{schrodinger1926stetige}%
  \BibitemOpen
  \bibfield  {author} {\bibinfo {author} {\bibfnamefont {E.}~\bibnamefont {Schr{\"o}dinger}},\ }\bibfield  {title} {\bibinfo {title} {Der stetige {\"u}bergang von der mikro-zur makromechanik},\ }\href@noop {} {\bibfield  {journal} {\bibinfo  {journal} {Naturwissenschaften}\ }\textbf {\bibinfo {volume} {14}},\ \bibinfo {pages} {664} (\bibinfo {year} {1926})}\BibitemShut {NoStop}%
\bibitem [{\citenamefont {Mandel}\ and\ \citenamefont {Wolf}(1995)}]{Mandel_Wolf_1995}%
  \BibitemOpen
  \bibfield  {author} {\bibinfo {author} {\bibfnamefont {L.}~\bibnamefont {Mandel}}\ and\ \bibinfo {author} {\bibfnamefont {E.}~\bibnamefont {Wolf}},\ }\href@noop {} {\emph {\bibinfo {title} {Optical Coherence and Quantum Optics}}}\ (\bibinfo  {publisher} {Cambridge University Press},\ \bibinfo {year} {1995})\BibitemShut {NoStop}%
\bibitem [{\citenamefont {Khan}\ \emph {et~al.}(2017)\citenamefont {Khan}, \citenamefont {Elser}, \citenamefont {Dirmeier}, \citenamefont {Marquardt},\ and\ \citenamefont {Leuchs}}]{khan2017quantum}%
  \BibitemOpen
  \bibfield  {author} {\bibinfo {author} {\bibfnamefont {I.}~\bibnamefont {Khan}}, \bibinfo {author} {\bibfnamefont {D.}~\bibnamefont {Elser}}, \bibinfo {author} {\bibfnamefont {T.}~\bibnamefont {Dirmeier}}, \bibinfo {author} {\bibfnamefont {C.}~\bibnamefont {Marquardt}},\ and\ \bibinfo {author} {\bibfnamefont {G.}~\bibnamefont {Leuchs}},\ }\bibfield  {title} {\bibinfo {title} {Quantum communication with coherent states of light},\ }\href@noop {} {\bibfield  {journal} {\bibinfo  {journal} {Philosophical Transactions of the Royal Society A: Mathematical, Physical and Engineering Sciences}\ }\textbf {\bibinfo {volume} {375}},\ \bibinfo {pages} {20160235} (\bibinfo {year} {2017})}\BibitemShut {NoStop}%
\bibitem [{\citenamefont {Caves}\ and\ \citenamefont {Drummond}(1994)}]{Caves94}%
  \BibitemOpen
  \bibfield  {author} {\bibinfo {author} {\bibfnamefont {C.~M.}\ \bibnamefont {Caves}}\ and\ \bibinfo {author} {\bibfnamefont {P.~D.}\ \bibnamefont {Drummond}},\ }\bibfield  {title} {\bibinfo {title} {Quantum limits on bosonic communication rates},\ }\href {https://doi.org/10.1103/RevModPhys.66.481} {\bibfield  {journal} {\bibinfo  {journal} {Rev. Mod. Phys.}\ }\textbf {\bibinfo {volume} {66}},\ \bibinfo {pages} {481} (\bibinfo {year} {1994})}\BibitemShut {NoStop}%
\bibitem [{\citenamefont {Iblisdir}\ \emph {et~al.}(2004)\citenamefont {Iblisdir}, \citenamefont {Van~Assche},\ and\ \citenamefont {Cerf}}]{app3}%
  \BibitemOpen
  \bibfield  {author} {\bibinfo {author} {\bibfnamefont {S.}~\bibnamefont {Iblisdir}}, \bibinfo {author} {\bibfnamefont {G.}~\bibnamefont {Van~Assche}},\ and\ \bibinfo {author} {\bibfnamefont {N.~J.}\ \bibnamefont {Cerf}},\ }\bibfield  {title} {\bibinfo {title} {Security of quantum key distribution with coherent states and homodyne detection},\ }\href {https://doi.org/10.1103/PhysRevLett.93.170502} {\bibfield  {journal} {\bibinfo  {journal} {Phys. Rev. Lett.}\ }\textbf {\bibinfo {volume} {93}},\ \bibinfo {pages} {170502} (\bibinfo {year} {2004})}\BibitemShut {NoStop}%
\bibitem [{\citenamefont {Heid}\ and\ \citenamefont {L\"utkenhaus}(2007)}]{app4}%
  \BibitemOpen
  \bibfield  {author} {\bibinfo {author} {\bibfnamefont {M.}~\bibnamefont {Heid}}\ and\ \bibinfo {author} {\bibfnamefont {N.}~\bibnamefont {L\"utkenhaus}},\ }\bibfield  {title} {\bibinfo {title} {Security of coherent-state quantum cryptography in the presence of gaussian noise},\ }\href {https://doi.org/10.1103/PhysRevA.76.022313} {\bibfield  {journal} {\bibinfo  {journal} {Phys. Rev. A}\ }\textbf {\bibinfo {volume} {76}},\ \bibinfo {pages} {022313} (\bibinfo {year} {2007})}\BibitemShut {NoStop}%
\bibitem [{\citenamefont {Braunstein}\ \emph {et~al.}(2001)\citenamefont {Braunstein}, \citenamefont {Cerf}, \citenamefont {Iblisdir}, \citenamefont {van Loock},\ and\ \citenamefont {Massar}}]{app5}%
  \BibitemOpen
  \bibfield  {author} {\bibinfo {author} {\bibfnamefont {S.~L.}\ \bibnamefont {Braunstein}}, \bibinfo {author} {\bibfnamefont {N.~J.}\ \bibnamefont {Cerf}}, \bibinfo {author} {\bibfnamefont {S.}~\bibnamefont {Iblisdir}}, \bibinfo {author} {\bibfnamefont {P.}~\bibnamefont {van Loock}},\ and\ \bibinfo {author} {\bibfnamefont {S.}~\bibnamefont {Massar}},\ }\bibfield  {title} {\bibinfo {title} {Optimal cloning of coherent states with a linear amplifier and beam splitters},\ }\href {https://doi.org/10.1103/PhysRevLett.86.4938} {\bibfield  {journal} {\bibinfo  {journal} {Phys. Rev. Lett.}\ }\textbf {\bibinfo {volume} {86}},\ \bibinfo {pages} {4938} (\bibinfo {year} {2001})}\BibitemShut {NoStop}%
\bibitem [{\citenamefont {Cerf}\ \emph {et~al.}(2005)\citenamefont {Cerf}, \citenamefont {Kr\"uger}, \citenamefont {Navez}, \citenamefont {Werner},\ and\ \citenamefont {Wolf}}]{app6}%
  \BibitemOpen
  \bibfield  {author} {\bibinfo {author} {\bibfnamefont {N.~J.}\ \bibnamefont {Cerf}}, \bibinfo {author} {\bibfnamefont {O.}~\bibnamefont {Kr\"uger}}, \bibinfo {author} {\bibfnamefont {P.}~\bibnamefont {Navez}}, \bibinfo {author} {\bibfnamefont {R.~F.}\ \bibnamefont {Werner}},\ and\ \bibinfo {author} {\bibfnamefont {M.~M.}\ \bibnamefont {Wolf}},\ }\bibfield  {title} {\bibinfo {title} {Non-gaussian cloning of quantum coherent states is optimal},\ }\href {https://doi.org/10.1103/PhysRevLett.95.070501} {\bibfield  {journal} {\bibinfo  {journal} {Phys. Rev. Lett.}\ }\textbf {\bibinfo {volume} {95}},\ \bibinfo {pages} {070501} (\bibinfo {year} {2005})}\BibitemShut {NoStop}%
\bibitem [{\citenamefont {Andersen}\ \emph {et~al.}(2005{\natexlab{a}})\citenamefont {Andersen}, \citenamefont {Josse},\ and\ \citenamefont {Leuchs}}]{app7}%
  \BibitemOpen
  \bibfield  {author} {\bibinfo {author} {\bibfnamefont {U.~L.}\ \bibnamefont {Andersen}}, \bibinfo {author} {\bibfnamefont {V.}~\bibnamefont {Josse}},\ and\ \bibinfo {author} {\bibfnamefont {G.}~\bibnamefont {Leuchs}},\ }\bibfield  {title} {\bibinfo {title} {Unconditional quantum cloning of coherent states with linear optics},\ }\href {https://doi.org/10.1103/PhysRevLett.94.240503} {\bibfield  {journal} {\bibinfo  {journal} {Phys. Rev. Lett.}\ }\textbf {\bibinfo {volume} {94}},\ \bibinfo {pages} {240503} (\bibinfo {year} {2005}{\natexlab{a}})}\BibitemShut {NoStop}%
\bibitem [{\citenamefont {Namiki}\ \emph {et~al.}(2006)\citenamefont {Namiki}, \citenamefont {Koashi},\ and\ \citenamefont {Imoto}}]{app8}%
  \BibitemOpen
  \bibfield  {author} {\bibinfo {author} {\bibfnamefont {R.}~\bibnamefont {Namiki}}, \bibinfo {author} {\bibfnamefont {M.}~\bibnamefont {Koashi}},\ and\ \bibinfo {author} {\bibfnamefont {N.}~\bibnamefont {Imoto}},\ }\bibfield  {title} {\bibinfo {title} {Cloning and optimal gaussian individual attacks for a continuous-variable quantum key distribution using coherent states and reverse reconciliation},\ }\href {https://doi.org/10.1103/PhysRevA.73.032302} {\bibfield  {journal} {\bibinfo  {journal} {Phys. Rev. A}\ }\textbf {\bibinfo {volume} {73}},\ \bibinfo {pages} {032302} (\bibinfo {year} {2006})}\BibitemShut {NoStop}%
\bibitem [{\citenamefont {Samuel L.~Braunstein}\ and\ \citenamefont {Kimble}(2000)}]{app9}%
  \BibitemOpen
  \bibfield  {author} {\bibinfo {author} {\bibfnamefont {C.~A.~F.}\ \bibnamefont {Samuel L.~Braunstein}}\ and\ \bibinfo {author} {\bibfnamefont {H.~J.}\ \bibnamefont {Kimble}},\ }\bibfield  {title} {\bibinfo {title} {Criteria for continuous-variable quantum teleportation},\ }\href {https://doi.org/10.1080/09500340008244041} {\bibfield  {journal} {\bibinfo  {journal} {Journal of Modern Optics}\ }\textbf {\bibinfo {volume} {47}},\ \bibinfo {pages} {267} (\bibinfo {year} {2000})},\ \Eprint {https://arxiv.org/abs/https://www.tandfonline.com/doi/pdf/10.1080/09500340008244041} {https://www.tandfonline.com/doi/pdf/10.1080/09500340008244041} \BibitemShut {NoStop}%
\bibitem [{\citenamefont {Arrazola}\ and\ \citenamefont {L\"utkenhaus}(2014)}]{Lutkenhaus2014}%
  \BibitemOpen
  \bibfield  {author} {\bibinfo {author} {\bibfnamefont {J.~M.}\ \bibnamefont {Arrazola}}\ and\ \bibinfo {author} {\bibfnamefont {N.}~\bibnamefont {L\"utkenhaus}},\ }\bibfield  {title} {\bibinfo {title} {Quantum communication with coherent states and linear optics},\ }\href {https://doi.org/10.1103/PhysRevA.90.042335} {\bibfield  {journal} {\bibinfo  {journal} {Phys. Rev. A}\ }\textbf {\bibinfo {volume} {90}},\ \bibinfo {pages} {042335} (\bibinfo {year} {2014})}\BibitemShut {NoStop}%
\bibitem [{\citenamefont {Pirandola}\ \emph {et~al.}(2018)\citenamefont {Pirandola}, \citenamefont {Bardhan}, \citenamefont {Gehring}, \citenamefont {Weedbrook},\ and\ \citenamefont {Lloyd}}]{pirandola2018advances}%
  \BibitemOpen
  \bibfield  {author} {\bibinfo {author} {\bibfnamefont {S.}~\bibnamefont {Pirandola}}, \bibinfo {author} {\bibfnamefont {B.~R.}\ \bibnamefont {Bardhan}}, \bibinfo {author} {\bibfnamefont {T.}~\bibnamefont {Gehring}}, \bibinfo {author} {\bibfnamefont {C.}~\bibnamefont {Weedbrook}},\ and\ \bibinfo {author} {\bibfnamefont {S.}~\bibnamefont {Lloyd}},\ }\bibfield  {title} {\bibinfo {title} {Advances in photonic quantum sensing},\ }\href@noop {} {\bibfield  {journal} {\bibinfo  {journal} {Nature Photonics}\ }\textbf {\bibinfo {volume} {12}},\ \bibinfo {pages} {724} (\bibinfo {year} {2018})}\BibitemShut {NoStop}%
\bibitem [{\citenamefont {Lami}\ \emph {et~al.}(2024)\citenamefont {Lami}, \citenamefont {Pedernales},\ and\ \citenamefont {Plenio}}]{TestingQuantumnessLami}%
  \BibitemOpen
  \bibfield  {author} {\bibinfo {author} {\bibfnamefont {L.}~\bibnamefont {Lami}}, \bibinfo {author} {\bibfnamefont {J.~S.}\ \bibnamefont {Pedernales}},\ and\ \bibinfo {author} {\bibfnamefont {M.~B.}\ \bibnamefont {Plenio}},\ }\bibfield  {title} {\bibinfo {title} {Testing the quantumness of gravity without entanglement},\ }\href {https://doi.org/10.1103/PhysRevX.14.021022} {\bibfield  {journal} {\bibinfo  {journal} {Phys. Rev. X}\ }\textbf {\bibinfo {volume} {14}},\ \bibinfo {pages} {021022} (\bibinfo {year} {2024})}\BibitemShut {NoStop}%
\bibitem [{\citenamefont {Holevo}(2013)}]{Holevo+2013}%
  \BibitemOpen
  \bibfield  {author} {\bibinfo {author} {\bibfnamefont {A.~S.}\ \bibnamefont {Holevo}},\ }\href {https://doi.org/doi:10.1515/9783110273403} {\emph {\bibinfo {title} {Quantum Systems, Channels, Information}}}\ (\bibinfo  {publisher} {De Gruyter},\ \bibinfo {address} {Berlin, Boston},\ \bibinfo {year} {2013})\BibitemShut {NoStop}%
\bibitem [{\citenamefont {Weedbrook}\ \emph {et~al.}(2012)\citenamefont {Weedbrook}, \citenamefont {Pirandola}, \citenamefont {Garc\'{\i}a-Patr\'on}, \citenamefont {Cerf}, \citenamefont {Ralph}, \citenamefont {Shapiro},\ and\ \citenamefont {Lloyd}}]{GaussianQI}%
  \BibitemOpen
  \bibfield  {author} {\bibinfo {author} {\bibfnamefont {C.}~\bibnamefont {Weedbrook}}, \bibinfo {author} {\bibfnamefont {S.}~\bibnamefont {Pirandola}}, \bibinfo {author} {\bibfnamefont {R.}~\bibnamefont {Garc\'{\i}a-Patr\'on}}, \bibinfo {author} {\bibfnamefont {N.~J.}\ \bibnamefont {Cerf}}, \bibinfo {author} {\bibfnamefont {T.~C.}\ \bibnamefont {Ralph}}, \bibinfo {author} {\bibfnamefont {J.~H.}\ \bibnamefont {Shapiro}},\ and\ \bibinfo {author} {\bibfnamefont {S.}~\bibnamefont {Lloyd}},\ }\bibfield  {title} {\bibinfo {title} {Gaussian quantum information},\ }\href {https://doi.org/10.1103/RevModPhys.84.621} {\bibfield  {journal} {\bibinfo  {journal} {Rev. Mod. Phys.}\ }\textbf {\bibinfo {volume} {84}},\ \bibinfo {pages} {621} (\bibinfo {year} {2012})}\BibitemShut {NoStop}%
\bibitem [{\citenamefont {Rosati}\ \emph {et~al.}(2018)\citenamefont {Rosati}, \citenamefont {Mari},\ and\ \citenamefont {Giovannetti}}]{rosati2018narrow}%
  \BibitemOpen
  \bibfield  {author} {\bibinfo {author} {\bibfnamefont {M.}~\bibnamefont {Rosati}}, \bibinfo {author} {\bibfnamefont {A.}~\bibnamefont {Mari}},\ and\ \bibinfo {author} {\bibfnamefont {V.}~\bibnamefont {Giovannetti}},\ }\bibfield  {title} {\bibinfo {title} {Narrow bounds for the quantum capacity of thermal attenuators},\ }\href@noop {} {\bibfield  {journal} {\bibinfo  {journal} {Nature communications}\ }\textbf {\bibinfo {volume} {9}},\ \bibinfo {pages} {4339} (\bibinfo {year} {2018})}\BibitemShut {NoStop}%
\bibitem [{\citenamefont {Xiang}\ \emph {et~al.}(2017)\citenamefont {Xiang}, \citenamefont {Zhang}, \citenamefont {Jiang},\ and\ \citenamefont {Rabl}}]{ThermalMicrowaveNetworks}%
  \BibitemOpen
  \bibfield  {author} {\bibinfo {author} {\bibfnamefont {Z.-L.}\ \bibnamefont {Xiang}}, \bibinfo {author} {\bibfnamefont {M.}~\bibnamefont {Zhang}}, \bibinfo {author} {\bibfnamefont {L.}~\bibnamefont {Jiang}},\ and\ \bibinfo {author} {\bibfnamefont {P.}~\bibnamefont {Rabl}},\ }\bibfield  {title} {\bibinfo {title} {Intracity quantum communication via thermal microwave networks},\ }\href {https://doi.org/10.1103/PhysRevX.7.011035} {\bibfield  {journal} {\bibinfo  {journal} {Phys. Rev. X}\ }\textbf {\bibinfo {volume} {7}},\ \bibinfo {pages} {011035} (\bibinfo {year} {2017})}\BibitemShut {NoStop}%
\bibitem [{\citenamefont {Bishop}\ and\ \citenamefont {Vourdas}(1987)}]{bishop1987coherent}%
  \BibitemOpen
  \bibfield  {author} {\bibinfo {author} {\bibfnamefont {R.}~\bibnamefont {Bishop}}\ and\ \bibinfo {author} {\bibfnamefont {A.}~\bibnamefont {Vourdas}},\ }\bibfield  {title} {\bibinfo {title} {Coherent mixed states and a generalised p representation},\ }\href@noop {} {\bibfield  {journal} {\bibinfo  {journal} {Journal of Physics A: Mathematical and General}\ }\textbf {\bibinfo {volume} {20}},\ \bibinfo {pages} {3743} (\bibinfo {year} {1987})}\BibitemShut {NoStop}%
\bibitem [{\citenamefont {Vourdas}\ and\ \citenamefont {Bishop}(1994)}]{CTS1}%
  \BibitemOpen
  \bibfield  {author} {\bibinfo {author} {\bibfnamefont {A.}~\bibnamefont {Vourdas}}\ and\ \bibinfo {author} {\bibfnamefont {R.~F.}\ \bibnamefont {Bishop}},\ }\bibfield  {title} {\bibinfo {title} {Thermal coherent states in the bargmann representation},\ }\href {https://doi.org/10.1103/PhysRevA.50.3331} {\bibfield  {journal} {\bibinfo  {journal} {Phys. Rev. A}\ }\textbf {\bibinfo {volume} {50}},\ \bibinfo {pages} {3331} (\bibinfo {year} {1994})}\BibitemShut {NoStop}%
\bibitem [{\citenamefont {Mollow}\ and\ \citenamefont {Glauber}(1967)}]{GlauberMatrixElem}%
  \BibitemOpen
  \bibfield  {author} {\bibinfo {author} {\bibfnamefont {B.~R.}\ \bibnamefont {Mollow}}\ and\ \bibinfo {author} {\bibfnamefont {R.~J.}\ \bibnamefont {Glauber}},\ }\bibfield  {title} {\bibinfo {title} {Quantum theory of parametric amplification. i},\ }\href {https://doi.org/10.1103/PhysRev.160.1076} {\bibfield  {journal} {\bibinfo  {journal} {Phys. Rev.}\ }\textbf {\bibinfo {volume} {160}},\ \bibinfo {pages} {1076} (\bibinfo {year} {1967})}\BibitemShut {NoStop}%
\bibitem [{\citenamefont {Marian}\ and\ \citenamefont {Marian}(1993)}]{MarianMatrixElem}%
  \BibitemOpen
  \bibfield  {author} {\bibinfo {author} {\bibfnamefont {P.}~\bibnamefont {Marian}}\ and\ \bibinfo {author} {\bibfnamefont {T.~A.}\ \bibnamefont {Marian}},\ }\bibfield  {title} {\bibinfo {title} {Squeezed states with thermal noise. i. photon-number statistics},\ }\href {https://doi.org/10.1103/PhysRevA.47.4474} {\bibfield  {journal} {\bibinfo  {journal} {Phys. Rev. A}\ }\textbf {\bibinfo {volume} {47}},\ \bibinfo {pages} {4474} (\bibinfo {year} {1993})}\BibitemShut {NoStop}%
\bibitem [{\citenamefont {Oz-Vogt}\ \emph {et~al.}(1991)\citenamefont {Oz-Vogt}, \citenamefont {Mann},\ and\ \citenamefont {Revzen}}]{CTS2}%
  \BibitemOpen
  \bibfield  {author} {\bibinfo {author} {\bibfnamefont {J.}~\bibnamefont {Oz-Vogt}}, \bibinfo {author} {\bibfnamefont {A.}~\bibnamefont {Mann}},\ and\ \bibinfo {author} {\bibfnamefont {M.}~\bibnamefont {Revzen}},\ }\bibfield  {title} {\bibinfo {title} {Thermal coherent states and thermal squeezed states},\ }\href@noop {} {\bibfield  {journal} {\bibinfo  {journal} {Journal of Modern Optics}\ }\textbf {\bibinfo {volume} {38}},\ \bibinfo {pages} {2339} (\bibinfo {year} {1991})}\BibitemShut {NoStop}%
\bibitem [{\citenamefont {Mann}\ \emph {et~al.}(1989)\citenamefont {Mann}, \citenamefont {Revzen}, \citenamefont {Nakamura}, \citenamefont {Umezawa},\ and\ \citenamefont {Yamanaka}}]{CTS3}%
  \BibitemOpen
  \bibfield  {author} {\bibinfo {author} {\bibfnamefont {A.}~\bibnamefont {Mann}}, \bibinfo {author} {\bibfnamefont {M.}~\bibnamefont {Revzen}}, \bibinfo {author} {\bibfnamefont {K.}~\bibnamefont {Nakamura}}, \bibinfo {author} {\bibfnamefont {H.}~\bibnamefont {Umezawa}},\ and\ \bibinfo {author} {\bibfnamefont {Y.}~\bibnamefont {Yamanaka}},\ }\bibfield  {title} {\bibinfo {title} {Coherent and thermal coherent state},\ }\href@noop {} {\bibfield  {journal} {\bibinfo  {journal} {Journal of mathematical physics}\ }\textbf {\bibinfo {volume} {30}},\ \bibinfo {pages} {2883} (\bibinfo {year} {1989})}\BibitemShut {NoStop}%
\bibitem [{\citenamefont {Marvian}(2020)}]{marvian2020coherence}%
  \BibitemOpen
  \bibfield  {author} {\bibinfo {author} {\bibfnamefont {I.}~\bibnamefont {Marvian}},\ }\bibfield  {title} {\bibinfo {title} {Coherence distillation machines are impossible in quantum thermodynamics},\ }\href@noop {} {\bibfield  {journal} {\bibinfo  {journal} {Nature communications}\ }\textbf {\bibinfo {volume} {11}},\ \bibinfo {pages} {25} (\bibinfo {year} {2020})}\BibitemShut {NoStop}%
\bibitem [{\citenamefont {Gour}\ and\ \citenamefont {Spekkens}(2008)}]{gour2008resource}%
  \BibitemOpen
  \bibfield  {author} {\bibinfo {author} {\bibfnamefont {G.}~\bibnamefont {Gour}}\ and\ \bibinfo {author} {\bibfnamefont {R.~W.}\ \bibnamefont {Spekkens}},\ }\bibfield  {title} {\bibinfo {title} {The resource theory of quantum reference frames: manipulations and monotones},\ }\href@noop {} {\bibfield  {journal} {\bibinfo  {journal} {New Journal of Physics}\ }\textbf {\bibinfo {volume} {10}},\ \bibinfo {pages} {033023} (\bibinfo {year} {2008})}\BibitemShut {NoStop}%
\bibitem [{\citenamefont {Marvian}(2012)}]{Marvian_thesis}%
  \BibitemOpen
  \bibfield  {author} {\bibinfo {author} {\bibfnamefont {I.}~\bibnamefont {Marvian}},\ }\emph {\bibinfo {title} {Symmetry, Asymmetry and Quantum Information, PhD thesis}},\ \href@noop {} {Ph.D. thesis},\ \bibinfo  {school} {University of Waterloo}, \bibinfo {address} {https://uwspace.uwaterloo.ca/handle/10012/7088} (\bibinfo {year} {2012})\BibitemShut {NoStop}%
\bibitem [{\citenamefont {Marvian}\ and\ \citenamefont {Spekkens}(2013)}]{marvian2013theory}%
  \BibitemOpen
  \bibfield  {author} {\bibinfo {author} {\bibfnamefont {I.}~\bibnamefont {Marvian}}\ and\ \bibinfo {author} {\bibfnamefont {R.~W.}\ \bibnamefont {Spekkens}},\ }\bibfield  {title} {\bibinfo {title} {The theory of manipulations of pure state asymmetry: I. basic tools, equivalence classes and single copy transformations},\ }\href@noop {} {\bibfield  {journal} {\bibinfo  {journal} {New Journal of Physics}\ }\textbf {\bibinfo {volume} {15}},\ \bibinfo {pages} {033001} (\bibinfo {year} {2013})}\BibitemShut {NoStop}%
\bibitem [{\citenamefont {Marvian}\ and\ \citenamefont {Spekkens}(2014{\natexlab{a}})}]{marvian2014extending}%
  \BibitemOpen
  \bibfield  {author} {\bibinfo {author} {\bibfnamefont {I.}~\bibnamefont {Marvian}}\ and\ \bibinfo {author} {\bibfnamefont {R.~W.}\ \bibnamefont {Spekkens}},\ }\bibfield  {title} {\bibinfo {title} {Extending noether's theorem by quantifying the asymmetry of quantum states},\ }\href@noop {} {\bibfield  {journal} {\bibinfo  {journal} {Nature communications}\ }\textbf {\bibinfo {volume} {5}},\ \bibinfo {pages} {3821} (\bibinfo {year} {2014}{\natexlab{a}})}\BibitemShut {NoStop}%
\bibitem [{\citenamefont {Marvian}\ and\ \citenamefont {Spekkens}(2014{\natexlab{b}})}]{marvianAsymPureStates}%
  \BibitemOpen
  \bibfield  {author} {\bibinfo {author} {\bibfnamefont {I.}~\bibnamefont {Marvian}}\ and\ \bibinfo {author} {\bibfnamefont {R.~W.}\ \bibnamefont {Spekkens}},\ }\bibfield  {title} {\bibinfo {title} {Asymmetry properties of pure quantum states},\ }\href {https://doi.org/10.1103/PhysRevA.90.014102} {\bibfield  {journal} {\bibinfo  {journal} {Phys. Rev. A}\ }\textbf {\bibinfo {volume} {90}},\ \bibinfo {pages} {014102} (\bibinfo {year} {2014}{\natexlab{b}})}\BibitemShut {NoStop}%
\bibitem [{\citenamefont {Piani}\ \emph {et~al.}(2016)\citenamefont {Piani}, \citenamefont {Cianciaruso}, \citenamefont {Bromley}, \citenamefont {Napoli}, \citenamefont {Johnston},\ and\ \citenamefont {Adesso}}]{PianiAsym}%
  \BibitemOpen
  \bibfield  {author} {\bibinfo {author} {\bibfnamefont {M.}~\bibnamefont {Piani}}, \bibinfo {author} {\bibfnamefont {M.}~\bibnamefont {Cianciaruso}}, \bibinfo {author} {\bibfnamefont {T.~R.}\ \bibnamefont {Bromley}}, \bibinfo {author} {\bibfnamefont {C.}~\bibnamefont {Napoli}}, \bibinfo {author} {\bibfnamefont {N.}~\bibnamefont {Johnston}},\ and\ \bibinfo {author} {\bibfnamefont {G.}~\bibnamefont {Adesso}},\ }\bibfield  {title} {\bibinfo {title} {Robustness of asymmetry and coherence of quantum states},\ }\href {https://doi.org/10.1103/PhysRevA.93.042107} {\bibfield  {journal} {\bibinfo  {journal} {Phys. Rev. A}\ }\textbf {\bibinfo {volume} {93}},\ \bibinfo {pages} {042107} (\bibinfo {year} {2016})}\BibitemShut {NoStop}%
\bibitem [{\citenamefont {Marvian}\ and\ \citenamefont {Spekkens}(2016)}]{marvian2016quantify}%
  \BibitemOpen
  \bibfield  {author} {\bibinfo {author} {\bibfnamefont {I.}~\bibnamefont {Marvian}}\ and\ \bibinfo {author} {\bibfnamefont {R.~W.}\ \bibnamefont {Spekkens}},\ }\bibfield  {title} {\bibinfo {title} {How to quantify coherence: Distinguishing speakable and unspeakable notions},\ }\href@noop {} {\bibfield  {journal} {\bibinfo  {journal} {Physical Review A}\ }\textbf {\bibinfo {volume} {94}},\ \bibinfo {pages} {052324} (\bibinfo {year} {2016})}\BibitemShut {NoStop}%
\bibitem [{\citenamefont {Caves}(1982)}]{CavesLinearAmp}%
  \BibitemOpen
  \bibfield  {author} {\bibinfo {author} {\bibfnamefont {C.~M.}\ \bibnamefont {Caves}},\ }\bibfield  {title} {\bibinfo {title} {Quantum limits on noise in linear amplifiers},\ }\href {https://doi.org/10.1103/PhysRevD.26.1817} {\bibfield  {journal} {\bibinfo  {journal} {Phys. Rev. D}\ }\textbf {\bibinfo {volume} {26}},\ \bibinfo {pages} {1817} (\bibinfo {year} {1982})}\BibitemShut {NoStop}%
\bibitem [{\citenamefont {Caves}\ \emph {et~al.}(2012)\citenamefont {Caves}, \citenamefont {Combes}, \citenamefont {Jiang},\ and\ \citenamefont {Pandey}}]{CavesAmplifier12}%
  \BibitemOpen
  \bibfield  {author} {\bibinfo {author} {\bibfnamefont {C.~M.}\ \bibnamefont {Caves}}, \bibinfo {author} {\bibfnamefont {J.}~\bibnamefont {Combes}}, \bibinfo {author} {\bibfnamefont {Z.}~\bibnamefont {Jiang}},\ and\ \bibinfo {author} {\bibfnamefont {S.}~\bibnamefont {Pandey}},\ }\bibfield  {title} {\bibinfo {title} {Quantum limits on phase-preserving linear amplifiers},\ }\href {https://doi.org/10.1103/PhysRevA.86.063802} {\bibfield  {journal} {\bibinfo  {journal} {Phys. Rev. A}\ }\textbf {\bibinfo {volume} {86}},\ \bibinfo {pages} {063802} (\bibinfo {year} {2012})}\BibitemShut {NoStop}%
\bibitem [{\citenamefont {Spagnolo}\ \emph {et~al.}(2012)\citenamefont {Spagnolo}, \citenamefont {Vitelli}, \citenamefont {Lucivero}, \citenamefont {Giovannetti}, \citenamefont {Maccone},\ and\ \citenamefont {Sciarrino}}]{PhaseEstimation12}%
  \BibitemOpen
  \bibfield  {author} {\bibinfo {author} {\bibfnamefont {N.}~\bibnamefont {Spagnolo}}, \bibinfo {author} {\bibfnamefont {C.}~\bibnamefont {Vitelli}}, \bibinfo {author} {\bibfnamefont {V.~G.}\ \bibnamefont {Lucivero}}, \bibinfo {author} {\bibfnamefont {V.}~\bibnamefont {Giovannetti}}, \bibinfo {author} {\bibfnamefont {L.}~\bibnamefont {Maccone}},\ and\ \bibinfo {author} {\bibfnamefont {F.}~\bibnamefont {Sciarrino}},\ }\bibfield  {title} {\bibinfo {title} {Phase estimation via quantum interferometry for noisy detectors},\ }\href {https://doi.org/10.1103/PhysRevLett.108.233602} {\bibfield  {journal} {\bibinfo  {journal} {Phys. Rev. Lett.}\ }\textbf {\bibinfo {volume} {108}},\ \bibinfo {pages} {233602} (\bibinfo {year} {2012})}\BibitemShut {NoStop}%
\bibitem [{\citenamefont {Petz}(1996)}]{PETZ_199681}%
  \BibitemOpen
  \bibfield  {author} {\bibinfo {author} {\bibfnamefont {D.}~\bibnamefont {Petz}},\ }\bibfield  {title} {\bibinfo {title} {Monotone metrics on matrix spaces},\ }\href {https://doi.org/https://doi.org/10.1016/0024-3795(94)00211-8} {\bibfield  {journal} {\bibinfo  {journal} {Linear Algebra and its Applications}\ }\textbf {\bibinfo {volume} {244}},\ \bibinfo {pages} {81} (\bibinfo {year} {1996})}\BibitemShut {NoStop}%
\bibitem [{\citenamefont {Bengtsson}\ and\ \citenamefont {Zyczkowski}(2006)}]{Bengtsson_Zyczkowski_2006}%
  \BibitemOpen
  \bibfield  {author} {\bibinfo {author} {\bibfnamefont {I.}~\bibnamefont {Bengtsson}}\ and\ \bibinfo {author} {\bibfnamefont {K.}~\bibnamefont {Zyczkowski}},\ }\href@noop {} {\emph {\bibinfo {title} {Geometry of Quantum States: An Introduction to Quantum Entanglement}}}\ (\bibinfo  {publisher} {Cambridge University Press},\ \bibinfo {year} {2006})\BibitemShut {NoStop}%
\bibitem [{\citenamefont {Pires}\ \emph {et~al.}(2016)\citenamefont {Pires}, \citenamefont {Cianciaruso}, \citenamefont {C\'eleri}, \citenamefont {Adesso},\ and\ \citenamefont {Soares-Pinto}}]{QFI_Metrics_2016}%
  \BibitemOpen
  \bibfield  {author} {\bibinfo {author} {\bibfnamefont {D.~P.}\ \bibnamefont {Pires}}, \bibinfo {author} {\bibfnamefont {M.}~\bibnamefont {Cianciaruso}}, \bibinfo {author} {\bibfnamefont {L.~C.}\ \bibnamefont {C\'eleri}}, \bibinfo {author} {\bibfnamefont {G.}~\bibnamefont {Adesso}},\ and\ \bibinfo {author} {\bibfnamefont {D.~O.}\ \bibnamefont {Soares-Pinto}},\ }\bibfield  {title} {\bibinfo {title} {Generalized geometric quantum speed limits},\ }\href {https://doi.org/10.1103/PhysRevX.6.021031} {\bibfield  {journal} {\bibinfo  {journal} {Phys. Rev. X}\ }\textbf {\bibinfo {volume} {6}},\ \bibinfo {pages} {021031} (\bibinfo {year} {2016})}\BibitemShut {NoStop}%
\bibitem [{\citenamefont {Marvian}(2022)}]{marvian2022operational}%
  \BibitemOpen
  \bibfield  {author} {\bibinfo {author} {\bibfnamefont {I.}~\bibnamefont {Marvian}},\ }\bibfield  {title} {\bibinfo {title} {Operational interpretation of quantum fisher information in quantum thermodynamics},\ }\href@noop {} {\bibfield  {journal} {\bibinfo  {journal} {Physical Review Letters}\ }\textbf {\bibinfo {volume} {129}},\ \bibinfo {pages} {190502} (\bibinfo {year} {2022})}\BibitemShut {NoStop}%
\bibitem [{\citenamefont {Andersen}\ \emph {et~al.}(2005{\natexlab{b}})\citenamefont {Andersen}, \citenamefont {Filip}, \citenamefont {Fiur\'a\ifmmode~\check{s}\else \v{s}\fi{}ek}, \citenamefont {Josse},\ and\ \citenamefont {Leuchs}}]{BeamSplitterProtocol1}%
  \BibitemOpen
  \bibfield  {author} {\bibinfo {author} {\bibfnamefont {U.~L.}\ \bibnamefont {Andersen}}, \bibinfo {author} {\bibfnamefont {R.}~\bibnamefont {Filip}}, \bibinfo {author} {\bibfnamefont {J.}~\bibnamefont {Fiur\'a\ifmmode~\check{s}\else \v{s}\fi{}ek}}, \bibinfo {author} {\bibfnamefont {V.}~\bibnamefont {Josse}},\ and\ \bibinfo {author} {\bibfnamefont {G.}~\bibnamefont {Leuchs}},\ }\bibfield  {title} {\bibinfo {title} {Experimental purification of coherent states},\ }\href {https://doi.org/10.1103/PhysRevA.72.060301} {\bibfield  {journal} {\bibinfo  {journal} {Phys. Rev. A}\ }\textbf {\bibinfo {volume} {72}},\ \bibinfo {pages} {060301} (\bibinfo {year} {2005}{\natexlab{b}})}\BibitemShut {NoStop}%
\bibitem [{\citenamefont {Marek}\ and\ \citenamefont {Filip}(2007)}]{BeamSplitterProtocol2}%
  \BibitemOpen
  \bibfield  {author} {\bibinfo {author} {\bibfnamefont {P.}~\bibnamefont {Marek}}\ and\ \bibinfo {author} {\bibfnamefont {R.}~\bibnamefont {Filip}},\ }\bibfield  {title} {\bibinfo {title} {Probabilistic purification of noisy coherent states},\ }\href@noop {} {\bibfield  {journal} {\bibinfo  {journal} {Quantum Info. Comput.}\ }\textbf {\bibinfo {volume} {7}},\ \bibinfo {pages} {609–623} (\bibinfo {year} {2007})}\BibitemShut {NoStop}%
\bibitem [{\citenamefont {Zhao}\ and\ \citenamefont {Chiribella}(2017)}]{Chiribella}%
  \BibitemOpen
  \bibfield  {author} {\bibinfo {author} {\bibfnamefont {X.}~\bibnamefont {Zhao}}\ and\ \bibinfo {author} {\bibfnamefont {G.}~\bibnamefont {Chiribella}},\ }\bibfield  {title} {\bibinfo {title} {Quantum amplification and purification of noisy coherent states},\ }\href {https://doi.org/10.1103/PhysRevA.95.042303} {\bibfield  {journal} {\bibinfo  {journal} {Phys. Rev. A}\ }\textbf {\bibinfo {volume} {95}},\ \bibinfo {pages} {042303} (\bibinfo {year} {2017})}\BibitemShut {NoStop}%
\bibitem [{\citenamefont {Barndorff-Nielsen}\ and\ \citenamefont {Gill}(2000)}]{Barndorff-Nielsen_2000}%
  \BibitemOpen
  \bibfield  {author} {\bibinfo {author} {\bibfnamefont {O.~E.}\ \bibnamefont {Barndorff-Nielsen}}\ and\ \bibinfo {author} {\bibfnamefont {R.~D.}\ \bibnamefont {Gill}},\ }\bibfield  {title} {\bibinfo {title} {Fisher information in quantum statistics},\ }\href {https://doi.org/10.1088/0305-4470/33/24/306} {\bibfield  {journal} {\bibinfo  {journal} {Journal of Physics A: Mathematical and General}\ }\textbf {\bibinfo {volume} {33}},\ \bibinfo {pages} {4481} (\bibinfo {year} {2000})}\BibitemShut {NoStop}%
\bibitem [{\citenamefont {Holevo}(2011)}]{holevo2011probabilistic}%
  \BibitemOpen
  \bibfield  {author} {\bibinfo {author} {\bibfnamefont {A.~S.}\ \bibnamefont {Holevo}},\ }\href@noop {} {\emph {\bibinfo {title} {Probabilistic and statistical aspects of quantum theory}}},\ Vol.~\bibinfo {volume} {1}\ (\bibinfo  {publisher} {Springer Science \& Business Media},\ \bibinfo {year} {2011})\BibitemShut {NoStop}%
\bibitem [{\citenamefont {Peres}(1980)}]{PeresClocks}%
  \BibitemOpen
  \bibfield  {author} {\bibinfo {author} {\bibfnamefont {A.}~\bibnamefont {Peres}},\ }\bibfield  {title} {\bibinfo {title} {{Measurement of time by quantum clocks}},\ }\href {https://doi.org/10.1119/1.12061} {\bibfield  {journal} {\bibinfo  {journal} {American Journal of Physics}\ }\textbf {\bibinfo {volume} {48}},\ \bibinfo {pages} {552} (\bibinfo {year} {1980})},\ \Eprint {https://arxiv.org/abs/https://pubs.aip.org/aapt/ajp/article-pdf/48/7/552/8500754/552\_1\_online.pdf} {https://pubs.aip.org/aapt/ajp/article-pdf/48/7/552/8500754/552\_1\_online.pdf} \BibitemShut {NoStop}%
\bibitem [{\citenamefont {Giovannetti}\ \emph {et~al.}(2001)\citenamefont {Giovannetti}, \citenamefont {Lloyd},\ and\ \citenamefont {Maccone}}]{giovannetti2001quantum}%
  \BibitemOpen
  \bibfield  {author} {\bibinfo {author} {\bibfnamefont {V.}~\bibnamefont {Giovannetti}}, \bibinfo {author} {\bibfnamefont {S.}~\bibnamefont {Lloyd}},\ and\ \bibinfo {author} {\bibfnamefont {L.}~\bibnamefont {Maccone}},\ }\bibfield  {title} {\bibinfo {title} {Quantum-enhanced positioning and clock synchronization},\ }\href@noop {} {\bibfield  {journal} {\bibinfo  {journal} {Nature}\ }\textbf {\bibinfo {volume} {412}},\ \bibinfo {pages} {417} (\bibinfo {year} {2001})}\BibitemShut {NoStop}%
\bibitem [{\citenamefont {Bu\ifmmode~\check{z}\else \v{z}\fi{}ek}\ \emph {et~al.}(1999)\citenamefont {Bu\ifmmode~\check{z}\else \v{z}\fi{}ek}, \citenamefont {Derka},\ and\ \citenamefont {Massar}}]{Buzek99}%
  \BibitemOpen
  \bibfield  {author} {\bibinfo {author} {\bibfnamefont {V.}~\bibnamefont {Bu\ifmmode~\check{z}\else \v{z}\fi{}ek}}, \bibinfo {author} {\bibfnamefont {R.}~\bibnamefont {Derka}},\ and\ \bibinfo {author} {\bibfnamefont {S.}~\bibnamefont {Massar}},\ }\bibfield  {title} {\bibinfo {title} {Optimal quantum clocks},\ }\href {https://doi.org/10.1103/PhysRevLett.82.2207} {\bibfield  {journal} {\bibinfo  {journal} {Phys. Rev. Lett.}\ }\textbf {\bibinfo {volume} {82}},\ \bibinfo {pages} {2207} (\bibinfo {year} {1999})}\BibitemShut {NoStop}%
\bibitem [{\citenamefont {Giovannetti}\ \emph {et~al.}(2002)\citenamefont {Giovannetti}, \citenamefont {Lloyd}, \citenamefont {Maccone},\ and\ \citenamefont {Shahriar}}]{GiovannettiClockSync}%
  \BibitemOpen
  \bibfield  {author} {\bibinfo {author} {\bibfnamefont {V.}~\bibnamefont {Giovannetti}}, \bibinfo {author} {\bibfnamefont {S.}~\bibnamefont {Lloyd}}, \bibinfo {author} {\bibfnamefont {L.}~\bibnamefont {Maccone}},\ and\ \bibinfo {author} {\bibfnamefont {M.~S.}\ \bibnamefont {Shahriar}},\ }\bibfield  {title} {\bibinfo {title} {Limits to clock synchronization induced by completely dephasing communication channels},\ }\href {https://doi.org/10.1103/PhysRevA.65.062319} {\bibfield  {journal} {\bibinfo  {journal} {Phys. Rev. A}\ }\textbf {\bibinfo {volume} {65}},\ \bibinfo {pages} {062319} (\bibinfo {year} {2002})}\BibitemShut {NoStop}%
\bibitem [{\citenamefont {Kwon}\ \emph {et~al.}(2018)\citenamefont {Kwon}, \citenamefont {Jeong}, \citenamefont {Jennings}, \citenamefont {Yadin},\ and\ \citenamefont {Kim}}]{KwonClockTradeOff}%
  \BibitemOpen
  \bibfield  {author} {\bibinfo {author} {\bibfnamefont {H.}~\bibnamefont {Kwon}}, \bibinfo {author} {\bibfnamefont {H.}~\bibnamefont {Jeong}}, \bibinfo {author} {\bibfnamefont {D.}~\bibnamefont {Jennings}}, \bibinfo {author} {\bibfnamefont {B.}~\bibnamefont {Yadin}},\ and\ \bibinfo {author} {\bibfnamefont {M.~S.}\ \bibnamefont {Kim}},\ }\bibfield  {title} {\bibinfo {title} {Clock--work trade-off relation for coherence in quantum thermodynamics},\ }\href {https://doi.org/10.1103/PhysRevLett.120.150602} {\bibfield  {journal} {\bibinfo  {journal} {Phys. Rev. Lett.}\ }\textbf {\bibinfo {volume} {120}},\ \bibinfo {pages} {150602} (\bibinfo {year} {2018})}\BibitemShut {NoStop}%
\bibitem [{\citenamefont {Bartlett}\ \emph {et~al.}(2007)\citenamefont {Bartlett}, \citenamefont {Rudolph},\ and\ \citenamefont {Spekkens}}]{BartlettRefFrames}%
  \BibitemOpen
  \bibfield  {author} {\bibinfo {author} {\bibfnamefont {S.~D.}\ \bibnamefont {Bartlett}}, \bibinfo {author} {\bibfnamefont {T.}~\bibnamefont {Rudolph}},\ and\ \bibinfo {author} {\bibfnamefont {R.~W.}\ \bibnamefont {Spekkens}},\ }\bibfield  {title} {\bibinfo {title} {Reference frames, superselection rules, and quantum information},\ }\href {https://doi.org/10.1103/RevModPhys.79.555} {\bibfield  {journal} {\bibinfo  {journal} {Rev. Mod. Phys.}\ }\textbf {\bibinfo {volume} {79}},\ \bibinfo {pages} {555} (\bibinfo {year} {2007})}\BibitemShut {NoStop}%
\bibitem [{\citenamefont {Giovannetti}\ \emph {et~al.}(2006)\citenamefont {Giovannetti}, \citenamefont {Lloyd},\ and\ \citenamefont {Maccone}}]{GiovannettiMetrology}%
  \BibitemOpen
  \bibfield  {author} {\bibinfo {author} {\bibfnamefont {V.}~\bibnamefont {Giovannetti}}, \bibinfo {author} {\bibfnamefont {S.}~\bibnamefont {Lloyd}},\ and\ \bibinfo {author} {\bibfnamefont {L.}~\bibnamefont {Maccone}},\ }\bibfield  {title} {\bibinfo {title} {Quantum metrology},\ }\href {https://doi.org/10.1103/PhysRevLett.96.010401} {\bibfield  {journal} {\bibinfo  {journal} {Phys. Rev. Lett.}\ }\textbf {\bibinfo {volume} {96}},\ \bibinfo {pages} {010401} (\bibinfo {year} {2006})}\BibitemShut {NoStop}%
\bibitem [{\citenamefont {Giovannetti}\ \emph {et~al.}(2011)\citenamefont {Giovannetti}, \citenamefont {Lloyd},\ and\ \citenamefont {Maccone}}]{giovannetti2011advances}%
  \BibitemOpen
  \bibfield  {author} {\bibinfo {author} {\bibfnamefont {V.}~\bibnamefont {Giovannetti}}, \bibinfo {author} {\bibfnamefont {S.}~\bibnamefont {Lloyd}},\ and\ \bibinfo {author} {\bibfnamefont {L.}~\bibnamefont {Maccone}},\ }\bibfield  {title} {\bibinfo {title} {Advances in quantum metrology},\ }\href@noop {} {\bibfield  {journal} {\bibinfo  {journal} {Nature photonics}\ }\textbf {\bibinfo {volume} {5}},\ \bibinfo {pages} {222} (\bibinfo {year} {2011})}\BibitemShut {NoStop}%
\bibitem [{\citenamefont {Chiribella}\ \emph {et~al.}(2013)\citenamefont {Chiribella}, \citenamefont {Yang},\ and\ \citenamefont {Yao}}]{chiribella2013quantum}%
  \BibitemOpen
  \bibfield  {author} {\bibinfo {author} {\bibfnamefont {G.}~\bibnamefont {Chiribella}}, \bibinfo {author} {\bibfnamefont {Y.}~\bibnamefont {Yang}},\ and\ \bibinfo {author} {\bibfnamefont {A.~C.-C.}\ \bibnamefont {Yao}},\ }\bibfield  {title} {\bibinfo {title} {Quantum replication at the heisenberg limit},\ }\href@noop {} {\bibfield  {journal} {\bibinfo  {journal} {Nature communications}\ }\textbf {\bibinfo {volume} {4}},\ \bibinfo {pages} {2915} (\bibinfo {year} {2013})}\BibitemShut {NoStop}%
\bibitem [{\citenamefont {Helstrom}(1969)}]{helstrom1969quantum}%
  \BibitemOpen
  \bibfield  {author} {\bibinfo {author} {\bibfnamefont {C.~W.}\ \bibnamefont {Helstrom}},\ }\bibfield  {title} {\bibinfo {title} {Quantum detection and estimation theory},\ }\href@noop {} {\bibfield  {journal} {\bibinfo  {journal} {Journal of Statistical Physics}\ }\textbf {\bibinfo {volume} {1}},\ \bibinfo {pages} {231} (\bibinfo {year} {1969})}\BibitemShut {NoStop}%
\bibitem [{\citenamefont {Marvian}\ and\ \citenamefont {Spekkens}(2014{\natexlab{c}})}]{marvian2014modes}%
  \BibitemOpen
  \bibfield  {author} {\bibinfo {author} {\bibfnamefont {I.}~\bibnamefont {Marvian}}\ and\ \bibinfo {author} {\bibfnamefont {R.~W.}\ \bibnamefont {Spekkens}},\ }\bibfield  {title} {\bibinfo {title} {Modes of asymmetry: the application of harmonic analysis to symmetric quantum dynamics and quantum reference frames},\ }\href@noop {} {\bibfield  {journal} {\bibinfo  {journal} {Physical Review A}\ }\textbf {\bibinfo {volume} {90}},\ \bibinfo {pages} {062110} (\bibinfo {year} {2014}{\natexlab{c}})}\BibitemShut {NoStop}%
\bibitem [{\citenamefont {Streltsov}\ \emph {et~al.}(2017)\citenamefont {Streltsov}, \citenamefont {Adesso},\ and\ \citenamefont {Plenio}}]{streltsov2017colloquium}%
  \BibitemOpen
  \bibfield  {author} {\bibinfo {author} {\bibfnamefont {A.}~\bibnamefont {Streltsov}}, \bibinfo {author} {\bibfnamefont {G.}~\bibnamefont {Adesso}},\ and\ \bibinfo {author} {\bibfnamefont {M.~B.}\ \bibnamefont {Plenio}},\ }\bibfield  {title} {\bibinfo {title} {Colloquium: Quantum coherence as a resource},\ }\href@noop {} {\bibfield  {journal} {\bibinfo  {journal} {Reviews of Modern Physics}\ }\textbf {\bibinfo {volume} {89}},\ \bibinfo {pages} {041003} (\bibinfo {year} {2017})}\BibitemShut {NoStop}%
\bibitem [{\citenamefont {{Blinnikov, S.}}\ and\ \citenamefont {{Moessner, R.}}(1998)}]{nearlyGaussian}%
  \BibitemOpen
  \bibfield  {author} {\bibinfo {author} {\bibnamefont {{Blinnikov, S.}}}\ and\ \bibinfo {author} {\bibnamefont {{Moessner, R.}}},\ }\bibfield  {title} {\bibinfo {title} {Expansions for nearly gaussian distributions},\ }\href {https://doi.org/10.1051/aas:1998221} {\bibfield  {journal} {\bibinfo  {journal} {Astron. Astrophys. Suppl. Ser.}\ }\textbf {\bibinfo {volume} {130}},\ \bibinfo {pages} {193} (\bibinfo {year} {1998})}\BibitemShut {NoStop}%
\bibitem [{\citenamefont {Katariya}\ and\ \citenamefont {Wilde}(2021{\natexlab{a}})}]{Katariya_2021}%
  \BibitemOpen
  \bibfield  {author} {\bibinfo {author} {\bibfnamefont {V.}~\bibnamefont {Katariya}}\ and\ \bibinfo {author} {\bibfnamefont {M.~M.}\ \bibnamefont {Wilde}},\ }\bibfield  {title} {\bibinfo {title} {Rld fisher information bound for multiparameter estimation of quantum channels},\ }\href {https://doi.org/10.1088/1367-2630/ac1186} {\bibfield  {journal} {\bibinfo  {journal} {New Journal of Physics}\ }\textbf {\bibinfo {volume} {23}},\ \bibinfo {pages} {073040} (\bibinfo {year} {2021}{\natexlab{a}})}\BibitemShut {NoStop}%
\bibitem [{\citenamefont {Katariya}\ and\ \citenamefont {Wilde}(2021{\natexlab{b}})}]{katariya2021geometric}%
  \BibitemOpen
  \bibfield  {author} {\bibinfo {author} {\bibfnamefont {V.}~\bibnamefont {Katariya}}\ and\ \bibinfo {author} {\bibfnamefont {M.~M.}\ \bibnamefont {Wilde}},\ }\bibfield  {title} {\bibinfo {title} {Geometric distinguishability measures limit quantum channel estimation and discrimination},\ }\href@noop {} {\bibfield  {journal} {\bibinfo  {journal} {Quantum Information Processing}\ }\textbf {\bibinfo {volume} {20}},\ \bibinfo {pages} {78} (\bibinfo {year} {2021}{\natexlab{b}})}\BibitemShut {NoStop}%
\bibitem [{\citenamefont {Braunstein}\ and\ \citenamefont {Caves}(1994)}]{BraunsteinGeometry}%
  \BibitemOpen
  \bibfield  {author} {\bibinfo {author} {\bibfnamefont {S.~L.}\ \bibnamefont {Braunstein}}\ and\ \bibinfo {author} {\bibfnamefont {C.~M.}\ \bibnamefont {Caves}},\ }\bibfield  {title} {\bibinfo {title} {Statistical distance and the geometry of quantum states},\ }\href {https://doi.org/10.1103/PhysRevLett.72.3439} {\bibfield  {journal} {\bibinfo  {journal} {Phys. Rev. Lett.}\ }\textbf {\bibinfo {volume} {72}},\ \bibinfo {pages} {3439} (\bibinfo {year} {1994})}\BibitemShut {NoStop}%
\bibitem [{\citenamefont {Morozova}\ and\ \citenamefont {Chentsov}(1991)}]{morozova1991markov}%
  \BibitemOpen
  \bibfield  {author} {\bibinfo {author} {\bibfnamefont {E.~A.}\ \bibnamefont {Morozova}}\ and\ \bibinfo {author} {\bibfnamefont {N.~N.}\ \bibnamefont {Chentsov}},\ }\bibfield  {title} {\bibinfo {title} {Markov invariant geometry on manifolds of states},\ }\href@noop {} {\bibfield  {journal} {\bibinfo  {journal} {Journal of Soviet Mathematics}\ }\textbf {\bibinfo {volume} {56}},\ \bibinfo {pages} {2648} (\bibinfo {year} {1991})}\BibitemShut {NoStop}%
\bibitem [{\citenamefont {Matsumoto}(2005)}]{matsumoto2005reverseestimationtheorycomplementality}%
  \BibitemOpen
  \bibfield  {author} {\bibinfo {author} {\bibfnamefont {K.}~\bibnamefont {Matsumoto}},\ }\href {https://arxiv.org/abs/quant-ph/0511170} {\bibinfo {title} {Reverse estimation theory, complementality between rld and sld, and monotone distances}} (\bibinfo {year} {2005}),\ \Eprint {https://arxiv.org/abs/quant-ph/0511170} {arXiv:quant-ph/0511170 [quant-ph]} \BibitemShut {NoStop}%
\bibitem [{\citenamefont {Paris}(2009)}]{paris2009quantum}%
  \BibitemOpen
  \bibfield  {author} {\bibinfo {author} {\bibfnamefont {M.~G.}\ \bibnamefont {Paris}},\ }\bibfield  {title} {\bibinfo {title} {Quantum estimation for quantum technology},\ }\href@noop {} {\bibfield  {journal} {\bibinfo  {journal} {International Journal of Quantum Information}\ }\textbf {\bibinfo {volume} {7}},\ \bibinfo {pages} {125} (\bibinfo {year} {2009})}\BibitemShut {NoStop}%
\bibitem [{\citenamefont {Gao}\ \emph {et~al.}(2023)\citenamefont {Gao}, \citenamefont {Li}, \citenamefont {Marvian},\ and\ \citenamefont {Rouz{\'e}}}]{gao2023sufficient}%
  \BibitemOpen
  \bibfield  {author} {\bibinfo {author} {\bibfnamefont {L.}~\bibnamefont {Gao}}, \bibinfo {author} {\bibfnamefont {H.}~\bibnamefont {Li}}, \bibinfo {author} {\bibfnamefont {I.}~\bibnamefont {Marvian}},\ and\ \bibinfo {author} {\bibfnamefont {C.}~\bibnamefont {Rouz{\'e}}},\ }\bibfield  {title} {\bibinfo {title} {Sufficient statistic and recoverability via quantum fisher information metrics},\ }\href@noop {} {\bibfield  {journal} {\bibinfo  {journal} {arXiv preprint arXiv:2302.02341}\ } (\bibinfo {year} {2023})}\BibitemShut {NoStop}%
\bibitem [{\citenamefont {{C.H. Bennett, G. Brassard, S. Popescu, B. Schumacher, J.A. Smolin and W.K. Wootters}}(1996)}]{Bennett:96}%
  \BibitemOpen
  \bibfield  {author} {\bibinfo {author} {\bibnamefont {{C.H. Bennett, G. Brassard, S. Popescu, B. Schumacher, J.A. Smolin and W.K. Wootters}}},\ }\bibfield  {title} {\bibinfo {title} {{Purification of Noisy Entanglement and Faithful Teleportation via Noisy Channels}},\ }\href@noop {} {\bibfield  {journal} {\bibinfo  {journal} {Phys. Rev. Lett.}\ }\textbf {\bibinfo {volume} {76}},\ \bibinfo {pages} {722} (\bibinfo {year} {1996})}\BibitemShut {NoStop}%
\bibitem [{\citenamefont {Winter}\ and\ \citenamefont {Yang}(2015)}]{winter2015operational}%
  \BibitemOpen
  \bibfield  {author} {\bibinfo {author} {\bibfnamefont {A.}~\bibnamefont {Winter}}\ and\ \bibinfo {author} {\bibfnamefont {D.}~\bibnamefont {Yang}},\ }\bibfield  {title} {\bibinfo {title} {Operational resource theory of coherence},\ }\href@noop {} {\bibfield  {journal} {\bibinfo  {journal} {arXiv preprint arXiv:1506.07975}\ } (\bibinfo {year} {2015})}\BibitemShut {NoStop}%
\bibitem [{\citenamefont {Regula}\ \emph {et~al.}(2018)\citenamefont {Regula}, \citenamefont {Fang}, \citenamefont {Wang},\ and\ \citenamefont {Adesso}}]{regula2018one}%
  \BibitemOpen
  \bibfield  {author} {\bibinfo {author} {\bibfnamefont {B.}~\bibnamefont {Regula}}, \bibinfo {author} {\bibfnamefont {K.}~\bibnamefont {Fang}}, \bibinfo {author} {\bibfnamefont {X.}~\bibnamefont {Wang}},\ and\ \bibinfo {author} {\bibfnamefont {G.}~\bibnamefont {Adesso}},\ }\bibfield  {title} {\bibinfo {title} {One-shot coherence distillation},\ }\href@noop {} {\bibfield  {journal} {\bibinfo  {journal} {Physical review letters}\ }\textbf {\bibinfo {volume} {121}},\ \bibinfo {pages} {010401} (\bibinfo {year} {2018})}\BibitemShut {NoStop}%
\bibitem [{\citenamefont {Zhao}\ \emph {et~al.}(2018)\citenamefont {Zhao}, \citenamefont {Liu}, \citenamefont {Yuan}, \citenamefont {Chitambar},\ and\ \citenamefont {Winter}}]{zhao2018one}%
  \BibitemOpen
  \bibfield  {author} {\bibinfo {author} {\bibfnamefont {Q.}~\bibnamefont {Zhao}}, \bibinfo {author} {\bibfnamefont {Y.}~\bibnamefont {Liu}}, \bibinfo {author} {\bibfnamefont {X.}~\bibnamefont {Yuan}}, \bibinfo {author} {\bibfnamefont {E.}~\bibnamefont {Chitambar}},\ and\ \bibinfo {author} {\bibfnamefont {A.}~\bibnamefont {Winter}},\ }\bibfield  {title} {\bibinfo {title} {One-shot coherence distillation: The full story},\ }\href@noop {} {\bibfield  {journal} {\bibinfo  {journal} {arXiv preprint arXiv:1808.01885}\ } (\bibinfo {year} {2018})}\BibitemShut {NoStop}%
\bibitem [{\citenamefont {Lami}\ \emph {et~al.}(2018)\citenamefont {Lami}, \citenamefont {Regula},\ and\ \citenamefont {Adesso}}]{lami2018generic}%
  \BibitemOpen
  \bibfield  {author} {\bibinfo {author} {\bibfnamefont {L.}~\bibnamefont {Lami}}, \bibinfo {author} {\bibfnamefont {B.}~\bibnamefont {Regula}},\ and\ \bibinfo {author} {\bibfnamefont {G.}~\bibnamefont {Adesso}},\ }\bibfield  {title} {\bibinfo {title} {Generic bound coherence under strictly incoherent operations},\ }\href@noop {} {\bibfield  {journal} {\bibinfo  {journal} {arXiv preprint arXiv:1809.06880}\ } (\bibinfo {year} {2018})}\BibitemShut {NoStop}%
\bibitem [{\citenamefont {Fang}\ \emph {et~al.}(2018)\citenamefont {Fang}, \citenamefont {Wang}, \citenamefont {Lami}, \citenamefont {Regula},\ and\ \citenamefont {Adesso}}]{ProbDistillation}%
  \BibitemOpen
  \bibfield  {author} {\bibinfo {author} {\bibfnamefont {K.}~\bibnamefont {Fang}}, \bibinfo {author} {\bibfnamefont {X.}~\bibnamefont {Wang}}, \bibinfo {author} {\bibfnamefont {L.}~\bibnamefont {Lami}}, \bibinfo {author} {\bibfnamefont {B.}~\bibnamefont {Regula}},\ and\ \bibinfo {author} {\bibfnamefont {G.}~\bibnamefont {Adesso}},\ }\bibfield  {title} {\bibinfo {title} {Probabilistic distillation of quantum coherence},\ }\href {https://doi.org/10.1103/PhysRevLett.121.070404} {\bibfield  {journal} {\bibinfo  {journal} {Phys. Rev. Lett.}\ }\textbf {\bibinfo {volume} {121}},\ \bibinfo {pages} {070404} (\bibinfo {year} {2018})}\BibitemShut {NoStop}%
\bibitem [{\citenamefont {Regula}\ \emph {et~al.}(2020)\citenamefont {Regula}, \citenamefont {Narasimhachar}, \citenamefont {Buscemi},\ and\ \citenamefont {Gu}}]{CoherenceManipulation}%
  \BibitemOpen
  \bibfield  {author} {\bibinfo {author} {\bibfnamefont {B.}~\bibnamefont {Regula}}, \bibinfo {author} {\bibfnamefont {V.}~\bibnamefont {Narasimhachar}}, \bibinfo {author} {\bibfnamefont {F.}~\bibnamefont {Buscemi}},\ and\ \bibinfo {author} {\bibfnamefont {M.}~\bibnamefont {Gu}},\ }\bibfield  {title} {\bibinfo {title} {Coherence manipulation with dephasing-covariant operations},\ }\href {https://doi.org/10.1103/PhysRevResearch.2.013109} {\bibfield  {journal} {\bibinfo  {journal} {Phys. Rev. Res.}\ }\textbf {\bibinfo {volume} {2}},\ \bibinfo {pages} {013109} (\bibinfo {year} {2020})}\BibitemShut {NoStop}%
\bibitem [{\citenamefont {Lami}\ \emph {et~al.}(2019)\citenamefont {Lami}, \citenamefont {Regula},\ and\ \citenamefont {Adesso}}]{GenericBoundCoherence}%
  \BibitemOpen
  \bibfield  {author} {\bibinfo {author} {\bibfnamefont {L.}~\bibnamefont {Lami}}, \bibinfo {author} {\bibfnamefont {B.}~\bibnamefont {Regula}},\ and\ \bibinfo {author} {\bibfnamefont {G.}~\bibnamefont {Adesso}},\ }\bibfield  {title} {\bibinfo {title} {Generic bound coherence under strictly incoherent operations},\ }\href {https://doi.org/10.1103/PhysRevLett.122.150402} {\bibfield  {journal} {\bibinfo  {journal} {Phys. Rev. Lett.}\ }\textbf {\bibinfo {volume} {122}},\ \bibinfo {pages} {150402} (\bibinfo {year} {2019})}\BibitemShut {NoStop}%
\bibitem [{\citenamefont {Chitambar}\ and\ \citenamefont {Gour}(2019)}]{ChitambarResourceTheories}%
  \BibitemOpen
  \bibfield  {author} {\bibinfo {author} {\bibfnamefont {E.}~\bibnamefont {Chitambar}}\ and\ \bibinfo {author} {\bibfnamefont {G.}~\bibnamefont {Gour}},\ }\bibfield  {title} {\bibinfo {title} {Quantum resource theories},\ }\href {https://doi.org/10.1103/RevModPhys.91.025001} {\bibfield  {journal} {\bibinfo  {journal} {Rev. Mod. Phys.}\ }\textbf {\bibinfo {volume} {91}},\ \bibinfo {pages} {025001} (\bibinfo {year} {2019})}\BibitemShut {NoStop}%
\bibitem [{\citenamefont {Lami}(2020)}]{LamiGrandTour}%
  \BibitemOpen
  \bibfield  {author} {\bibinfo {author} {\bibfnamefont {L.}~\bibnamefont {Lami}},\ }\bibfield  {title} {\bibinfo {title} {Completing the grand tour of asymptotic quantum coherence manipulation},\ }\href {https://doi.org/10.1109/TIT.2019.2945798} {\bibfield  {journal} {\bibinfo  {journal} {IEEE Transactions on Information Theory}\ }\textbf {\bibinfo {volume} {66}},\ \bibinfo {pages} {2165} (\bibinfo {year} {2020})}\BibitemShut {NoStop}%
\bibitem [{\citenamefont {Tajima}\ \emph {et~al.}(2022)\citenamefont {Tajima}, \citenamefont {Takagi},\ and\ \citenamefont {Kuramochi}}]{tajima2022universal}%
  \BibitemOpen
  \bibfield  {author} {\bibinfo {author} {\bibfnamefont {H.}~\bibnamefont {Tajima}}, \bibinfo {author} {\bibfnamefont {R.}~\bibnamefont {Takagi}},\ and\ \bibinfo {author} {\bibfnamefont {Y.}~\bibnamefont {Kuramochi}},\ }\bibfield  {title} {\bibinfo {title} {Universal trade-off structure between symmetry, irreversibility, and quantum coherence in quantum processes},\ }\href@noop {} {\bibfield  {journal} {\bibinfo  {journal} {arXiv preprint arXiv:2206.11086}\ } (\bibinfo {year} {2022})}\BibitemShut {NoStop}%
\bibitem [{\citenamefont {Yamaguchi}\ and\ \citenamefont {Tajima}(2023)}]{yamaguchi2023beyond}%
  \BibitemOpen
  \bibfield  {author} {\bibinfo {author} {\bibfnamefont {K.}~\bibnamefont {Yamaguchi}}\ and\ \bibinfo {author} {\bibfnamefont {H.}~\bibnamefont {Tajima}},\ }\bibfield  {title} {\bibinfo {title} {Beyond iid in the resource theory of asymmetry: An information-spectrum approach for quantum fisher information},\ }\href@noop {} {\bibfield  {journal} {\bibinfo  {journal} {Physical Review Letters}\ }\textbf {\bibinfo {volume} {131}},\ \bibinfo {pages} {200203} (\bibinfo {year} {2023})}\BibitemShut {NoStop}%
\bibitem [{\citenamefont {Tajima}\ and\ \citenamefont {Takagi}(2024)}]{tajima2024gibbs}%
  \BibitemOpen
  \bibfield  {author} {\bibinfo {author} {\bibfnamefont {H.}~\bibnamefont {Tajima}}\ and\ \bibinfo {author} {\bibfnamefont {R.}~\bibnamefont {Takagi}},\ }\bibfield  {title} {\bibinfo {title} {Gibbs-preserving operations requiring infinite amount of quantum coherence},\ }\href@noop {} {\bibfield  {journal} {\bibinfo  {journal} {arXiv preprint arXiv:2404.03479}\ } (\bibinfo {year} {2024})}\BibitemShut {NoStop}%
\bibitem [{\citenamefont {Takagi}\ and\ \citenamefont {Shiraishi}(2022)}]{takagi2022correlation}%
  \BibitemOpen
  \bibfield  {author} {\bibinfo {author} {\bibfnamefont {R.}~\bibnamefont {Takagi}}\ and\ \bibinfo {author} {\bibfnamefont {N.}~\bibnamefont {Shiraishi}},\ }\bibfield  {title} {\bibinfo {title} {Correlation in catalysts enables arbitrary manipulation of quantum coherence},\ }\href@noop {} {\bibfield  {journal} {\bibinfo  {journal} {Physical Review Letters}\ }\textbf {\bibinfo {volume} {128}},\ \bibinfo {pages} {240501} (\bibinfo {year} {2022})}\BibitemShut {NoStop}%
\bibitem [{\citenamefont {Serafini}(2017)}]{serafini2017quantum}%
  \BibitemOpen
  \bibfield  {author} {\bibinfo {author} {\bibfnamefont {A.}~\bibnamefont {Serafini}},\ }\href@noop {} {\emph {\bibinfo {title} {Quantum continuous variables: a primer of theoretical methods}}}\ (\bibinfo  {publisher} {CRC press},\ \bibinfo {year} {2017})\BibitemShut {NoStop}%
\bibitem [{\citenamefont {Cohen}(1966)}]{PDistribution}%
  \BibitemOpen
  \bibfield  {author} {\bibinfo {author} {\bibfnamefont {L.}~\bibnamefont {Cohen}},\ }\bibfield  {title} {\bibinfo {title} {{Generalized Phase-Space Distribution Functions}},\ }\href {https://doi.org/10.1063/1.1931206} {\bibfield  {journal} {\bibinfo  {journal} {Journal of Mathematical Physics}\ }\textbf {\bibinfo {volume} {7}},\ \bibinfo {pages} {781} (\bibinfo {year} {1966})},\ \Eprint {https://arxiv.org/abs/https://pubs.aip.org/aip/jmp/article-pdf/7/5/781/19200984/781\_1\_online.pdf} {https://pubs.aip.org/aip/jmp/article-pdf/7/5/781/19200984/781\_1\_online.pdf} \BibitemShut {NoStop}%
\bibitem [{\citenamefont {Yadavalli}(2025)}]{CoherentThermalState2025}%
  \BibitemOpen
  \bibfield  {author} {\bibinfo {author} {\bibfnamefont {S.~A.}\ \bibnamefont {Yadavalli}},\ }\href@noop {} {\bibinfo {title} {Coherent thermal state mathematica code}},\ \bibinfo {howpublished} {\url{https://github.com/shivakshary/CoherentThermalState_Mathematica}} (\bibinfo {year} {2025})\BibitemShut {NoStop}%
\bibitem [{{\relax DLMF}()}]{NIST:DLMF}%
  \BibitemOpen
  {\relax DLMF},\ \href {https://dlmf.nist.gov/} {\bibinfo {title} {{\it NIST Digital Library of Mathematical Functions}}},\ \bibinfo {howpublished} {\url{https://dlmf.nist.gov/}, Release 1.2.0 of 2024-03-15},\ \bibinfo {note} {f.~W.~J. Olver, A.~B. {Olde Daalhuis}, D.~W. Lozier, B.~I. Schneider, R.~F. Boisvert, C.~W. Clark, B.~R. Miller, B.~V. Saunders, H.~S. Cohl, and M.~A. McClain, eds.}\BibitemShut {Stop}%
\bibitem [{\citenamefont {Evans}(2018)}]{evans2018measure}%
  \BibitemOpen
  \bibfield  {author} {\bibinfo {author} {\bibfnamefont {L.}~\bibnamefont {Evans}},\ }\href@noop {} {\emph {\bibinfo {title} {Measure theory and fine properties of functions}}}\ (\bibinfo  {publisher} {Routledge},\ \bibinfo {year} {2018})\BibitemShut {NoStop}%
\bibitem [{\citenamefont {Magnus}\ \emph {et~al.}(1967)\citenamefont {Magnus}, \citenamefont {Oberhettinger}, \citenamefont {Soni},\ and\ \citenamefont {Wigner}}]{SpecialFunctionsTextbook}%
  \BibitemOpen
  \bibfield  {author} {\bibinfo {author} {\bibfnamefont {W.}~\bibnamefont {Magnus}}, \bibinfo {author} {\bibfnamefont {F.}~\bibnamefont {Oberhettinger}}, \bibinfo {author} {\bibfnamefont {R.~P.}\ \bibnamefont {Soni}},\ and\ \bibinfo {author} {\bibfnamefont {E.~P.}\ \bibnamefont {Wigner}},\ }\bibfield  {title} {\bibinfo {title} {{Formulas and Theorems for the Special Functions of Mathematical Physics}},\ }\href {https://doi.org/10.1063/1.3034082} {\bibfield  {journal} {\bibinfo  {journal} {Physics Today}\ }\textbf {\bibinfo {volume} {20}},\ \bibinfo {pages} {81} (\bibinfo {year} {1967})},\ \Eprint {https://arxiv.org/abs/https://pubs.aip.org/physicstoday/article-pdf/20/12/81/8266873/81\_1\_online.pdf} {https://pubs.aip.org/physicstoday/article-pdf/20/12/81/8266873/81\_1\_online.pdf} \BibitemShut {NoStop}%
\bibitem [{\citenamefont {Barnett}\ \emph {et~al.}(2018)\citenamefont {Barnett}, \citenamefont {Ferenczi}, \citenamefont {Gilson},\ and\ \citenamefont {Speirits}}]{BarnettMGF}%
  \BibitemOpen
  \bibfield  {author} {\bibinfo {author} {\bibfnamefont {S.~M.}\ \bibnamefont {Barnett}}, \bibinfo {author} {\bibfnamefont {G.}~\bibnamefont {Ferenczi}}, \bibinfo {author} {\bibfnamefont {C.~R.}\ \bibnamefont {Gilson}},\ and\ \bibinfo {author} {\bibfnamefont {F.~C.}\ \bibnamefont {Speirits}},\ }\bibfield  {title} {\bibinfo {title} {Statistics of photon-subtracted and photon-added states},\ }\href {https://doi.org/10.1103/PhysRevA.98.013809} {\bibfield  {journal} {\bibinfo  {journal} {Phys. Rev. A}\ }\textbf {\bibinfo {volume} {98}},\ \bibinfo {pages} {013809} (\bibinfo {year} {2018})}\BibitemShut {NoStop}%
\end{thebibliography}%

\onecolumngrid
%\appendix

\newpage

\maketitle
\vspace{-5in}
\begin{center}

\Large{Appendix}
\end{center}
\appendix

\onecolumngrid

%$|l-\gamma_\text{in}^2|\leq R\sigma$, for any  $R\ll \gamma_\text{in}$, i.e., $R\in o(\gamma_\text{in})$, we have

\section{Coherence distillation using Gaussian channels}\label{appendix: gaussian prop}
{In this section, we prove Proposition \ref{prop: gaussian}. We prove the result of the proposition by computing the optimal distillation error for general phase-insensitive Gaussian channels in Appendix~\ref{appendix: phase insenstive gaussian channels}. We explore the general state transformations possible for coherent thermal states using passive unitaries in Appendix~\ref{appendix: general transformations passive}. Using this, in Appendix~\ref{appendix: vacuum distillation} we show that beam splitters and a single vacuum ancilla suffices to achieve the optimal distillation error possible for phase-insensitive Gaussian channels.}

\subsection{Preliminaries}
In all the above mentioned cases, we will be considering $N$ bosonic modes with the Hamiltonian
\be
{H} = \sum_{i=1}^N \omega \ {a}^\dag_i{a}_i \ .
\ee
As always,
\be
 {\rho}_\text{th}(\beta):=\frac{\exp({-\omega \beta  {a}^\dag  {a}})}{\Tr(\exp({- \omega \beta  {a}^\dag  {a}}))}\ 
\ee
is the thermal state of a Harmonic Oscillator,
\be
 {D}(\alpha)=\exp(\alpha  {a}^\dag-\alpha^\ast  {a})\ 
\ee
is the Weyl displacement operator, and
\begin{align}
   \rho(\beta,\alpha):=D(\alpha)\rho_{\text{th}}(\beta)D^\dagger(\alpha)
\end{align}
is the coherent thermal state. Finally, the fidelity of a coherent thermal state $\rho(\beta', \alpha')$ with the desired coherent state $|\alpha\rangle$ is a quantity we will be interested in. We compute it here:
\bes\label{fidelity}
\begin{align}
\braket{\alpha|\rho(\beta', \alpha')|\alpha}&=\langle\alpha|  {D}(\alpha')  {\rho}_\text{th}(\beta') {D}^\dag(\alpha')|\alpha\rangle\\
&=\langle\alpha-\alpha'|   {\rho}_\text{th}(\beta')|\alpha-\alpha'\rangle\\ &= (1-e^{-\beta' })e^{-|\alpha-\alpha'|^2} \sum_{s=0}^\infty \frac{}{} e^{-\beta'  s} \frac{|\alpha-\alpha'|^{2s}}{s!} \\ &= (1-e^{-\beta' })e^{-|\alpha-\alpha'|^2} \sum_{s=0}^\infty   \frac{\big(e^{-\beta'  } |\alpha-\alpha'|^{2}\big)^s}{s!}\\ &= (1-e^{-\beta' })  \exp[-{|\alpha-\alpha'|^2} \times (1-e^{-\beta'  })]\\ &=  \frac{1}{Z'} \times \exp[-\frac{{|\alpha-\alpha'|^2}}{Z'}]\ ,
\end{align}
\ees
where $Z'={(1-e^{-\omega\beta' })^{-1}}=1+ n_{\text{th}}(\beta')$ is the partition function of a Harmonic oscillator in the inverse temperature  $\beta'$, and $n_{\text{th}}(\beta')$ is the mean number of thermal excitations in the Harmonic oscillator.

\newpage

\subsection{Distillation with phase-insensitive Gaussian channels (Proof of Proposition \ref{prop: gaussian})}\label{appendix: phase insenstive gaussian channels}

In this section, we prove Proposition \ref{prop: gaussian}. In Sec.~\ref{subsec: concentration dilution reversibility}, we note that the passive \textit{concentration channel} that processes $n$ copies of a coherent thermal state $\rho(\beta,\alpha)^{\otimes n}$ to $\rho(\beta,\sqrt{n}\alpha)$ is reversible. Thus, instead of considering a channel from $n$ modes, we can study phase-insensitive Gaussian channels from one mode to one mode where the input is $\rho(\beta,\sqrt{n}\alpha)$ (doing so does not affect optimality).

Any phase-insensitive Gaussian channel affects a Gaussian state with covariance matrix $\sigma$ and phase vector $\mathbf{r}$ as
\begin{align}
    &\sigma \mapsto x^2\sigma + y I\\
    & \mathbf{r} \mapsto x\mathbf{r}
\end{align}
where $y\geq|x^2-1|$ \cite{serafini2017quantum}. The effect of such a channel on coherent thermal states is easy to analyze. The covariance matrix of $\rho(\beta, \alpha)$ is $(1+2 n_{\text{th}}(\beta))I$. Thus, applying the phase-insensitive Gaussian channel characterized by $x,y$ yields the covariance matrix, 
\begin{align}
 (1+2 n_{\text{th}}(\beta))I \mapsto x^2 (1+2 n_{\text{th}}(\beta))I + y I = (x^2(1+2 n_{\text{th}}(\beta))+y)I.
\end{align}
Thus, the new temperature satisfies,
\begin{align}
    1+2n_{\text{th}}(\beta') = x^2(1+2 n_{\text{th}}(\beta))+y \implies n_{\text{th}}(\beta') = x^2 n_{\text{th}}(\beta) +  \frac{(x^2+y-1)}{2}
\end{align}
In other words, a general phase-insensitive Gaussian channel maps
\begin{align}
    \rho(\beta, \alpha) \mapsto \rho(\beta', x \alpha)
\end{align}
where $n_{\text{th}}(\beta') = x^2 n_{\text{th}}(\beta) +  \frac{(x^2+y-1)}{2}$ and $y\geq|x^2-1|$. 
So given the input $\rho(\beta, \sqrt{n} \alpha)$, the output is $\rho(\beta', x \sqrt{n} \alpha)$. From Eq.(\ref{fidelity}), the fidelity of this output state with $\ket{\alpha}$ is 
\begin{align}
    \braket{\alpha|\rho(\beta', x_n \sqrt{n} \alpha)|\alpha} = \frac{1}{Z_n'} \times \exp\Big[-\frac{|x_n \sqrt{n} \alpha-\alpha|^2}{Z_n'}\Big]
\end{align}
where $Z'_n=1+n_\text{th}(\beta'_n) = 1+ x_n^2 n_{\text{th}}(\beta) +  \frac{(x_n^2+y_n-1)}{2}$, and we have added the subscript $n$ to the parameters $x$, $y$ and $Z'$ for clarity. As $n\rightarrow \infty$, the distillation fidelity approaches $1$. Thus, $x_n \sqrt{n} \alpha$ approaches $\alpha$. To this end we shall define the phase-space \textit{bias} as
\begin{align}\label{CV 1}
    \delta_n := x_n\sqrt{n} \alpha - \alpha.
\end{align}

Then the fidelity can be written as
\begin{align}\label{eq infid 2}
    \braket{\alpha|\rho(\beta', x_n \sqrt{n} \alpha)|\alpha} = \frac{e^{-\frac{|\delta_n|^2}{Z'_n}}}{Z'_n}
\end{align}
where $Z'_n$ becomes $$Z'_n= \frac{1}{2}+ \frac{y_n}{2} + \Big(n_\text{th}(\beta) + \frac{1}{2}\Big)\Big(1 + \frac{\delta_n}{\alpha}\Big)^2 \times \frac{1}{n}.$$
Now, we shall upperbound the fidelity. Clearly,
\begin{align}
    \braket{\alpha|\rho(\beta', x_n \sqrt{n} \alpha)|\alpha} \leq \frac{1}{Z'_n} \ .
\end{align}
From this expression, it is apparent that minimizing $y_n$ will further upperbound the fidelity. The minimum value that $y_n$ can take is 
\begin{align}
    y_n \geq |x_n^2-1| = 1 - \Big(1 + \frac{\delta_n}{\alpha}\Big)^2 \times \frac{1}{n} \equiv y^\text{min}_n \ ,
\end{align}
where we note that $\big(1 + \frac{\delta_n}{\alpha}\big)^2 \times \frac{1}{n} \ll 1$ when $n \gg 1$. Substituting $y_n=y_n^\text{min}$ in $Z'_n$ yields
\begin{align}
    {Z_n'}^\text{min} = 1 + \Big( 1 + \frac{\delta_n}{\alpha}\Big)^2 \times \frac{n_\text{th}(\beta)}{n} \ .
\end{align}
Thus,
\begin{align}
    \braket{\alpha|\rho(\beta', x_n \sqrt{n} \alpha)|\alpha} \leq \frac{1}{Z'_n} \leq \frac{1}{{Z_n'}^\text{min}} = \frac{1}{1 + \big(1 + \frac{\delta_n}{\alpha}\big)^2 \times \frac{n_\text{th}(\beta)}{n}}\ .
\end{align}
Because $\delta_n \rightarrow 0$ as $n \rightarrow \infty$, the the leading order behavior of ${\big(1 + \frac{\delta_n}{\alpha}\big)^2 \times \frac{n_\text{th}(\beta)}{n}}$ in $n$ will be $n_\text{th}(\beta)/n$, in the large $n\gg 1$ regime. As a consequence,
\begin{align}
     \limsup_{n \rightarrow \infty}\  n \times (1-\braket{\alpha|\rho(\beta', x_n \sqrt{n} \alpha)|\alpha}) \geq n_\text{th}(\beta)\ .
\end{align}

This lowerbound is saturated by a simple linear optical protocol comprised of passive unitaries, illustrated in Figure \ref{first real}. Thus, the optimal infidelity factor for Gaussian phase-insensitive channels is $\delta^{\text{opt-Gaussian}}=n_\text{th}(\beta)$. We detail the state transformations possible with passive unitaries in the following. 

\subsection{General State Transitions with passive unitaries}\label{appendix: general transformations passive}
In this subsection, we prove the general transitions possible for coherent thermal states using passive unitaries.  A notable two-mode passive unitary is a beam splitter, which is characterized by the unitary matrix
\be\label{beam splitter bogo}
\left(
\begin{array}{c}
b_1   \\
b_2
\end{array}
\right) =\left(
\begin{array}{cc}
t  &  r  \\
-r^\ast  & t^\ast
\end{array}
\right)\left(
\begin{array}{c}
a_1   \\
a_2
\end{array}
\right)\ ,
\ee
where $t$ and $r$ are respectively the transmission and reflection coefficients of the beam splitter satisfying $|t|^2+|r|^2=1$. Recall that any such transformation can be realized by a sequence of beam splitters and phase shifters.

\begin{proposition}\label{prop1}
Suppose we apply a passive unitary on the initial state {$ {\rho}=\bigotimes_{i=1}^N  {D}(\alpha_i) \rho_\text{th}(\beta_i)  {D}^\dag(\alpha_i)$}.  Then, the output state of each output mode is in the form $\mathcal{D}(\alpha') {\rho}_\text{th}(\beta')\mathcal{D}(\alpha')^\dag$ for $\alpha'\in\mathbb{C}$, with 
\begin{align}
\alpha'=\sum_{j=1}^N c_j \alpha_j\ ,
\end{align}
and 
\begin{align}
n_\text{th}(\beta')=\sum_{j=1}^N |c_j|^2\times n_\text{th}(\beta_j)\ ,
\end{align}
where $c_j\in \mathbb{C}$ are arbitrary coefficients satisfying the normalization condition $\sum_{j=1}^N|c_j|^2=1$, and $n_{\text{th}}(\beta)=e^{-\beta \omega}(1-e^{-\beta \omega})^{-1}$ is the mean number of thermal excitations for a Harmonic Oscillator in the inverse temperature $\beta$.  Hence, its fidelity with a desired state $|\alpha\rangle$ is $\exp(-|\alpha-\alpha'|^2/Z')/Z'$, where $Z'=(1-e^{-\beta' \omega })^{-1}$.
\end{proposition}
This proposition means that $\alpha'\in\mathbb{C}$ can be any complex number satisfying
\be
|\alpha'|\le \sqrt{\sum_{j=1}^N |\alpha_j|^2}\ , 
\ee
which follows from the Cauchy–Schwarz inequality.
Similarly $\beta'$ can take any value in the range
\be
\beta_{\text{min}} \le \beta'\le \beta_{\text{max}}\ ,
\ee
where $\beta_{\text{min}}=\text{min}\{\beta_j\}$ and $\beta_{\text{max}}=\text{max}\{\beta_j\}$.\\

Note that Eq.(\ref{eq: bs alpha}-\ref{eq: bs n}) follows as a direct application of Proposition \ref{prop1} to the unitary matrix of the beam splitter in Eq.(\ref{beam splitter bogo}).

\subsection*{Proof of Proposition \ref{prop1}}\label{proof passive}
In the following, we prove the above proposition. The output modes are related to the input modes via 
\be
 {b}^\dag_j= {U}^\dag_{\text{BS}}  {a}^\dag_j  {U}_{\text{BS}}=  \sum_{i=1}^N {u}_{ji}\   {a}^\dag_i\ ,
\ee
where ${u}_{ji}: 1\le i,j\le N$ is a unitary matrix. This implies 
\be
 {U}_{\text{BS}} {a}^\dag_i  {U}^\dag_{\text{BS}}=\sum_{r=1}^N u^\ast_{ri}\   {a}^\dag_r \ ,
\ee
and therefore
\be
 {U}_{\text{BS}} \big[\bigotimes_{i=1}^N  {D}(\alpha_i) \big]  {U}^\dag_{\text{BS}}= \bigotimes_{j=1}^N  {D}(\sum_{i} {u}^\ast_{ji} \alpha_i)\ .
\ee
Therefore,
\be
{U}_{\text{BS}}\big[\bigotimes_{i=1}^N  {D}(\alpha_i) \rho_\text{th}(\beta_i)  {D}^\dag(\alpha_i)\big] {U}_{\text{BS}}^\dag =
\big[\bigotimes_{j=1}^N  {D}(\sum_{i=1}^N u^\ast_{ji} \alpha_i)\big]  {\tau}\big[\bigotimes_{j=1}^N  {D}(\sum_{i=1}^N u^\ast_{ji} \alpha_i)\big]^\dag\ ,
\ee
where 
\be
 {\tau}= {U}_{\text{BS}} \big[\bigotimes_{i=1}^N  \rho_\text{th}(\beta_i) \big] {U}^\dag_{\text{BS}}=  {U}_{\text{BS}}   {\rho}  {U}_{\text{BS}}^\dagger\ . 
\ee
By doing so, we have effectively separated the effect of the passive unitary on the displacement from its effect on the thermal state. To determine $ {\tau}$ we first note that it is a Gaussian state and that all its first moments vanish, i.e., 
\be
\Tr( {\tau}  {a}_j)= \Tr( {\rho}  {b}_j) =  \sum_{i} {u}_{ji} \Tr( {\rho}  {a}_i)=0\ .
\ee
because $\rho$ is an $N$-fold tensor of thermal states. Furthermore, 
\be
\Tr( {a}^\dag_{j'} {a}^\dag_j {\tau})=0\ .
\ee
Note that both of the above equations also follow from  the fact that the unitary $ {U}_\text{BS}$ is a passive unitary and the initial state $ {\rho}$ commutes with the total Hamiltonian  $\omega \sum_{j=1}^N  {a}^\dag_{j} {a}_j$ (indeed this initial state commutes with each individual Hamiltonian $\omega {a}^\dag_{j} {a}_j$).  Thus, $\tau$ will be an incoherent Gaussian state on $N$ modes, and more importantly the reduced state on each mode $k$ will be a thermal state. To determine the temperature of each mode we note that
\be
 {b}^\dag_{j'} {b}_j=\sum_{i,i'} {u}_{j'i'} {u}^\ast_{ji}\    {a}^\dag_{i'}  {a}_i\ ,
\ee
which implies
\be
\Tr( {a}^\dag_{j'} {a}_j {\tau})=\Tr( {b}^\dag_{j'} {b}_j {\rho})=\sum_{i,i'} {u}_{j'i'} {u}^\ast_{ji} \times \Tr( {a}^\dag_{i'}  {a}_i  {\rho})=\sum_{i} {u}_{j'i} {u}^\ast_{ji} \times \Tr( {a}^\dag_{i}  {a}_i  {\rho}) \   .
\ee
where we used the unitarity of $U_{\text{BS}}$ in the last equality. This, in particular, implies that the reduced state of each mode $k\in\{1,\cdots, n\}$  is a thermal state $ {\rho}_\text{th}(\beta'_k)$ with 

\be
n_{\text{th}}(\beta'_k)=\Tr( {a}^\dag_{k} {a}_k {\tau})=\sum_{j=1}^N |{u}_{kj}|^2 \times \Tr( {a}^\dag_{j} {a}_j {\rho})=\sum_{j=1}^N |{u}_{kj}|^2 \times n_{\text{th}}(\beta_j) \   ,
\ee
where 
\be
n_{\text{th}}(\beta_j)=\Tr( {a}^\dag_{j} {a}_j  {\rho}_\text{th}(\beta_j) )=\frac{e^{-\beta_j \omega}}{1- e^{-\beta_j \omega}}\ .
\ee

Fidelity of the output state in proposition \ref{prop1} follows from Eq.(\ref{fidelity}).

\subsection{Coherence distillation with beam splitters}\label{appendix: vacuum distillation}
Now we look at how this optimal performance can be achieved using only beam splitters. Suppose we are given $n$ copies of noisy version of this state, 
\be
 {\rho}= \bigotimes_{i=1}^n {D}(\alpha)  {\rho}_\text{th}(\beta)  {D}^\dag(\alpha)\ .
\ee
Then, from Proposition \ref{prop1} we choose  $N=n+1$, i.e., we consider an additional mode in the vacuum state, which corresponds to $\beta=\infty$, or, equivalently $n_\text{th}(\beta)=0$. We label this mode as mode $j=N$. Then, choosing 
\begin{align}
c_{i}&=\frac{\cos \theta}{\sqrt{n}}\ \ \  : i=1,\cdots, n\\ 
c_N&=\sin\theta\ .
\end{align}
Then, we can obtain
any state with 
\begin{align}
\alpha'&=\cos \theta \sqrt{n}\times \alpha \ \ ,\ \ 
n_\text{th}(\beta')=\cos^2 \theta \times n_\text{th}(\beta)\ ,
\end{align}
where
\be
n_\text{th}(\beta)=\frac{e^{-\beta \omega}}{1- e^{-\beta \omega}}\ , 
\ee
is the expected number of excitations. Therefore, choosing $\cos\theta=1/\sqrt{n}$ we find $n_\text{th}(\beta')=n_\text{th}(\beta)/n$, i.e., we can suppress the expected number of excitations by a factor of $n$, and achieve  $\alpha'=\alpha$.

It follows that the fidelity of this state with the desired state $|\alpha\rangle$ from Eq.(\ref{fidelity}) is
  
\be
\frac{1}{1+n_\text{th}(\beta)/n}= 1-\frac{n_\text{th}(\beta)}{n}+\mathcal{O}\Big(\frac{1}{n^2}\Big)\ .
\ee
Thus, the infidelity factor is the optimal $\delta(\beta)=n_\text{th}(\beta)$ among Gaussian phase-insensitive channels.

\subsection*{Details regarding beam splitter circuits in Figure \ref{first real}}

Fig.~\ref{first real} in the main body illustrates two explicit passive unitary circuits of beam splitters and a single vacuum ancilla that achieve the same optimal infidelity factor. The aforementioned protocol corresponds to the \textit{First Realization} in the figure. 

The performance of the \textit{Second Realization} is a simple consequence of Eq.(\ref{eq: bs alpha}-\ref{eq: bs n}): In the $j^\text{th}$ use of the beam splitter in the circuit ($j=1, \cdots, n$), we realize the transformation
$$\rho(\beta_{j-1}, \sqrt{\frac{j-1}{n}}\ \alpha) \otimes \rho(\beta, \alpha)\longrightarrow \rho(\beta_j, \sqrt{\frac{j}{n}}\ \alpha),$$ where $$n_\text{th}(\beta_j)=(1-\frac{1}{n^2})\times n_\text{th}(\beta_{j-1})+\frac{1}{n^2} \times n_\text{th}(\beta)\ =[1-(1-\frac{1}{n^2})^{j+1}]\times n_\text{th}(\beta)$$ 
if we assume that the initial ancilla is the vacuum ($\beta_0=\infty$). The overall effect of each weak interaction on this mode can be approximated as a unitary transformation that realizes the Weyl displacement by $\alpha/\sqrt{n}$ up to $\mathcal{O}(1/n^2)$ corrections, along with a temperature increase of $n_\text{th}(\beta)/n^2$ up to $\mathcal{O}(1/n^{3})$ corrections. Then, for $n \gg 1$ the error becomes $\epsilon_n={n_\text{th}(\beta)}/{n}+\mathcal{O}(n^{-2})$.

\newpage

\section{Divide and Distill Strategy (Proof of Lemma \ref{lemma: divide and distill})}\label{appendix: lemma divide distill proof}

Consider the energy distribution of an arbitrary single-mode state $\tau$ with respect to the Hamiltonian $H=a^\dagger a$, i.e., $\braket{l|\tau|l}$, where $\ket{l}$ are the Fock states. Then, applying the  Markov's inequality to this  energy distribution,  we find that
\begin{align}
    1- \braket{0|\tau|0} \leq \Tr(\tau a^\dagger a)\ .
\end{align}
Now we apply this inequality to $\tau= D^\dagger(\alpha') \sigma' D(\alpha')$, which implies
\be
1-\langle 0|\tau| 0\rangle=  1-\langle \alpha'|\sigma'|\alpha'\rangle \le  \Tr(\sigma' (a-\alpha')^\dagger (a-\alpha')) \ ,
\ee
where we have used  $D(\alpha')\ket{0}=\ket{\alpha'}$, and  $D(\alpha')a D^\dagger(\alpha')= a - \alpha'$. We conclude that
\be\label{infid ineq}
1-\langle \alpha'|\sigma'|\alpha'\rangle \le \Tr(\sigma' a^\dagger a) - |\Tr(\sigma' a)|^2 + |\Tr(\sigma' a)-\alpha'|^2\ ,
\ee
as stated in Lemma \ref{lemma: divide and distill}.  We can interpret $\Tr(\sigma' a^\dagger a) - |\Tr(\sigma' a)|^2$ as the variance, and $|\Tr(\sigma' a)-\alpha'|^2$ as the squared bias. These transform simply upon the use of concentration channels. 

Now suppose $\sigma'$ is obtained 
by applying the concentration channel on $m$ copies of $\sigma$,
as $\mathcal{C}_m(\sigma^{\otimes m})=\sigma'$.  
 Recall that the concentration channel is realized by a passive linear optical transformation. For such a transformation, 
 the output annihilation operator $b$ is related to the input modes  $a_j: j=1,\cdots, m$, via a linear transformation $b=\sum_j r_j a_j$, for a normalized complex vector $(r_1,\cdots, r_m)^T\in \mathbb{C}^m$. 
 Assuming the $m$ input modes are initially uncorrelated 
 in the joint state $\rho_\text{tot}=\bigotimes_{j=1}^m \rho_j$, this implies that for the output we have
\begin{align}
&\Tr(\rho_\text{tot} b)=\sum_{j=1}^m  r_j\Tr(\rho_j a_j) \label{m1 eqn lemma}\\
&\Tr(\rho_{\text{tot}} b^\dagger b)= \sum_{i,j}^m r_i^{*} r_j \Tr(\rho_{\text{tot}} a_i^\dagger a_j)
   = \sum_{i}^m |r_i|^2 \Tr(\rho_i a_i^\dagger a_i) + \sum_{i\neq j}^m r_i^{*} r_j \Tr(\rho_i a_i)^* \Tr(\rho_j a_j)\\ 
   &|\Tr(\rho_{\text{tot}} b)|^2 = \Big| \sum_{j=1}^m r_j \Tr(\rho_j a_j)\Big|^2= \sum_{j}^m |r_j|^2 |\Tr(\rho_j a_j)|^2 + \sum_{i\neq j}^m r_i^{*} r_j \Tr(\rho_i a_i)^* \Tr(\rho_j a_j)\ , 
\end{align}
which imply 
\be \label{m2 eqn lemma}
\Tr(\rho_\text{tot} b^\dag b)-|\Tr(\rho_\text{tot} b)|^2=\sum_{j=1}^m |r_j|^2 \big[\Tr(\rho_j a_j^\dag a_j)-|\Tr(\rho a_j)|^2\big]\ .
\ee
For the concentration map $\mathcal{C}_m$ we have $r_j=1/\sqrt{m}$ for all $j$. Then the RHS of Eq.(\ref{infid ineq}) can be further simplified to
\begin{align}
1-\langle \alpha'|\sigma'|\alpha'\rangle &\le \Tr(\sigma' a^\dagger a) - |\Tr(\sigma' a)|^2 + |\Tr(\sigma' a)-\alpha'|^2\ \\ & = \Tr(\sigma a^\dagger a) - |\Tr(\sigma a)|^2 + m|\Tr(\sigma a)-\alpha|^2\ ,
\end{align}
where we used {Eq.(\ref{m1 eqn lemma}), and} that $\alpha'=\sqrt{m} \alpha$ in the equality.
\newpage

\section{Matrix elements of coherent thermal states in the large $\alpha$ regime {(Proof of Eq.(\ref{limit E}) and Lemma \ref{lem4})}}\label{appendix: matrix elements large alpha}
In this section, we show that
\begin{lemma}[Restating of Lemma \ref{lem4}]
Consider coherent thermal state $\rho(\beta,\alpha)=D(\alpha) e^{-\beta a^\dag a} D(\alpha)^\dag/\Tr(e^{-\beta a^\dag a} )$ where $\alpha$ is real. Then, for $\alpha\gg \max\{1,  n_\beta^2\}$, ignoring terms of order $\mathcal{O}(e^{-\sqrt{\alpha}/n_\beta})$ we have
\begin{align}\label{E appendix}
   E(\alpha)\equiv \sum_{l=0}^\infty \rho_{l,l} [|c_l^\text{opt}|^2-1] = \sum_{l=0}^\infty \rho_{l,l}\ \times \Big(\Big[\frac{\rho_{l-1,l-1}}{\rho_{l-1,l}}\Big]^2 - 1\Big) = \frac{1}{\alpha^2}\frac{n_\beta(n_\beta+1)}{1+2n_\beta} +\mathcal{O}\Big(\frac{(1+2n_\beta)^{3/2}}{\alpha^3}\Big) \ ,
\end{align}
where {$c_l^{\text{opt}}=\rho_{l-1,l-1}/\rho_{l-1,l}$}, and 
  $\rho_{m,n}\equiv\braket{m|\rho(\beta, \alpha)|n}$. Furthermore, $n_\beta\equiv n_\text{th}(\beta)$ is the mean number of thermal excitations. This, in particular, implies that
\begin{align}\label{E lim appendix}
    \lim_{\alpha \rightarrow \infty}  |\alpha|^2 E(\alpha) = \frac{n_\beta}{2} + \frac{n_\beta}{2+ 4 n_\beta} = \frac{n_\beta(n_\beta+1)}{1+2 n_\beta}\ .
\end{align}
\end{lemma}

We explain the strategy to prove the expression for $E(\alpha)$. As noted, $\rho_{l,l}$ is approximately Gaussian in the large $\alpha$ regime, with mean $\alpha^2$ and variance $\sigma^2 = (1+2n_\beta) \alpha^2$. This allows us to partition our range $l$ into a typical and atypical range, i.e., 
\begin{align}
    E(\alpha)\equiv \sum_{l=0}^\infty \rho_{l,l} \times [|c_l^\text{opt}|^2-1] = \sum_{l \in \text{typical}}  \rho_{l,l}\times [|c_l^\text{opt}|^2 - 1] + \sum_{l \not\in \text{typical}}  \rho_{l,l}\times [|c_l^\text{opt}|^2 - 1]  \ 
\end{align}
where we have defined the typical regime as 
\be
\alpha^2-R \sigma \le l \le \alpha^2 + R \sigma
\ee
for large $\alpha$, or, equivalently, in terms of the rescaled parameter $r=(l-\alpha^2)/\sigma$ 
\be
-R \le r\le R \ .
\ee
Here, $R$ is properly chosen to satisfy the following two properties: 
\begin{enumerate}
    \item Inside the typical regime  $\rho_{l,l}$ can be approximated by a Gaussian distribution with up to quadratic corrections, in the form seen in Lemma \ref{eqn: matrix elements final}. {This, in turn, implies that $|c_l^\text{opt}|^2$ can be approximated by the quadratic form in Eq.(\ref{app45}).} To derive this approximation, we require that within the typical regime {$|r|\leq R \ll \alpha^\frac13$}. Or, more precisely, $\lim_{\alpha\rightarrow \infty} \frac{R^3(\alpha)}{\alpha}=0$. {Here, the power $1/3$ ensures that the corrections in approximating a Poisson distribution with a Gaussian distribution remain small (see Lemma \ref{poisson gaussian approx}).}

\item The contribution of the terms outside the typical regime is
\begin{align}\label{atypical bound}
    \sum_{l \not\in \text{typical}}  \rho_{l,l}\times [|c_l^\text{opt}|^2 - 1] = \mathcal{O}\Big(e^{n_\beta + \frac12 - \sqrt{1+2n_\beta}R}\Big)
\end{align}
which goes to zero %\green{assuming $n_\beta\ll R^2\ll \alpha^{2/3}$} 
{assuming $\max\{n_\beta ,1 \}\ll R^2$. 
%which is guaranteed by our assumption $\alpha\gg \max\{1, n_\beta^2\}$, 
Recall that we must also choose $R\ll \alpha^{1/3}$. These two conditions can be simultaneously satisfied because we consider $\alpha\gg \max\{1, n_\beta^2\}$.} Here, $\mathcal{O}$(.) suppresses terms of order of powers of $1/\alpha$ in the exponent.
\end{enumerate}

To compute $|c^\text{opt}_l|^2$ in the typical regime, we  first evaluate $\rho_{l,l}$,  and $\rho_{l,l+1}$. {In particular, as summarized in Lemma \ref{eqn: matrix elements final}, we find a useful expansion of $\rho_{l,l}$, and  $\rho_{l,l+1}$ in the large $\alpha$ regime, and carefully determine deviations of  $\rho_{l,l}$ from Gaussianity, which could be of independent interest.} Then, as Lemma \ref{c_l lemma} shows, inside the typical regime $|c^\text{opt}_l|^2$ is close to 1. More precisely, it can be approximated by the quadratic form
\begin{align}\label{app45}
    |c_l^\text{opt}|^2 = 1 + \frac{r}{\sigma} + \Big[\frac{2 n_\beta (n_\beta+1)^2}{2 n_\beta+1}-\frac{n_\beta (3 n_\beta+2) r^2}{2 n_\beta+1}\Big]\frac{1}{\sigma^2} +\mathcal{O}\Big(\frac{(1+2n_\beta)^{3}}{\sigma^3}\Big)\ ,
\end{align}
where we recall that $r=(l-\alpha^2)/\sigma$. Then to compute $\alpha^2 E(\alpha)$, we show that the contribution of terms within the typical range is at most (upto leading order in $1/\alpha$)
\be
\sum_{l \in \text{typical}}  \rho_{l,l}\times [|c_l^\text{opt}|^2 - 1] = \frac{n_\beta(n_\beta+1)}{1+2n_\beta}\frac{1}{\alpha^2} + \mathcal{O}\Big(\frac{(1+2n_\beta)^{3/2}}{\alpha^3}\Big) \ .
\ee
On the other hand, in the atypical regime, $|c^\text{opt}_l|^2$ can be significantly larger than 1, although it remains bounded as
\begin{align}
  1\le   |c_l^\text{opt}|^2 \leq (1+n_\beta)^2 \times \frac{l}{\alpha^2} \ .
\end{align}
Therefore, the contribution of the terms outside typical regime is
\be
0\le \sum_{l \not\in \text{typical}}  \rho_{l,l}\times [|c_l^\text{opt}|^2 - 1] \le  \frac{(1+n_\beta)^2}{\alpha^2} \times \sum_{l \not\in \text{typical}} \rho_{l,l} \times l \ .
\ee

{In Appendix \ref{chernoff bound}, we use Chernoff-type concentration bounds for $\rho_{l,l}$ to conclude that atypical contributions to probability weight, the first and second moment are exponentially small. Specifically,
\begin{align}
    \sum_{l \not\in \text{typical}} \rho_{l,l} \times l^k= \mathcal{O}\Big(e^{n_\beta + \frac12 - \sqrt{1+2n_\beta}R}\Big) \ :\ k=0,1,2.
\end{align}
This yields the scaling seen in Eq.(\ref{atypical bound}).}\\

\subsection{Computation of diagonal elements  $\rho_{l,l}$}

Recall that the thermal state $\rho_\text{th}(\beta)$ can be expanded in the overcomplete basis of coherent states using its $P$ distribution,
\begin{align}
    \rho_\text{th}(\beta)= \int_{\mathbb{C}} \frac{d^2 \gamma}{\pi} \frac{e^{-\frac{|\gamma|^2}{n_\beta}}}{n_\beta} \ketbra{\gamma}{\gamma}\ ,
\end{align}
where the integral is an area integral over the entire complex plane \cite{PDistribution}. Then, the coherent thermal state $\rho(\beta, \alpha)\equiv D(\alpha) \rho_\text{th}(\beta) D^\dagger(\alpha)$ is,
\begin{align}
    \rho(\beta, \alpha) = \int_{\mathbb{C}} \frac{d^2 \gamma}{\pi} \frac{e^{-\frac{|\gamma|^2}{n_\beta}}}{n_\beta}\  D(\alpha)\ketbra{\gamma}{\gamma}D^\dagger(\alpha) = \int_{\mathbb{C}} \frac{d^2 \gamma}{\pi} \frac{e^{-\frac{|\gamma|^2}{n_\beta}}}{n_\beta} \ketbra{\alpha+\gamma}{\alpha+\gamma} \ .
\end{align}

The diagonal matrix elements in the Fock basis are
\begin{align}
    \rho_{l,l}\equiv\braket{l|\rho(\beta, \alpha)|l}=\int_{\mathbb{C}} \frac{d^2\gamma}{ \pi} \frac{e^{-\frac{|\gamma|^2}{n_\beta}}}{n_\beta} |\braket{\alpha+ \gamma| l}|^2.
\end{align}
We are interested in evaluating this expression in the large $\alpha$ regime.\\ 

Noticing the Gaussian-weighted integral over $\gamma$, we partition the integral as
\begin{align}\label{eq: rho matrix element 1}
    \rho_{l,l}&=\int_{|\gamma|< t(\alpha)} \frac{d^2\gamma}{\pi} \frac{e^{-\frac{|\gamma|^2}{n_\beta}}}{n_\beta} |\braket{\alpha+ \gamma| l}|^2 + \int_{|\gamma|> t(\alpha)} \frac{d^2\gamma}{\pi} \frac{e^{-\frac{|\gamma|^2}{n_\beta}}}{n_\beta} |\braket{\alpha+ \gamma| l}|^2\nonumber\\
    &=\int_{|\gamma|< t(\alpha)} \frac{d^2\gamma}{\pi} \frac{e^{-\frac{|\gamma|^2}{n_\beta}}}{n_\beta} |\braket{\alpha+ \gamma| l}|^2\ + \ \mathcal{O}(e^{-{t(\alpha)^2}/{n_\beta}})
\end{align}
where $t(\alpha)$ is an $\alpha$-dependent cut-off. For concreteness, let us pick $t(\alpha)=\alpha^{1/4}$. Then, for large $\alpha$,  the second term is $\mathcal{O}(e^{-\sqrt{\alpha}/n_\beta})$. {This term is small in the regime $\alpha \gg n_\beta^2$ that we are considering.} Next,  we focus on the first term.\\

\subsubsection*{Evaluation of $|\braket{\alpha+ \gamma| l}|^2$}

We wish to evaluate the energy distribution of $|\braket{\alpha+ \gamma| l}|^2$ when $\alpha \gg 1$ and $\alpha \gg n_\beta^2$, given the cutoff $|\gamma|<\alpha^{1/4}$, where $\ket{l}$ is the Fock basis state. To this end, let us first consider the energy distribution of an arbitrary coherent state $\ket{\alpha}$, i.e., $|\braket{\alpha|l}|^2$. It is well-known that for coherent states the energy distribution has the Poisson distribution, \cite{ECG_coherent1, Glauber63}, that is,
\begin{align}
    |\braket{\alpha|l}|^2 = e^{-|\alpha|^2} \frac{|\alpha|^{2l}}{l!}:\ \ \ \ \ l \in \{0,1,2,...\} .
\end{align}
Then, for large $|\alpha| \gg 1$, this Poisson distribution can be approximated as a Gaussian. For simplicity, let us assume $\alpha$ is real and positive. We prove in Appendix \ref{proof: poisson gaussian approx} that,
\begin{lemma}[Deviations of a Poisson Distribution from a Gaussian]
    Consider a Poisson distribution with positive real parameter $\alpha^2$, i.e., with probabilities 
    $$p_l=e^{-\alpha^2} \times \frac{\alpha^{2l}}{l!}:\ \ \ \ \ l=0,1,2,\cdots\ .$$ 
    {Assuming $\alpha\gg1$ and for $l$ satisfying $|l-\alpha^2|< R \alpha$, where $R\ll \alpha^\frac13$,}
    the Poisson distribution can be approximated as     
    \begin{align}\label{poisson 1}
        p_l = \frac{e^{-\frac{s^2}{2}}}{\sqrt{2 \pi} \alpha}\Big[1 + \frac{a_1}{\alpha} + \frac{a_2}{\alpha^2} + \mathcal{O}\Big(\frac{s^9}{\alpha^{3}}\Big)\Big]:\ \ \ \ \ l \in [\alpha^2 - R \alpha, \alpha^2 + R \alpha]\ ,
    \end{align}
    where $s=(l-\alpha^2)/\alpha$, and 
    \begin{align}
        &a_1=\frac{s^3-3s}{6}\ ,\ \ a_2=\frac{s^6-12s^4+27s^2-6}{72}\ .
    \end{align}
\label{poisson gaussian approx}\end{lemma}

Next, we apply this Lemma to find an approximation of the energy distribution for 
coherent state $|\alpha+\gamma\rangle$, i.e., the distribution $p_l=|\braket{\alpha+\gamma|l}|^2$.  
 {Recall that, without loss of generality, we have assumed $\alpha$ is real and positive.} Assuming $\alpha \gg 1$, thanks to the cutoff $|\gamma|\leq \alpha^{1/4}$, we have $|\alpha+\gamma|\gg 1$, for all $\gamma$ within the disk. That is, when $\alpha\gg 1$, $\gamma$ can be viewed as a small perturbation to $\alpha$. 
Then, Eq.(\ref{poisson 1}) gives an expansion of $p_l=|\braket{\alpha+\gamma|l}|^2$ in terms of the powers of $|\alpha+\gamma|^{-1}$. But, since $\gamma$ is small we can replace this by an expansion in terms of powers of $|\alpha|^{-1}$.  
Therefore, considering $\gamma = x + i y$ and assuming $\alpha$ is real, we get the following expansions 
\bes\label{eq: expansions}
\begin{align}
\frac{1}{|\alpha+\gamma|}&=\frac{1}{\alpha}\Big[1 - \frac{x}{\alpha} + \frac{2x^2-y^2}{\alpha^2}+ \mathcal{O}\Big(\frac{|\gamma|^3}{\alpha^3}\Big)\Big]\\
\frac{l-|\alpha+\gamma|^2}{|\alpha+\gamma|}&=\frac{(\alpha^2+ s\alpha) -|\alpha+\gamma|^2}{|\alpha+\gamma|} \nonumber\\
&=s- 2x + \frac{x^2-y^2-sx}{\alpha} + \frac{2 s x^2-s y^2-2 x^3+4 x y^2}{2\alpha^2} + \mathcal{O}\Big(\frac{s|\gamma|^4}{\alpha^3}\Big)\label{eq: s expansion}
\end{align}
\ees
{which, in turn, implies}
\begin{align}
\exp\Big[-\frac{1}{2}\Big(\frac{l-|\alpha+\gamma|^2}{|\alpha+\gamma|}\Big)^2\Big]&= \exp\Big[-\frac{1}{2}(s-2x)^2\Big] \times \Bigg[ 1 + \frac{c_1}{\alpha} + \frac{c_2}{\alpha^2}
+\mathcal{O}\Big(s^6 \frac{|\gamma|^9}{\alpha^3}\Big)
\Bigg] \label{eq: exp(-s^2) expansion}
\end{align}

where
\begin{align*}
    &c_1=(s-2x)(sx-x^2+y^2)\\
    &c_2=s^4 x^2+s^3 \left(2 x y^2-6 x^3\right)+s^2 \left(13 x^4-x^2 \left(10 y^2+3\right)+y^4+y^2\right)\nonumber\\
    &\ \ \ \ \ \ \ \ \ -4 s \left(3 x^5-2 x^3 \left(2 y^2+1\right)+x y^2 \left(y^2+2\right)\right)+4 x^6-x^4 \left(8 y^2+5\right)+2 x^2 y^2 \left(2 y^2+5\right)-y^4
\end{align*}
and we have retained the same definition of $s=(l-\alpha^2)/\alpha$. These expansions are evaluated using Mathematica (this Mathematica notebook can be found in \cite{CoherentThermalState2025}). Note that Eq.(\ref{eq: exp(-s^2) expansion}) is the exponential of $-1/2$ times the square of Eq.(\ref{eq: s expansion}). The terms suppressed within $\mathcal{O}((.)/\alpha^3)$ are included and explicitly in the Mathematica computation, but only their dominant behavior is mentioned here in the interest of clarity.

Substituting the expansions in Eq.(\ref{eq: expansions}) {and Eq.(\ref{eq: exp(-s^2) expansion})} into Eq.(\ref{poisson 1}) and simplifying yields
\begin{align}\label{braket term 2}
    |\braket{\alpha+\gamma|l}|^2 = \frac{\exp[{-\frac{(s-2x)^2}{2}}]}{\sqrt{2 \pi} \alpha} \Big[1+ \frac{b_1}{\alpha} + \frac{b_2}{\alpha^2} + \mathcal{O}\Big(\frac{s^9 |\gamma|^9}{\alpha^3}\Big)\Big]  
\end{align}
where
\begin{align}
    b_1=&\frac{s^3}{6}+s \left(-x^2+y^2-\frac{1}{2}\right)+\frac{2}{3} \left(x^3-3 x y^2\right)\ ,\\
    b_2=&\frac{s^6}{72}+\frac{1}{6} s^4 \left(-x^2+y^2-1\right)\nonumber\\
    &+\frac{1}{9} s^3 \left(x^3-3 x y^2\right)+\frac{1}{8} s^2 \left(4 x^4+x^2 \left(4-8 y^2\right)+4 y^4-4 y^2+3\right) \nonumber\\
    &+\frac{1}{3} s \left(-2 x^5+x^3 \left(8 y^2+1\right)-3 x \left(2 y^4+y^2\right)\right)\nonumber\\
    &+\frac{1}{36} \left(8 x^6-6 x^4 \left(8 y^2+3\right)+36 x^2 y^2 \left(2 y^2+3\right)-3 \left(6 y^4+1\right)\right)\ ,
\end{align}
where, again, the $\mathcal{O}((.)/\alpha^3)$ term displays only the leading order behavior for clarity.
%where the correction term $\mathcal{O}((s-2x)^9/\alpha^3)$ follows from the expansion in Eq.(\ref{eq: s expansion}), and this dominates all other correction terms that have power $1/\alpha^3$ noted above, \blue{namely, $s^6/\alpha^3$. Note that this correction term is an odd-degree polynomial in $(s-2x)$.} 
{Observe that $b_i$ is a degree $3i$ polynomial in $s$, for $i=1,2$. Also recall from the Lemma statement that $|s|\leq R \ll \alpha^{1/3}$. As a consequence, the correction terms $|b_i/\alpha^i|\ll 1: i=1,2$ for large $\alpha$. The same holds true from the correction term too.}

\subsubsection*{{Diagonal matrix elements $\rho_{l,l}$}}
Now we substitute this expression in Eq.(\ref{braket term 2}) into the integral in Eq.(\ref{eq: rho matrix element 1}). Recall that the integral is over the disk $|\gamma|\leq {\alpha}^{1/4}$ weighted by a Gaussian {probability distribution} centered at the origin with variance $n_\beta$. We know that, {for any complex function $f$ bounded by $|f(\gamma)|\le 1: \gamma\in\mathbb{C}$}, 
$$\int_{|\gamma|^2<\alpha^{1/4}} \frac{d^2\gamma}{\pi} \frac{e^{-\frac{|\gamma|^2}{n_\beta}}}{n_\beta} f(\gamma) = \int_\mathbb{C} \frac{d^2\gamma}{\pi} \frac{e^{-\frac{|\gamma|^2}{n_\beta}}}{n_\beta} \ f(\gamma) -\  \mathcal{O}(e^{-\sqrt{\alpha}/n_\beta}).$$
So,
\begin{align}\label{eqn: integral 1}
\rho_{l,l}=\int_\mathbb{C}   \frac{d^2\gamma}{\pi} \frac{e^{-\frac{|\gamma|^2}{n_\beta}}}{n_\beta}\ \ \Bigg[ \frac{\exp[{-\frac{(s-2x)^2}{2}}]}{\sqrt{2 \pi} \alpha} \Big[1+ \frac{b_1}{\alpha} + \frac{b_2}{\alpha^2} + \mathcal{O}\Big(\frac{s^9|\gamma|^9}{\alpha^3}\Big)\Big]\Bigg] + \mathcal{O}(e^{-\sqrt{\alpha}/n_\beta})\  
\end{align}
{where we used the expression in Eq.(\ref{braket term 2}) for $|\braket{\alpha+ \gamma| l}|^2$.}   This integral can be computed easily by noting that 
\begin{align}
    \int_\mathbb{C} \frac{d^2\gamma}{ \pi} \frac{e^{-\frac{|\gamma|^2}{n_\beta}}}{n_\beta} = \int_{-\infty}^{\infty} \frac{dx}{\sqrt{ \pi}} \frac{e^{-\frac{x^2}{n_\beta}}}{\sqrt{n_\beta}} \int_{-\infty}^{\infty} \frac{dy}{\sqrt{ \pi}} \frac{e^{-\frac{y^2}{n_\beta}}}{\sqrt{n_\beta}}\ .
\end{align}
Then Eq.(\ref{eqn: integral 1}) can be evaluated term-by-term by simply considering the moments of said Gaussians. Note that the term $e^{-(s-2x)^2}$ in Eq.(\ref{braket term 2}) changes the mean and variance of the Gaussian that weights the $x$ integral. We evaluate this integral using Mathematica too. We thus get,
\begin{align}\label{final matrix elem 1}
    \rho_{l,l} = \frac{e^{-\frac{r^2}{2}}}{\sqrt{2 \pi}\sigma} \Bigg[1 + \frac{f_1(n_\beta,r)}{\sigma} + \frac{1}{2}\frac{f_2(n_\beta,r)}{\sigma^2} + \mathcal{O}\Big(\frac{(1+2n_\beta)^3}{\sigma^3}r^9\Big)\Bigg] + \mathcal{O}(e^{-\sqrt{\alpha}/n_\beta})
\end{align}
where we define $\sigma=\sqrt{1+2n_\beta} \alpha$ along with $r=(l-\alpha^2)/\sigma= s/\sqrt{1+2n_\beta}$, and
\begin{align}
    &f_1(n, r)=\frac{1}{2n+1}\Bigg[\frac{\left(6 n^2+6 n+1\right)}{6}r^3+\frac{\left(-2 n^2-4 n-1\right)}{2}r\Bigg]\\
    &f_2(n, r)=\frac{1}{(2n+1)^2}\Bigg[\frac{\left(6 n^2+6 n+1\right)^2}{36}r^6 -\frac{\left(21 n^4+48 n^3+36 n^2+10 n+1\right) }{3}r^4\nonumber\\
    &\ \ \ \ \ \ \ \ \ \ \ \ \ \ \ \ \ \ \ \ \ \ \ \ \ \ \ \ \ \ \ \ \ \ \ \ \ \ \ \ \ +\frac{\left(20 n^4+72 n^3+72 n^2+24 n+3\right)}{4}r^2 -\frac{-6 n^4+12 n^2+6 n+1}{6}\Bigg] \ .
\end{align}
Observe that $f_1$ is degree 3, and $f_2$ is degree 6 polynomial in $r$ with rational co-efficients in terms of $n_\beta$. We note that $\alpha$ and $\sigma=\sqrt{1+2n_\beta}$ are the actual mean and variance of $\rho_{l,l}$ in the large $\alpha$ regime. We compute these moments using the moment-generating function of $\rho_{l,l}$ in Appendix \ref{appendix: moment generating function}.

{Lastly, we draw attention to the the error term in the above expression, i.e., $\mathcal{O}((1+2n_\beta)^3/\sigma^3)$. We obtain this by taking the exact form of the $\mathcal{O}(s^9 |\gamma|^9/\alpha^3)$ term in Eq.(\ref{braket term 2}) and substitute it in the integral in Eq.(\ref{eqn: integral 1}). Using that $r=\sqrt{1+2n_\beta} s$ and $\sigma=\sqrt{1+2n_\beta} \alpha$, we get
$$\int_{-\infty}^\infty dx \frac{e^{-\frac{x^2}{n_\beta}}}{n_\beta}  \frac{e^{-\frac{(s-2x)^2}{2}}}{\sqrt{2 \pi} \alpha} \times \mathcal{O}(\frac{s^9 |\gamma|^9}{\alpha^3}) = \frac{e^{-\frac{r^2}{2}}}{\sqrt{2 \pi} \sigma} \times \mathcal{O}\Big(\frac{(1+2n_\beta)^3}{\sigma^3} r^9 \Big)$$  
where $\mathcal{O}(1+2n_\beta)$ indicates a rational function that is a ratio of polynomials in $1+2n_\beta$, where the difference between the degree of the numerator and denominator is $3$. Furthermore, this error term is a polynomial in $r$ with $n_\beta$-dependent coefficients of degree at most $9$. %\blue{We shall suppress this $r$-dependence in this error term from here on, unless otherwise noted. %because $|r|\leq R \ll \alpha^{1/3}$ 
}

\subsection{Computation of the {off-diagonal elements} $\rho_{l,l+1}$}
Computation of $\rho_{l,l+1}$ follows in the same manner as above. Specifically, in Eq.(\ref{eq: rho matrix element 1}), we replace
$$|\braket{\alpha+\gamma|l}|^2 \rightarrow \braket{l|\alpha+\gamma}\braket{\alpha+\gamma|l+1} = |\braket{\alpha+\gamma|l}|^2 \times \frac{(\alpha+\gamma)^*}{\sqrt{l+1}}.$$
Hence, to determine $\rho_{l,l+1}$ we simply add an extra ${(\alpha+\gamma)^*}/{\sqrt{l+1}}$ term to the integral in Eq.(\ref{eqn: integral 1}), where $l=|\alpha+\gamma|^2 + s |\alpha+ \gamma|$, and retain the approximate expression for $|\braket{\alpha+\gamma|l}|^2$ considered in Eq.(\ref{braket term 2}) as before. The integral is evaluated term-by-term, exactly as above, after substituting $\gamma=x+iy$. We use Mathematica for this computation too. This yields
\begin{align}\label{final matrix elem 2}
        \rho_{l,l+1}= \frac{e^{-\frac{r^2}{2}}}{\sqrt{2 \pi}\sigma} \Bigg[1 + \frac{g_1(n_\beta, r)}{\sigma} + \frac12 \frac{g_2(n_\beta, r)}{\sigma^2} + \mathcal{O}\Big(\frac{(1+2n_\beta)^3}{\sigma^3}r^9\Big)\Bigg] + \mathcal{O}(e^{-\sqrt{\alpha}/n_\beta})
    \end{align}
where again we define $\sigma=\sqrt{1+2n_\beta} \alpha$ along with $r=(l-\alpha^2)/\sigma= \sqrt{1+2n_\beta} s$, and
\begin{align}
    &g_1(n, r)=\frac{1}{(2n+1)}\Bigg[\frac{\left(6 n^2+6 n+1\right)}{6}r^3+{\left(-n^2-3 n-1\right)} r \Bigg]\\
    &g_2(n, r)=\frac{1}{(2n+1)^2}\Bigg[\frac{\left(6 n^2+6 n+1\right)^2}{36}r^6-\frac{\left(42 n^4+108 n^3+90 n^2+28 n+3\right)}{6}r^4\nonumber\\
    &\ \ \ \ \ \ \ \ \ \ \ \ \ \ \ \ \ \ \ \ \ \ \ \ \ \ \ \ \ \ \ \ \ \ \ \ \ \ \ \ \ +{\left(5 n^4+26 n^3+33 n^2+14 n+2\right)}r^2-\frac{18 n^4+60 n^3+84 n^2+42 n+7}{6}\Bigg]\ .
\end{align}\\

These two matrix elements are summarized in the following lemma:
\begin{lemma}\label{eqn: matrix elements final}
Consider a coherent thermal state $\rho(\beta, \alpha)$ with {$\alpha\gg \max\{1, n_\beta^2\}$}. Denote its matrix elements in the Fock basis as $\rho_{m,n}:=\braket{m|\rho(\beta, \alpha)|n}$. Let $n_\text{th}(\beta)=(e^{\beta \omega}-1)^{-1}$ be the expected number of thermal excitations. Furthermore, define
\begin{align}
    &\sigma^2:=(1+2n_\beta)\alpha^2\\
    &r:=(l-\alpha^2)/\sigma
\end{align}
{Consider integer} $l \in [\alpha^2 - R \sigma, \alpha^2 + R \sigma ]$, where $R \ll \alpha^{\frac13}$. 
Then,
    \begin{align}
        \rho_{l,l}= \frac{e^{-\frac{r^2}{2}}}{\sqrt{2 \pi}\sigma} \Bigg[1 + \frac{f_1(n_\beta, r)}{\sigma} + \frac12 \frac{f_2(n_\beta, r)}{\sigma^2} + \mathcal{O}\Big(\frac{(1+2n_\beta)^3}{\sigma^3}r^9\Big)\Bigg] + \mathcal{O}(e^{-\sqrt{\alpha}/n_\beta})
    \end{align}
where we recall the implicit $l$ dependence in $r$, and
\begin{align}
    &f_1(n, r)=\frac{1}{2n+1}\Bigg[\frac{\left(6 n^2+6 n+1\right)}{6}r^3+\frac{\left(-2 n^2-4 n-1\right)}{2}r\Bigg]\\
    &f_2(n, r)=\frac{1}{(2n+1)^2}\Bigg[\frac{\left(6 n^2+6 n+1\right)^2}{36}r^6 -\frac{\left(21 n^4+48 n^3+36 n^2+10 n+1\right) }{3}r^4\\
    &\ \ \ \ \ \ \ \ \ \ \ \ \ \ \ \ \ \ \ \ \ \ \ \ \ \ \ \ \ \ \ \ \ \ \ \ \ \ \ \ \ +\frac{\left(20 n^4+72 n^3+72 n^2+24 n+3\right)}{4}r^2 -\frac{-6 n^4+12 n^2+6 n+1}{6}\Bigg]
\end{align}
which are degree $3$ and degree $6$ polynomials in $r$ respectively. Similarly,
\begin{align}
        \rho_{l,l+1}= \frac{e^{-\frac{r^2}{2}}}{\sqrt{2 \pi}\sigma} \Bigg[1 + \frac{g_1(n_\beta, r)}{\sigma} + \frac12 \frac{g_2(n_\beta, r)}{\sigma^2} + \mathcal{O}\Big(\frac{(1+2n_\beta)^3}{\sigma^3}r^9\Big)\Bigg] + \mathcal{O}(e^{-\sqrt{\alpha}/n_\beta})
    \end{align}
where
\begin{align}
    &g_1(n, r)=\frac{1}{(2n+1)}\Bigg[\frac{\left(6 n^2+6 n+1\right)}{6}r^3+{\left(-n^2-3 n-1\right)} r \Bigg]\\
    &g_2(n, r)=\frac{1}{(2n+1)^2}\Bigg[\frac{\left(6 n^2+6 n+1\right)^2}{36}r^6-\frac{\left(42 n^4+108 n^3+90 n^2+28 n+3\right)}{6}r^4\\
    &\ \ \ \ \ \ \ \ \ \ \ \ \ \ \ \ \ \ \ \ \ \ \ \ \ \ \ \ \ \ \ \ \ \ \ \ \ \ \ \ \ +{\left(5 n^4+26 n^3+33 n^2+14 n+2\right)}r^2-\frac{18 n^4+60 n^3+84 n^2+42 n+7}{6}\Bigg]
\end{align}
which also are degree $3$ and degree $6$ polynomials in $r$ respectively.
\end{lemma}

{Here too, observe that $f_i$ is a degree $3i$ polynomial in $r$, for $i=1,2$. So the condition that $|r|\leq R \ll \alpha^{1/3}$ ensures that $|f_i/\sigma^i| \ll 1$ for large $\alpha$, justifying this expansion. The same holds true for the $1/\sigma^3$ correction term too because it would be a degree $9$ polynomial in $r$. All these properties extend exactly to $g_i$ too.}

\subsection{Evaluating $|c^\text{opt}_{l+1}|^2$}\label{c_l subsection}
We are interested in the ratio
\begin{align}
    |c^\text{opt}_{l+1}|^2 = \Big|\frac{\rho_{l,l}}{\rho_{l,l+1}}\Big|^2\ .
\end{align}

For $l \in [\alpha^2 - R \sigma, \alpha^2 + R \sigma ]$,  where the coefficient $R \ll \alpha^\frac13$, we can expand this in powers of $1/\sigma$ when $\alpha \gg 1$. To this end, recall that $|f_i/\sigma^i| \ll 1$ and $|g_i/\sigma^i| \ll 1$ for large $\alpha$, for $i=1,2$. We shall introduce $f_3$ and $g_3$ to indicate the degree $9$ polynomials in $r$ present in the correction term. Because $|f_i/\sigma^i|\ll 1$ and $|g_i/\sigma^i|\ll 1$ for $i=1,2,3$, we can Taylor expand this ratio as
\begin{align}
    \frac{\rho_{l,l}}{\rho_{l,l+1}} &= \frac{1 + f_1/\sigma + f_2/\sigma^2 + \mathcal{O}(f_3/\sigma^3)}{1 + g_1/\sigma + g_2/\sigma^2 + \mathcal{O}(g_3/\sigma^3)}\\
    &= 1+ \frac{f_1-g_1}{\sigma} + \frac{f_2-f_1 f_2 +g_1^2 - g_2}{\sigma^2} + \mathcal{O}\Big(\frac{f_3-f_2 g_1 - g_1^3 + f_1 (g_1^2 - g_2) + 2 g_1 g_2 -g_3}{\sigma^3}\Big)
\end{align}
By squaring this expression, and collecting terms, we get
\begin{align}
    \Big|\frac{\rho_{l,l}}{\rho_{l,l+1}}\Big|^2 = 1 + \frac{2(f_1-g_1)}{\sigma} &+ \frac{f_1^2+2f_2 -4f_1g_1+3g_1^2-2g_2}{\sigma^2}\\ &+\mathcal{O}\Big(\frac{f_1f_2+f_3-f_1^2g_1- 2f_2 g_1 + 3 f_1g_1^2 -2g_1^3 - 2 f_1 g_2 + 3g_1 g_2 - g_3}{\sigma^3}\Big)
\end{align}
These expansions are computed using Mathematica. Now we substitute the exact expressions for $f_1$, $f_2$, $g_1$, $g_2$ mentioned in Lemma \ref{eqn: matrix elements final}, and simplify to get
\begin{align}\label{eqn: c_l+1}
    |c^\text{opt}_{l+1}|^2&=1 + \frac{r}{\sigma} + \Big[\frac{2 n_\beta^3+4 n_\beta^2+4 n_\beta+1}{2 n_\beta+1}-\frac{n_\beta (3 n_\beta+2) }{2 n_\beta+1}r^2\Big]\frac{1}{\sigma^2} + \mathcal{O}\Big(\frac{(1+2n_\beta)^{3}}{\sigma^3}\Big)
\end{align}
where we recall that $r=(l-\alpha^2)/\sigma$ and $\sigma=\sqrt{1+2n_\beta} \alpha$. Notice that the $1/\sigma^3$ correction term remains $\mathcal{O}((1+2n_\beta)^3/\sigma^3)$. This follows by noting that $f_i$ and $g_i$ are degree $i$ polynomials in $(1+2n_\beta)$ for $i=1,2$, and, $f_3$ and $g_3$ are degree $1$ polynomials in $(1+2n_\beta)$. Thus, the largest degree of $(1+2n_\beta)$ possible from the terms noted in the correction term is $3$. \footnote{The word `degree' for rational functions in $(1+2n_\beta)$ here means the \textit{difference} between the degree of the numerator and denominator polynomials.} Furthermore, $f_i$ and $g_i$ are degree $3i$ polynomials in $r$ with $n_\beta$-dependent coefficients. Thus, the correction term would be atmost a degree $9$ polynomial in $r$. Finally, recall that $|r| \leq R \ll \alpha^\frac13$. This ensures that all the terms in above expression are much smaller than $1$.\\

To compute the limit in Eq.(\ref{limit E}), we must consider $|c^\text{opt}_l|^2$ instead. We can easily evaluate this by replacing $r \rightarrow r - 1/\sigma$ in the above equation and simplifying. This is because $r=(l-\alpha^2)/\sigma$. Doing so yields,
\begin{lemma}[Restating of Lemma \ref{c_l lemma2}]\label{c_l lemma}
{For $\alpha \gg \max\{1, n_\beta^2\}$ and $l$ in the typical interval, or equivalently $l \in [\alpha^2 - R \sigma, \alpha^2 + R \sigma]$ with coefficient $R \ll \alpha^\frac13$ and $\sigma=\sqrt{1+2n_\beta} \alpha$}, the optimal $|c_l^\text{opt}|^2 = |\rho_{l-1,l-1}/\rho_{l-1,l}|^2$ takes the form
\begin{align}\label{c_l eqn}
    |c_l^\text{opt}|^2 = 1 + \frac{r}{\sigma} + \Big[\frac{2 n_\beta (n_\beta+1)^2}{2 n_\beta+1}-\frac{n_\beta (3 n_\beta+2) }{2 n_\beta+1}r^2\Big]\frac{1}{\sigma^2} +\mathcal{O}\Big(\frac{(1+2n_\beta)^{3}}{\sigma^3}\Big)\ ,
\end{align}
where $|r|\leq R \ll \alpha^\frac13$.
\end{lemma}
Note that we have ignored the exponentially-suppressed term $\mathcal{O}(e^{-\frac{\sqrt{\alpha}}{n_\beta}})$ because this is dominated by the included $\mathcal{O}(\sigma^{-3})$ term in the $\alpha \rightarrow \infty$ limit. Furthermore, $\mathcal{O}((1+2n_\beta)^3)$ in the error term is a polynomial in $r$ with $n_\beta$-dependent coefficients, with degree at most $9$.

\subsection{Limit Computation (Proof of Eq.(\ref{E appendix}) and Eq.(\ref{E lim appendix}))}\label{limit comp app}
Next, we calculate the limit 
\begin{align}
 \lim_{\alpha \rightarrow \infty} \alpha^2 E(\alpha)  =\lim_{\alpha \rightarrow \infty} \sum_{l=0}^\infty \rho_{l,l} \times \alpha^2 [|c^\text{opt}_l|^2 - 1]= \frac{n_\beta(n_\beta+1)}{1+2n_\beta} \ .
\end{align}
We first evaluate $E(\alpha)$. Recall that the typical regime of $l$ is defined as $l \in [\alpha^2 - R \sigma, \alpha^2 + R \sigma]$, where $\sigma=\sqrt{1+2n_\beta} \alpha$ and $R \ll \alpha^{1/3}$; its complement is defined as the atypical range of $l$. Note $R$ can be made very large, as long as $\lim_{\alpha \rightarrow \infty} R^3/\alpha = 0$. We shall consider $R=\alpha^{1/4}$ when $\alpha\gg \max\{1, n_\beta^2\}$. Now we partition the summation over all $l$ in $E(\alpha)$ into the typical and atypical ranges,
\begin{align}\label{typ atyp}
    E(\alpha)=&\sum_{l=0}^\infty \rho_{l,l} [|c_l^\text{opt}|^2-1]\nonumber\\
    =& \underbrace{\sum_{l \in \text{typical}} \rho_{l,l} \Big[\frac{r}{\sigma} + \Big[\frac{2 n_\beta (n_\beta+1)^2}{2 n_\beta+1}-\frac{n_\beta (3 n_\beta+2) r^2}{2 n_\beta+1}\Big] \frac{1}{\sigma^2}\Big]}_\text{typical term}\nonumber\\
    &+ \underbrace{\sum_{l \in \text{typical}} \rho_{l,l} \times \mathcal{O}\Big(\frac{(1+2n_\beta)^{3}}{\sigma^3}\Big)}_\text{approximation error} 
    + \underbrace{\sum_{l \not\in \text{typical}} \rho_{l,l} [|c_l^\text{opt}|^2-1]}_\text{atypical error} \ .
\end{align}

For large $\alpha\gg \max\{1, n_\beta^2\}$, the typical term turns out to be the only meaningful contribution to the limit because it does not decay faster than $1/\alpha^2$. Computing this term is rather simple: in Appendix \ref{chernoff bound}, we show that the contribution of the \textit{atypical} values of $l$ to the moments of $r$ is exponentially small. That is, $\sum_{l \not \in \text{typical}} \rho_{l,l} \times r^k = \mathcal{O}(e^{-R}) : k=0,1,2$. So, the summation over the typical range can be replaced by a summation over \textit{all} $l$, at a cost that is exponentially small in $\alpha$ because we picked $R=\alpha^{1/4}$. This simplifies the typical term to depend only on the first and second moments of $r$; these are computed in Appendix \ref{mu1 mu2 appendix}.

On the other hand, the two error terms simplify to $\mathcal{O}\big(\frac{(1+2n_\beta)^3}{\sigma^3}\big)$ and $\mathcal{O}(e^{-R})$. The first approximation error is a rational function in $(1+2n_\beta)$ with degree at most $3$, whole divided by $\sigma^3$. Whereas the second atypical error (where we suppressed $n_\beta$-dependence) can be made exponentially small in $\alpha$ by choosing $R=\alpha^{1/4}$ as above. Then, $\alpha^2$ times these errors vanishes in the limit $\alpha\rightarrow \infty$.\\

We shall first evaluate the error terms in the following, and then evaluate the typical term. As mentioned, we shall implicitly assume $R=\alpha^{1/4}$ to simplify arguments.

\subsubsection*{\textbf{Atypical Error}}
For all $l$, we show in Eq.(\ref{c_l upperbound}) that
\begin{align}
     |c_l^\text{opt}|^2 - 1 \leq |c_l^\text{opt}|^2 \leq (1+n_\beta)^2 \times \frac{l}{\alpha^2} \ .
\end{align}
Then,
\begin{align}\label{atypical c_l limit}
    \sum_{l \not\in \text{typical}}  \rho_{l,l}\times \alpha^2[|c_l^\text{opt}|^2 - 1] &\leq (n_\beta + 1)^2 \times \sum_{l \not\in \text{typical}}  \rho_{l,l}\times l = \mathcal{O}(e^{-R}) \ .
\end{align}
As mentioned, these atypical contributions, $\sum_{l \not \in \text{typical}} \rho_{l,l} \times l^k = \mathcal{O}(e^{-R}) : k=0,1,2$, are computed in Appendix \ref{chernoff bound} (see Eq.(\ref{chernoff eqns})). Here, $\mathcal{O}(.)$ suppressed the $n_\beta$-dependence.

\subsubsection*{\textbf{Approximation Error}}
Recall that the term being summed over in the approximation term is the error term in Eq.(\ref{c_l eqn}). As mentioned there, this is a polynomial in $r$ with $n_\beta$-dependent coefficients, with degree at most $9$, whole divided by $\sigma^3$. We shall write it out explicitly as
\begin{align}
    \mathcal{O}\Big(\frac{(1+2n_\beta)^3}{\sigma^3}\Big) = \frac{1}{\sigma^3} \sum_{i=0}^9 g_i(n_\beta) \  r^i \ ,
\end{align}
where $g_i(n_\beta)$ are $n_\beta$-dependent coefficients {that are independent of $\sigma$. Furthermore, $g_i(n_\beta)$ are rational functions in $n_\beta$ such that the difference between the degree of the $n_\beta$ polynomial in the numerator and denominator is at most $3$.} Substituting this into the summation in the approximation term yields
\begin{align}
    \frac{1}{\sigma^3} \sum_{i=0}^9 g_i(n_\beta) \times \Big(\sum_{l \in \text{typical}} \rho_{l,l} \times r^i\Big) \ .
\end{align}
{In Appendix \ref{r exp appendix 2} we show that because $\rho_{l,l}$ can be approximated by a Gaussian (as given in Lemma \ref{eqn: matrix elements final}), for any fixed non-negative integer $k$}
\begin{align}\label{k gen mid}
    \sum_{l \in \text{typical}} \rho_{l,l} \times r^k = c_k + o(1)\ ,
\end{align}
where $c_k$ is {a standard Gaussian moment that} is a constant independent of $\alpha$ and $n_\beta$, and the $o(1)$ is vanishing in the limit $\alpha \rightarrow \infty$. {We also noted our choice of $R=\alpha^{1/4}$ to drop the exponentially-suppressed term}. Using this, the approximation term simplifies to 
\begin{align}
    \frac{1}{\sigma^3} \sum_{i=0}^9 g_i(n_\beta) \times (c_i + o(1)) \ .
\end{align}
Thus, even after the summation over typical $l$, this term continues to be $\mathcal{O}({(1+2n_\beta)^3}/{\sigma^3})$, albeit with no $r$ dependence any more. Now we focus on the typical term in Eq.(\ref{typ atyp}).
\subsubsection*{\textbf{Typical Term}}
The function whose expectation is being computed in the typical term is a quadratic polynomial in $r$. In Appendix \ref{chernoff bound}, we show that $\sum_{l \not\in \text{typical}} \rho_{l,l} \times l^k = \mathcal{O}(e^{-R}): k=0,1,2$. Because $r=(l-\alpha^2)/\sigma$, it follows that 
\begin{align}
    \sum_{l \not\in \text{typical}} \rho_{l,l} \times r^k = \sum_{l \not\in \text{typical}} \rho_{l,l} \times \big(\frac{l-\alpha^2}{\sigma}\big)^k= \mathcal{O}(e^{-R})\ : \ k=0,1,2 \ .
\end{align}
Thus, we can replace the summation over the typical range with a summation over all $l$, at the cost of an $\mathcal{O}(e^{-R})$ term (we ignore this in the following computation because we already considered it in the error terms above). So, the typical term can be written as
\begin{align}\label{typ 5}
    \sum_{l=0}^\infty \rho_{l,l}\times \Bigg[\frac{r}{\sigma} + \Big[\frac{2 n_\beta (n_\beta+1)^2}{2 n_\beta+1}-\frac{n_\beta (3 n_\beta+2) r^2}{2 n_\beta+1}\Big] \frac{1}{\sigma^2}\Bigg] \ .
\end{align}
In Appendix \ref{mu1 mu2 appendix}, we directly compute the first and second moment of $r$ with respect to the $\rho_{l,l}$ distribution,
\bes\label{k12}
\begin{align}
    &\sum_{l =0}^\infty \rho_{l,l} \times r = \frac{n_\beta}{\sqrt{1+2n_\beta}}\frac{1}{\alpha} \ ,\\
    &\sum_{l =0}^\infty \rho_{l,l} \times r^2 = 1 + \frac{n_\beta}{\alpha^2}  \ .
\end{align}
\ees
Using these, we can evaluate the summation by in Eq.(\ref{typ 5}) exactly. Given $\sigma=\sqrt{1+2n_\beta} \alpha$, the typical term simplifies to
\begin{align}
    \frac{n_\beta(1+n_\beta)}{1+2n_\beta} \frac{1}{\alpha^2} - \frac{n_\beta^2 (2+3n_\beta)}{(1+2n_\beta)^2} \frac{1}{\alpha^4} \ .
\end{align}

\subsubsection*{\textbf{Combining the three terms}}
We get $E(\alpha)$ by combining this typical term with the two error terms. That is,
\begin{align}
    E(\alpha) = \frac{n_\beta(1+n_\beta)}{1+2n_\beta} \frac{1}{\alpha^2} - \frac{n_\beta^2 (2+3n_\beta)}{(1+2n_\beta)^2} \frac{1}{\alpha^4} + \mathcal{O}\Big(\frac{(1+2n_\beta)^{3/2}}{\alpha^3}\Big) + \mathcal{O}(e^{-\alpha^{1/4}}) \ .
\end{align}
Absorbing the $1/\alpha^4$ term into the $\mathcal{O}(\alpha^{-3})$ error term yields Eq.(\ref{E appendix}). Then taking the limit $\alpha \rightarrow \infty$ of $\alpha^2 E(\alpha)$ yields Eq.(\ref{E lim appendix}).

\subsection{Evaluation of first and second moments (Proof of Eq.(\ref{k12}))}\label{mu1 mu2 appendix}
The quantities $\sum_{l=0}^\infty \rho_{l,l} \times r^k: k=1,2$ can be computed directly, given $r=(l-\alpha^2)/\sigma$. {Observe that $$\sum_{l=0}^\infty \rho_{l,l}\times r^k = \frac{1}{\sigma^k}\sum_{l=0}^\infty \rho_{l,l}\times (l-\alpha^2)^k= \frac{1}{\sigma^k}\Tr(\rho(\beta, \alpha)[a^\dagger a - |\alpha|^2]^k).$$ }
To this end, we first note that
\begin{align}
    &\Tr(\rho(\beta, \alpha)\ a^\dagger a)=\Tr(\rho_\text{th}(\beta) D_\alpha^\dagger a^\dagger a D_\alpha)=n_\beta + |\alpha|^2\label{Tr(a)}\\
    &\Tr(\rho(\beta, \alpha) a^\dagger a\  a^\dagger a) = \Tr(\rho_\text{th}(\beta) D_\alpha^\dagger a^\dagger a\  a^\dagger a D_\alpha) = n_\beta + 2 n_\beta^2 + 4|\alpha|^2 n_\beta + |\alpha|^2 + |\alpha|^4\label{Tr(aa)}. 
\end{align}
These follow from $D_\alpha^\dagger a D_\alpha = a+\alpha$, $\Tr(\rho_\text{th}(\beta) a^\dagger a)= n_\beta$, and 
\bes
\begin{align}
\Tr(\rho_\text{th}(\beta) a^\dagger a\ a^\dagger a)&= (1-\lambda)\sum_{l=0} l^2 \times \lambda^l\\ 
&= (1-\lambda) \times \lambda \frac{d}{d \lambda} \Big( \lambda \frac{d}{d \lambda} \Big( \sum_{l=0} \lambda^l \Big) \Big)\\ 
&= (1-\lambda) \times \lambda \frac{d}{d \lambda} \Big( \lambda \frac{d}{d \lambda} \Big( \frac{1}{1-\lambda}\Big) \Big)\\
&= \frac{\lambda(1+\lambda)}{(1-\lambda)^2}\\
&= n_\beta + 2 n_\beta^2 \ .
\end{align}
\ees
We denoted $\lambda=n_\beta/(1+n_\beta)$ for brevity in the above derivation. From Eq.(\ref{Tr(a)}), clearly 
\begin{align}
    \Tr(\rho(\beta, \alpha)a^\dagger a)-|\alpha|^2 = n_\beta.
\end{align}
Similarly, from Eq.(\ref{Tr(a)}) and Eq.(\ref{Tr(aa)}), we get that 
\bes
\begin{align}
    \Tr(\rho(\beta, \alpha)[a^\dagger a - |\alpha|^2]^2)
    &= \Tr(\rho(\beta, \alpha) \ a^\dagger a \ a^\dagger a) - 2 |\alpha|^2  \Tr(\rho(\beta, \alpha) a^\dagger a) + |\alpha|^4 \\
    &=(1+2n_\beta)(n_\beta + |\alpha|^2) \ .
\end{align}
\ees
{Then,
\bes
\begin{align}
    &\sum_{l =0}^\infty \rho_{l,l} \times r = \frac{n_\beta}{\sqrt{1+2n_\beta}}\frac{1}{\alpha} \ ,\\
    &\sum_{l =0}^\infty \rho_{l,l} \times r^2 = 1 + \frac{n_\beta}{\alpha^2}  \ .
\end{align}
\ees}

\subsection{Expectation value of $r^k$ (Proof of Eq.(\ref{k gen mid}))}\label{r exp appendix 2}
Recall from Lemma \ref{eqn: matrix elements final} that within the typical regime, the diagonal matrix elements of $\rho(\beta, \alpha)$ are
\begin{align}\label{rho_ll r}
    \rho_{l,l}= \frac{e^{-\frac{r^2}{2}}}{\sqrt{2 \pi}\sigma} \Bigg[1 + \frac{f_1(n_\beta, r)}{\sigma} + \frac12 \frac{f_2(n_\beta, r)}{\sigma^2} + \mathcal{O}\Big(\frac{(1+2n_\beta)^3}{\sigma^3}r^9\Big)\Bigg] + \mathcal{O}(e^{-\sqrt{\alpha}/n_\beta})\ ,    
\end{align}
where $r=(l-\alpha^2)/\sigma$. Now recall that $f_1$ and $f_2$ are $3$ and $6$ degree polynomials in $r$ respectively, with $n_\beta$-dependent coefficients. And similarly, $\mathcal{O}({(1+2n_\beta)^3 r^9}/{\sigma^3})$ is a polynomial in $r$ with $n_\beta$-dependent coefficients, with degree at most $9$, whole divided by $\sigma^3$. To isolate the Gaussian approximation from its corrections,
we shall rewrite the above expression as  
\begin{align}
    \rho_{l,l} = \frac{e^{-\frac{r^2}{2}}}{\sqrt{2 \pi}\sigma} + \frac{e^{-\frac{r^2}{2}}}{\sqrt{2 \pi}\sigma} \times \frac{\Delta(r)}{\sigma} 
\end{align}
where $\Delta(r)$ is a polynomial in $r$ of degree at most $9$, whose coefficients are $\mathcal{O}(1)$ in $\alpha$, or consequently in $\sigma$. Thus, 
\begin{align}\label{Delta 1}
    \sum_{l \in \text{typical}} \rho_{l,l} \times {r^k}=\sum_{l \in \text{typical}} \frac{e^{-\frac{r^2}{2}}}{\sqrt{2 \pi}\sigma} \times {r^k} + \frac{1}{\sigma}\sum_{l \in \text{typical}} \frac{e^{-\frac{r^2}{2}}}{\sqrt{2 \pi}\sigma} \times ({\Delta(r) \times r^k}) \ .
\end{align}
So, to simplify this expression, we shall focus on evaluating the term
\begin{align}
    \sum_{l \in \text{typical}} \frac{e^{-\frac{r^2}{2}}}{\sqrt{2 \pi}\sigma} \times {r^k} \ ,
\end{align}
where $k$ is a non-negative integer. In the large $\alpha$ limit, we shall replace this summation with an integral (note that the limit of large $\alpha$ is equivalent to the limit of large $\sigma$ because $\sigma=\sqrt{1+2n_\beta} \alpha$). That is, 
\begin{align}
    \lim_{\alpha \rightarrow \infty} \sum_{l \in \text{typical}} \frac{e^{-\frac{r^2}{2}}}{\sqrt{2 \pi}\sigma} {r^k}  = \frac{1}{\sqrt{2 \pi}}  \lim_{\sigma \rightarrow \infty} \sum_{l \in \text{typical}} \frac{1}{\sigma} \times {e^{-\frac{r^2}{2}}}  {r^k} =  \int_{-R}^R dx\  \frac{e^{-\frac{x^2}{2}}}{\sqrt{2 \pi}} x^k =: c_{k,R} 
\end{align}
where we make the replacement $r=(l-\alpha^2)/\sigma \rightarrow x$, and $1/\sigma \rightarrow dx$, and recall that the summation is over the typical range $l \in [\alpha^2 - R \sigma, \alpha^2 + R \sigma]$. Clearly, $c_{k,R}$ is a constant independent of $\sigma$ and $n_\beta$. Therefore, 
\begin{align}\label{g1}
    \sum_{l \in \text{typical}} \frac{e^{-\frac{r^2}{2}}}{\sqrt{2 \pi}\sigma} {r^k} = c_{k,R} + o(1) \ ,
\end{align}
where $o(1)$ is vanishing in the limit $\sigma \rightarrow \infty$. Now observe that $\Delta(r) \times r^k$ is atmost a $k+9$ degree polynomial in $r$ whose coefficients are $\mathcal{O}(1)$ in $\alpha$, or consequently in $\sigma$. So using Eq.(\ref{g1}), 
\begin{align}
    \sum_{l \in \text{typical}} \frac{e^{-\frac{r^2}{2}}}{\sqrt{2 \pi}\sigma} \times ({\Delta(r) \times r^k}) =  \mathcal{O}(1)
\end{align}
where $\mathcal{O}(1)$ pertains to variable $\sigma$; we have, in effect, suppressed the $n_\beta$-dependence. Substituting this in Eq.(\ref{Delta 1}) yields
\begin{align}\label{g7}
    \sum_{l \in \text{typical}} \rho_{l,l} \times r^k  &= \Big(c_{k,R} + o(1)\Big) + \frac{\mathcal{O}(1)}{\sigma}\\
    &= c_{k,R} + o(1)\  ,
\end{align}
where we noted in the second equality that $\mathcal{O}(1)/\sigma$ is $o(1)$ in the variable $\sigma$. Finally, recall that $$c_{k,R}:=\int_{-R}^R dx\  \frac{e^{-\frac{x^2}{2}}}{\sqrt{2 \pi}} x^k.$$ Let $c_k$ denote the actual $k^\text{th}$ moment. {From standard Gaussian tail bounds, we know that $c_{k,R} = c_k + \mathcal{O}(e^{-R^2})$. Now we note the choice of $R=\alpha^{1/4}$ made in Sec. \ref{limit comp app}. Thus, the term $\mathcal{O}(e^{-R^2})=\mathcal{O}(e^{-\alpha^{1/2}})$ is also $o(1)$. Substituting this into Eq.(\ref{g7}) finally yields
\begin{align}
    \sum_{l \in \text{typical}} \rho_{l,l} \times r^k = c_k + o(1) \ .
\end{align}}

\subsection{{Beyond} Gaussian Approximation of Poisson Distribution: Proof of Lemma \ref{poisson gaussian approx}}\label{proof: poisson gaussian approx}
\begin{proof}
    We are interested in approximating the quantity 
    $$p_l=e^{-\alpha^2} \times \frac{\alpha^{2l}}{l!}$$
    in the typical range, i.e., $l \in [\alpha^2 - R \alpha, \alpha^2 + R \alpha]$ where we shall start with assuming $R \ll \alpha$. It then implies that all the typical $l$ satisfy $l \gg 1$ when $\alpha \gg 1$. Then, we can simplify $l!$ using the Stirling's approximation \cite{NIST:DLMF} as follows,
    \begin{align}
    \ln(l!) &= l \ln l -l + \ln \sqrt{2 \pi l}  + \ln\Big(1+ \frac{1}{12l} + \mathcal{O}\Big(\frac{1}{l^2}\Big)\Big)\\
    &=l \ln l - l + \ln \sqrt{2 \pi l} + \frac{1}{12 l} + \mathcal{O}\Big(\frac{1}{l^2}\Big)
\end{align}
where we used that $\ln(1+x)=x +\mathcal{O}(x^2)$ for $x\ll 1$ in the second line. Thus,
\begin{align}
    \ln\Big[e^{-\alpha^2} \frac{\alpha^{2l}}{l!}\Big] =-\ln \sqrt{2 \pi l} - l \ln\Big(\frac{l}{\alpha^2}\Big) + (l-\alpha^2) - \frac{1}{12l} + \mathcal{O}\Big(\frac{1}{l^2}\Big).
\end{align}
Now we reparameterize the index $l$ with $s=(l-\alpha^2)/\alpha$. 
This yields
\begin{align}
    \ln\Big[e^{-\alpha^2} \frac{\alpha^{2l}}{l!}\Big] = - \ln \sqrt{2 \pi (\alpha^2 + s \alpha)} - (\alpha^2 + s \alpha) \ln \Big(1 + \frac{s}{\alpha}\Big) + s \alpha  - \frac{1}{12(\alpha^2 + s \alpha)} + \mathcal{O}\Big(\frac{1}{(\alpha^2 + s \alpha)^2}\Big)
\end{align}

Because $|s|\leq R \ll \alpha$, we can expand the above expression in powers of $1/\alpha$. This yields

\begin{align}
    \ln\Big[e^{-\alpha^2} \frac{\alpha^{2l}}{l!}\Big] = -\frac{s^2}{2} - \ln \sqrt{2 \pi \alpha^2}+ \frac{1}{ \alpha}\Big(\frac{s^3-3s}{6}\Big) -\frac{1}{\alpha ^2}\Big(\frac{s^4-3s^2+1}{12}\Big) +\mathcal{O}\Big(\frac{s^5}{\alpha ^3}\Big)\ .
\end{align}

Taking the exponential on both sides yields
\begin{align}
    e^{-\alpha^2} \frac{\alpha^{2l}}{l!} = \frac{e^{-\frac{s^2}{2}}}{\sqrt{2 \pi} \alpha} \times \exp\Big[\frac{1}{ \alpha}\Big(\frac{s^3-3s}{6}\Big) -\frac{1}{\alpha ^2}\Big(\frac{s^4-3s^2+1}{12}\Big) +\mathcal{O}\Big(\frac{s^5}{\alpha ^3}\Big)\Big].
\end{align}
To Taylor expand the quantity in the exponential, we require that $s^3/\alpha$ is small for large $\alpha \gg 1$. Thus, we update our requirement to the tighter condition $R \ll \alpha^{1/3}$, {which ensures $|s^3/\alpha|\ll 1$, $|s^4/\alpha^2|\ll 1$ and $|s^5/\alpha^3|\ll 1$}. Thus, the quantity in the exponential can be Taylor expanded as
\begin{align}
    e^{-\alpha^2} \frac{\alpha^{2l}}{l!} = \frac{e^{-\frac{s^2}{2}}}{\sqrt{2 \pi} \alpha} \times \Big[1+\frac{1}{ \alpha}\Big(\frac{s^3-3s}{6}\Big) + \frac{1}{\alpha ^2}\Big(\frac{s^6-12s^4+27s^2-6}{72}\Big) +\mathcal{O}\Big(\frac{s^9}{\alpha^3}\Big)\Big].
\end{align}
\end{proof}

\newpage

\section{Numerical Verification of Approximations derived in Appendix \ref{appendix: matrix elements large alpha}}\label{app num verify all 2}

\subsection{Moments of coherent thermal states (Numerical Verification of Lemma \ref{eqn: matrix elements final})}\label{app num verify}
{To numerically verify Lemma \ref{eqn: matrix elements final}, here we consider the following 3 moments for coherent thermal state}:
\begin{align}
    &\nu_1(\beta, \alpha)\equiv \Tr(a\  \rho(\beta, \alpha)) = \sum_{l=0}^\infty \sqrt{l+1}\  \rho_{l,l+1} = \alpha \label{moment 1 test}\\
    &\nu_2(\beta, \alpha)\equiv \Tr(a^\dagger a\  \rho(\beta, \alpha)) = \sum_{l=0}^\infty {l}\  \rho_{l,l} = |\alpha|^2 + n_\beta \label{moment 2 test}\\
    &\nu_3(\beta, \alpha)\equiv \Tr(a^\dagger a^\dagger a\  \rho(\beta, \alpha)) = \sum_{l=0}^\infty {l}\sqrt{l+1}\  \rho_{l,l+1} =   2 \alpha^* n_\beta + \alpha^* |\alpha|^2\label{moment 3 test}
\end{align}
We can test our approximate formulae for $\rho_{l,l}$ and $\rho_{l,l+1}$ in Lemma \ref{eqn: matrix elements final} by seeing whether they numerically recover the results of Eq.(\ref{moment 1 test}-\ref{moment 3 test}). We test this and find agreement for a wide range of $\beta$ and $\alpha$ in Fig.~\ref{numverifyrho}.

{To compute the summation in Eq.(\ref{moment 1 test}-\ref{moment 3 test}), we sum over $l$ in the \textit{typical} range $[\alpha^2-50 \sigma, \alpha^2+50 \sigma]$. Then, the error in this estimate for these three estimates (indexed here by 
$i=1,2,3$) becomes
\begin{align}\label{error moment verify}
    E_i \equiv \sum_{l \in [\alpha^2-50 \sigma, \alpha^2+50 \sigma]} x_i(l) \times \Bigg(\frac{e^{-\frac{r^2}{2}}}{\sqrt{2 \pi}\sigma} \times \mathcal{O}\Big(\frac{(1+2n_\beta)^3}{\sigma^3} r^9\Big) + \mathcal{O}(e^{-\sqrt{\alpha}/n_\beta})\Bigg)\ ,
\end{align}
where $r=(l-\alpha^2)/\sigma$, and $x_1(l)=\sqrt{l+1}$, $x_2(l)=\sqrt{l}$ and $x_3(l)=l\sqrt{l+1}$. The numerical estimate yields an error substantially lesser than the expected error $E_i$.}

\subsection{Comparison with Edgeworth Series}
There are standard ways of computing the deviations of a distribution from a Gaussian in terms of its higher cumulants beyond the mean and variance. For instance, we can consider the \textit{Edgeworth} series or the \textit{Gram Charlier A} series for a single sample $n=1$, upto the fourth cumulant \cite{nearlyGaussian}. Specifically, a distribution $\Phi_{l}$ can be approximated as
\begin{align}\label{edgeworth}
    \Phi^\text{approx}_l \approx \frac{e^{\frac{-r^2}{2}}}{\sqrt{2 \pi} \sigma} \Big[1+ \frac{\kappa_3 H_3(r)}{6 \sigma^3} + \frac{\kappa_4 H_4(r)}{24 \sigma^4}\Big] \ , 
\end{align}
where $r=(l-\mu)/\sigma$, $\kappa_i$ is the $i^\text{th}$ cumulant, and $H_i$ is the $i^\text{th}$ Legendre polynomial. 

We shall apply this approximation to the diagonal elements of a coherent thermal state $\rho_{l,l}$ by substituting its culumants in Eq.(\ref{edgeworth}). Specifically, its third and fourth culumants are
\begin{align}
    &\kappa_3=\alpha ^2+2 n_\beta^3+6 \alpha ^2 n_\beta^2+3 n_\beta^2+6 \alpha ^2 n_\beta+n_\beta\ , \\
    &\kappa_4=\alpha ^2+6 n_\beta^4+24 \alpha ^2 n_\beta^3+12 n_\beta^3+36 \alpha ^2 n_\beta^2+7 n_\beta^2+14 \alpha ^2 n_\beta+n_\beta\ .
\end{align}
In Figure \ref{edgew}, we plot the relative error between the diagonal elements from Edgeworth approximation and the exact formula as seen in Eq.(\ref{eq: rho_ml final}). Similarly, we consider the same for the approximation in Lemma \ref{eqn: matrix elements final}. We observe that our approximation is significantly  better than the standard Edgeworth approximation.

\begin{figure}
    \centering
    \includegraphics[width=0.9\textwidth]{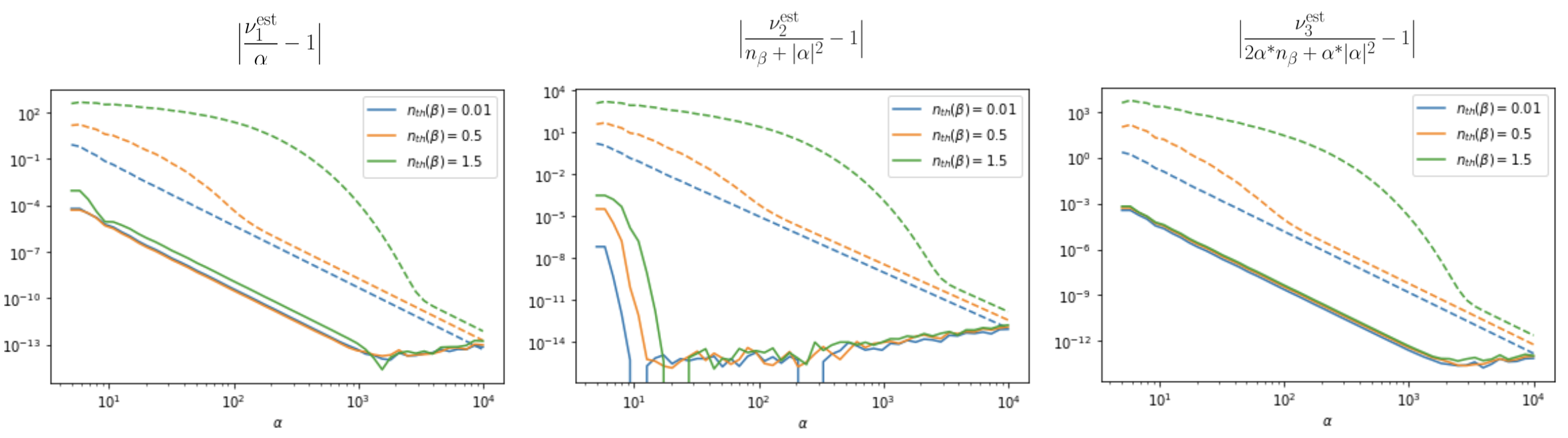}
\caption{\textbf{Error in estimating moments -- }{In this plot  $\nu_i^\text{est}(\beta, \alpha): i=1,2,3$
denote the numerical estimates of the moments  $\nu_i(\beta, \alpha): i=1,2,3$ calculated in Eq.(\ref{moment 1 test}-\ref{moment 3 test}). These numerical estimates are calculated by applying the approximation formulae in Lemma \ref{eqn: matrix elements final}. We plot the relative error of estimation $|\nu_i^\text{est}(\beta, \alpha)/\nu_i(\beta, \alpha)-1|$ in bold lines as a function of $\alpha$ for various temperatures. To compute the summation in Eq.(\ref{moment 1 test}-\ref{moment 3 test}), we sum over $l$ in the \textit{typical} range $[\alpha^2-50 \sigma, \alpha^2+50 \sigma]$, yielding a truncation error of $\mathcal{O}(e^{-50})\approx \mathcal{O}(10^{-22})$, which is clearly ignorable for purposes of this plot. We also plot the expected error $E_i/\nu_i^\text{est}: i=1,2,3$ for each of the estimates based on Eq.(\ref{error moment verify}) as a dashed line with the corresponding color. As expected, the relative error reduces with increasing $\alpha$, all the while remaining less than the expected error (dashed line). But for the first and third-moment estimates, between $\alpha\sim10^3-10^4$, we hit the limits of machine precision, i.e., our computation encounters quantities smaller than $10^{-16}$ (E.g., $r$ would take values $\mathcal{O}(1/\alpha)$ near the mean, and the formulae in Lemma \ref{eqn: matrix elements final} have $r^6$ terms, which would clearly be less that $10^{-16}$ for $\alpha>\  \sim 450$). Similarly, in the second moment estimate, the relative errors themselves would be approaching the limits of machine precision. This is clear from the fact that all the estimates become erratic once they reach $\sim 10^{-14}$.}}
    \label{numverifyrho}
\end{figure}

\begin{figure}
    \centering
    \includegraphics[width=0.33\textwidth]{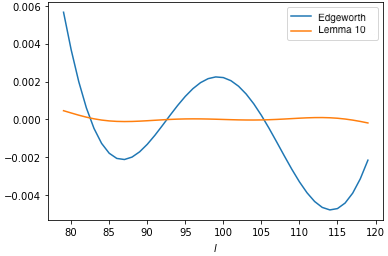}
    \caption{\textbf{Comparison between approximation for $\rho_{l,l}$ from the Edgeworth approximation, and  Lemma \ref{eqn: matrix elements final} approximation --} We compute the relative error of $\rho_{l,l}$ from these approximations with the exact expression from Eq.(\ref{eq: rho_ml final}). Here, we plot the relative error for the case where $\gamma_\text{in}=10$ across various values of $l$. Clearly, the Lemma \ref{eqn: matrix elements final} approximation is significantly better than the Edgeworth approximation.}
    \label{edgew}
\end{figure}

\begin{figure}
    \centering
    \includegraphics[width=0.33\textwidth]{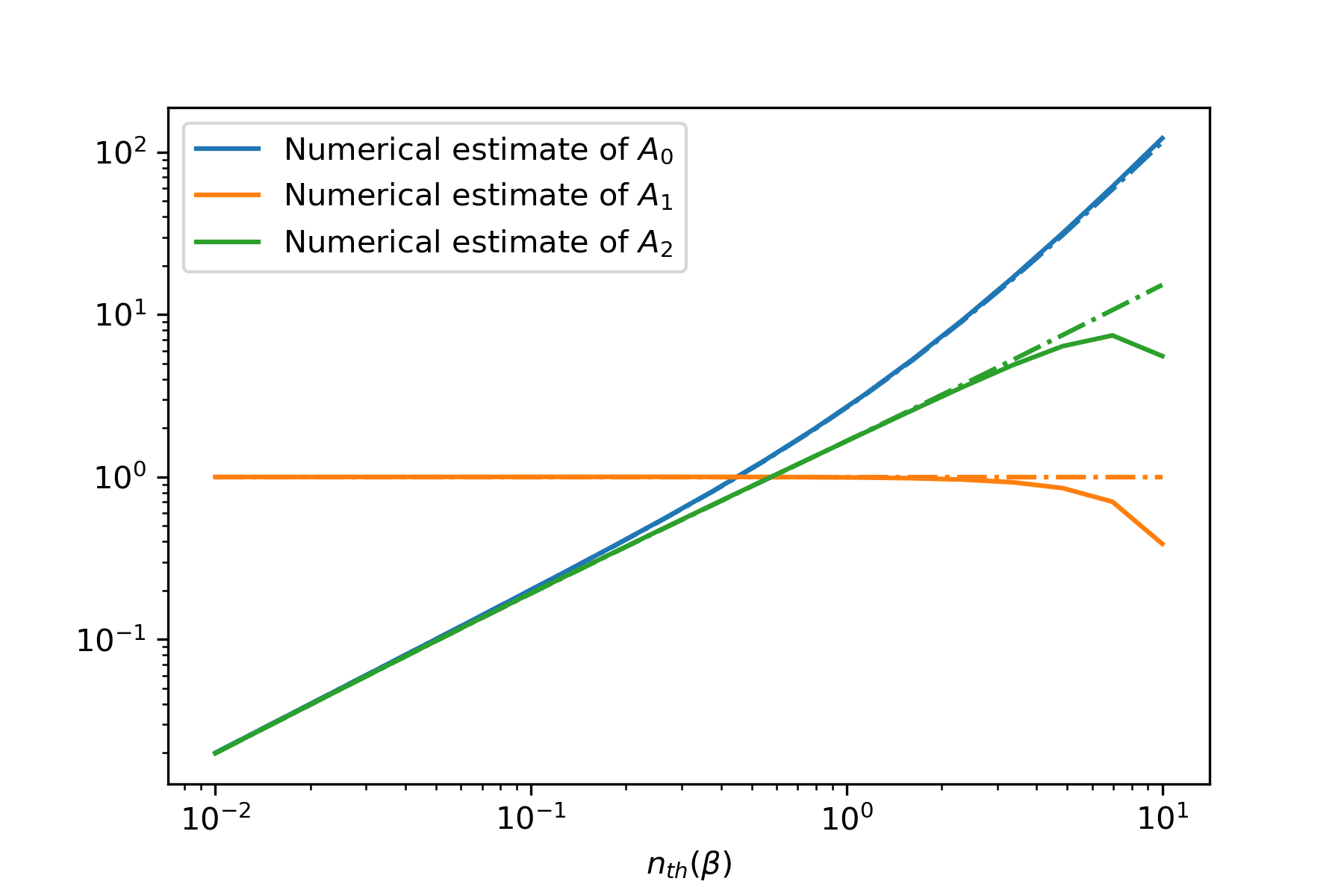}
    \includegraphics[width=0.28\textwidth]{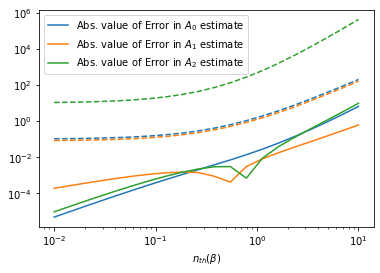}
    \caption{{\textbf{Numerical verification of quadratic approximation for $|c_l^\text{opt}|^2$ in Eq.(\ref{eq: c_l}) -- }}
    We can numerically estimate the co-efficients of the quadratic polynomial in Eq.(\ref{eq: c_l}) by computing the `zero-th', first and second discrete derivatives of $|\rho_{l-1,l-1}/\rho_{l-1,l}|^2$ at $l=\alpha^2$ by evaluating the exact values at $l=\alpha^2, \alpha^2 \pm 1$ (see Eq.(\ref{eq: rho_ml final}) in the Appendix for the exact formula for $\rho_{l-1,l-1}$ and $\rho_{l-1,l}$). In the left figure, we fix $\alpha=10$ and observe that these numerical estimates (plotted as bold lines) closely match the exact expressions that are listed in Eq.(\ref{eqn: A coeffs}) (plotted as dotted lines) across a wide range of temperatures $n_\text{th}(\beta)$. Our approximation formulae assume {$\alpha \gg \max\{1,n_\beta^2\}$}, thus we only consider up till $n_\beta \sim \sqrt{10}$. Furthermore, in the right plot we compute the absolute value of the error between the numerical estimates and the exact expressions, and observe that they are within the order of magnitude of the approximation error as expected from Eq.(\ref{eq: c_l}) (plotted as dashed lines with the corresponding color). {Note that we have ignored the exponentially-suppressed term $\mathcal{O}(\exp[{-{\sqrt{\alpha}}/{n_\beta}}])$ in the error. Nonetheless, because $\mathcal{O}(\exp[{-{\sqrt{\alpha}}/{n_\beta}}])$ is positive, including it would have only increased the expected error in the plot, thus, further confirming the fitness of our approximation.}}
    \label{fig:c_l verify}
\end{figure}

\subsection{Numerical Verification of quadratic approximation of $|c_l^\text{opt}|^2$ in Eq.(\ref{c_l eqn})}
In this subsection, we numerically verify Eq.(\ref{c_l eqn}). We start by rewriting the expression as
\begin{align}\label{eq: c_l}
    |c_l^\text{opt}|^2 =1+\frac{A_0}{\sigma^2}+  {A_1} \frac{r}{\sigma}  -A_2 \frac{r^2}{\sigma^2}  +\mathcal{O}\big(\frac{(1+2 n_\text{th}(\beta))^3}{\sigma^3} %+\exp({-\frac{\sqrt{\gamma_\text{in}}}{n_\text{th}(\beta)}})
    \big) \ ,
\end{align}
  
where $r=(l-\alpha^2)/\sigma$ and
\begin{align}\label{eqn: A coeffs}
A_0=\frac{2 n_\text{th}(\beta) \times (n_\text{th}(\beta)+1)^2}{2n_\text{th}(\beta)+1} , 
\ A_1= 1 ,\ 
A_2=\frac{n_\text{th}(\beta) \times (3 n_\text{th}(\beta)+2)}{2n_\text{th}(\beta)+1}\ ,
\end{align}
and we ignore terms of order $\mathcal{O}(e^{-\sqrt{\alpha}/n_\beta})$. We also note that this approximation assumes $\alpha\gg \max\{n_\beta^2,1\}$. While we formally prove this approximation in Appendix~\ref{c_l subsection}, we numerically verify it in Fig.~\ref{fig:c_l verify}. Note that while $A_0$ and $A_2$ vanish at zero temperature, the coefficient of the linear term is independent of temperature, namely $A_1=1$. In particular, at zero temperature this formula simplifies to  $|c_l^\text{opt}|^2= l/ \alpha^2$, which indeed holds for any pure coherent state $|\alpha\rangle$, regardless of the value of  $\alpha$.

\newpage

\section{Performance of the strong-input weak-output optimal protocol (Proof of Eq.(\ref{eqn: moment asymptotics}))}\label{appendix: performance of strong-input weak output}
In this section, we study the performance of the strong-input weak-output optimal protocol, and its output state. The channel is defined as  
\be\label{KO channel}
\mathcal{K}(\cdot)=\sum_{l=0}^\infty K_l(\cdot) K_l^\dag\ ,
 \ee
with Kraus operators
\begin{align}\label{KO 1}
    K_l = \frac{1}{\sqrt{1+|c_l|^2 |\gamma_\text{out}|^2 }} \ketbra{0}{l} + \frac{ c_{l+1} \gamma_\text{out}}{\sqrt{1+|c_{l+1}|^2 |\gamma_\text{out}|^2 }} \ketbra{1}{l+1}\ ,
\end{align}
 where $c_0=0$, and arbitrary $c_l \in \mathbb{C}$ for $l\geq 0$. We pick an optimal choice of $c_l$ to be
 \begin{align}\label{opt c_l}
     c_l^\text{opt} = \frac{\rho_{l-1,l-1}}{\rho_{l-1,l}}\ 
 \end{align}
where $\rho_{m,n}\equiv \braket{m|\rho|n}$. For simplicity, we assume $\gamma_\text{in}$ and $\gamma_\text{out}$ are real and positive from here on. Clearly, for any input state $\rho$,  the output state is restricted to the 2D subspace spanned by $\{|0\rangle,|1\rangle\}$. Then, in this basis, the output state is given by
\begin{align}\label{tau app 1}
\tau=\mathcal{K}(\rho)=\begin{pmatrix}
    1-\Tr(a^\dagger a\ \tau) & \Tr(a\ \tau)\\
 \Tr(a\ \tau) & \Tr(a^\dagger a\ \tau)
\end{pmatrix} \ ,
\end{align}
{where  we have used the fact that $\Tr(a\ \tau)$ is real.}

\subsection{First Moment Computation}
From the definition of the channel $\mathcal{K}$ and its Kraus operators mentioned above, we can directly compute that 
\begin{align}\label{first moment eq1}
    \Tr(a\ \tau) &= \braket{1|\tau|0} \\
    &= \gamma_\text{out} \sum_{l=0}^\infty  c_{l+1} \braket{l+1|\rho|l} \times q_l\\
    &= \gamma_\text{out}\sum_{l=0}^\infty  \rho_{l+1,l} \times c_{l+1} \times q_l\\
    &= \gamma_\text{out}\sum_{l=0}^\infty \rho_{l,l} \times q_l,
\end{align}
where we noted in the first equality that $\tau$ is a state restricted to the 2D Hilbert space spanned by $\ket{0}$ and $\ket{1}$, in the last equality we substituted the optimal $c^{\text{opt}}_l$ along with the input $\rho=\rho(\beta, \gamma_\text{in})$, and defined
\begin{align}\label{q_l defn}
    q_l \equiv \frac{1}{\sqrt{1+|c_l^\text{opt}|^2 \gamma_\text{out}^2}} \times \frac{1}{\sqrt{1+|c_{l+1}^\text{opt}|^2 \gamma_\text{out}^2}}.
\end{align}\\

\textit{\textbf{Restricting summation to typical $l$ --}} In Appendix \ref{appendix: matrix elements large alpha}, for $\gamma_\text{in} \gg 1$, we considered the typical range of $l$, that is $l \in [\gamma_\text{in}^2 - R \sigma, \gamma_\text{in}^2 + R \sigma]$, where $\sigma=\sqrt{1+ 2n_\beta} \gamma_\text{in}$ and $R \ll \gamma_\text{in}^{1/3}$ (e.g., see Lemma \ref{eqn: matrix elements final}). {We can, for instance, implicitly assume that $R=\gamma_\text{in}^{1/4}$, given $\gamma_\text{in} \gg \max\{1, n_\beta^2\}$.} We shall split the summation into this typical and atypical regime
\begin{align}
    \Tr(\tau a)&=\gamma_\text{out} \sum_{l \in \text{typical}} \rho_{l,l} \times q_l + \gamma_\text{out} \sum_{l \not\in \text{typical}} \rho_{l,l} \times q_l \\ 
    &=\gamma_\text{out} \sum_{l \in \text{typical}} \rho_{l,l} \times q_l + \mathcal{O}(e^{-R})\label{eq1}
\end{align}
where the atypical summation is exponentially suppressed. This follows because $q_l\leq 1$ for any $c_l$ for all $l$. Thus, the probability weight in the atypical regime is {$\mathcal{O}(e^{-R})$ (see Sec. \ref{chernoff bound} for details regarding atypical contributions).} \\

\textit{\textbf{Simplifying $q_l$ for typical $l$ --}} As mentioned in Lemma \ref{c_l lemma}, within the typical range of $l$, {i.e., $l \in [\gamma_\text{in}^2 - R \sigma, \gamma_\text{in}^2 + R \sigma]$}
\begin{align}\label{m1 eq 1}
    |c_l^\text{opt}|^2 = 1 + \frac{r}{\sigma} + \mathcal{O}\Big(\frac{r^2}{\sigma^2}\Big)
\end{align}
for $\gamma_\text{in}\gg 1$, where $\sigma=\sqrt{1+2n_\beta} \gamma_\text{in}$ and $r=(l-\gamma_\text{in}^2)/\sigma$. Here, and here on, the $\mathcal{O}(.)$ suppresses a temperature-dependent constant that is independent of $\gamma_\text{in}$ and $\gamma_\text{out}$. Consequently, \begin{align}\label{m1 eq 2}
    |c_{l+1}^\text{opt}|^2&= 1+ \frac{r+1/\sigma}{\sigma} + \mathcal{O}\Big(\frac{r^2}{\sigma^2}\Big)\nonumber\\ &=|c_l^\text{opt}|^2+\mathcal{O}\big(\frac{1+r^2}{\gamma_\text{in}^{2}}\big) \ .
\end{align} 
Recall that $|r| \leq R \ll \gamma_\text{in}^{1/3}$. {Now consider the strong-input weak-output regime, i.e., satisfying the conditions $\gamma_\text{in} \gg 1$ while $\gamma_\text{in} \times  \gamma_\text{out}\ll 1$. Then $|r|/\gamma_\text{in} \ll 1$ and $|r^2|/\gamma_\text{in}^2 \ll 1$. Consequently Eq.(\ref{m1 eq 1}) simplifies to}
\begin{align}\label{c_l small}
|c^\text{opt}_l|^2\times \gamma_\text{out}^2   &= \gamma_\text{out}^2 + \gamma_\text{out}^2 \times \frac{r}{\sqrt{1+2n_\beta} \gamma_\text{in}} + \gamma_\text{out}^2 \times \mathcal{O}\Big(\frac{r^2}{\sigma^2}\Big)\nonumber\\ 
&= \gamma_\text{out}^2 + \mathcal{O}\Big(r \times \frac{\gamma_\text{out}^2}{\gamma_\text{in}}\Big) \ . 
\end{align}
  
Thus, clearly $$|c^\text{opt}_l|^2 \times \gamma_\text{out}^2 \ll 1$$ in the strong-input weak-output regime. Similarly, $$|c^\text{opt}_{l+1}|^2 \times \gamma_\text{out}^2 \ll 1$$  in the strong-input weak-output regime. We can see this by substituting Eq.(\ref{c_l small}) in Eq.(\ref{m1 eq 2})
\begin{align}
    |c^\text{opt}_{l+1}|^2\times \gamma_\text{out}^2   = |c^\text{opt}_{l}|^2 \times \gamma_\text{out}^2 + \mathcal{O}\big(\frac{1+r^2}{\gamma_\text{in}^{2}}\big) \times \gamma_\text{out}^2 = \gamma_\text{out}^2 + \mathcal{O}\Big(r \times \frac{\gamma_\text{out}^2}{\gamma_\text{in}}\Big) + \mathcal{O}\big(\frac{\gamma_\text{out}^2(1+r^2)}{\gamma_\text{in}^{2}}\big) \ .
\end{align}
Using this, we shall expand $q_l$ defined in Eq.(\ref{q_l defn}) in powers of $|c^\text{opt}_l|^2 \times \gamma_\text{out}^2$ and $|c^\text{opt}_{l+1}|^2 \times \gamma_\text{out}^2$. In the strong-input weak-output regime, $q_l$ in the typical regime expands to
\begin{align}\label{eq6}
    q_l & = \frac{1}{\sqrt{1+|c_l^\text{opt}|^2 \gamma_\text{out}^2}} \times \frac{1}{\sqrt{1+|c_{l+1}^\text{opt}|^2 \gamma_\text{out}^2}}\\ 
    &= 1 - \gamma_\text{out}^2 \times \frac{1}{2}\Big[|c_l^\text{opt}|^2 + |c_{l+1}^\text{opt}|^2\Big] + \mathcal{O}((|c_l^\text{opt}|^2+|c_{l+1}^\text{opt}|^2)^2 \gamma_\text{out}^4) \\
    &= 1 - \gamma_\text{out}^2|c_l^\text{opt}|^2 + \mathcal{O}\Big(\frac{1+r^2}{\gamma_\text{in}^{2}}\gamma_\text{out}^2\Big) + \mathcal{O}((1+\frac{r}{\gamma_\text{in}})^2 \gamma_\text{out}^4) \ , \label{eq2}
\end{align}
where we used Eq.(\ref{m1 eq 1}-\ref{m1 eq 2}) in the last equality. Substituting Eq.(\ref{eq2}) in Eq.(\ref{eq1}), we get
\begin{align}\label{tau 1 mid}
    \Tr(\tau a)= \gamma_\text{out} &-  \sum_{l \in \text{typical}} \rho_{l,l} \times \Bigg[ \gamma_\text{out}^3 |c_l^\text{opt}|^2\Bigg]\nonumber\\ 
    &+ \sum_{l \in \text{typical}} \rho_{l,l} \times \Bigg[\mathcal{O}\Big(\frac{1+r^2}{\gamma_\text{in}^{2}}\gamma_\text{out}^3\Big) + \mathcal{O}((1+\frac{r}{\gamma_\text{in}})^2 \gamma_\text{out}^5) \Bigg] + \mathcal{O}(e^{-R})
\end{align}
when $\gamma_\text{in} \gg 1$ while $\gamma_\text{in} \times  \gamma_\text{out}\ll 1$ .

\textit{\textbf{Evaluating exact error terms -- }} In Eq.(\ref{Def:E}), we defined
\begin{align}\label{eq5}
    E(\gamma_\text{in}) \equiv \sum_{l=0}^\infty \rho_{l,l} [|c_l^\text{opt}|^2 -1] \ .
\end{align}
Furthermore in Eq.(\ref{atypical c_l limit}), we noted that the atypical contribution to this summation is
\begin{align}
    \sum_{l \not \in \text{typical}} \rho_{l,l} [|c_l^\text{opt}|^2 -1] = \mathcal{O}(e^{-R}).
\end{align}
Thus, 
\begin{align}\label{E typical}
    \sum_{l \in \text{typical}} \rho_{l,l} [|c_l^\text{opt}|^2 -1] = E(\gamma_\text{in})\ ,
\end{align}
{where we ignored the $\mathcal{O}(e^{-R})$ term because it was already considered in Eq.(\ref{tau 1 mid}).} Substituting this into Eq.(\ref{tau 1 mid}) yields
\begin{align}\label{tau 1 mid 2}
    \Tr(\tau a)= \gamma_\text{out} &- \gamma_\text{out}^3 [1+E(\gamma_\text{in})]\nonumber\\
    &+ \sum_{l \in \text{typical}} \rho_{l,l} \times \Bigg[ \mathcal{O}\Big(\frac{1+r^2}{\gamma_\text{in}^{2}}\gamma_\text{out}^3\Big) + \mathcal{O}((1+\frac{r}{\gamma_\text{in}})^2 \gamma_\text{out}^5) \Bigg] + \mathcal{O}(e^{-R}) \ .
\end{align}
In Eq.(\ref{k12}) in Appendix \ref{limit comp app}, we showed that
\bes\label{eq4}
\begin{align}
    &\sum_{l=0}^\infty\rho_{l,l} \times r = \frac{n_\beta}{\sqrt{1+2n_\beta}}\frac{1}{\gamma_\text{in}}\ ,\\
    &\sum_{l =0}^\infty \rho_{l,l} \times r^2 = 1 + \frac{n_\beta}{\gamma_\text{in}^2}\ ,
\end{align}
\ees
where $r=(l-\gamma_\text{in}^2)/\sigma$. 
Using Eq.(\ref{eq4}) we can evaluate the error term in Eq.(\ref{tau 1 mid 2}) {because they are polynomials in $r$ with the leading power in $r$ indicated in $\mathcal{O}(.)$, and $r$ is the only $l$-dependent quantity. Note that because the atypical contributions are exponentially small, we shall replace the summation to be over \textit{all} $l$ at a cost that is exponentially small in $R$, i.e., $\mathcal{O}(e^{-R})$. Upon evaluating this summation using Eq.(\ref{eq4}), the error terms simplify to}
\begin{align}
    &\sum_{l \in \text{typical}} \rho_{l,l} \times \mathcal{O}\Big(\frac{1+r^2}{\gamma_\text{in}^{2}}\gamma_\text{out}^3\Big) = \mathcal{O}(\gamma_\text{in}^{-2} \gamma_\text{out}^3) , \\
    &\sum_{l \in \text{typical}} \rho_{l,l} \times \mathcal{O}((1+\frac{r}{\gamma_\text{in}})^2 \gamma_\text{out}^5) = \mathcal{O}((1+\gamma_\text{in}^{-2})\gamma_\text{out}^5) \ .
\end{align}
{where we again ignored the $\mathcal{O}(e^{-R})$ terms. Noting that $E(\gamma_\text{in})=\mathcal{O}(\gamma_\text{in}^{-2})$ as seen in Eq.(\ref{E appendix})}, the first moment finally becomes
\begin{align}\label{final m1}
    \Tr(\tau a)= \gamma_\text{out} - {\gamma_\text{out}^3} + \mathcal{O}(\gamma_\text{in}^{-2} \gamma_\text{out}^3) + \mathcal{O}((1+\gamma_\text{in}^{-2})\gamma_\text{out}^5) + \mathcal{O}(e^{-R}),
\end{align}
when $\gamma_\text{in} \gg 1$ while $\gamma_\text{in} \times  \gamma_\text{out}\ll 1$, and where the $\mathcal{O}(.)$ indicates a temperature-dependent constant that is independent of $\gamma_\text{in}$ and $\gamma_\text{out}$, and $R \ll \gamma_\text{in}^{1/3}$.

\subsection{Second Moment Computation}
Just as above, from the definition of $\mathcal{K}$ and its Kraus operators with the optimal $c_l^\text{opt}$, the second moment becomes
\begin{align}
    \Tr(a^\dagger a\ \tau)=\langle 1|\tau|1\rangle&=\sum_{l=0}^\infty \rho_{l+1,l+1} \times g_{l+1}\\
    &= \sum_{l=0}^\infty \rho_{l,l} \times g_{l} \ 
\end{align}
where we used that $c_0^\text{opt}=0$ implies $g_0=0$ in the last equality, and defined
$$g_{l} \equiv \frac{|c_{l}^\text{opt}|^2 \gamma_\text{out}^2}{1+ |c_{l}^\text{opt}|^2 \gamma_\text{out}^2}\ .$$

\textit{\textbf{Restricting summation to Typical $l$ --}} We proceed in lines similar to the first moment computation. We first split the summation into the typical and atypical ranges of $l$. And then, because $g_l\leq 1$ for all $l$, the atypical summation is $\mathcal{O}(e^{-R})$ where $R \ll \gamma_\text{in}^{1/3}$. This yields
\begin{align}
    \Tr(a^\dagger a\ \tau) &= \sum_{l \in \text{typical}} \rho_{l,l} \times g_{l} +  \sum_{l \not\in \text{typical}} \rho_{l,l} \times g_{l}\\
    & = \sum_{l \in \text{typical}} \rho_{l,l} \times g_{l} +  \mathcal{O}(e^{-R})\label{eq10} \ .
\end{align}

\textit{\textbf{Simplifying $g_l$ in typical $l$ --}} 
We saw earlier that Eq.(\ref{c_l small}) in the strong-input weak-output regime for typical $l$, i.e., $l \in [\gamma_\text{in}^2 - R \sigma, \gamma_\text{in}^2 + R \sigma]$, implies that
\begin{align}
    |c_l^\text{opt}|^2 \times \gamma_\text{out}^2 \ll 1\ .
\end{align}
Therefore, we can expand $g_l$ in powers of $|c_l^\text{opt}|^2 \times \gamma_\text{out}^2$. That is,
\begin{align}\label{eq8}
    g_l &= \frac{|c_{l}^\text{opt}|^2 \gamma_\text{out}^2}{1+ |c_{l}^\text{opt}|^2 \gamma_\text{out}^2}\\
    &= \gamma_\text{out}^2 |c_{l}^\text{opt}|^2 + \mathcal{O}(|c_{l}^\text{opt}|^4 \gamma_\text{out}^4) \\
    &= \gamma_\text{out}^2 |c_{l}^\text{opt}|^2 + \mathcal{O}((1+ 2\frac{r}{\gamma_\text{in}}) \gamma_\text{out}^4) \label{eq11}
\end{align}
where, in the last equality we used that
\begin{align}\label{m2 eq 2}
    |c_l^\text{opt}|^4 = 1 + 2 \frac{r}{\sigma} + \mathcal{O}\Big(\frac{r^2}{\sigma^2}\Big) 
\end{align}
in the typical range of $l$. This follows from simply squaring Eq.(\ref{m1 eq 1}). {We also noted that $|r|/\gamma_\text{in} \ll 1$ and $|r^2|/\gamma_\text{in}^2 \ll 1$. Recall that we implicitly assumed $R=\gamma_\text{in}^{1/4}$.}\\

Substituting Eq.(\ref{eq11}) in Eq.(\ref{eq10}), we get
\begin{align}\label{tau 2 mid}
    \Tr(a^\dagger a \tau)= \sum_{l \in \text{typical}} \rho_{l,l} \times \Big[\gamma_\text{out}^2 |c_{l}^\text{opt}|^2\Big] + \sum_{l \in \text{typical}} \rho_{l,l} \times \Big[\mathcal{O}((1+2\frac{r}{\sigma})\gamma_\text{out}^4)\Big] + \mathcal{O}(e^{-R})\ .
\end{align}\\

\textit{\textbf{Evaluating exact error terms -- }} 
{We can evaluate the above summation using Eq.(\ref{E typical}) and Eq.(\ref{eq4}). Then, the second moment expression} in Eq.(\ref{tau 2 mid}) finally becomes
\begin{align}\label{final m2}
    \Tr(a^\dagger a\ \tau)=  \gamma_\text{out}^2 [1+ E(\gamma_\text{in})] + \mathcal{O}((1+\gamma_\text{in}^{-2})\gamma_\text{out}^4) + \mathcal{O}(e^{-R})\ 
\end{align}
when $\gamma_\text{in} \gg 1$ while $\gamma_\text{in} \times  \gamma_\text{out}\ll 1$, and where the $\mathcal{O}(.)$ indicates a temperature-dependent constant that is independent of $\gamma_\text{in}$ and $\gamma_\text{out}$, and $R \ll \gamma_\text{in}^{1/3}$.

\subsection{Fidelity Computation}\label{fid comp app}
Recall the density matrix of $\tau$ defined in Eq.(\ref{tau app 1}),
\begin{align}
\tau=\mathcal{K}(\rho)=\begin{pmatrix}
    1-\Tr(a^\dagger a\ \tau) & \Tr(a\ \tau)\\
 \Tr(a\ \tau) & \Tr(a^\dagger a\ \tau)
\end{pmatrix} \ .
\end{align}
Its fidelity with the coherent state $|\gamma_\text{out}\rangle$
is 
\begin{align}
\langle\gamma_\text{out}|\tau|\gamma_\text{out}\rangle&=e^{-|\gamma_\text{out}|^2} (\tau_{00}+|\gamma_\text{out}|^2 \tau_{11}+2\gamma_\text{out}\tau_{01})\\ &=e^{-|\gamma_\text{out}|^2} (1-(1-|\gamma_{\text{out}}|^2)\Tr(\tau a^\dag a) +2\gamma_\text{out}\Tr(\tau a))  \ ,  
\end{align}
where we have assumed $\gamma_\text{out}$ is real. Using the first and second moment expressions in Eq.(\ref{final m1}) and Eq.(\ref{final m2}), for $\gamma_\text{in} \gg 1$ while $\gamma_\text{in} \times  \gamma_\text{out}\ll 1$, we simplify the expression to
\begin{align}
    \langle\gamma_\text{out}|\tau|\gamma_\text{out}\rangle & = (1-\gamma_\text{out}^2) \Big[1- (1-\gamma_\text{out}^2)\gamma_\text{out}[1+E(\gamma_\text{in})] + 2 \gamma_\text{out} (\gamma_\text{out}-{\gamma_\text{out}^3[1 + \mathcal{O}(\gamma_\text{in}^{-2})]})\Big] + \mathcal{O}(\gamma_\text{out}^5)\\
    & = (1-\gamma_\text{out}^2)\nonumber\\ 
    &\ \ \ \ \  - \ (1-\gamma_\text{out}^2)^2 \gamma_\text{out}[1+E(\gamma_\text{in})]\nonumber\\ 
    &\ \ \ \ \  + \ 2 \gamma_\text{out} (1- \gamma_\text{out}^2)(\gamma_\text{out}-{\gamma_\text{out}^3[1+ \mathcal{O}(\gamma_\text{in}^{-2})]}) + \mathcal{O}(\gamma_\text{out}^5) \label{fid0}
\end{align}
Here, we have focused on terms that only contribute up to $\gamma_\text{out}^4$. Now we further expand the second and third term in the above expression in powers of $\gamma_\text{out}$, which yields
\begin{align}
    (1-\gamma_\text{out}^2)^2 \gamma_\text{out}[1+E(\gamma_\text{in})] = [1+E(\gamma_\text{in})]\gamma_\text{out}^2 + 2[1+E(\gamma_\text{in})]\gamma_\text{out}^4 + \mathcal{O}(\gamma_\text{out}^6) \ , \label{fid1} 
\end{align}
and
\begin{align}
     {2 \gamma_\text{out} (1- \gamma_\text{out}^2)(\gamma_\text{out}-\gamma_\text{out}^3[1+ \mathcal{O}(\gamma_\text{in}^{-2})]) = 2 \gamma_\text{out}^2 - 2[2+\mathcal{O}(\gamma_\text{in}^{-2})] \gamma_\text{out}^4 + \mathcal{O}(\gamma_\text{out}^6)} \ . \label{fid2}
\end{align}
Finally, substituting Eq.(\ref{fid1}) and Eq.(\ref{fid2}) into Eq.(\ref{fid0}), {noting that $E(\gamma_\text{in})=\mathcal{O}(\gamma_\text{in}^{-2})$} and simplifying yields
\begin{align}\label{final infid 1}
1-\langle\gamma_\text{out}|\tau|\gamma_\text{out}\rangle= \gamma_\text{out}^2 E(\gamma_\text{in}) + \mathcal{O}((1-\gamma_\text{in}^{-2}) \gamma_\text{out}^4) \ ,
\end{align}
where the $\mathcal{O}(.)$ indicates a temperature-dependent constant that is independent of $\gamma_\text{in}$ and $\gamma_\text{out}$.\\

\subsubsection*{Explicit form of fidelity for numerics}
For the sake of completeness, we present the explicit form of the summation to compute the fidelity. For simplicity, we consider the desired output state in the form $\ket{\gamma_\text{out}}=\sqrt{1- \gamma_\text{out}^2} \ket{0} + \gamma_\text{out} \ket{1}$ because $\gamma_\text{out} \ll 1$. Then, from Eq.(\ref{KO channel}-\ref{opt c_l}) we get
\begin{align} \label{explicit fid formula} \braket{\gamma_\text{out}|\tau|\gamma_\text{out}} = \sum_{l=0}^\infty \rho_{l,l} \times  (f_1+2f_2) +  \rho_{l+1,l+1} \times f_3
\end{align}
where 
\begin{align}
    f_1 &= \frac{1-\gamma_\text{out}^2}{1+|c_l^{\text{opt}}|^2 \gamma_\text{out}^2} \ ,\\
    f_2 &= \frac{\gamma_\text{out}^2 \sqrt{1-\gamma_\text{out}^2}}{\sqrt{1+|c_l^{\text{opt}}|^2 \gamma_\text{out}^2}\sqrt{1+|c_{l+1}^{\text{opt}}|^2 \gamma_\text{out}^2}}\ ,\\
    f_3 &= \frac{|c_{l+1}^{\text{opt}}|^2 \gamma_\text{out}^4}{1+|c_{l+1}^{\text{opt}}|^2 \gamma_\text{out}^2} \ .
\end{align}

\subsubsection{Truncation error in numerical verification of Eq.(\ref{kjj})}\label{final infid}

In Fig.(\ref{fig: num verify E}) we numerically compute this fidelity by considering 400 terms in the summation in Eq.(\ref{explicit fid formula}). The contribution to the fidelity beyond 400 terms is
\begin{align}
    \sum_{l=400}^\infty \rho_{l,l} \times (f_1+2f_2) +  \rho_{l+1,l+1} \times f_3  \leq 4 \sum_{l=400}^\infty \rho_{l,l}
\end{align}
because $f_i \leq 1: i=1,2,3$, when $\gamma_\text{out}<1$. {In our plot, $\gamma_\text{in}$ ranges from $1$ to $5$. From the Chernoff bound stated in Eq.(\ref{m0 r}), it is clear that considering $400$ terms for the largest $\gamma_\text{in}=5$ corresponds choosing $R\approx (400-5^2)/5$ at least. This ensures a truncation error at most of order $\mathcal{O}(e^{-75}) \approx \mathcal{O}(10^{-33})$. Consequently, the contribution to the RHS of Eq.(\ref{kjj}) would be at most of order $\mathcal{O}(10^{-25})$ if we consider $\gamma_\text{out}=10^{-4}$ used in the plot.}

\subsection{Summary}
{Recall from Eq.(\ref{E appendix}) that $E(\gamma_\text{in})=\mathcal{O}(\gamma_\text{in}^{-2})$. Then, Eq.(\ref{final m1}), Eq.(\ref{final m2}) and Eq.(\ref{final infid 1}), up to leading order in powers of $\gamma_\text{out}$ and $\gamma_\text{in}^{-1}$ simplify to}
\bes
\begin{align}
     \gamma_\text{out}-\Tr(\tau a)&=   \gamma_\text{out}^3 \times  [1+ \mathcal{O}(\gamma_\text{in}^{-2})] + \mathcal{O}(\gamma_\text{out}^5)\ , \\   
        \Tr(\tau a^\dagger a)-\gamma_\text{out}^2&=  \gamma_\text{out}^2 \times E(\gamma_\text{in}) +  \mathcal{O}(\gamma_\text{out}^4)\ ,
    \\
   1-\langle\gamma_\text{out}|\tau|\gamma_\text{out}\rangle
 &= \gamma_\text{out}^2 \times E(\gamma_\text{in})
+\mathcal{O}(\gamma_\text{out}^4)\ ,
\end{align}
\ees
As mentioned, $\mathcal{O}(.)$ here indicates a temperature-dependent constant that is independent of $\gamma_\text{in}$ and $\gamma_\text{out}$. We have also ignored the exponentially-supressed term $\mathcal{O}(e^{-R})$, where $R \ll \gamma_\text{in}^{1/3}$.

\newpage
\section{Error analysis for the optimal protocol (Proof of Theorem \ref{optimal theorem})}\label{Sec:overview}

In this section, we prove Theorem \ref{optimal theorem}. To this end, we use Lemma \ref{lemma: divide and distill} to upperbound the infidelity in terms of the first and second moments, and then use Eq.(\ref{eqn: moment asymptotics}) to analyze their asymptotics when $n \rightarrow \infty$.\\

{In Theorem \ref{optimal theorem}, for concreteness we set $B_n = \lfloor n^{3/4}\rfloor$. To simplify our analysis slightly, we shall consider $B_n = \lceil n^{3/4}\rceil$ in the following.} This implies
\begin{align}\label{eqn: gamma asymptotics}
    \gamma_\text{in}=\frac{\sqrt{n}\alpha}{\sqrt{B_n}}\leq\alpha n^{1/8}  \ \  \text{and}\ \  \   \gamma_\text{out}=\frac{\alpha}{\sqrt{B_n}}\leq\alpha n^{-3/8} \ ,
\end{align}
with $n={|\gamma_\text{in}|^2}/{|\gamma_\text{out}|^2}$. Then, observe that $\mathcal{C}_{B_n}\circ \mathcal{K}_n^{\otimes B_n}\circ\mathcal{D}_{B_n}$ is an instance of the scheme in Eq.(\ref{divide distill scheme}) in Lemma \ref{lemma: divide and distill} where $m=B_n$. 

As shown in Fig. \ref{fig: divide distill}, let $\sigma_n$ be the output of the channel $\mathcal{K}_n$ whose desired target is $\ket{\alpha/\sqrt{B_n}}$. Similarly, let $\sigma'_n$ be the output of the channel $\mathcal{C}_{B_n}\circ \mathcal{K}_n^{\otimes B_n}\circ\mathcal{D}_{B_n}$ whose desired target would consequently be $\ket{\alpha}$.

Then from Eq.(\ref{eqn: infid divide distill}), it follows that
\begin{align}
    1-\braket{\alpha|\sigma'_n|\alpha}
     &\leq \Tr(\sigma'_n a^\dagger a) - |\Tr(\sigma'_n a)|^2 +
     |\Tr(\sigma'_n a)-\alpha|^2\\
     &=\Tr(\sigma_n a^\dagger a) - |\Tr(\sigma_n a)|^2 + B_n \Big|\Tr(\sigma_n a)-\frac{\alpha}{\sqrt{B_n}}\Big|^2\ \label{infid 7}
\end{align}
for all $n$. Then, $n$ times the infidelity further satisfies
\begin{align}\label{tr1}
    n[1-\braket{\alpha|\sigma'_n|\alpha}]
     &\le n\big[\Tr(\sigma_n a^\dagger a) - \gamma_\text{out}^2\big]
     \\ &+ n\big[\gamma_\text{out}^2-\Tr(\sigma_n a)^2\big]\nonumber\\ 
    &+ n {B_n}[
    \gamma_\text{out}-\Tr(\sigma_n a)]^2\nonumber.
\end{align}
The asymptotic behavior of the quantities in each line in Eq.(\ref{tr1}) can be inferred using the first and second moments expressions in Eq.(\ref{eqn: moment asymptotics}) (these are computed in Appendix \ref{appendix: performance of strong-input weak output}). Specifically, from the first moment in Eq.(\ref{eqn: m1 main}) its follows that $0\le \gamma_\text{out}-\Tr(\sigma_n a)$. This implies that the second line in Eq.(\ref{tr1}) is upperbounded by
\begin{align}
n\big[\gamma_\text{out}^2-\Tr(\sigma_n a)^2\big] &\le 2n \gamma_\text{out}[\gamma_\text{out}-\Tr(\sigma_n a)]\\ 
&= \frac{2\gamma_\text{in}^2}{\gamma_\text{out}}[\gamma_\text{out}-\Tr(\sigma_n a)] \ \nonumber.
\end{align}
By substituting the same Eq.(\ref{eqn: m1 main}) into the above equation, we get that
\begin{align}
    &\frac{2\gamma_\text{in}^2}{\gamma_\text{out}}[\gamma_\text{out}-\Tr(\sigma_n a)] \\
    &= \frac{2\gamma_\text{in}^2}{\gamma_\text{out}} \Big[\gamma_\text{out}^3 \big[ 1 + \mathcal{O}(\gamma_\text{in}^{-2})\big] + \mathcal{O}(\gamma_\text{out}^5)\Big]\nonumber\\
    &= 2 (\gamma_\text{in} \times \gamma_\text{out})^2 \times \Big[\big[ 1 +\mathcal{O}(\gamma_\text{in}^{-2})\big] + \mathcal{O}(\gamma_\text{out}^2)\Big] \ . \nonumber
\end{align}
{Since $\gamma_{\text{in}}\times \gamma_{\text{out}}\leq \alpha^2 n^{-1/4}$ from Eq.(\ref{eqn: gamma asymptotics})}, in the limit  $n\rightarrow \infty$, $\gamma_{\text{in}}\times \gamma_{\text{out}}\rightarrow 0$, and clearly, $\gamma_\text{in}^{-1}, \gamma_\text{out} \rightarrow 0$. Thus, the second line of Eq.(\ref{tr1}) vanishes in the limit $n\rightarrow \infty$. Similarly, the square root of the third line of Eq.(\ref{tr1}) equals
\begin{align}
    &\sqrt{n B_n} [\gamma_\text{out} - \Tr(\sigma_n a)]\\
    &= \frac{\gamma_\text{in}}{\gamma_\text{out}} \times \frac{\alpha}{\gamma_\text{out}} \times \Big[\gamma_\text{out}^3 \big[ 1  + \mathcal{O}(\gamma_\text{in}^{-2})\big] + \mathcal{O}(\gamma_\text{out}^5)\Big]\nonumber\\
    &= \alpha \times (\gamma_\text{in} \times \gamma_\text{out}) \times \Big[\big[ 1 +\mathcal{O}(\gamma_\text{in}^{-2})\big] + \mathcal{O}(\gamma_\text{out}^2)\Big] \ ,\nonumber
\end{align}
where we note that $\sqrt{B_n} = \alpha/\gamma_\text{out}$ in the first equality. For the same reason as above, the term in the third line of Eq.(\ref{tr1}) also vanishes in the limit $n\rightarrow \infty$. Finally, considering the second moment expression in Eq.(\ref{eqn: m2 main}), the first line of Eq.(\ref{tr1}) becomes 
\begin{align}
&\lim_{n\rightarrow \infty} n\big[\Tr(\sigma_n a^\dagger a) - \gamma^2_{\text{out}}\big]\\
&= \lim_{n\rightarrow \infty} \gamma^2_{\text{in}}\times E(\gamma_{\text{in}}) + \mathcal{O}(\gamma_\text{in}^2 \times \gamma_\text{out}^2)\nonumber\\ 
&= \delta^\text{opt}(\beta)\nonumber \ ,
\end{align}
where we again note that $\gamma_{\text{in}}\times \gamma_{\text{out}}\rightarrow 0$ as $n \rightarrow \infty$, along with Eq.(\ref{limit E}) in the last line. We note that Eq.(\ref{limit E}) is proved in Appendix \ref{appendix: matrix elements large alpha}.

\subsection*{Infidelity in terms of $n$ and $B_n$}
We shall reconsider Eq.(\ref{tr1}) in terms of $n$ and $B_n$, using the relations \begin{align}\label{gamma n Bn}
    \gamma_\text{in}=\frac{\sqrt{n}\alpha}{\sqrt{B_n}}  \ \  \text{and}\ \  \   \gamma_\text{out}=\frac{\alpha}{\sqrt{B_n}} \ ,
\end{align}
with $n={|\gamma_\text{in}|^2}/{|\gamma_\text{out}|^2}$.

For large $\gamma_\text{in}$, or equivalently for $n\gg1$ when $B_n=o(n)$\footnote{Here, we use the asymptotic notations: 
$B_n=o(n)$ means $\lim_{n \rightarrow \infty} {B_n}/{n} = 0$ and $B_n=\omega(\sqrt{n})$ means $\lim_{n \rightarrow \infty} {B_n}/{\sqrt{n}} = \infty$.}, Eq.(\ref{E with error}) implies that these quantities scale (up to leading order in powers of $n$ and $B_n$) as 
\bes\label{moment scalings 2}
\begin{align}
    \gamma_\text{out} - \Tr(\sigma_n a) &= \frac{\alpha^3}{B_n^{3/2}} + \frac{\alpha\ (1+\delta^\text{opt}(\beta))}{n\sqrt{B_n}} + \mathcal{O}\Big(\frac{1}{B^{5/2}}\Big) + \mathcal{O}\Big(\frac{1}{n^{3/2}}\Big)\\
    \Tr(\sigma_n a^\dagger a) - \gamma_\text{out}^2 &= \frac{\delta^\text{opt}(\beta)}{n} + \mathcal{O}\Big(\frac{1}{B_n^2}\Big)+ \mathcal{O}\Big(\frac{\sqrt{B_n}}{n^{3/2}}\Big)
\end{align}
\ees
where $\mathcal{O}(.)$ suppresses a temperature and $\alpha$ dependent constant that is independent of $n$ and $B_n$. By considering only the leading order terms in Eq.(\ref{moment scalings 2}), the infidelity factor estimate in Eq.(\ref{tr1}) simplifies, up to leading order in powers of $n$ and $B_n$, to
\begin{align}
    n[1-\braket{\alpha|\sigma'_n|\alpha}]
     &\le \delta^{\text{opt}}(\beta) + \mathcal{O}\Big(\frac{1}{{B_n}}+\frac{1}{{n}}\Big) + \mathcal{O}\Big(\frac{n}{B_n^2}\Big)\ .
\end{align}
To ensure that this quantity approaches $\delta^\text{opt}(\beta)$ as $n \rightarrow \infty$, we require that $B_n$ satisfies the asymptotic scalings $o(n)$ \textit{and} $\omega(\sqrt{n})$. 
And clearly, any $B_n \propto n^{1-s}$, for an arbitrarily small $s>0$, would converge to $\delta^\text{opt}(\beta)$ the fastest with increasing $n$.

\subsubsection{General Requirements}
In the above discussion, we focused on finding the optimal distillation protocol, using the divide and distill protocol explained in Sec. \ref{subsection: divide distill}. Suppose, instead of the optimal channel $\mathcal{K}$ introduced in Sec. \ref{subsection: weak output regime}, we consider another channel $\mathcal{K}'$ within the divide and distill strategy to obtain a protocol for converting $n$ copies of state $\rho(\beta,\alpha)$ to a state close to $\alpha$. Specifically, consider the phase-insensitive channel $\mathcal{K}'$ that transforms the coherent thermal input $\rho(\beta, \gamma_\text{in})$ to {a state close to} coherent state $\ket{\gamma_\text{out}}$ when $\gamma_\text{in} \gg 1$ and $\gamma_\text{out} \ll 1$. In the following, we show that if the output state  $\sigma:=\mathcal{K}'(\rho(\beta, \gamma_\text{in}))$ satisfies
\bes\label{req 7}
\begin{align}
    &{\rm Var}(\rho):=\Tr(\rho a^\dagger a)-|\Tr(\rho a)|^2 = \mathcal{O}\Big(\frac{\gamma_\text{out}^2}{\gamma_\text{in}^2}\Big)\\
    & |\Tr(\sigma a)- \gamma_\text{out}| = \mathcal{O}\Big(\frac{\gamma_\text{out}^2}{\gamma_\text{in}}\Big) \ 
\end{align}
\ees
in the $\gamma_\text{in} \gg 1$ and $\gamma_\text{out} \ll 1$ regime, then the overall channel obtained from applying $\mathcal{K}'$ in the divide and distill strategy will have finite infidelity factor.

%In particular, as assumed above, we use this channel $B_n$ times in parallel, so that the input state to channel $\mathcal{K}'$ becomes the thermal state $\rho(\beta,\gamma_{\text{in}})$ with $\gamma_\text{in}= \sqrt{n}\alpha/\sqrt{B_n}$.  

%We ask under what conditions on the output of channel $\mathcal{K}'$ does the overall channel obtained from applying $\mathcal{K}'$ in divide and distill strategy have finite infidelity factor. 

To derive these conditions in Eq.(\ref{req 7}), we briefly recall the divide and distill protocol: we first apply passive operations to convert $n$ copies of $\rho(\beta, \alpha)$ to $B_n$ copies of the coherent thermal input $\rho(\beta, \gamma_\text{in})$ with $\gamma_\text{in}= \sqrt{n}\alpha/\sqrt{B_n}$. Then, we apply the phase-insensitive protocol $\mathcal{K}'$ that transforms each copy of $\rho(\beta, \gamma_\text{in})$ to a state $\sigma:=\mathcal{K}'(\rho(\beta, \gamma_\text{in}))$ that is close to $\ket{\gamma_\text{out}}$ with $\gamma_\text{out}=\alpha/\sqrt{B_n}$. This channel is applied $B_n$ times in parallel. Then, finally, the $B_n$ copies of $\sigma$ are combined into a single mode $\sigma'_n := \mathcal{C}_{B_n}(\sigma^{\otimes B_n})$. Then, as seen in Eq.(\ref{infid 7}), the infidelity of the output of this protocol $\sigma'_n$ with the coherent state $\ket{\alpha}$ satisfies
\begin{align}
    n\times [1- \braket{\alpha|\sigma'_n|\alpha}] \leq n \times {\rm Var}(\sigma) +  n B_n \times |\Tr(\sigma a)- \gamma_\text{out}|^2 \ , 
\end{align}
where ${\rm Var}(\rho):=\Tr(\rho a^\dagger a)-|\Tr(\rho a)|^2$. Now, recall that for a given $\alpha$, the parameters $n$ and $B_n$ uniquely specify $\gamma_\text{in}$ and $\gamma_\text{out}$. That is,
\begin{align*}
    \gamma_\text{in}=\frac{\sqrt{n}\alpha}{\sqrt{B_n}}  \ \  \text{and}\ \  \   \gamma_\text{out}=\frac{\alpha}{\sqrt{B_n}} \ .
\end{align*}
Using this, we can rewrite the right-hand side of the above inequality in terms of $\gamma_\text{in}$ and $\gamma_\text{out}$ as
\begin{align}
    n\times [1- \braket{\alpha'|\sigma'_n|\alpha'}] \leq \frac{\gamma_\text{in}^2}{\gamma_\text{out}^2} \times {\rm Var}(\sigma) +  \frac{\alpha^2 \gamma_\text{in}^2}{\gamma_\text{out}^4} \times |\Tr(\sigma a)- \gamma_\text{out}|^2 \ .
\end{align}
Furthermore, for the channel $\mathcal{K}'$ that we are considering, the limit $n \rightarrow \infty$ would be equivalent to the limit of $\gamma_\text{in} \rightarrow \infty$ and $\gamma_\text{out} \rightarrow 0$. Then, for the left- hand side quantity to remain bounded in the $\gamma_\text{in} \rightarrow \infty$ and $\gamma_\text{out} \rightarrow 0$ limits, we require Eq.(\ref{req 7}).\\

For our choice of optimal protocol detailed in Eq.(\ref{subsection: weak output regime}), from Eq.(\ref{eqn: moment asymptotics}) it follows that
\bes
\begin{align}
    &{\rm Var}(\sigma) = \gamma_\text{out}^2 E(\gamma_\text{in}) + \mathcal{O}(\gamma_\text{out}^4)
    \label{vv1}\\
    & |\Tr(\sigma a)- \gamma_\text{out}| = \gamma_\text{out}^3[1+ E(\gamma_\text{in})] +  \mathcal{O}(\gamma_\text{out}^5)
    \label{vv2}
\end{align}
\ees
where we only considered the leading order terms. These terms satisfy Eq.(\ref{req 7}). Specifically, Eq.(\ref{vv1}) is $\mathcal{O}\big(\frac{\gamma_\text{out}^2}{\gamma_\text{in}^2}\big)$ because in the limit $\gamma_\text{in} \rightarrow \infty$, the quantity $\gamma_\text{in}^2 E(\gamma_\text{in})$ is finite. Similarly, Eq.(\ref{vv2}) is $\mathcal{O}\big(\frac{\gamma_\text{out}^2}{\gamma_\text{in}}\big)$ because in the limit $\gamma_\text{in} \rightarrow \infty$ and $\gamma_\text{out} \rightarrow 0$, $\gamma_\text{out}^3/(\gamma_\text{out}^2/\gamma_\text{in}) =\gamma_\text{in} \times \gamma_\text{out} \rightarrow 0$ as is assumed in the strong-input weak-output regime. So, for that matter, Eq.(\ref{vv2}) is $o\big(\frac{\gamma_\text{out}^2}{\gamma_\text{in}}\big)$ which is stronger than the requirement for finite infidelity, and this ensures the contribution of the bias term in the infidelity factor vanishes asymptotically.

\newpage

\section{Quantum Fisher Information Metrics for Displaced Incoherent States (Proof of Lemma \ref{lemma: qfi computation})} \label{appendix: incoherent qfi computation}

Consider a closed system 
evolving under a (possibly time-dependent) Hamiltonian $H(t)$, with density operator $\rho(t)$, at time $t$.  Then, for any arbitrary unitary $U(t)$, we obtain 

\begin{align}
\frac{d}{dt}\rho(t)&=-\mathrm{i}[H, \rho]=-\mathrm{i}\ U [ \tilde{H}\ ,\   \tilde{\rho}] U^\dag\ ,
\end{align}
or equivalently 
\begin{align}
U^\dag[\frac{d}{dt}\rho(t)] U&=-\mathrm{i}  [ \tilde{H}\ ,\   \tilde{\rho}] \ ,
\end{align}
  
where we have suppressed the $t$ dependence in $U$ and $H$, and defined $\tilde{H}=U^\dag H U$ and 
$\tilde{\rho}= U^\dag\rho U$. 
As we see below this is particularly useful for finding the matrix elements of $d\rho(t)/dt$ in the eigenbasis of $\rho(t)$.

 We apply this observation to the case of a Harmonic Oscillator in the initial state $D(\alpha) \rho_\text{incoh} D(\alpha)^\dag$, evolving under Hamiltonian $H=\omega a^\dag a$.  This time evolution results in  the 1-parameter family of states
\begin{align}
    \rho(t)=e^{-\mathrm{i} \omega a^\dagger a  t}D(\alpha) \rho_\text{incoh} D^\dagger(\alpha) e^{\mathrm{i} \omega a^\dagger a t} = D(\alpha e^{-\mathrm{i}\omega t}) \rho_\text{incoh} D^\dagger(\alpha e^{-\mathrm{i}\omega t})\ \ \ \ : \ t \in \mathbb{R}
\end{align}
where $[\rho_\text{incoh}, H]=0$.  Our goal is to find the matrix elements of ${d}\rho(t)/dt$ in the eigenbasis of $\rho(t)$. 
Therefore,  we choose $U(t)=D(\alpha e^{-\mathrm{i} \omega t})$, which means 
\be
\tilde{\rho}(t)=U^\dag(t)\rho(t) U(t)=\rho_\text{incoh}\ .
\ee
Furthermore, because $D(\alpha e^{-\mathrm{i}\omega t})^\dag \ a\  D(\alpha e^{-\mathrm{i}\omega t})= a+\alpha e^{-\mathrm{i}\omega t}$, 
\be
\tilde{H}=U^\dag(t) H U(t)
=\tilde{H}(t)=\omega \Big[ a^\dag a +a^\dag \alpha e^{-\mathrm{i}\omega t}+a \alpha^\ast e^{\mathrm{i}\omega t}+|\alpha|^2\Big]\ .
\ee
  
Therefore, since  $[\rho_\text{incoh}, H]=0$  we  obtain 
\begin{align}
U^\dag(t) \frac{d}{dt}\rho(t) U(t)=D^\dag(\alpha e^{-\mathrm{i} \omega t}) \frac{d}{dt}\rho(t) D(\alpha e^{-\mathrm{i} \omega t})&=-i\omega  [ \alpha e^{-\mathrm{i}\omega t} a^\dag+\alpha^\ast e^{\mathrm{i}\omega t} a\ ,\   \rho_\text{incoh}]\ .
\end{align}
Finally, we note that  $\rho_{\text{incoh}}=\sum_n p_n |n\rangle\langle n|$ is diagonal in the Fock basis.
The matrix elements of the two sides of the above equation in this basis are given 
by 
\begin{align}
\langle k|U^\dag(t)[\dot{\rho}(t)]U(t)|j\rangle  \equiv\langle\widetilde{k}| \frac{d}{dt}\rho(t) |\widetilde{j}\rangle&= -\mathrm{i}\omega e^{\mathrm{i}(k-j)\omega t} \langle{k}| [ \alpha  a^\dag+\alpha^\ast  a \ ,\   \rho_\text{incoh}] | j\rangle\\
&=\mathrm{i}\omega\ e^{\mathrm{i}(k-j)\omega t} \times (p_k -p_j) \times (\alpha^* \braket{k|a|j} + \alpha \braket{k|a^\dagger|j})\ ,
\end{align}
where  $|\widetilde{n}\rangle=D(\alpha e^{-\mathrm{i} \omega t} )|n\rangle$ 
is the eigenbasis of ${\rho}(t)$. 
Then, using the fact that $\langle j|a^\dag|k\rangle=\sqrt{k+1}\delta_{j,k+1}$ we find that
\begin{align}
    |\langle\tilde{j}|\dot{\rho}(t)|\tilde{k}\rangle |^2 = \delta_{j, k+1} \ \omega^2(p_k - p_{k+1})^2 \times |\alpha|^2 \times (k+1)\ . 
\end{align}
Substituting this into Eq.(\ref{eqn: qfi metric}) yields
\begin{align}
    \textbf{g}^{f}_\rho(\dot{\rho}, \dot{\rho})={2\omega^2|\alpha|^2}\sum_k (k+1) p_{k+1} \times  \frac{({p_k}/{p_{k+1}}-1)^2}{f(p_k/p_{k+1})}\ .
\end{align}

\subsection{Derivation of formulae for QFI and purity of coherence (Proof of Eq.(\ref{eq: gen F_H P_H}))}\label{F_H and P_H deriv appendix}
Next, we derive the formulae for $F_H$ and $P_H$ defined in Eq.(\ref{eq: gen F_H P_H}) by noting that they are the minimal and maximal QFI metrics respectively. These correspond to the formulae previously noted in \cite{marvian2020coherence}. \\

Recall that a general QFI metric is characterized by a function $f: \mathbb{R}_+\rightarrow \mathbb{R}_+$, such that for a density operator $\rho$ with spectral decomposition 
\begin{align}\label{eq: rho spec}
    \rho=\sum_i p_i |i\rangle\langle i|\ ,
\end{align}
within a single-parameter family of states $\{\rho(t)\}_{t \in \mathbb{R}}$, it holds that
\begin{align}\label{eq: metric}
\textbf{g}^{f}_\rho(\dot{\rho}, \dot{\rho})=\sum_{i} \frac{\dot{\rho}^2_{ii}}{p_i} +2 \sum_{i<j} c_f(p_i,p_j) |\dot{\rho}_{ij}|^2\ ,
\end{align}
where $\rho_{ij}=\langle i|\rho|j\rangle$. Here, we  evaluate this metric for state $\rho$ in the single-parameter family defined by a Hamiltonian evolution by $H$,
\begin{align}\label{single param fam}
    \rho(t)=e^{-\mathrm{i}Ht} \rho  e^{\mathrm{i}Ht}
\end{align}
for initial state $\rho$. %For brevity, we shall denote the state at $t$ as $\rho \equiv \rho(t)$.

We simplify the right-hand side of Eq.(\ref{eq: metric}) by substituting the spectral decomposition of $\rho$ seen in Eq.(\ref{eq: rho spec}) into it. Using  $\dot{\rho} = - \mathrm{i} [H, \rho]$, we find that 
\bes
\begin{align} 
    \dot{\rho}_{ij}(t) &= \braket{i|\dot{\rho}(t)|j}\\
    &{= -\mathrm{i}\braket{i|[H,\rho(t)]|j}}\\
    &{={-\mathrm{i} e^{\mathrm{i} (E_i-E_j)t }}\braket{i|[H,\rho]|j}} \\
    &= -\mathrm{i} {e^{\mathrm{i} (E_i-E_j)t }} \bra{i}[H , \sum_k p_k \ketbra{k}{k}]\ket{j}\\
    &= - \mathrm{i} {e^{\mathrm{i} (E_i-E_j)t }}  \sum_k \bra{i}(p_k H\ketbra{k}{k} - p_k \ketbra{k}{k}H)\ket{j}\\
    & = -\mathrm{i} {e^{\mathrm{i} (E_i-E_j)t }} (p_i-p_j) \times \braket{i|H|j} \ .
\end{align}
\ees
Thus, the first term in Eq.(\ref{eq: metric}) is zero. Because $f$ is self-inverse, it follows that $yf(x/y)=xf(y/x)$. Thus, $c_f(x,y)=c_f(y,x)$. Using this, and that $|\braket{i|H|j}|=|\braket{j|H|i}|$, we can simplify the second term to get
\begin{align}\label{metric mid}
    \textbf{g}^{f}_\rho(\dot{\rho}, \dot{\rho})&= 2 \sum_{i<j} \frac{1}{p_j f(p_i/p_j)} \times (p_i-p_j)^2 \times |\braket{i|H|j}|^2\\
    &= \sum_{i,j} \frac{1}{p_j f(p_i/p_j)} \times (p_i-p_j)^2 \times |\braket{i|H|j}|^2 \ .
\end{align}
Now we recall the minimal and maximal normalized self-inverse functions, that respectively correspond to $F_H$ and $P_H$ functions
\begin{align}
f^{\text{SLD}}(x)=\frac{1+x}{2}\  ,\  \ 
f^{\text{RLD}}(x)=\frac{2x}{1+x}\ .
\end{align}
Then, 
\begin{align}\label{f simplify}
    &f^{\text{SLD}}(p_i/p_j) = \frac{p_i + p_j}{2 p_j}\  ,\  \ 
f^{\text{RLD}}(p_i/p_j) = \frac{2p_i}{p_i+p_j} \ .
\end{align}
Substituting these in Eq.(\ref{metric mid}) yields the formulae in Eq.(\ref{eq: gen F_H P_H}). We derive this in the following,
\bes
\begin{align}
    \textbf{g}^{\text{SLD}}_\rho(\dot{\rho}, \dot{\rho})&= \sum_{i,j} \frac{1}{p_j f^{\text{SLD}}(p_i/p_j)} \times (p_i-p_j)^2 \times |\braket{i|H|j}|^2\\
    & = \sum_{i,j} \frac{2(p_i-p_j)^2}{p_i+p_j} \times |\braket{i|H|j}|^2\ .
\end{align}
\ees
Similarly,
\bes
\begin{align}
    \textbf{g}^{\text{RLD}}_\rho(\dot{\rho}, \dot{\rho})&= \sum_{i,j} \frac{1}{p_j f^{\text{RLD}}(p_i/p_j)} \times (p_i-p_j)^2 \times |\braket{i|H|j}|^2\\
    &= \sum_{i,j} \frac{p_i+p_j}{2 p_i p_j} \times (p_i-p_j)^2 \times |\braket{i|H|j}|^2\\
    &=\frac12\sum_{i,j} \Big(\frac{1}{p_j} - \frac{1}{p_i}\Big) \times (p_i^2-p_j^2)  \times |\braket{i|H|j}|^2\\
    &=\frac12\sum_{i,j} \Big(\frac{1}{p_j}\Big) \times (p_i^2-p_j^2)  \times |\braket{i|H|j}|^2 - \frac12\sum_{i,j} \Big(\frac{1}{p_i}\Big) \times (p_i^2-p_j^2)  \times |\braket{i|H|j}|^2\\
    & = \sum_{i,j} \frac{p_i^2-p_j^2}{p_j}   \times |\braket{i|H|j}|^2 \ ,
\end{align}
\ees
where, in the last line, we observed that the second summation is the negative of the first summation. To see this, exchange the index labels $i,j$ in the second summation, and recall that $|\braket{i|H|j}|=|\braket{j|H|i}|$.
\newpage

\section{Ultimate Limits of Distillation Error (Proof of Eq.(\ref{thm first part}) in Theorem \ref{theorem: universal bounds})}\label{appendix: universal bounds}

Ref. \cite{marvian2020coherence} shows that if using TI operations $\mathcal{E}_n$,  $n$ copies of a system with state $\rho$ and Hamiltonian $H_1$ can be converted to a copy of system with Hamiltonian $H_2$ in state $\rho'$, such that $\rho'=\mathcal{E}_n(\rho^{\otimes n})$, and if  $\langle\phi|\rho'|\phi\rangle=1-\epsilon_n$, then 
\be
 \liminf_{n \rightarrow \infty} n \times \epsilon_n \geq \frac{V_{H_2}(\ketbra{\phi}{\phi})}{P_{H_1}(\rho)}\ .
\ee
Here, we give an overview of this argument. Furthermore, we verify that this claim remains valid even when Hamiltonians $H_1$ and $H_2$ are unbounded, provided that state $|\phi\rangle$ has finite first and second moments of energy (in  Ref.\cite{marvian2020coherence} this result is presented for the case of systems with bounded Hamiltonians. However, as we verify in the following, the result does not rely on this assumption). 

To simplify the notion we assume $H_1=H_2$. However, it is worth noting that this assumption can be easily relaxed. First, recall the following result of \cite{marvian2020coherence}:

\begin{lemma}\label{variance lowerbound lemma} \cite{marvian2020coherence}
Let  $\delta\equiv 1-\langle\phi|\sigma|\phi\rangle $ be the infidelity of pure state $\ket{\phi}$ and state $\sigma$. Let $\ket{\psi}$ be the eigenvector of $\sigma$ with the largest eigenvalue. Then, the fidelity of $\ket{\psi}$ and $\ket{\phi}$ is lower bounded by $|\langle\phi|\psi\rangle|^2\ge 1-2\delta$.  Furthermore, 
\begin{align}
P_H(\sigma)&\ge V_H(\ketbra{\psi}{\psi})\times (\frac{(1-\delta)^2}{\delta}-1) \label{1ododo} , \\ 
F_H(\sigma)&\ge V_H(\ketbra{\psi}{\psi})\times  4(1-2\delta)^2\ .\label{2ododo}
\end{align}
\end{lemma}
Consider the sequence of distillation channels $\{\mathcal{E}_n\}_n$ that outputs $\sigma_n:=\mathcal{E}_n(\rho^{\otimes n})$ for each $n$. Let 
$$\epsilon_n:= 1- \braket{\phi|\sigma_n|\phi}$$ be
 the infidelity of this state with the desired target state $\ket{\phi}$. 
 For generality, we are \textit{not} assuming $\rho$ are coherent thermal states, \textit{nor} that $\ket{\phi}$ are coherent states at this point. 
 
 Now we define $\ket{\psi_n}$ to be the eigenvector of $\sigma_n$ with the maximum eigenvalue. Then, Lemma \ref{variance lowerbound lemma} implies that
$$1-|\braket{\phi|\psi_n}|^2\leq 2\epsilon_n.$$
%Clearly, as $n\rightarrow \infty$, $\epsilon_n\rightarrow 0$, and thus $|\braket{\phi|\psi_n}|^2\rightarrow 1$.
Furthermore, Eq.(\ref{1ododo}) implies that for all $n$,
\begin{align}
    P_H(\sigma_n) \geq V_H(\ketbra{\psi_n}{\psi_n})\times \Big[\frac{(1-\epsilon_n)^2}{\epsilon_n} -1\Big].
\end{align}
Because the purity of coherence is a monotone and is additive, we see that $P_H(\sigma_n)\leq P_H(\rho^{\otimes n}) = n \times P_H(\rho)$. Thus,
\begin{align}
    n \geq \frac{V_H(\ketbra{\psi_n}{\psi_n})}{P_H(\rho)}\times \Big[\frac{(1-\epsilon_n)^2}{\epsilon_n} -1\Big].
\end{align}
Multiplying both sides with $\epsilon_n$, we find
\begin{align}\label{eqn: POC middle step}
     n \times \epsilon_n \geq  \frac{V_H(\ketbra{\psi_n}{\psi_n})}{P_H(\rho)} \times \epsilon_n \Big[\frac{(1-\epsilon_n)^2}{\epsilon_n} -1\Big]
\end{align}
In general, when  $\epsilon_n \rightarrow 0$ and consequently $1-|\braket{\phi|\psi_n}|^2 \rightarrow 0$, the limit of variance of $|\psi_n\rangle$, even if it exists, can be different from the variance of $|\phi\rangle$.  That is, two states being close in fidelity does \textit{not} necessarily imply that their variance is close. To ensure such an implication, we require that the first and second moments of $\ket{\phi}$ are finite. In the following subsection, we prove that
\begin{align}\label{varaince bound}
  \liminf_{n \rightarrow \infty} V_H(\psi^{(n)}) \geq {V}_H(\phi) \ .
\end{align}
Substituting this into the LHS of Eq.(\ref{eqn: POC middle step}) after taking the limit of $n \rightarrow \infty$ yields
\begin{align}
    \liminf_{n \rightarrow \infty} n \times \epsilon_n \geq \frac{V_H(\ketbra{\phi}{\phi})}{P_H(\rho)} .
\end{align}
Note that the first and second moments are clearly finite for coherent states because their energy spectrum is a Poisson distribution.

\subsection{Proof of Eq.(\ref{varaince bound})}
For brevity, we denote the probability distributions $\phi_E:=|\braket{\phi|E}|^2$ and $\psi^{(n)}_E:= |\braket{\psi_n|E}|^2$ from here on. We start with our assumption that the mean and variance of $\ket{\phi}$, i.e.,  $\sum_{E=0}^\infty E^k \phi_E$ exist, and are finite for $k=1,2$. Furthermore, we know that generally for all pairs of probability distributions, the total variation distance and the fidelity are related as
\begin{align}
    \sum_{E=0}^\infty |\phi_E - \psi^{(n)}_E| \leq \sqrt{1-|\braket{\phi|\psi_n}|^2}.
\end{align}
Thus, 
\begin{align}\label{TVD convergence}
    \sum_{E=0}^\infty |\phi_E - \psi^{(n)}_E| \leq \sqrt{2\epsilon_n}.
\end{align}
So from Eq.(\ref{TVD convergence}), it is clear that as $n\rightarrow \infty$, $\epsilon_n \rightarrow 0$, and consequently the entries $\psi^{(n)}_E \rightarrow \phi_E$ for all $E$. Thus the distributions converge point-wise, which we indicate by $\psi^{(n)} \xrightarrow{p.w.} \phi$.

{It can be easily seen that if the support of  $\phi$ and $\psi^{(n)}$  are bounded  for all $n$, then Eq.(\ref{varaince bound}) holds. Then, to prove that it also holds for distributions with unbounded support, we truncate the distribution by applying a functional $f_M$ to the distributions, which makes them restricted to $E\le M$.
In particular, consider a functional $f_M$ of a distribution $\psi$ over $E=0,1,2...$, that outputs a `truncated' distribution as described in Appendix \ref{appendix fm}. As we show in an example in the next section, we can choose this functional satisfying the following properties:}
\begin{enumerate}
    \item $f_M(\psi)$ has a bounded support for any $\psi$ and $M$, specifically ${f_M(\psi)}_E=0$ for all $E> M$. 
    \item The variance reduces upon such truncation: $V_H(f_M(\psi))\leq V_H(\psi)$. 
    \item $\lim_{M \rightarrow \infty}V_H(f_M(\psi)) = V_H(\psi)$.
\end{enumerate}
We check these for an example of such a {truncation functional in }Appendix \ref{appendix fm}. Using this truncation functional, we see that for any $M\geq 0$,
\bes\label{V 1}
\begin{align}
    V_H(\psi^{(n)}) - V_H(\phi) = & \ \ \ \ V_H(\psi^{(n)}) - V_H(f_M(\psi^{(n)}))\\
                                  & + V_H(f_M(\psi^{(n)})) - V_H(f_M(\phi))\\
                                  & + V_H(f_M(\phi)) - V_H(\phi) \ .
\end{align}
\ees
The first line in the right-hand side of this expression is lowerbounded by zero because the assumed truncation does not increase the variance. Thus,
\bes\label{V 2}
\begin{align}
    V_H(\psi^{(n)}) - V_H(\phi) \geq & \  V_H(f_M(\psi^{(n)})) - V_H(f_M(\phi))\\
                                  & + V_H(f_M(\phi)) - V_H(\phi) \ .
\end{align}
\ees
Next we take the limit $\lim_{M \rightarrow \infty}\  \liminf_{n \rightarrow \infty}$ on both sides. Observe that the first line in the right-hand side goes to zero in the limit $n \rightarrow \infty$ because $f_M(\psi^{(n)})$ and $f_M(\phi)$ have bounded support, and the former sequence converges point-wise to the latter target. And then the second line goes to zero in the limit $M \rightarrow \infty$ because
$$ \lim_{M \rightarrow \infty} V_H(f_M(\phi))= V_H(\phi)\ .$$ 
Finally, the left-hand side of Eq.(\ref{V 2}) has no $M$ dependence. Therefore, in the $\lim_{M \rightarrow \infty}\  \liminf_{n \rightarrow \infty}$ limit it simplifies to
\begin{align}
   \liminf_{n \rightarrow \infty} V_H(\psi^{(n)}) - V_H(\phi) \geq 0\ ,
\end{align}
which is Eq.(\ref{varaince bound}).

\subsection{The truncation functional $f_M$}\label{appendix fm}
Let us define the truncation functional $f_M$ as 
\begin{align}
    {f_M(\psi)}_E=
    \begin{cases}
    \psi_E &:E<M\\
             \sum_{E\geq M} \psi_E &:E=M\\
            0 &:E>M
    \end{cases}
\end{align}
where ${f_M(\psi)}_E$ indicates the probability of $f_M(\psi)$ at $E$.  Clearly, \textbf{Property 1} is satisfied by construction.

\textbf{Property 3} -- {First, we note that this property can be easily shown using Markov's inequality, provided that one assumes that the $2+\epsilon$ moment is finite, for some $\epsilon>0$ (in general, if $k+\epsilon$ moment is finite, then by Markov's inequality, we have $\lim_{M\rightarrow \infty} M^k \times \sum_{E\ge M}\phi_{E}=0$. This immediately implies property 3). However, as we explain in the following, the same conclusion can be reached by assuming only the second moment, and not necessarily moment $2+\epsilon$, is finite.}

To show {Property 3}, it suffices to show that the first and second moments of $f_M(\phi)$ approach $\phi$ as $M \rightarrow \infty$. For brevity, let us denote the entries of $f_M(\phi)$ as $\phi_E^{(M)}:={f_M(\phi)}_E$, where ${f_M(\phi)}_E$ indicates the probability of $f_M(\phi)$ at $E$. So, we wish to show,
\begin{align}\label{k moment diffs}
    \lim_{M \rightarrow \infty} \Bigg(\sum_{E=0}^\infty \phi_E E^{k} - \sum_{E=0}^\infty \phi_E^{(M)} E^{k}\Bigg) = 0 : k=1,2\ .
\end{align}
The above expression before the limit can be rewritten as
\bes
\begin{align}
    \sum_{E=0}^\infty \phi_E E^{k} - \sum_{E=0}^\infty \phi_E^{(M)} E^{k} &= \sum_{E=M}^\infty \phi_E (E^k -M^k)\\
    &= \sum_{E=0}^\infty \phi_E \times (E^k - M^k) \ \delta(E \geq M)\\
    & = \sum_{E=0}^\infty \phi_E \times g_M (E) \ ,
\end{align}
\ees
where we defined $$g_M(E)\equiv (E^k - M^k) \ \delta(E \geq M): k = 1,2$$
such that
\begin{align}
    \delta(E \geq M) =\begin{cases}
        &1 : E\geq M\ ,\\
        &0 : E<M\ .
    \end{cases}\ 
\end{align}
So, we shall equivalently prove 
\begin{align}\label{g_M 0}
    \lim_{M \rightarrow \infty} \sum_{E=0}^\infty \phi_E\times g_M(E) =0 \ .
\end{align}
From its definition, it is clear that $g_M(E)$ \textit{point-wise} converges to the zero function $\mathbf{0}(E)=0, \ \forall E$. Furthermore, observe that,
\bes
\begin{align}
     &g_M(E) \leq E^k: k =1,2\ ,\ \ \ \ \forall M, \forall E \ , \label{g_M 1} \\
     &\sum_{E=0}^\infty \phi_E \times E^k < \infty \ : k =1,2. \label{g_M 2}
\end{align}
\ees
In other words, the sequence of functions $g_M$ indexed by $M$ is dominated by $E^k$ (Eq.(\ref{g_M 1})), and $E^k$ is an integrable function because the first and second moments of $\phi_E$ are assumed to be finite (Eq.(\ref{g_M 2})). Thus, we can apply \textit{Dominated Convergence Theorem} \cite{evans2018measure} to switch the order of the limit and summation in Eq.(\ref{g_M 0}) and evaluate,
\bes
\begin{align}
    \lim_{M \rightarrow \infty} \sum_{E=0}^\infty \phi_E\times g_M(E) &= \sum_{E=0}^\infty \phi_E\times \lim_{M \rightarrow \infty} g_M(E) \\
    &= \sum_{E=0}^\infty \phi_E\times \mathbf{0}(E) \\
    & = 0 \ .
\end{align}
\ees
This proves Eq.(\ref{g_M 0}), which shows \textbf{Property 3} is satisfied.

\textbf{Property 2 -- }What remains to be shown is that such truncation reduces the variance. This can be easily shown using the convexity of variance. 

To see this, we first consider a modification of the aforementioned truncation, 
\begin{align}\label{fx psi}
    {f^{(x)}_M(\psi)}_E=
    \begin{cases}
    \psi_E &:E<M\\
             \sum_{E\geq M} \psi_E &:E=x\\
            0 &: \text{otherwise} \ .
    \end{cases}
\end{align}
where ${f^{(x)}_M(\psi)}_E$ indicates the probability of $f^{(x)}_M(\psi)$ at $E$. In other words, we place all the probability weight after the cutoff $M$ at the point $E=x: x\geq M$. Now note that for all $E$, the entry of the distribution $\psi$ at $E$ can be decomposed as
\begin{align}
    \psi_E = \sum_{x=M}^\infty \lambda_x \times {f^{(x)}_M(\psi)}_E 
\end{align}
where $\lambda_x = {\psi_x}/{\sum_{E\geq M} \psi_E}$. In other words, we have decomposed $\psi$ into a convex mixture of $f^{(x)}_M(\psi)$ with weights $\lambda_x$. In the next subsection, using the convexity of the variance, we show that
\begin{align}\label{x exists}
    \exists\  x \geq M \ \ \text{such that } \  {\rm Var}(\psi) \geq {\rm Var}(f_M^{(x)}(\psi))\ .
\end{align}

Now we show that this inequality holds true {in particular for} $x=M$. To see this, we evaluate the variance 
\begin{align}
    {\rm Var}(f_M^{(x)}(\psi)) &= \sum_{E=0}^{M-1} \psi_E E^2 + \delta_M x^2 - (\sum_{E=0}^{M-1} \psi_E E + \delta_M x)^2 \\
    & = \sum_{E=0}^{M-1} \psi_E E^2 - \tilde{\mu}^2 + (\delta_M - \delta_M^2) x^2 - 2 \tilde{\mu} \delta_M x
\end{align}
where we defined $\delta_M=\sum_{E\geq M} \psi_E$, and $\tilde{\mu}=\sum_{E=0}^{M-1} \psi_E E$ for brevity. This quantity is a quadratic in $x$ that finds its minimum at 
$$x_\text{min} = \frac{\tilde{\mu}}{1-\delta_M}\ .$$
Noting that $\tilde{\mu}=\sum_{E=0}^{M-1} \psi_E  E \leq M (1-\delta_M)$, we see that $x_\text{min} \leq M$. Thus, shifting the $x$ in Eq.(\ref{x exists}) down to $x=M$ will only further reduce the variance because the expression is a quadratic in $x$. Thus, 
\begin{align}
    {\rm Var}(\psi) \geq {\rm Var}(f_M^{(x)}(\psi)) \geq {\rm Var}(f_M^{(M)}(\psi)) = {\rm Var}(f_M(\psi))\ ,
\end{align}
{where, the first inequality follows from Eq.(\ref{x exists}), the second inequality follows from argument we just made, and the last equality follows because $f_M^{(M)}(\psi)=f_M(\psi)$ for any $\psi$ by construction.}

\subsubsection{Proof of Eq.(\ref{x exists})}
In general, consider a random variable $Q$ with distribution $q$ that is a convex mixture of random variable $P_i$ with distributions $p_i$, that is,
$$q=\sum_i \lambda_i p_i$$
where $\sum_i \lambda_i=1$. Now note that for any random variable, $\mathbb{E}[(Q-\mu)^2]= \mathbb{E}[Q^2] - \mu^2$, where $\mu:= \mathbb{E}[Q]$ is its mean. Therefore, 
\begin{align}
    \mathbb{E}[(Q-\mu)^2]= \mathbb{E}[Q^2] - \mu^2 = \sum_i \lambda_i \mathbb{E}[P_i^2] - \mu^2 = \sum_i \lambda_i (\sigma_i^2 + \mu_i^2) - \mu^2 \geq \sum_i \lambda_i \sigma_i^2
\end{align}
where $\mu$ and $\sigma^2$ are the mean and variance of $q$, $\mu_i$ and $\sigma_i^2$ are the variance of $p_i$. In the last inequality, we note that 
$$\sum_i \lambda_i \mu_i^2 \geq  (\sum_i \lambda_i \mu_i)^2 = \mu^2 \ ,$$
which follows from the convexity of the function $t^2$. Applying this to Eq.(\ref{fx psi}) yields
\begin{align}
    {\rm Var}(\psi) \geq \sum_{x=M}^\infty \lambda_x {\rm Var}(f_M^{(x)}(\psi)) \ .
\end{align}
Therefore,
\begin{align}
    \exists\ x \geq M \ \ \text{such that } \  {\rm Var}(\psi) \geq {\rm Var}(f_M^{(x)}(\psi))\ .
\end{align}

\newpage

\section{Measure-and-prepare distillation protocols}\label{appendix: measure prepare}
In this Appendix, we consider measure-and-prepare protocols where we assume knowledge of the magnitude of the displacement $\alpha$ (i.e., $|\alpha|$), and estimate its phase using $n$ copies of the coherent thermal state. As mentioned in the main body, we can begin our protocol by first utilizing the concentration channel to combine $n$ copies of $\rho(\beta, \alpha)^{\otimes n}$ to produce $\rho(\beta, \sqrt{n} \alpha)$ without affecting the overall fidelity (see Sec.~\ref{subsec: concentration dilution reversibility}). Then, in the regime of large $n$, i.e., in the high amplitude regime, the protocol proceeds by estimating the phase of $\alpha$ by measuring the state $\rho(\beta, \sqrt{n} \alpha)$, and then preparing $\ket{\alpha}$.

\subsection{Optimal Distillation using Canonical Phase Measurement}\label{appendix: optimal MP canonical}
In this subsection, we derive the fidelity of a measure-and-prepare protocol that uses the well-known canonical phase measurement \cite{holevo2011probabilistic} to estimate the phase. The canonical phase measurement is characterized by the POVM elements $M(\phi)=\sum_{m,n} M_{m,n}(\phi) \ketbra{m}{n}$ have matrix elements
\be\label{canonical phase msmt POVM}
M_{mn}(\phi)=\frac{e^{\mathrm{i}\phi (m-n)}d\phi}{2\pi}\ : \phi\in[0,2\pi)\ .
\ee
Then the corresponding measure-and-prepare channel is
\be\label{K_n canon}
\mathcal{K}_n(\cdot)=\int \frac{d\phi}{2\pi}\  \Tr(M(\phi)\ \cdot\ )\  |e^{\mathrm{i}\phi}|\alpha| \rangle\langle e^{\mathrm{i}\phi}|\alpha||\ .
\ee
To see that this channel is phase-insensitive, we show that the channel remains unchanged upon time-evolving the single-mode input and output. That is, 
\begin{align}
    \mathcal{K}_n (e^{-\mathrm{i}Ht} \rho e^{\mathrm{i}Ht})  = e^{-\mathrm{i}Ht} \mathcal{K}_n (\rho) e^{\mathrm{i}Ht}
\end{align}
for all $t \in \mathbb{R}$. We first evaluate the left hand side,
\begin{align}
    \mathcal{K}_n (e^{-\mathrm{i}Ht} \rho e^{\mathrm{i}Ht})&=\int \frac{d\phi}{2\pi}\  \Tr(M(\phi)\ e^{-\mathrm{i}Ht} \rho e^{\mathrm{i}Ht} )\  |e^{\mathrm{i}\phi}|\alpha| \rangle\langle e^{\mathrm{i}\phi}|\alpha||\ \\
    &= \int \frac{d\phi}{2\pi}\  \Tr( e^{\mathrm{i}Ht} M(\phi)\ e^{-\mathrm{i}Ht} \rho)\  |e^{\mathrm{i}\phi}|\alpha| \rangle\langle e^{\mathrm{i}\phi}|\alpha||\ \\
    &=\int \frac{d\phi}{2\pi}\  \Tr( M(\phi-\omega t)\ \rho)\  |e^{\mathrm{i}\phi}|\alpha| \rangle\langle e^{\mathrm{i}\phi}|\alpha||\ \\
    &=\int \frac{d\phi}{2\pi}\  \Tr( M(\phi)\ \rho)\  |e^{\mathrm{i}(\phi+ \omega t)}|\alpha| \rangle\langle e^{\mathrm{i}(\phi+ \omega t)}|\alpha||\
\end{align} 
where we used in the third line that the POVM defined in Eq.(\ref{canonical phase msmt POVM}) satisfies $e^{\mathrm{i}Ht} M(\phi) e^{-\mathrm{i}Ht} = M(\phi-\omega t)$ when $H=\omega a^\dagger a$. In the last line, we simply changed variables $\phi \mapsto \phi + \omega t$. The right hand side evaluates to the same expression, 
\begin{align}
    e^{-\mathrm{i}Ht} \mathcal{K}_n(\rho) e^{\mathrm{i}Ht} &=  \int \frac{d\phi}{2\pi}\  \Tr(M(\phi)\ \rho )\ \  e^{-\mathrm{i}Ht} |e^{\mathrm{i}\phi}|\alpha| \rangle\langle e^{\mathrm{i}\phi}|\alpha|| e^{\mathrm{i}Ht}\\
    &= \int \frac{d\phi}{2\pi}\  \Tr(M(\phi)\ \rho )\ \  |e^{\mathrm{i}(\phi+\omega t)}|\alpha| \rangle\langle e^{\mathrm{i}(\phi+\omega t)}|\alpha|| 
\end{align}
because $e^{-\mathrm{i}Ht} \ket{|\alpha|e^{\mathrm{i} \phi}} = \ket{|\alpha|e^{\mathrm{i} (\phi+\omega t)}}$ when $H=\omega a^\dagger a$. Thus, the channel $\mathcal{K}_n$ is phase-insensitive.\\

Given the infidelity to the desired output state $\ket{\alpha}$,
\begin{align}
    \epsilon^{\text{can}} = 1 - \braket{\alpha|\mathcal{K}_n(\rho(\beta, \sqrt{n} \alpha))|\alpha} \ ,
\end{align}
we can achieve the optimal distillation error possible among all measure-and-prepare protocols (see discussion in Sec.(\ref{section: universal bounds proof})), i.e.,
\begin{align}\label{optimal statement MP}
    \lim_{n \rightarrow \infty} n \times \epsilon^{\text{can}} = \delta^{\text{opt-MP}}(\beta) = \frac{n_\text{th}(\beta)}{2} + \frac{1}{4}. 
\end{align}
To prove this, we first recall that for all measure-and-prepare protocols (including the channel $\mathcal{K}_n$ defined in Eq.(\ref{K_n canon})),  $$\limsup_{n \rightarrow \infty} n \times \epsilon^{\text{MP}}\geq \frac{n_\text{th}(\beta)}{2} + \frac{1}{4}.$$ 
Then Eq.(\ref{optimal statement MP}) follows from the statement
\begin{align}\label{canonical phase mst statement}
    \lim_{n \rightarrow \infty} n \times \epsilon^{\text{can}}\leq \frac{n_\text{th}(\beta)}{2} + \frac{1}{4}.
\end{align}
We prove the statement in Eq.(\ref{canonical phase mst statement}) in the subsequent section using Lemma \ref{lemma: divide and distill}.

\subsection*{Infidelity Computation}
Recall that we have concentrated all the coherence in $n$ input modes into a single mode via passive operations, i.e., $\rho(\beta, |\alpha|)^{\otimes n} \leftrightarrow \rho(\beta, \sqrt{n} |\alpha|)\equiv \rho$. Furthermore, because the channel is phase-insensitive by construction, we have assumed that phase of the input state is zero without loss of generality. Then, the output state is
\begin{align}\label{output Kn MP}
    \mathcal{K}_n(\rho) = \int \frac{d\phi}{2 \pi}\ p(\phi|0)\ \ketbra{e^{\mathrm{i} \phi}|\alpha|}{e^{\mathrm{i} \phi}|\alpha|}
\end{align}
where $p(\phi|0)$ is the probability that our canonical phase measurement estimates phase $\phi$ given the input phase $0$. Using the POVM elements of this measurement defined in Eq.(\ref{canonical phase msmt POVM}), 
\begin{align}\label{p(phi|0)}
    p(\phi|0)
    =\Tr(M(\phi) \rho)= \sum_{m, n} \rho_{m,n} e^{\mathrm{i}\phi(m-n)}.
\end{align} 
The infidelity of this output state with the desired output state $\ket{|\alpha|}$ is $1-\braket{|\alpha||\mathcal{K}_n(\rho)||\alpha|}$. Using the first line of Lemma \ref{lemma: divide and distill}, we see that
\begin{align}
    1-\braket{|\alpha||\mathcal{K}_n(\rho)||\alpha|} \leq \Tr(\mathcal{K}_n(\rho) a^\dagger a) - |\Tr(\mathcal{K}_n(\rho) a)|^2 + \big||\alpha| - \Tr(\mathcal{K}_n(\rho) a)\big|^2.
\end{align}
At the end of this section, we show that
\begin{align}
    &\Tr(\mathcal{K}_n(\rho) a)= |\alpha| \times \sum_m |\rho_{m,m+1}| \label{MP moment 1}\\
    &\Tr(\mathcal{K}_n(\rho) a^\dagger a)= |\alpha|^2\label{MP moment 2}
\end{align}
If we denote $\mathcal{F}_1=\sum_m |\rho_{m,m+1}|$ as done in \cite{PhaseEstimation12}, we see that
\begin{align}
    1-\braket{|\alpha||\mathcal{K}_n(\rho)||\alpha|} &\leq \Tr(\mathcal{K}_n(\rho) a^\dagger a) - |\Tr(\mathcal{K}_n(\rho) a)|^2 + \big||\alpha| - \Tr(\mathcal{K}_n(\rho) a)\big|^2\\
    &= |\alpha|^2 - |\alpha|^2 \mathcal{F}_1^2 + |\alpha|^2 [1-\mathcal{F}_1]^2\\
    &= 2 |\alpha|^2 [1-\mathcal{F}_1].\label{canon eqn mid}
\end{align}
Now we recall that $\rho\equiv \rho(\beta, \sqrt{n} |\alpha|)$ and that we are interested in the regime where $n\gg1$. Thus, $\rho$ is a coherent thermal state in the large amplitude regime. In Eq.(21) of \cite{PhaseEstimation12} it is shown that
in this large amplitude regime,
\begin{align}\label{eq: F1}
    \mathcal{F}_1 = 1 -\frac{2n_\text{th}(\beta) +1}{8n |\alpha|^2} + \mathcal{O}\Big(\frac{1}{n^2}\Big)\ ,
\end{align}
where we noted that the displacement of the input state is $\sqrt{n}|\alpha|$.\footnote{We can reproduce the result in Eq.(\ref{eq: F1}) from our formulae in Lemma \ref{eqn: matrix elements final} by evaluating the equivalent form $\sum_m |\rho_{m,m+1}| = \sum_m \rho_{m,m} {|c^\text{opt}_{m+1}|^{-1}}$.} Substituting this into Eq.(\ref{canon eqn mid}) yields
\begin{align}
    1-\braket{|\alpha||\mathcal{K}_n(\rho)||\alpha|} \leq \frac{2n_\text{th}(\beta) +1}{4n} + \mathcal{O}\Big(\frac{1}{n^2}\Big).
\end{align}
Eq.(\ref{canonical phase mst statement}) follows straightforwardly.

\subsubsection*{First Moment Computation: Eq.(\ref{MP moment 1})}
From the output state defined in Eq.(\ref{output Kn MP}),
\begin{align}
    \Tr(\mathcal{K}_n(\rho) a) = \Tr\Big[\int \frac{d\phi}{2 \pi} p(\phi|0)\ a \ketbra{e^{\mathrm{i} \phi}|\alpha|}{e^{\mathrm{i} \phi}|\alpha|}\Big] = |\alpha| \int \frac{d\phi}{2 \pi} p(\phi|0)\ e^{\mathrm{i} \phi}
\end{align}
By substituting the explicit expression of $p(\phi|0)$ seen in Eq.(\ref{p(phi|0)}), we get 
\begin{align}
    =|\alpha| \int \frac{d\phi}{2 \pi} \ \sum_{m, n} \rho_{m,n} e^{\mathrm{i}\phi(m-n)}  e^{\mathrm{i} \phi} = |\alpha| \sum_m \rho_{m,m+1} = |\alpha| \times \sum_m |\rho_{m,m+1}|
\end{align}
where we exchanged the order of the integral and summation in the second equality and simplified noting that $m$ and $n$ are integers. In the last equality we noted that all $\rho_{m,m+1}$ are non-negative and real when the displacement is $|\alpha|$, i.e., is also non-negative and real (this is obvious, e.g., from Eq.(\ref{eq: rho_ml final})).

\subsubsection*{Second Moment Computation: Eq.(\ref{MP moment 2})}
Just as above, from the output state defined in Eq.(\ref{output Kn MP}), we see that
\begin{align}
    \Tr(\mathcal{K}_n(\rho) a^\dagger a) = \Tr\Big[\int \frac{d\phi}{2 \pi} p(\phi|0)\ a \ketbra{e^{\mathrm{i} \phi}|\alpha|}{e^{\mathrm{i} \phi}|\alpha|} a^\dagger \Big] = |\alpha|^2 \int \frac{d\phi}{2 \pi} p(\phi|0) = |\alpha|^2.
\end{align}

\subsection{Distillation using Heterodyne detection}\label{appendix: heterodyne measure prepare}

We focus on distillation protocols based on Heterodyne detection with POVM elements $\frac{d^2\alpha'}{\pi}  |\alpha'\rangle\langle\alpha'|: \alpha'\in\mathbb{C}$, which satisfy
\be
\int \frac{d^2\alpha'}{\pi}  |\alpha'\rangle\langle\alpha'|=\mathbb{I}\ .
\ee
Upon performing this measurement on the state $\rho(\beta, \sqrt{n}\alpha)$ in the limit $n\rightarrow\infty$, the outcome of this measurement will be concentrated around $\alpha'\approx \sqrt{n} \alpha$. Then, we prepare the state $|\alpha'/\sqrt{n}\rangle$. This means the overall phase-insensitive channel is
\be
\mathcal{E}_n(\cdot)=\int \frac{d^2\alpha'}{\pi} \ \  |\frac{\alpha'}{\sqrt{n}}\rangle\langle{\alpha'}|(\cdot) |{\alpha'}\rangle\langle\frac{\alpha'}{\sqrt{n}}|.
\ee
The fidelity achieves
\begin{align}\label{protocol 1}
&\langle\alpha|\mathcal{E}_n(\rho(\beta,\sqrt{n}\alpha))|\alpha\rangle \nonumber\\
&=\int \frac{d^2\alpha'}{\pi} \ \  |\langle\alpha|\frac{\alpha'}{\sqrt{n}}\rangle|^2\langle{\alpha'}|\rho(\beta,\sqrt{n}\alpha)|{\alpha'}\rangle\ \\
&=1-\frac{n_{\text{th}}(\beta)+1}{n} + \mathcal{O}\Big(\frac{1}{n^2}\Big)
\end{align}
We compute this expression below.
\subsubsection{Evaluation of Protocol}
We start with writing out the integral explicitly.
\begin{align}
\langle\alpha|\mathcal{E}_n(\rho(\beta,\sqrt{n}\alpha))|\alpha\rangle&=\int \frac{d^2\alpha'}{\pi} \ \  |\langle\alpha|\frac{\alpha'}{\sqrt{n}}\rangle|^2\langle{\alpha'}|\rho(\beta,\sqrt{n}\alpha)|{\alpha'}\rangle\ \\
    &= \int \frac{d^2\alpha'}{\pi} \ \  |\langle\alpha|\frac{\alpha'}{\sqrt{n}}\rangle|^2\langle{\alpha'-\sqrt{n}\alpha}|\rho(\beta)|{\alpha' -\sqrt{n}\alpha}\rangle\
\end{align}
We consider each term. Observe that,
\begin{align}
    &\langle{\alpha'}|\rho(\beta,\sqrt{n}\alpha)|{\alpha'}\rangle\ \\
    &=\langle{\alpha'-\sqrt{n}\alpha}|\rho(\beta)|{\alpha' -\sqrt{n}\alpha}\rangle\ \\
    &= (1-e^{-\beta \omega})\sum_{s=0} e^{-\beta \omega s} |\langle{\alpha'-\sqrt{n}\alpha}|s\rangle|^2\\
    &= (1-e^{-\beta \omega})\times e^{-|\alpha'-\sqrt{n}\alpha|^2}\sum_{s=0} e^{-\beta \omega s} \frac{1}{s!} |\alpha'-\sqrt{n}\alpha|^{2s}\\
    &= (1-e^{-\beta \omega}) \exp[-(1-e^{-\beta \omega})  |\alpha'-\sqrt{n}\alpha|^{2}]
\end{align}
and, 
\begin{align}
&|\langle\alpha|\frac{\alpha'}{\sqrt{n}}\rangle|^2= \exp[-|\alpha-\frac{\alpha'}{\sqrt{n}}|^2]=\exp[-\frac{1}{n}|\sqrt{n}\alpha-{\alpha'}|^2].
\end{align}
So, the fidelity simplifies to,
\begin{align}
    &\int \frac{d^2\alpha'}{\pi} \ \  |\langle\alpha|\frac{\alpha'}{\sqrt{n}}\rangle|^2\langle{\alpha'-\sqrt{n}\alpha}|\rho(\beta)|{\alpha' -\sqrt{n}\alpha}\rangle\ \\
    &= \int \frac{d^2\alpha'}{\pi} \times \exp[-\frac{1}{n}|\sqrt{n}\alpha-{\alpha'}|^2] \times (1-e^{-\beta \omega}) \exp[-(1-e^{-\beta \omega})  |\alpha'-\sqrt{n}\alpha|^{2}]\\
    &= (1-e^{-\beta \omega}) \int \frac{d^2\alpha'}{\pi} \exp[-\Big(1-e^{-\beta \omega} + \frac{1}{n}\Big) \times |\alpha'-\sqrt{n}\alpha|^{2}]
\end{align}
Now we can perform a simple linear change of variables $y=\sqrt{1-e^{-\beta \omega} + \frac{1}{n}}\times [\alpha'-\sqrt{n}\alpha]$. {This yields,}
\begin{align}
    &= (1-e^{-\beta \omega}) \frac{1}{\Big(1-e^{-\beta \omega} + \frac{1}{n}\Big)} \int \frac{d^2 y}{\pi} \exp[-|y|^{2}]\\
    &= \frac{1-e^{-\beta \omega}}{1-e^{-\beta \omega} + \frac{1}{n}}\\ 
    &=\frac{1}{1+(n_{\text{th}}(\beta)+1)/n}=1-\frac{n_{\text{th}}(\beta)+1}{n}+ \mathcal{O}\Big(\frac{1}{n^2}\Big)
\end{align}
\newpage

\section{Matrix elements of  coherent thermal states}

In this section, we use useful formulae for the matrix elements of coherent thermal states to estimate the atypical contribution to the moments of its energy distribution (in the Fock-basis), $\rho_{l,l}$. We do so by computing the moment-generating function for $\rho_{l,l}$, and then use standard Chernoff bounds to bound the weight of the tails of the distribution. 

We first review the formulae for the matrix elements that appeared previously in literature, e.g., see Eq.(6.13-6.14) in \cite{GlauberMatrixElem}, and Eq.(5.17) in \cite{MarianMatrixElem}. We shall restate these formulae in the following Lemma in a form that is useful for our purposes.

\begin{lemma} [\cite{GlauberMatrixElem, MarianMatrixElem}]
    Consider the coherent thermal state $\rho(\beta,\alpha)=D(\alpha) e^{-\beta a^\dag a} D(\alpha)^\dag/\Tr(e^{-\beta a^\dag a} )$. For $0\leq m\leq l$, its matrix elements in the Fock basis $\braket{m|\rho(\beta, \alpha)|l}\equiv\rho_{m,l}$ can be evaluated as
    
\begin{align}
    \rho_{m,l}&=\frac{e^{-\frac{|\alpha|^2}{n_\beta+1}}}{n_\beta+1} \times \sqrt{\frac{{{m}!}}{{l!}}} \frac{n_\beta^m}{(n_\beta+1)^{l}} (\alpha^*)^{l-m} \times  L^{(l-m)}_m\Big(-\frac{|\alpha|^2}{n_\beta(n_\beta+1)}\Big)\ \ \ \ \  &: 0\leq m\leq l \\
    &=e^{-\frac{|\alpha|^2}{n_\beta+1}}  {\frac{1}{(n_\beta+1)^{m+l+1}}} \times \frac{\alpha^{m} \alpha^{*l}}{\sqrt{{m}!}\sqrt{l!}} \times \sum^{m}_{k=0} {k! {m \choose k}{l \choose k}} \Big[\frac{n_\beta(n_\beta+1)}{|\alpha|^2}\Big]^{k}\ \ \ \ \  &: 0\leq m\leq l \label{eq: rho_ml final} 
\end{align}
where $L_n^{(k)}$ are the generalized Laguerre polynomials defined as
\begin{align}
    L^{(k)}_n(x) = \sum_{i=0}^n (-1)^i {n + k \choose n - i} \frac{(x)^i}{i!}\ .
\end{align}
The diagonal elements can be further simplified to 
\begin{align} \label{rho_ll final formula}
    \rho_{l,l}=\frac{e^{-\frac{|\alpha|^2}{n_\beta+1}}}{n_\beta+1} \bigg(\frac{n_\beta}{n_\beta+1}\bigg)^l \times {L}_l \bigg(-\frac{|\alpha|^2}{n_\beta(n_\beta+1)}\bigg)
\end{align}
where ${L}_l$ is the $l^\text{th}$ degree Laguerre polynomial \cite{SpecialFunctionsTextbook} defined as {$L^{(0)}_l$, or equivalently}
\begin{align}
    L_n(|x|^2) = \sum_{i=0}^n (-1)^i {n\choose n - i} \frac{(|x|^2)^i}{i!} =  e^{|x|^2/2} \times \braket{n|D(x)|n}
\end{align}
where $D(x)$ is the Weyl displacement operator.
\end{lemma}

\subsection{Upperbound on $|c^{\text{opt}}_l|^2$}
Using Eq.(\ref{eq: rho_ml final}) for the matrix elements of the coherent thermal state for $\alpha$ real, we get that
\begin{align}
    \rho_{l,l}&= e^{-\frac{|\alpha|^2}{n_\beta + 1}} \frac{1}{(1+n_\beta)^{2l +1}} \frac{|\alpha|^{2l}}{l!} \times \sum_{k=0}^l {l \choose k}\frac{l!}{(l-k)!} \Big[\frac{n_\beta(n_\beta+1)}{|\alpha|^2}\Big]^k \\
    &=e^{-\frac{|\alpha|^2}{n_\beta + 1}} \frac{1}{(1+n_\beta)^{2l +1}} {|\alpha|^{2l}} \times l! \times \sum_{k=0}^l \frac{1}{k!(l-k)!}\frac{1}{(l-k)!} \Big[\frac{n_\beta(n_\beta+1)}{|\alpha|^2}\Big]^k\ ,
\end{align}
and
\begin{align}
    \rho_{l,l+1}&= e^{-\frac{|\alpha|^2}{n_\beta + 1}} \frac{1}{(1+n_\beta)^{2l +2}} \frac{|\alpha|^{2l} \alpha}{l! \sqrt{l+1}} \times \sum_{k=0}^l {l \choose k}\frac{(l+1)!}{(l+1-k)!} \Big[\frac{n_\beta(n_\beta+1)}{|\alpha|^2}\Big]^k\\
    &= e^{-\frac{|\alpha|^2}{n_\beta + 1}} \frac{1}{(1+n_\beta)^{2l +2}} {|\alpha|^{2l} \alpha}{\sqrt{l+1}} \times l! \times \sum_{k=0}^l \frac{1}{k!(l-k)!}\frac{1}{(l+1-k)!} \Big[\frac{n_\beta(n_\beta+1)}{|\alpha|^2}\Big]^k \\
    & \geq e^{-\frac{|\alpha|^2}{n_\beta + 1}} \frac{1}{(1+n_\beta)^{2l +2}} \frac{|\alpha|^{2l} \alpha}{\sqrt{l+1}} \times l! \times \sum_{k=0}^l \frac{1}{k!(l-k)!}\frac{1}{(l-k)!} \Big[\frac{n_\beta(n_\beta+1)}{|\alpha|^2}\Big]^k \ ,
\end{align}
where we used in the last inequality that $(l+1-k)!=(l-k)! (l+1-k) \leq (l-k)! (l+1)$ because $k\geq 0$. Thus,
\begin{align}
    c_{l+1}^\text{opt}:=\frac{\rho_{l,l}}{\rho_{l,l+1}}  \leq (n_\beta+1) \times \frac{1}{\alpha} \sqrt{l+1} \ .
\end{align}
 Thus, 
\begin{align}\label{c_l upperbound}
    |c_l^\text{opt}|^2-1 \leq |c_l^\text{opt}|^2 \leq (n_\beta+1)^2 \times  \frac{l}{\alpha^2}.
\end{align}

\subsection{Moment-Generating Function of $\rho_{l,l}$}\label{appendix: moment generating function}
In this subsection, we derive the moment generating function for the distribution $\rho_{l,l}$ for $l\in \{0,1,2,\cdots\}$. 

From Eq.(\ref{rho_ll final formula}), we see that 
\begin{align}
    M(t)=\sum_{l=0}^\infty \rho_{l,l} e^{t l}= \frac{e^{-\frac{|\alpha|^2}{n_\beta+1}}}{n_\beta+1} \sum_{l=0}^\infty p^l \times {L}_l(z)
\end{align}
where $z=-\frac{|\alpha|^2}{n_\beta(n_\beta+1)}$ and $p=\frac{n_\beta e^t}{n_\beta+1}$. We shall assume $\alpha$ is real and positive for convenience.
%To this end, we pick $t=\ln(1+(\alpha^2+n_\beta)^{-1})$. 
Next, recall the generating function of Laguerre polynomials \cite{SpecialFunctionsTextbook},
$$\sum_{l=0}^\infty p^l \times {L}_l(z) = \frac{1}{1-p} e^{-\frac{pz}{1-p}}.$$
Thus, 
\begin{align}\label{M(t) for rho_ll}
    M(t)= \frac{\exp\Big[\frac{\alpha^2(e^t-1)}{1-(e^t-1)n_\beta}\Big]}{1-(e^t-1)n_\beta} \ .
\end{align}
{To ensure non-negative values of $M(t)$, we require that $p = e^t n_\beta/(1+n_\beta) < 1$, or equivalently $e^t < 1+ n_\beta^{-1}$.} The first and second moments (assumed $\alpha$ is real and positive for simplicity) are
\begin{align}
    &M'(0)=\alpha^2+ n_\beta\\
    &M''(0)=\alpha^4+4 \alpha^2 n_\beta+\alpha^2+2 n_\beta^2+n_\beta \ .
\end{align}
Then the variance of the distribution is
\begin{align}
    {M''(0)-[M'(0)]^2} = 2 \alpha^2 n_\beta+\alpha^2+n_\beta^2+n_\beta\ .
\end{align}
In the large {$\alpha \gg \max\{1,n_\beta^2\}$} regime, the mean and variance are well approximated by $\alpha^2$ and $(1+2n_\beta)\alpha^2$ respectively.

\subsubsection*{Moment-Generating Function from P-distribution}
An alternate method for computing the moment-generating function is presented in \cite{BarnettMGF}. Consider a single-mode state $\rho$ with P-distribution defined as
\begin{align}
    \rho = \int_\mathbb{C} {d^2\gamma}\ P(\gamma) \ketbra{\alpha}{\alpha} \ .
\end{align}
Then, the moment-generating function for its energy (Fock-basis) distribution is mentioned in Eq.(A8) in Appendix A of \cite{BarnettMGF} as
\begin{align}
    M(\mu) = \int  {d^2\gamma}\  e^{-\mu|\gamma|^2} P(\gamma) \ ,
\end{align}
where $\mu=1-e^t$ with respect to our notation. Recall that the P-distribution of a coherent thermal state $\rho(\beta, \alpha)$ is $$P(\gamma)=\frac{1}{\pi n_\beta}e^{-\frac{|\gamma-\alpha|^2}{n_\beta}}.$$
Substituting this into the integral formula for the $M(\mu)$ and then substituting $\mu=1-e^t$ yields the same result as seen in Eq.(\ref{M(t) for rho_ll}).

\subsection{Bounding atypical contributions using Chernoff Bound}\label{chernoff bound}
In Appendix \ref{appendix: matrix elements large alpha}, we partition the support of $\rho_{l,l}$, i.e., $l \in \{0,1,2,...\}$ into a typical and atypical range. Specifically, when $\alpha \gg 1$, the typical range is
\begin{align}
    \alpha^2 - R \sigma \leq l \leq \alpha^2 + R \sigma \ ,
\end{align}
where $R \ll \alpha^{\frac13}$, and we define $\sigma=\sqrt{1+2n_\beta} \alpha$. The complement of the typical range is the atypical range. 

We have seen that the `zero-th', first and second moment of $\rho_{l,l}$ are
\bes
\begin{align}
    &\sum_{l=0}^\infty \rho_{l,l} =M(0)= 1\\
    &\sum_{l=0}^\infty \rho_{l,l} \times l =M'(0)= \alpha^2 + n_\beta\\
    &\sum_{l=0}^\infty \rho_{l,l} \times l^2 = M''(0)= \alpha^4+4 \alpha^2 n_\beta+\alpha^2+2 n_\beta^2+n_\beta
\end{align}
\ees
In this section, we show that the contribution of the atypical range to these summations are exponentially small. Specifically,
\begin{align} \label{chernoff eqns}
   \sum_{l \not\in \text{typical}} \rho_{l,l} \times l^m= \mathcal{O}\Big(e^{n_\beta + \frac12 - \sqrt{1+2n_\beta}R}\Big) \ : \ m=0,1,2 \ ,
\end{align}
where $l\in [\alpha^2 - R \sigma, \alpha^2 + R \sigma]$ is the typical range of $l$, and we have  suppressed higher order terms in $1/\alpha$ in the exponent. To show this, we use {the standard Chernoff bound} for computing contributions from the right and left tail of $\rho_{l,l}$.

{In Appendix \ref{limit comp app}, we pick $R=\alpha^{1/4}$ to compute the a limit $\alpha \rightarrow \infty$. Because $r=(l-\alpha^2)/\sigma$ where $\sigma=\sqrt{1+2n_\beta}\alpha$, it clearly follows that 
\begin{align} \label{chernoff eqns r}
   \sum_{l \not\in \text{typical}} \rho_{l,l} \times r^m= \mathcal{O}\Big(e^{n_\beta + \frac12 - \sqrt{1+2n_\beta} \alpha^{1/4}}\Big) \ : \ m=0,1,2 \ .
\end{align}
}

\subsubsection{`Zero-th' moment}
In this subsection, we consider the contribution of the right and left tails to the total probability weight.\\

\textbf{Right Tail -- }Assume $t\geq 0$. Noting that $l\geq 0$, 
\begin{align}\label{mn0}
    \sum_{l=T}^\infty \rho_{l,l} = \sum_{l=T}^\infty \rho_{l,l} \times \frac{e^{tl}}{e^{tl}}\leq e^{-tT} \times \sum_{l=0}^\infty \rho_{l,l} \times e^{t l}= e^{-t T} \times M(t) \ 
\end{align}
where $M(t)$ is the moment generating function defined in Eq.(\ref{appendix: moment generating function}). Because we require $e^t n_\beta/(1+n_\beta) < 1$, and we have that $t\geq 0$, we must choose $t$ such that
\begin{align}
    1 \leq e^t < 1 + \frac{1}{n_\beta} \ . 
\end{align}
We shall pick $t'=\log\Big[1+ \frac{1}{\alpha+ n_\beta}\Big]$. Then Eq.(\ref{M(t) for rho_ll}) simplifies to
\begin{align}
    M(t')= \frac{\alpha+n_\beta}{\alpha} \times \exp(\alpha)  \ .
\end{align}
Thus, Eq.(\ref{mn0}) becomes
\begin{align}
    \sum_{l=T}^\infty \rho_{l,l} \leq \frac{\alpha + n_\beta}{\alpha} \times \exp\Big[\alpha - \log\Big(1+ \frac{1}{\alpha+n_\beta}\Big)T\Big] \label{m0 r}\ .
\end{align}
To analyze the right tail, we set $T=\alpha^2+ R \sigma = \alpha^2+ R \sqrt{1+2n_\beta} \alpha$. Then, in the large $\alpha$ regime, the term in the exponent expanded in powers of $1/\alpha$ becomes
\begin{align}\label{exp right}
    \alpha - \log\Big(1+ \frac{1}{\alpha+n_\beta}\Big) (\alpha^2 + R \sqrt{1+2n_\beta} \alpha)=n_\beta + \frac12 - \sqrt{1+2n_\beta}R + \mathcal{O}\Big(\frac{R}{\alpha}\Big) \ .
\end{align}
Thus, ignoring the polynomial prefactor, Eq.(\ref{m0 r}) becomes
\begin{align}
    \sum_{l>\alpha^2+R \sigma}^\infty \rho_{l,l} = \mathcal{O}\Big(e^{n_\beta + \frac12 - \sqrt{1+2n_\beta}R + \mathcal{O}({R}/{\alpha})}\Big) \ . 
\end{align}

\textbf{Left Tail -- }Assume $t\leq 0$. Noting that $l\geq 0$, 
\begin{align}\label{mn0l}
    \sum_{l=0}^T \rho_{l,l} = \sum_{l=0}^T \rho_{l,l} \times \frac{e^{tl}}{e^{tl}}\leq e^{-tT} \times \sum_{l=0}^\infty \rho_{l,l} \times e^{t l}= e^{-t T} \times M(t) \ .
\end{align}
Because {$t\leq 0$}, any choice of $t$ ensures $e^t n_\beta/(1+n_\beta) < 1$. We shall pick $t'=-\log\Big[1+\frac{1}{\alpha-n_\beta-1}\Big]$. Then Eq.(\ref{M(t) for rho_ll}) simplifies to
\begin{align}
    M(t')= \frac{\alpha-n_\beta}{\alpha} \times \exp(-\alpha)  \ .
\end{align}
Thus, Eq.(\ref{mn0l}) becomes
\begin{align} \label{m0 l}
    \sum_{l=0}^T \rho_{l,l} \leq \frac{\alpha - n_\beta}{\alpha} \times \exp\Big[-\alpha + \log\Big(1+ \frac{1}{\alpha-n_\beta-1}\Big)T\Big] \ .
\end{align}
To analyze the left tail, we set $T=\alpha^2- R \sigma = \alpha^2- R \sqrt{1+2n_\beta} \alpha$. Then, in the large $\alpha$ regime, the term in the exponent expanded in powers of $1/\alpha$ becomes
\begin{align}\label{exp left}
    -\alpha + \log\Big(1+ \frac{1}{\alpha-n_\beta-1}\Big)(\alpha^2 - R \sqrt{1+2n_\beta}\alpha) = n_\beta + \frac12 - \sqrt{1+2n_\beta}R + \mathcal{O}\Big(\frac{R}{\alpha}\Big)\ .
\end{align}
Thus, ignoring the polynomial prefactor, Eq.(\ref{m0 l}) becomes
\begin{align}
\sum_{l<\alpha^2+R \sigma} \rho_{l,l} = \mathcal{O}\Big(e^{n_\beta + \frac12 - \sqrt{1+2n_\beta}R + \mathcal{O}({R}/{\alpha})}\Big) \ . 
\end{align}

\textbf{Summary -- } Combining the left and right tail contributions, we see that for $\alpha \gg 1$, 
\begin{align}
    \sum_{l \not\in \text{typical}} \rho_{l,l} = \mathcal{O}\Big(e^{n_\beta + \frac12 - \sqrt{1+2n_\beta}R}\Big)
\end{align}
where $l\in [\alpha^2 - R \sigma, \alpha^2 + R \sigma]$ is the typical range of $l$, and we have  suppressed higher order terms in $1/\alpha$ in the exponent.

\subsubsection{First moment}
In this subsection, we consider the contribution of the right and left tails to the first moment.\\

\textbf{Right Tail -- }Assume $t\geq 0$. Noting that $l\geq 0$,
\begin{align}\label{mn1r}
   \sum_{l=T}^\infty \rho_{l,l}\times l = \sum_{l=T}^\infty \rho_{l,l} \times \frac{l e^{tl}}{e^{tl}}\leq e^{-tT} \times \sum_{l=0}^\infty \rho_{l,l} \times l e^{t l}= e^{-t T} \times \frac{d}{dt}\Big[\sum_{l=0}^\infty \rho_{l,l} e^{tl} \Big]= e^{-t T} \times \frac{dM(t)}{dt} \ 
\end{align}
where $M(t)$ is the moment generating function defined in Eq.(\ref{appendix: moment generating function}). Just as earlier, because we require $e^t n_\beta/(1+n_\beta) < 1$, and we have that $t\geq 0$, we must choose $t$ such that
\begin{align}
    1 \leq e^t < 1 + \frac{1}{n_\beta} \ . 
\end{align}
Just as earlier, we shall pick $t'=\log\Big[1+ \frac{1}{\alpha+ n_\beta}\Big]$. Taking the derivative of $M(t)$ seen in Eq.(\ref{M(t) for rho_ll}) and substituting $t'=\log\Big[1+ \frac{1}{\alpha+ n_\beta}\Big]$ yields,
\begin{align}
    M'(t')=\frac{(\alpha+ n_\beta)(1+ \alpha+ n_\beta)(\alpha^2+n_\beta+\alpha n_\beta)}{\alpha^2} \times \exp(\alpha) \ .
\end{align}
We compute this with Mathematica. Thus, Eq.(\ref{mn1r}) becomes
\begin{align}\label{m1 r}
    \sum_{l=T}^\infty \rho_{l,l}\times l \leq \frac{(\alpha+ n_\beta)(1+ \alpha+ n_\beta)(\alpha^2+n_\beta+\alpha n_\beta)}{\alpha^2} \times \exp\Big[\alpha - \log\Big(1 + \frac{1}{\alpha+n_\beta}\Big)T\Big]  \ .
\end{align}
Notice that the exponential term in this expression is the same as Eq.(\ref{m0 r}). So if we ignore the polynomial prefactors, choosing $T=\alpha^2 + R \sigma$ yields the same leading order behavior as in the `zero-th' moment for large $\alpha$. That is,
\begin{align}
    \sum_{l>\alpha^2+R \sigma}^\infty \rho_{l,l} \times l= \mathcal{O}\Big(e^{n_\beta + \frac12 - \sqrt{1+2n_\beta}R + \mathcal{O}({R}/{\alpha})}\Big) \ . 
\end{align}

\textbf{Left Tail -- } Assume $t \leq 0$. Noting $l\leq 0$,
\begin{align}\label{mn1l}
   \sum_{l=0}^T \rho_{l,l}\times l = \sum_{l=0}^T \rho_{l,l} \times \frac{l e^{tl}}{e^{tl}}\leq e^{-tT} \times \sum_{l=0}^\infty \rho_{l,l} \times l e^{t l}= e^{-t T} \times \frac{d}{dt}\Big[\sum_{l=0}^\infty \rho_{l,l} e^{tl} \Big]= e^{-t T} \times \frac{dM(t)}{dt} \ .
\end{align}
Because {$t\leq 0$}, any choice of $t$ ensures $e^t n_\beta/(1+n_\beta) < 1$. We shall pick $t'=-\log\Big[1+\frac{1}{\alpha-n_\beta-1}\Big]$. Taking the derivative of $M(t)$ seen in Eq.(\ref{M(t) for rho_ll}) and substituting $t'=-\log\Big[1+ \frac{1}{\alpha- n_\beta-1}\Big]$ yields,
\begin{align}
    M'(t')=\frac{(\alpha- n_\beta)(1+ \alpha- n_\beta)(\alpha^2+n_\beta-\alpha n_\beta)}{\alpha^2} \times \exp(-\alpha)\ .
\end{align}
We compute this with Mathematica. Thus, Eq.(\ref{mn1l}) becomes
\begin{align}
    \sum_{l=0}^T \rho_{l,l} \times l \leq \frac{(\alpha- n_\beta)(1+ \alpha- n_\beta)(\alpha^2+n_\beta-\alpha n_\beta)}{\alpha^2} \times \exp\Big[-\alpha + \log\Big(1+ \frac{1}{\alpha-n_\beta-1}\Big)T\Big] \label{m1 l}
\end{align}
Notice that the exponential term in this expression is the same as Eq.(\ref{m0 l}). So if we ignore the polynomial prefactors, choosing $T=\alpha^2 - R \sigma$ yields the same leading order behavior as in the `zero-th' moment for large $\alpha$. That is,
\begin{align}
 \sum_{l<\alpha^2+R \sigma}\rho_{l,l} \times l = \mathcal{O}\Big(e^{n_\beta + \frac12 - \sqrt{1+2n_\beta}R + \mathcal{O}({R}/{\alpha})}\Big) \ . 
\end{align}

\textbf{Summary -- } Combining the left and right tail contributions, we see that for $\alpha \gg 1$, 
\begin{align}
    \sum_{l \not\in \text{typical}} \rho_{l,l} \times l = \mathcal{O}\Big(e^{n_\beta + \frac12 - \sqrt{1+2n_\beta}R}\Big)
\end{align}
where $l\in [\alpha^2 - R \sigma, \alpha^2 + R \sigma]$ is the typical range of $l$, and we have suppressed higher order terms in $1/\alpha$ in the exponent.

\subsubsection{Second Moment}
The contribution of the left and right tails to the second moment can be computed exactly as above. We shall not repeat the entire calculation, but mention the important modification from above. Because $l\geq 0$, observe that
\begin{align}
   \sum_{l=T}^\infty \rho_{l,l}\times l^2 = \sum_{l=T}^\infty \rho_{l,l} \times \frac{l^2 e^{tl}}{e^{tl}}\leq e^{-tT} \times \sum_{l=0}^\infty \rho_{l,l} \times l^2 e^{t l}= e^{-t T} \times \frac{d^2}{dt^2}\Big[\sum_{l=0}^\infty \rho_{l,l} e^{tl} \Big]= e^{-t T} \times \frac{d^2 M(t)}{dt^2} \ 
\end{align}
when $t\geq 0$, and 
\begin{align}
\sum_{l=0}^T \rho_{l,l}\times l^2 = \sum_{l=0}^T \rho_{l,l} \times \frac{l^2 e^{tl}}{e^{tl}}\leq e^{-tT} \times \sum_{l=0}^\infty \rho_{l,l} \times l^2 e^{t l}= e^{-t T} \times \frac{d^2}{dt^2}\Big[\sum_{l=0}^\infty \rho_{l,l} e^{tl} \Big]= e^{-t T} \times \frac{d^2M(t)}{dt^2} \ 
\end{align}
when $t\leq 0$. By making the choices $t'=\log\Big[1+\frac{1}{\alpha + n_\beta}\Big]$, $T=\alpha^2 + R \sigma$  and $t'=-\log\Big[1+\frac{1}{\alpha - n_\beta - 1}\Big]$, $T=\alpha^2 - R \sigma$ for the right and left tails respectively, we again get the same exponential-decay scaling as observed earlier, that is
\begin{align}
    \sum_{l>\alpha^2+R \sigma} \rho_{l,l} \times l^2 = \mathcal{O}\Big(e^{n_\beta + \frac12 - \sqrt{1+2n_\beta}R + \mathcal{O}({R}/{\alpha})}\Big) \ , 
\end{align}
and 
\begin{align}
  \sum_{l<\alpha^2+R \sigma}\rho_{l,l} \times l^2 = \mathcal{O}\Big(e^{n_\beta + \frac12 - \sqrt{1+2n_\beta}R + \mathcal{O}({R}/{\alpha})}\Big) \ . 
\end{align}
This is because taking derivatives changes the polynomial prefactor, but does not change the exponential term in the upperbounds that determines the leading order behavior.\\

\textbf{Summary -- } Combining the left and right tail contributions, we see that for $\alpha \gg 1$, 
\begin{align}
    \sum_{l \not\in \text{typical}} \rho_{l,l} \times l^2 = \mathcal{O}\Big(e^{n_\beta + \frac12 - \sqrt{1+2n_\beta}R}\Big)
\end{align}
where $l\in [\alpha^2 - R \sigma, \alpha^2 + R \sigma]$ is the typical range of $l$, and we have suppressed higher order terms in $1/\alpha$ in the exponent.

\end{document}